\documentclass[oneside,12pt]{leeds-style/Classes/myThesis}

\usepackage[dvipsnames]{xcolor}
\usepackage[pass]{geometry}
\usepackage{amsmath}
\usepackage{amsfonts}
\usepackage{amssymb}
\usepackage{amsthm}
\usepackage{mathrsfs}
\usepackage{tikz-cd}
\usepackage{tikz}
\graphicspath{{leeds-style/ThesisFigs/PNG/}{leeds-style/ThesisFigs/PDF/}{leeds-style/ThesisFigs/}{img/}}

\ifpdf
    \pdfcatalog { /PageMode (/UseOutlines)
                  /OpenAction (fitbh)  }
\fi

\title{Geometric models of soliton vortex dynamics}
\ifpdf
  \author{\href{mailto:mmrig@leeds.ac.uk}{Ren\'e Israel Garc\'ia Lara}}
  \crest{\includegraphics[width=54mm]{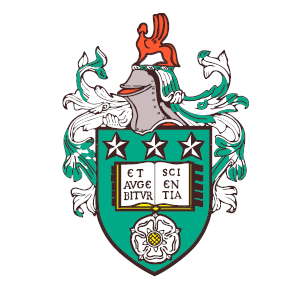}}
  \collegeordept{\href{https://eps.leeds.ac.uk/maths}{School of Mathematics}}
  \university{\href{https://www.leeds.ac.uk}{University of Leeds}}

\else
  \author{Ren\'e Israel Garc\'ia Lara}
  \crest{\includegraphics[width=54mm]{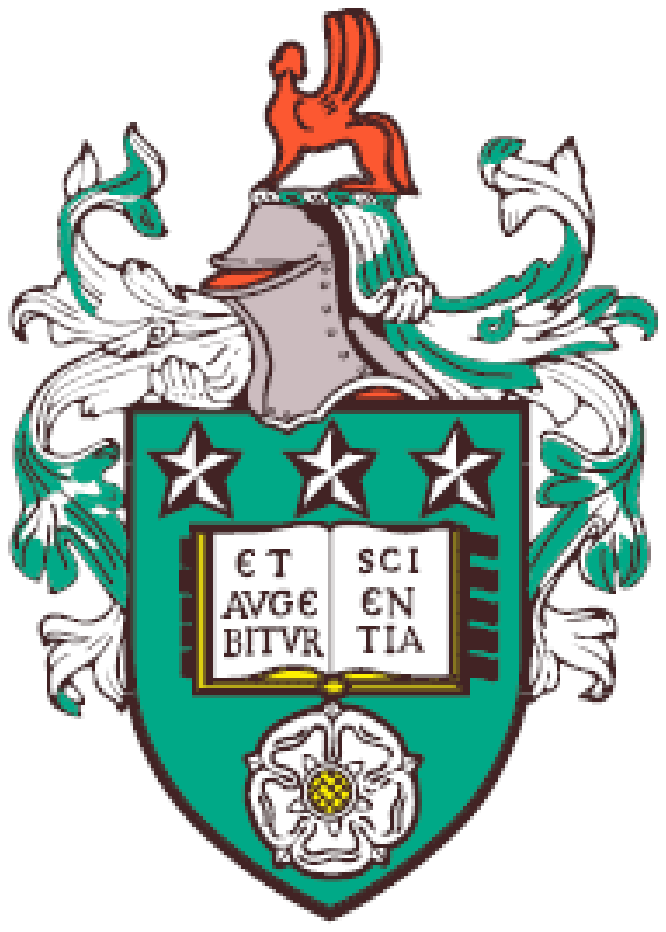}}
  \collegeordept{School of mathematics}
  \university{University of Leeds}
\fi

\degree{Doctor of Philosophy}
\degreedate{March 2021}

\usepackage{leeds-style/StyleFiles/watermark}

\def\higgsField{\phi}
\def\gaugePotential{A}
\def\curvatureField{F}
\def\neutralField{N}
\def\vortexSet{\mathcal{P}}
\def\antivortexSet{\mathcal{Q}}
\def\fieldConfSpace{\mathcal{A}}
\def\gaugeTransfSpace{\mathcal{G}}
\def\configurationSpace{\mathscr{C}}
\def\bigTantent{\mathcal{T}}
\def\vertSpace{\mathcal{V}}
\def\connectionSpace{\mathscr{A}}
\def\gaugeTransfSpace{\mathscr{G}}
\def\LagrangianDensity{\mathcal{L}}
\def\projectiveSpace{\mathbb{P}}

\usepackage{xparse}

\NewDocumentCommand\set{m}{
  \left\{#1\right\}
}

\NewDocumentCommand\st{m}{
  \;\left\lvert\;{#1}\right.
}

\newcommand*\half{\frac{1}{2}}
\newcommand*{\projSp}{\projectiveSpace}

\newcommand*\plane{{\mathbb{R}^2}}
\newcommand*\surface{\Sigma}
\newcommand*\sphere{{\mathbb{S}^2}}

\DeclareMathOperator\Aut{Aut}

\newcommand*\reals{\mathbb{R}}%
\newcommand*\cpx{\mathbb{C}}%
\newcommand*\laplacian{\Delta}
\newcommand*\grad{\nabla}
\newcommand*\del{\partial}
\newcommand*\disk{\mathbb{D}}
\newcommand*\order{\mathcal{O}}

\newcommand*\moduli{\mathcal{M}}
\newcommand*\eval[2]{\left. #1 \right\rvert_{#2}}
\newcommand*\conj[1]{\overline{#1}}

\newcommand*\Lsp{\mathrm{L}}

\DeclareMathOperator\Arg{Arg}

\NewDocumentCommand\lproduct{sm}{
  \IfBooleanTF#1
    {
        \left\langle{#2}\right\rangle
    }{
        \langle{#2}\rangle
    }
}

\newcommand*\dv[1]{\frac{d}{d #1}}

\NewDocumentCommand\pdv{mo}{
 \IfValueTF{#2}
 {\frac{\del #2}{\del #1}}
 {\frac{\del}{\del #1}}
}

\NewDocumentCommand\abs{sm}{
  \IfBooleanTF#1
  {
    \left\lvert {#2} \right\rvert}
  {
    \lvert {#2} \rvert
  }
}

\NewDocumentCommand\norm{sm}{
  \IfBooleanTF#1
  {
    \left\lvert\left\lvert {#2} \right\rvert\right\rvert
  }
  {
    \lvert\lvert {#2} \rvert\rvert
  }
}

\NewDocumentCommand\brk{r()}{
    \left( {#1} \right)
}

\DeclareMathOperator\Ker{Ker}

\newtheorem{theorem}{Theorem}
\numberwithin{theorem}{chapter}
\newtheorem*{theorem*}{Theorem}

\newtheorem{proposition}[theorem]{Proposition}
\newtheorem{lemma}[theorem]{Lemma}
\newtheorem{corollary}[theorem]{Corollary}
\newtheorem{definition}[theorem]{Definition}
\newtheorem{conjecture}[theorem]{Conjecture}
\newtheorem{theoremdef}[theorem]{Theorem/Definition}

\newenvironment{example}{\textbf{Example}. }{
  \begin{flushright}
    $\square$
    \end{flushright}
}

\begin{document}
\renewcommand\baselinestretch{1.2}
\baselineskip=18pt plus1pt

\maketitle
\setcounter{secnumdepth}{2}
\setcounter{tocdepth}{2}

\frontmatter


\begin{abstracts}        
    In this work we focus on BPS solutions 
    of the gauged $O(3)$ Sigma model, originally due to Schroers, and use these 
    ideas 
    to study the geometry of the moduli space. The model has an asymmetry 
    parameter 
    $\tau$ breaking the symmetry of vortices and antivortices on the field 
    equations.     
    It is shown that the moduli space is incomplete both on the Euclidean plane 
    and 
    on a compact surface. 
    On the Euclidean plane, the $\mathrm{L}^2$ metric 
    on the moduli space is approximated for well separated cores and results 
    consistent
    with similar approximations for the Ginzburg-Landau functional are found.
    The scattering angle of approaching vortex-antivortex pairs of different 
    effective mass is computed numerically and is shown to be different from 
    the 
    well known scattering of approaching Ginzburg-Landau vortices.  
    The volume of the moduli space for general $\tau$ is computed
    for the case of the round sphere and flat tori.
    
    The model on a compact surface is deformed introducing 
    a neutral field and a Chern-Simons term.
    A lower bound for the Chern-Simons constant $\kappa$ such that the extended 
    model admits a solution is shown to exist, and if the total number of 
    vortices and antivortices are different, the existence of an upper bound 
    is  also shown. Existence of multiple solutions to the governing elliptic 
    problem is established on a compact surface as well as the existence of two 
    limiting behaviours as $\kappa \to 0$. A localization formula for the 
    deformation is found for both Ginzburg-Landau and the $O(3)$ Sigma model 
    vortices  
    and it is shown that it can be extended to the coalescense set. This rules 
    out the possibility that this is Kim-Lee's term in the case of 
    Ginzburg-Landau vortices, moreover, the deformation term is compared 
    on the plane with the Ricci form of the surface and it is shown they are 
    different, hence also discarding that this is the term proposed by 
    Collie-Tong to model vortex dynamics with Chern-Simons interaction.
\end{abstracts}




\tableofcontents

\mainmatter

\newcommand*{\gp}{A}
\newcommand*{\hf}{\phi}

\chapter{Introduction}\label{c:intro}


This work is about the geometry of moduli spaces of vortices and antivortices 
 on a Riemann surface $\surface$. We are interested mostly in the 
 gauged $O(3)$ Sigma 
 model, 
 where the fields are represented by a connection 
 $\gp$ and a section $\hf$ of a fibre bundle 
 with fibres 
 diffeomorphic to $\mathbb{P}^1$, the Riemann sphere. 
 We say $\hf$ is a Higgs field with target the Riemann sphere. 
 Static solutions of the field equations modulo gauge equivalence 
 form the moduli space of vortices and antivortices, each solution is 
 determined by the cores of the fields: 
 the preimages of the north pole (vortex points) and the south pole (antivortex 
 points). It can be proved the total number of the cores is enumerable and if 
 $\surface$ is compact, it is finite. We will assume without 
 loss of generality this is the 
 case, even though $\surface$ can be the complex plane. The dynamics 
of slowly varying fields can be described by geodesic motion of curves 
on the moduli space~\cite{manton_remark_1982} with  
a metric called the $\Lsp^2$ metric. This metric is K\"ahler and well 
understood for the moduli space of vortices of the Ginzburg-Landau
functional, in which case it is known that the moduli space is a complete
metric space and if the ambient surface is compact the moduli space is also
compact, hence of finite volume. 

The $O(3)$ Sigma model we will study is 
asymmetric, vortices and antivortices have different
effective mass, moreover, 
the existence of two types of cores 
means vortices and 
antivortices cannot coalesce, therefore, a natural question is if
the moduli space is still complete. 
Another question we address is how the asymmetry 
affects the volume of the moduli space. These questions were addressed for the 
symmetric case in the
reference~\cite{romao2018}. The techniques used in the reference however
do not apply in general, we developed analytical tools to 
extend the results to the asymmetric case.

Later, we add a Chern-Simons deformation to the model and
describe the change in the dynamics of the fields on the moduli space. 
The deformation is tuned by means of a deformation constant 
$\kappa$ which we assume small. 
It turns out that the dynamics of the theory is described by 
geodesic motion perturbed with a connection term proportional to 
$\kappa$, i.e. a term 
dependent on the velocity of the cores. Our model resembles the model of Kim and
Lee~\cite{kim_first_2002} with the difference that the target is the
sphere and there are two types of cores to consider. 
It is well known for several related models with Chern-Simons deformations 
that multiple solutions of the field equations occur. 
We study the problem of existence and multiplicity of 
solutions to the field equations of the deformed 
$O(3)$ Sigma model, the main result is that even though 
multiple solutions of the equations can exist, 
there is a minimal deformation, such that 
no matter which configuration of vortices and antivortices on the 
moduli space we choose, we can find exactly one 
solution close to the undeformed solution of the 
$O(3)$ Sigma model. 

We conclude with a description of the chapters of the thesis.

In chapter~\ref{ch:pre} we describe the ideas of localization in
abstract terms. Our approach is general and suits equally well
Ginzburg-Landau vortices as well as the $O(3)$ Sigma model, with the benefit 
that it
makes clear what we mean by adding a Chern-Simons term. 
We also present
analytical results that are common to other parts of the next
chapters. 

In chapter~\ref{c:vav-euclidean} we focus on the $O(3)$ Sigma 
model on the
euclidean plane. We study asymmetric vortex-antivortex pairs, 
supporting our
analysis with numerical evidence of the behaviour of colliding
vortex-antivortex pairs. 
 We compute the metric on the moduli space of vortex-antivortex 
 pairs numerically 
and use this computations to study the scattering of
approaching cores. 
The main result is theorem~\ref{thm:vav-euc-incompleteness} 
which 
says that the moduli space is incomplete.

In chapter~\ref{c:vav-compact} we move to a compact ambient
surface. 
The main results are the incompleteness of the moduli space of 
vortex-antivortex pairs,  
theorem~\ref{thm:moduli-space-incomplete-compact-surface}, 
and the computation of the volume of 
 the moduli space for the round sphere and for flat tori in 
 theorem~\ref{thm:vol-torus}, confirming 
 a general conjecture by Rom\~ao-Speight~\cite{romao2018} 
in these cases.

Chapter~\ref{ch:cs-moduli} is devoted to the study of Chern-Simons
deformations on compact surfaces. 
We prove the existence of multiple solutions for small 
deformations of the $O(3)$ Sigma model if the number of vortices and 
antivortices is different and find bounds for 
the deformation constant.
We also solve the 
field equations numerically on the sphere for two configurations of 
vortices and antivortices at antipodal positions. 
The main result is theorem~\ref{thm:limit-kappa-0-hf}, 
describing the behaviour of the solutions to the field 
equations. 
We finalise the chapter
applying the localization technique to vortices of the Ginzburg-Landau 
model and vortices/antivortices of the $O(3)$ Sigma model, both with a 
Chern-Simons deformation. We found that dynamics is deviated from  
geodesic motion by a connection term consistent with 
previous results of
Kim-Lee~\cite{kim_first_2002} and Collie-Tong~\cite{collie_dynamics_2008}, 
and compared our result with theirs.

\let\gp\undefined 
\let\hf\undefined


\newcommand*{\confSp}{\mathcal{C}}
\newcommand*{\gaugeGp}{\mathcal{G}}
\newcommand*{\hf}{\higgsField}
\newcommand*{\connSp}{\connectionSpace}
\newcommand*{\fieldSp}{\fieldConfSpace}
\newcommand*{\bigTangent}{\mathcal{T}}
\newcommand*{\gp}{\gaugePotential}
\newcommand*{\vertSp}{\vertSpace}
\newcommand*{\Energy}{\mathrm{E}}
\newcommand*{\stgp}{\varphi}
\newcommand*{\Diff}{\mathrm{D}}
\newcommand*{\mf}{B}
\newcommand*{\ef}{e}
\newcommand*{\vset}{\vortexSet}
\newcommand*{\avset}{\antivortexSet}
\newcommand*{\moduliSp}{\moduli}
\newcommand*{\pair}{(\hf, \gp)}
\newcommand*{\Lagrangian}{\mathrm{L}}
\newcommand*{\dvr}[1]{\del\cdot{#1}}
\newcommand*{\vb}[1]{\mathbf{#1}}
\newcommand*{\north}{N}
\newcommand*{\south}{-N}
\newcommand*{\cDiff}{\mathcal{D}}
\newcommand*{\eform}{e}
\newcommand*{\Bform}{B}
\newcommand*{\Jop}{\operatorname{J}}
\newcommand*{\vol}{\mathrm{Vol}}
\newcommand*{\fnu}{h}
\newcommand*{\sgn}{s}
\newcommand*{\kmetric}{\mathrm{K}}

\newcommand*{\spL}{\mathrm{L}}

\chapter{Preliminaries}\label{ch:pre}

This chapter is for basic definitions and results of field theory 
that we will use in the successive. To study the geometry 
of the moduli space of vortices we need several analytical tools, this 
chapter is intended to be a bridge between field theory and analysis. 

In section~\ref{sec:field-th} we introduce the $O(3)$ Sigma model, which 
will play a central role all along the thesis.

In section~\ref{sec:loc-formulas} we discuss a localization formula 
for the $O(3)$ Sigma model, we compute a metric for the moduli space 
of vortices and antivortices, the $\spL^2$ metric, 
and 
prove that 
it is K\"ahler.


Section~\ref{c:intro-gov-elliptic} is about the analytic properties of 
the Taubes equation, this is the elliptic PDE that guarantees the 
existence of the moduli space of vortices and antivortices. Several 
theorems of analysis are introduced in this section to keep them 
collected in the same place for further reference. 
In subsection~\ref{subsec:smooth-par-dep} we prove that the 
solution to the Taubes equation depends differentiably on the 
position of the vortices and antivortices.

In section~\ref{sec:top-meth} we state less known theorems of functional 
analysis about compact non-linear operators that we will need later.

\let\spL         \undefined

\section{Field theory on complex line bundles}\label{sec:field-th}

In this section we introduce notation and a few facts about
$\projSp^1$ fibre bundles that will be required for most of the
work. 

Let us start considering a principal $U(1)$ bundle $U(1) \to P \to
\reals\times\surface$, where $\surface$ is a Riemann surface. No
further assumption on $\surface$ is needed. Let $M$ be an
$n$-dimensional manifold, such that there exists a homomorphism
\begin{equation}
\label{eq:rep-u1-on-M}
\rho: U(1) \to \Aut(M),
\end{equation}
from the structure group to the group of automorphisms of $M$. The
word automorphism means that if $M$ has an extra structure, for
example, if is a symplectic or K{\"a}hler manifold, then $\rho$ should
preserve this structure. Let $F$ be the fibre bundle associated to $\rho$,
\begin{equation}
F = (\reals\times\surface) \times_{\rho} M.
\end{equation}

Recall a connection form on
$P$ is a $\mathfrak{u}(1)$ valued form $\omega$ on $P$, such that the
kernel $\Ker(\omega)$ defines the horizontal sub-bundle of $TP$. Since
$U(1)$ is one dimensional, we can identify $\omega$ with a
regular form. For any local section $s_a:U_a \subset
\reals\times\surface \to P$, the connection is given by a local form
$A_a = s_a^{*}(\omega)$ such that in any overlap $U_a\cap U_b \neq
\emptyset$ there is a transition function $\theta_{ab}: U_a\cap U_b \to \reals$ satisfying the condition,  
\begin{equation}
A_b = A_a + d\theta_{ab}.
\end{equation}

$U(1)$ is an abelian group, hence the adjoint representation of the
structure group is trivial,  the group of gauge transformations in
this case is
\begin{align}
 \gaugeGp = C^{\infty}(\reals\times\surface, U(1)). 
\end{align}

The space of connections $\connSp$ is an affine space: for any two
connection forms $\omega, \omega' \in \connSp$,
the difference $\omega - \omega'$ determines a unique 1-form $A \in
\Omega^1(\reals\times\surface)$ such that if $s_a:U_a\to P$ is 
a local trivialisation, then $s_a^*(\omega - \omega')$ is the restriction of
$A$ to $U_a$. Therefore $\connSp$ is in bijection with 
$\Omega^1(\reals\times\surface)$, the space of 1-forms on
$\reals\times\surface$. Let $\fieldSp = \Gamma
F \times \connSp$ be the space of pairs of fields $(\hf, A)$,
consisting of a section, $\hf: \reals\times\surface \to F$, 
%
%
and a connection form $A \in \Omega^1(\reals\times\surface)$.

\par

The quotient
$\fieldSp/\gaugeGp$ is the configuration space $\confSp$. If $M =
\sphere$, then $\rho$ has two antipodal fixed points, the north and
south poles. We choose one that we will denote as $N$ and call it
the north pole. In this setting $F$ is a $\projSp^1$ bundle, the
fibres are modelled on the complex projective line. The fact that $\rho$
represents the unitary group by rotations of the sphere lets us pull
the north pole back into a section $\north: \reals\times\surface \to
F$. The south pole can also be pulled back into another section, that
we denote by $\south$, however we must emphasise that $F$ lacks any
algebraic structure conferring other meaning to the name than a mere
notation. We also denote by $X \in \mathfrak{X}(\sphere)$ the Killing
field generated by $\rho$,
\begin{align}
X_p = \eval{\dv{s}}{s = 0} \brk(\rho\brk(e^{is})\cdot p), \qquad p \in \sphere.
\end{align}

A section $\hf: \reals\times \surface \to F$ is determined completely
by the family of maps $\hf_\alpha: U_\alpha \to \sphere$ defined for
each trivialising neighbourhood $U_\alpha \subset
\reals\times\surface$, if $U_{\alpha\beta} = U_\alpha \cap
U_\beta \neq \emptyset$, we have
\begin{align}
  \hf_\beta(x) = \rho(\exp(i\theta_{\alpha\beta}(x)))\cdot
\hf_\alpha(x), \qquad x \in U_{\alpha\beta} = U_\alpha \cap
U_\beta.  
\end{align}

Since $\rho$ acts by isometries, we can define the product
$\lproduct{\north,\hf}$ using the trivialisations: for $x \in
U_{\alpha}$,
\begin{align}
 \lproduct{\north(x),\hf(x)} = \lproduct{N, \hf_\alpha(x)}.  
\end{align}

We also define the covariant derivative of $\hf$ as the section
\begin{align}
 \cDiff \hf: \reals\times\surface \to
T^{*}(\reals\times\surface)\otimes \hf^*(TF) 
\end{align}
determined by the trivialisations  
$\cDiff\hf_{\alpha}: U_{\alpha} \to T^*U_\alpha\otimes T\sphere$ as,
\begin{align}\label{eq:cov-der-line-bundle}
\cDiff \hf_{\alpha} = \vb d \hf_\alpha - \vb\gp_{\alpha}\otimes
X_{\hf_{\alpha}}, 
\end{align}
where $\vb d \hf_\alpha: TU_\alpha \to TF$ can be split into
its temporal and spatial components,
\begin{align}
\vb d \phi_\alpha = dt\otimes\del_t\hf_\alpha + d\hf_\alpha, \qquad
  d\hf_\alpha(t,\cdot)\in T^{*}\surface\otimes (\hf_\alpha(t,\cdot))^*(TF).
\end{align}

Likewise, $\vb\gp_{\alpha} = \gp^0_{\alpha}\,dt + \gp_{\alpha}$, where $\gp^0_{\alpha} \in
C^{\infty}(U_{\alpha})$ and $\gp_{\alpha}(t,\cdot) \in
\Omega^1(U_{\alpha})$. If we define
\begin{align}
\Diff_t\hf_\alpha &= \del_t\hf_\alpha - \gp^0_{\alpha} \otimes X_{\hf_{\alpha}}, &
\Diff \hf_\alpha  &= d\hf_\alpha - \gp_{\alpha}\otimes X_{\hf},
\end{align}
then,
\begin{align}
\cDiff \hf_\alpha = dt \otimes \Diff_t\hf_\alpha + \Diff \hf_\alpha.
\end{align}

We introduce a Lorentzian metric as follows. If $g$ denotes a
Riemannian metric in $\surface$ then the 
metric in $\reals \times \surface$ is the product $dt^2 - g$. This
metric induces a metric in
$\Omega^2(\reals\times\surface)$. Recall the curvature
form $\omega \in \Omega^2(\reals\times\surface)$ is given in a local
trivialisation by $\omega = d\vb\gp_{\alpha}$ and define the electric and
magnetic forms, as the forms $\eform\in\Omega^1(\reals\times\surface)$ and
$\Bform\in\Omega^2(\reals\times\surface)$ respectively, such that,
\begin{align}
\omega = dt\wedge \eform + \Bform,
\end{align}
and for fixed $t$, $\eform(t,\cdot)\in\Omega^1(\surface)$,
$\Bform(t,\cdot)\in\Omega^2(\surface)$.

Although $\abs {d\hf_\alpha}$ is gauge
dependent, at the intersection $U_{\alpha\beta}$ of any two
trivialisation neighbourhoods, $\abs{\cDiff\hf_\alpha} =
\abs{\cDiff\hf_\beta}$, hence we can define 
\begin{align}
\norm{\cDiff\hf(t, \cdot)}^2 = \norm{\Diff_t\hf}^2 - \norm{\Diff\hf}^2.
\end{align}

With all these definitions, we can express the gauged $O(3)$
Lagrangian as,
\begin{align}\label{eq:o3-lag}
  \Lagrangian_{O(3)} = \half \brk( \norm{\Diff_t\hf}^2 +
  \norm{\eform}^2 - (\norm{\Diff\hf}^2 + \norm{\Bform}^2 + \norm{ \tau
  - \lproduct{\north, \hf}}^2)),
\end{align}
where the asymmetry parameter $\tau\in (-1, 1)$ determines the vaccuum
manifold and if $\surface$ is non-compact, we must add suitable
boundary conditions to $\hf$ and $\gp$ to guarantee convergence of the
norms. The $O(3)$ Lagrangian admits Bogomolny type static
solutions in the temporal gauge, in which $\gp_\alpha^0 = 0$. In this
gauge, the total conserved energy of a time independent pair of fields
$(\hf, \gp)$ is
\begin{align}
  \Energy = \half\brk(\norm{\Diff\hf}^2 + \norm{\Bform}^2 + \norm{ \tau
  - \lproduct{\north, \hf}}^2).
\end{align}

The temporal covariant derivative $\Diff_{\gp}$ can be decompose into
holomorphic and anti-holomorphic parts, 
\begin{align}
\Diff_{\gp} = \del_{\gp} + \conj\del_{\gp}, 
\end{align}
where in a local holomorphic coordinate chart $U_{\alpha}$ in which
$\hf$ trivialises as $\hf_{\alpha}:U_{\alpha} \to \sphere$,
\begin{align}
\del_{\gp}\hf_{\alpha} &= \half\brk(\Diff_{\gp}\hf_{\alpha}(\del_1) -
\hf_{\alpha}\times \Diff_{\gp}\hf_{\alpha}(\del_2)), &
\conj\del_{\gp}\hf_{\alpha} &= \half\brk(\Diff_{\gp}\hf_{\alpha}(\del_1) +
\hf_{\alpha}\times \Diff_{\gp}\hf_{\alpha}(\del_2)),
\end{align}
%

We will consider the sets
\begin{align}
\vset &= \hf^{-1}(\north), &
\avset &= \hf^{-1}(-\north),
\end{align}
which we call the set of vortices and antivortices. The term vortex is of 
wide use for the Abelian Higgs model, where it refers to the zeros of the 
Higgs field. Both theories, the Abelian Higgs model and the $O(3)$ Sigma model, 
 have similarities, for example the $U(1)$ symmetry of the fields, hence it is 
 natural to refer to vortices of the
$O(3)$ Sigma model, on the other hand, the term antivortex, which is also 
used in the literature, stresses the distinction with vortices, since 
 vortices and antivortices cannot coalesce.  
 We assume
that both sets are finite. In proposition~\ref{prop:energy-bog-sols} 
we define the Bogomolny equations. 

\begin{proposition}\label{prop:energy-bog-sols}
  If $(\hf, \gp)$ is a solution of the Bogomolny equations,
  \begin{align}
    \label{eq:bog1}
    \conj\del_{\gp}\hf &= 0,\\
    \label{eq:bog2}
    *\Bform &= \lproduct{\north, \hf} - \tau,
  \end{align}
  then the pair minimises the energy of the $O(3)$ Lagrangian and the
  minimum energy is,
  \begin{align}
   \Energy = 2\pi (1 - \tau)\,k_+ + 2\pi (1 + \tau)\, k_-.
  \end{align}
\end{proposition}

Proposition~\ref{prop:energy-bog-sols} should be attributed to several
authors who proved it for the different cases. On the plane it was proved by
Schroers~\cite{schroers_bogomolnyi_1995} for $\tau = 1$ and later for 
general $\tau$ in \cite{schroers_spectrum_1996}. On a compact
manifold for $\tau = 0$ it was proved by Sibner, Sibner and
Yang~\cite{sibner2010}. Speight and R\~omao~\cite{romao2018} give another proof 
which is suitable for both a compact surface and the euclidean plane,  which we 
adapt. 

\begin{proof}
  We distinguish two cases. Firstly, let us assume that $\surface$ is
  compact. We can choose an open and dense set $U \subset \surface$
  holomorphic to the unit disc such that it  contains $\vset \cup
  \avset$. Since $U$ is contractible, the restriction $F\mid_U$ can be
  trivialised. In this trivialisation, $\hf$ is equivalent to a
  function $\varphi: U \to \sphere$. Since the action of $U(1)$ in the
  sphere is Hamiltonian, we can consider the moment map $\mu: \sphere
  \to \reals$,
  \begin{align}
    \mu(p) = \lproduct{N, p} - \tau.
\end{align}

If $\omega$ denotes the symplectic form in the sphere, then
$d\mu = \iota_X\omega$. Let us denote by
$\Jop: T\sphere \to T\sphere$, $\Jop_x(v) = x\times v$ the almost
complex structure on the sphere. Recall the basic identity,
\begin{align}
\lproduct{\Jop v, w} = \omega(v, w), \qquad v,w \in T_x\sphere.
\end{align}
 We will use the Bogomolny trick,
  \begin{align}
0 &\leq \half\brk(\norm{\Diff_1\varphi + \Jop \Diff_2\varphi}^2 +
\norm{*\Bform - \mu\circ\varphi}^2) \nonumber \\
&= \Energy + \lproduct{\Diff_1\varphi, \Jop \Diff_2\varphi} -
\lproduct{*\Bform, \mu\circ\varphi} \nonumber\\
&= \Energy + \lproduct{\del_1\varphi, \Jop \del_2\varphi}
+ \int_U \omega(X_{\varphi}, \gp_1\del_2\varphi - \gp_2\del_1\varphi)
\,\vol -
\lproduct{*\Bform, \mu\circ\varphi} \nonumber\\
&= \Energy + \int_U \brk(-\varphi^{*}\omega + \gp\wedge
d(\mu\circ\varphi) - \Bform \wedge \mu\circ\varphi ) \nonumber\\
&= \Energy - \int_U \brk( \varphi^{*}\omega + d\brk(\mu\circ\varphi\cdot\gp)).
\end{align}

Note that $\varphi^{*}\omega + d(\mu\circ\varphi\cdot\gp)$ is gauge
invariant and can be extended to all of $\surface$. Introducing
spherical coordinates $(\vartheta, \varrho)$ in $\sphere$ with $\varrho$
the azimuthal angle, we define the one form,
\begin{align}
\varpi = \varphi^{*}(d\varrho) - \gp \in \Omega^1(U\setminus
\vset\cup\avset),
\end{align}
and note that $\varpi$ is gauge invariant and therefore also extends
to $\surface\setminus \vset\cup\avset$. If we denote by $\disk_{\epsilon}$
a collection of disjoint $\epsilon$-disks, each one centred at one point
$x\in\vset\cup\avset$, then,
\begin{align}
\int_U \brk( \varphi^{*}\omega + d\brk(\mu\circ\varphi\cdot\gp)) &=
-\int_{\surface\setminus 
\vset\cup\avset}d(\lproduct{\north,\hf}\varpi)  - \tau
\int_{\surface}\Bform
\label{eq:int-pb-symp-form-sphere}\\
&= \lim_{\epsilon\to 0} \int_{\del \disk_{\epsilon}}
\lproduct{\north,\hf}\varpi - \tau\int_{\surface}\Bform
\nonumber\\
&= 2\pi (k_+ + k_-) - 2\pi \tau (k_+ - k_-)
\nonumber\\
&= 2\pi (1 - \tau)\,k_+ + 2\pi (1 + \tau)\,k_-.
\nonumber
\end{align}

Hence,
\begin{align}
\Energy \geq 2\pi (1 - \tau)\,k_+ + 2\pi (1 + \tau)\,k_-,
\end{align}
and the energy is minimised if $(\hf, \gp)$ is a solution to the 
Bogomolny equations. If $\surface$ instead is the Euclidean plane,
we have to assume that $\Diff\hf$, $\Bform$ and $\mu\circ\hf$ are
$\Lsp^2$ sections of their respective bundles. In this case
we can take $U = \plane$, and most of the proof follows verbatim the
previous steps, except that to compute the integral
\eqref{eq:int-pb-symp-form-sphere} we must suppose that the fields
satisfy the boundary condition,
\begin{align}
\lim_{\abs{x}\to\infty}(\lproduct{\north,\hf} - \tau) &= 0.
\end{align}

\end{proof}

We started assuming the sets $\vset$ and $\avset$ where finite and found
that a pair $(\hf, \gp)$ of solutions to the Bogomolny equations
minimises the static energy. In the compact case, the assumption about the
size of the sets is redundant, the proof for $\tau = 0$ found
in~\cite{sibner2010} can be adapted to the asymmetric case.

\begin{proposition}\label{prop:degree}
  If $(\hf, \gp)$ is a solution to the Bogomolny equations, then $\vset$
  and $\avset$ are discrete. In particular, if $\surface$ is compact,
  these are finite sets. Moreover, if $x \in \vset\cup\avset$, then
  $\hf(x)$ is of finite degree, in the sense that there is a unique positive
  integer $d$ such that if $x \in \vset\cup\avset$ and $\varphi: U \to
  \cpx$, $\pi:V\subset \sphere \to \cpx$, are holomorphic coordinates
  about $x$ and $\hf(x)$ with $\varphi(x) = \pi(\hf(x)) = 0$,  then
  there is a 
  smooth function $R: \varphi(U) \to \pi(V)$ such that,
  \begin{align}
    \pi\circ\hf\circ \varphi^{-1}(z) = z^dR(z), \qquad \forall z \in
    \varphi(U), 
\end{align}
but $R(0) \neq 0$.
\end{proposition}

\begin{proof}
 Suppose $x \in \vset$ and $\hf_{\alpha}: U_{\alpha} \to \sphere$
 is a local trivialisation in an holomorphic chart $\varphi_{\alpha}:
 U_{\alpha} \to \cpx$ with $\varphi_{\alpha}(x) = 0$. Let $\pi_-:
 \sphere\setminus\set{-N} \to \cpx$ be south pole stereographic 
 projection and let
 \begin{align}
  \psi_{\alpha} = \pi_-\circ \hf_{\alpha}\circ\varphi_{\alpha}^{-1}:
  \varphi_{\alpha}(U_{\alpha}\setminus\avset) \to \cpx.
\end{align}
Since $\pi_-$ is a 
holomorphic local diffeomorphism, the first Bogomolny equation is
equivalent in these charts to,
\begin{align}\label{eq:bog1-equiv}
\conj\del \psi_{\alpha} = \half(-\gp_2 + \gp_1 i)\psi_{\alpha}.
\end{align}

If $\gp$ is smooth, by the $\conj\del$-Poincare lemma, there exists a
smooth function $w: \varphi_{\alpha}(U_{\alpha}\setminus\avset)\to
\cpx$ such that $\conj\del w = \half(-\gp_2 + \gp_1 i)$, hence the function
$e^w\psi_{\alpha}$ is holomorphic, $\conj\del (e^w\psi_{\alpha}) = 0$ 
and the zero set of $\psi_{\alpha}$ is discrete unless $\psi_{\alpha}
\equiv 0$ which is impossible because it violates the Bogomolny 
equations. This proves that $\vset$ is a discrete set. Since
$e^w\psi_{\alpha}$ is holomorphic, the assertion about the degree
follows in these charts and since the degree is an holomorphic
invariant, this proves the claim for any other holomorphic
chart. Using the north 
pole stereographic projection proves similar claims for $\avset$.
\end{proof}

We say that $x \in \vset$ is the position of a single vortex if the
degree is 1 and 
similarly for $x \in \avset$, if the degree is 1 we say that $x$ is
the position of a single antivortex. We will denote the size of the
sets  $\vset$, $\avset$ as $k_{\pm}$ respectively, where we count each
vortex and antivortex with multiplicity.

For any solution $(\hf,  \gp)$ to the 
Bogomolny equations, we define the function $\fnu: \surface\setminus
\vset\cup\avset \to \reals$, 
\begin{align}
  \fnu = \log\brk(\frac{1 - \lproduct{\north, \hf}}{1 +
  \lproduct{\north, \hf}}).
\end{align}

If we define the map $\psi_{\alpha}: \pi_-\circ\hf_{\alpha}:
U_{\alpha} \to \cpx$ as in the proof of proposition~\ref{prop:degree},
where $\hf_{\alpha}: U_{\alpha}\to \sphere$ represents $\hf$ in a local
trivialisation $U_{\alpha} \subset \surface \setminus \vset\cup\avset$, and
$\pi_-:\sphere\setminus \set{-N} \to \cpx$ is south pole's stereographic
projection, then $\exp(\fnu) = \abs{\psi_{\alpha}}^2$  and 
$\log \psi_{\alpha} = \frac{\fnu}{2} + \chi_{\alpha} i$, where the
argument function $\chi_{\alpha}:U_{\alpha} \to \reals$ is gauge
dependent. By equation~\eqref{eq:bog1-equiv},
\begin{align}
- \frac{1}{4}\laplacian (\log \psi_{\alpha}) = \half\del (-\gp_2 +
\gp_1 i) = \frac{1}{4} \brk( -(\del_1\gp_2 - \del_2\gp_1) + (\del_1\gp_1 +
\del_2\gp_2) i ),
\end{align}
where $\laplacian$ is geometer's laplacian, which in
the holomorphic coordinates we are considering is of the form 
$\laplacian = -e^{-\Lambda}(\del_1^2 + \del_2^2)$, where
$e^{\Lambda}$ is the conformal factor of the metric. Taking the real
part of the previous equation, we find,
\begin{align}
-\laplacian h = -2 *\Bform = 2\brk( \frac{e^h - 1}{e^h + 1}
+ \tau).
\end{align}

If $x \in \vset \cup \avset$ has degree $d_x$, we can extend the
definition of $h$ to the core set $\vset\cup\avset$ by requiring it to
be a solution to~\cite{schroers_spectrum_1996}, 
\begin{align}
\label{eq:taubes}
-\laplacian \fnu = 2 \brk(\frac{e^{\fnu} - 1}{e^{\fnu} + 1}
+ \tau) + 4\pi\sum_{x \in \vset} d_x\delta_x -  4\pi\sum_{x\in\avset}
d_x\delta_x,
 \end{align}
 where $\delta_{x}$ is Dirac's measure concentrated at $x$. 
  For any
 test function~$\varphi \in C^{\infty}_0(\surface)$,
 \begin{align}
   \int_{\surface}\varphi\,\delta_x = \varphi(x). 
 \end{align}
 Notice $\delta_{x}$ includes the  measure on $\surface$.  We will
 call equation~\eqref{eq:taubes}  
 \emph{the Taubes equation}, as is analogous to the equation studied by
 Taubes for the Ginzburg-Landau functional \cite{taubes1980}. The Taubes 
 equation as given by~\eqref{eq:taubes} was also obtained by Schroers 
 in~\cite{schroers_spectrum_1996}.

\section{Localization }
\label{sec:loc-formulas}

The idea of a localization formula originates in the work of
Strachan~\cite{strachan_lowvelocity_1992}.  It was later generalised by
Samols~\cite{samols1992} and is based on ideas about geodesic
approximation originating in~\cite{manton_remark_1982}. 
From his
work, Strachan and Samols developed approximations to the 
dynamics of the Abelian Higgs model in the moduli space of static 
solutions of the field equations, later, Stuart proved 
in~\cite{stuart_dynamics_1994} that the 
moduli space approximation is correct. 
The results of Stuart also extended to 
other field theories, for example in~\cite{demoulini_adiabatic_2009} 
Demoulini-Stuart proved a moduli space approximation to the dynamics of 
the Chern-Simons-Schr\"odinger model 
proposed by Manton in~\cite{manton_first_1997}. On the other hand, 
for some field theories it is possible to find an explicit formula for a 
metric on the moduli space governing the dynamics, such that it 
only depends of local data, i.e., the position of the cores of the field 
$\hf$. Over the time, the localization formula has been refined and
extended to other field 
theories, e.g. Chern-Simons vortices 
\cite{kim_vortex_1994,collie_dynamics_2008} 
or Ginzburg-Landau vortices with electric
and magnetic impurities \cite{tong_vortices_2014}.
We can describe in an unified way the idea behind localization if we
introduce the $\Lsp^2$ metric in the space of fields modulo gauge
transformations. By this we mean the space of sections of a given
$U(1)$ fibre bundle as described on section~\ref{sec:field-th}.
There are several situations in which this space is
finite dimensional, for example for BPS solitons of the Ginzburg-Landau
functional. In this case, there are rigorous proofs 
of this fact \cite{taubes1980,yang_strings_1999}. We make no
assumption on finite dimensionality though, since the theory can be
written in full generality. We restrict the previous field theoretic
setup to the static case and 
think of $\fieldConfSpace \to \confSp$ as an infinite dimensional
principal $\gaugeTransfSpace$-bundle \cite{nagy2017}. A curve
$(\hf_s, A_s) :I \to \fieldSp$ is said to be differentiable, if 
for any $x \in \surface$, the curves $s\mapsto \hf_s(x)$, $s\mapsto A_s(x)$, 
where, 
\begin{align}
\hf_s(x) : I \to F, &&
A_s(x) : I \to T_x\surface,
\end{align}
are differentiable. For a differentiable curve in field space, the variation is
the pair $(\delta\hf,\delta A)$, 
\begin{align}
\delta\hf: \surface \to \hf^{*}TF, && \delta A \in \Omega^1(\surface),
\end{align}
of pointwise derivatives:
\begin{align}
\delta\hf(x) = \eval{\dv{s}}{s = 0} \hf_s(x), &&
\delta A(x) = \eval{\dv{s}}{s=0} A_s(x).
\end{align}

We will think of the space of variations as the tangent space of
$\fieldSp$ and denote it as $\bigTangent$. If $\alpha \in
\gaugeGp$ is a gauge transformation, the fields transform as, 
\begin{align}
\label{eq:gauge-transform-fields}
e^{i\alpha}\cdot\hf, &&  \gp + d\alpha,
\end{align}
where the product $e^{i\alpha}\cdot\hf$ is to be understood as the
action of $e^{i\alpha}$ in $\hf$ via the representation
$\rho$. Equation \eqref{eq:gauge-transform-fields} defines an action 
$\alpha * (\hf, \gp)$, of the gauge group in the space of
fields. This action extends naturally to tangent space. By an abuse in
notation, let us denote by $X$  
the vector field induced in target space by this action, then
$\gaugeGp$ acts in $\bigTangent$ as,
\begin{align}
\label{eq:u1-action-bigtangent}
\alpha * (\delta\hf, \delta\gp) = (\delta\hf + \alpha X_{\hf},
\delta\gp + d\alpha). 
\end{align}

Moreover, the vertical space,
\begin{equation}
\label{eq:vert-space}
\vertSp_{(\hf, \gp)} = \set{(\alpha X_{\hf}, d\alpha) \st \alpha \in
  \gaugeGp},
\end{equation}
determines a sub-bundle of $\bigTangent$ whose fibre is in bijection with the
Lie algebra~$\gaugeGp$ of gauge transformations~\cite{nagy2017}.

The Riemannian metrics in $\surface$ and $M$  extend to metrics in
the cotangent bundle $T^{*}\surface$ and $F$ respectively, which on
the other hand, extend to a metric in the space of fields: if $(\hf,
\gp) \in \fieldSp$ and $(\delta\hf, \delta\gp) \in \bigTangent_{(\hf,
  \gp)}$, the $\Lsp^2$-metric is the product of metrics
induced by the Riemannian structure in the domain and the target space, 
\begin{equation}
\label{eq:l2-metric-tuples}
\norm{(\delta\hf,\delta\gp)}_{\fieldSp}^2 =
\norm{\delta\hf}^2_{\Lsp^2(\surface,F)} +
\norm{\delta\gp}^2_{\Lsp^2(\surface)}. 
\end{equation}
$\confSp$ is the relevant space for applications, as two field
configurations differing by a gauge transformation are regarded as
physically the same. In analogy to a finite dimensional vector bundle,
the $\Lsp^2$-metric 
 can be used to split
$\bigTangent$ in a direct sum of the vertical space $\vertSp$ and its
orthogonal complement. If the quotient $\confSp$ has a
finite dimensional differentiable structure, this complement can be
identified with its tangent space. This is not necessarily the
case, however we can consider that the orthogonal complement
describes tangent vectors to $\confSp$, whether this space is finite
dimensional or not. Hence, the orthogonal complement describes the
dynamics of curves $[(\hf_s, \gp_s)] : I \to \confSp$ with a lift to
field space, even if the quotient lacks regularity. 

Given $(\delta\hf, \delta{\gp}) \in \bigTangent_{(\hf,\gp)}$, let
$\beta \in \gaugeGp$ be the projection onto
$\vertSp_{(\hf,\gp)}$ with respect to the $\Lsp^2$ product. If
$\alpha\in\gaugeGp$ represents another  
arbitrary vertical vector at $(\hf, \gp)$, then 
\begin{equation}
\label{eq:loc-perp-cond}
 \lproduct{(\delta\hf - \beta X_{\hf}, \delta \gp - d\beta), (\alpha
  X_{\hf}, d\alpha)} = 0.
\end{equation}

Since $\alpha$ is arbitrary, the perpendicularity condition is
equivalent to the equation, 
\begin{equation}
  \label{eq:perp-condition}
  \left(\laplacian + \abs{X_{\hf}}^2\right)\beta =
   \lproduct{X_{\hf}, \delta\hf} + d^{*}\delta\gp, 
\end{equation}
where $d^{*}: \Omega^1(\surface) \to \Omega^0(\surface)$ is the
codifferential, $d^{*} = -*d*$. What is interesting about 
the perpendicularity condition is that it is
independent of the theory because no Lagrangian or functional for
the fields was necessary to deduce it. At the same time, we can
talk of kinetic energy in configuration space, at least
for curves $[\hf_s,\gp_s]$ admitting a lift to $\fieldSp$. 
For such a curve, we could define its instant energy as,
\begin{equation}
  \label{eq:energy-conf-class}
  \Energy[\delta\hf, \delta{\gp}] =
  \frac{1}{2}\norm{(\delta\hf^{\perp}, \delta{\gp}^{\perp})}_{\fieldSp}^2.
\end{equation}

We think of the kinetic energy of a dynamic pair $(\hf, \gp)$ of
solutions to the field equations slowly varying in time, as
approximated by the energy of the 
variation of a static pair of solutions. 
In this way, we reduce the full theory in spacetime to variations of
the static solutions   to the Bogomolny equations. With this point of
view, the components  of the gauge potential are curves defined in some
interval $I \subset \reals$,
\begin{align}
\label{eq:loc-stgp-split}
\gp_0: I \subset \reals \to
C^{\infty}(\surface), &&
\gp: I \subset \reals \to \Omega^1(\surface).
\end{align}

Likewise for the electric and magnetic fields. Let us
denote by $\fieldSp'$ the subset of $\fieldSp$ of solutions of the 
Bogomolny equations, and by 
$\moduliSp$ the quotient space $\fieldSp'/\gaugeGp$. There is 
 a bundle inclusion,
\begin{equation}
  \label{eq:bundle-comm-diag}
\begin{tikzcd}
  \fieldSp' \arrow[d] \arrow[r, hook] & \fieldSp \arrow[d] \\
  \moduliSp             \arrow[r,hook]  & \confSp
\end{tikzcd}  
\end{equation}
and since the Bogomolny equations are gauge invariant, both bundles share
the same vertical space. Therefore, the orthogonal projection onto
$\vertSp$ is the same.

\begin{example}(\textit{Localization of Ginzburg-Landau vortices}).
  A an example we consider the Ginzburg-Landau
  functional. In this case the target space is $\cpx$ and we can think of
  sections $\hf$ as complex valued functions $\reals \times \surface
  \to \cpx$. As described above, static fields are the same as pairs
  $(\hf, \gp)$ of a function $\hf: \surface \to \cpx$ and a
  connection $\gp$ on a principal bundle $U(1)\to P \to \surface$. As
  it turns out~\cite{manton_topological_2004}, 
  static configurations in the radiation gauge minimise the energy
  \begin{align}
    \Energy = \half\brk(\norm{\Diff\hf}^2 + \norm{\Bform}^2 +
    \frac{1}{4}\,\norm{ 1 - \abs{\hf}^2}^2)
  \end{align}
  and satisfy the following Bogomolny equations,
  \begin{align}
    \conj\del_{\gp}\hf &= 0,\\
    *\Bform &= \half(1 - \abs \hf^2).
  \end{align}

  The action of $U(1)$ on the target manifold gives rise to the vector
  field $X_{\hf} = i\hf$. As is well known from the work of Taubes,
  solutions to the field equations modulo gauge equivalence are
  determined by the zeros of $\hf$, $(p_1,
  \ldots,p_n)$. If we let the zeros vary with respect to a parameter,
  $p_k(s)$, $k = 1, \ldots,n$, 
  identified as the \emph{time} parameter, then the
  perpendicularity condition is equivalent to,
\begin{equation}
(\laplacian + \abs{\hf}^2)\beta = -\frac{i}{2} \left( \hf
  \dot\hf^{\dagger} - \dot\hf \hf^{\dagger} \right) + d^{*}\dot\gp,
\end{equation}
and the projection of the variation on the horizontal subspace
of $\bigTangent$ is,
\begin{align}
\dot\hf^{\perp} &= \dot{\hf} - i\hf\beta, &
\dot\gp^{\perp} &= \dot\gp - d\beta.
\end{align}

Since the variation is determined by variations of the zeros, if each
$p_k$ is in the same open and dense holomorphic neighbourhood $U$,
\begin{align}
\dot\hf &= \dot p_k\, \pdv{p_k}[\hf], &
\dot\gp &= \dot p_k\, \pdv{p_k}[\gp] .
\end{align}

In the sequel we make the convention that
repeated indices represent sums. If $\beta_k$ is the projection onto
vertical space 
corresponding to the variation $(\del_{p_k}\hf, \del_{p_k}\gp)$, then
$\beta = \dot 
p_k\,\beta_k$. If we denote the pair
$(\hf, \gp)$ by $\Phi$,
the instant energy of a trajectory in the moduli space is
therefore, 
\begin{align}
\nonumber
\Energy[\dot\Phi] &= \half \norm{\dot\Phi^{\perp}}^2_{\fieldSp}  \\
  \nonumber
  &= \half \dot p_k\,\dot p_r \lproduct*{
  (\del_{p_k}\Phi)^{\perp}, (\del_{p_r}\Phi)^{\perp}}_{\fieldSp}\\
  &= \half\, \dot p_k \dot p_r\, g_{p_kp_r}.
\label{eq:gl-kin-energy-moduli}
\end{align}

The coefficients $g_{p_kp_r}$ determine a metric in the moduli
space. Manton proposed an interpretation of this metric
in~\cite{manton_remark_1982}. In our language, the static energy in
the 
sub-bundle $\fieldSp'$ must be preserved by solutions of the 
Bogomonly equations, because they are energy minimisers. Thence,
$\Energy[\dot\Phi]$ approximates the energy of slow moving solutions
of the full field equations. Equation
\eqref{eq:gl-kin-energy-moduli} opens the possibility to study the
dynamics of the full field equations as geodesic motion in a finite
dimensional manifold. It was Samols who proved that
this metric depends only in the first derivatives of $\hf$ at the
zeros \cite{samols1992} of the Higgs field $\hf$, obtaining the formula
bearing his name on $\plane$,
\begin{align}\label{eq:loc-formula-samols}
  ds^2 = \pi\,\sum_{rs}\left(
    \delta_{rs} + 2\,\del_rb_s
  \right)\,dp_r\,\conj{dp_s},
\end{align}
where the coefficients depend on the position of the zeros of $\hf$, 
in fact, if $h = \log |\hf|^2$, then,
\begin{align}
b_s = 2\partial_{z}\vert_{z=p_s}(h - \log |z - p_s|^2),
\end{align} 
which explains why \eqref{eq:loc-formula-samols} is called a localization 
formula, in the sense that the data needed to compute the metric is only 
local to  the position of the zeros of $\hf$.
\end{example}

\subsection{Localization of BPS solitons of the gauged O(3) Sigma
  model}\label{subsec:loc-bps-solitons}
 
Having discussed localization of Ginzburg-Landau vortices as example,
we turn attention to the gauged $O(3)$ Sigma model and apply the same
technique in detail. Let $\varphi_{\alpha}: U_{\alpha}
\to \cpx$ be a holomorphic chart admitting a trivialisation on
$U_{\alpha}$ such that $\hf$ is equivalent to a  function $\hf_{\alpha}:
U_{\alpha} \to \sphere$. As before, let us define the stereographic
projection of $\hf_{\alpha}$ as $\psi_{\alpha} = \pi_{-} \circ
\hf_{\alpha} \circ \varphi_{\alpha}^{-1}$. In this chart,
$\psi_{\alpha} = \exp({h/2 + i\chi_{\alpha}})$ where the function $h$ can be
extended to a
well defined gauge invariant function on
$\surface\setminus\vset\cup\avset$; however, 
$\chi_{\alpha}$ is only defined on
$U_{\alpha}\setminus\vset\cup\avset$ modulo $2\pi$. If $U_{\beta}$ is
another holomorphic chart, we can also define a related function
$\chi_{\beta}$ in $U_{\beta}$, if the domains overlap, then for all
$x$ in the intersection $U_{\alpha\beta}$,
\begin{align}\label{eq:chi-beta-alpha}
  \chi_{\beta} = \chi_{\alpha} + \theta_{\alpha\beta} + 2\pi n, \qquad
  n \in \mathbb{Z},
\end{align}
where $\theta_{\alpha\beta}: U_{\alpha\beta} \to \reals$ are 
transition functions. Therefore $d\chi_{\beta} = d\chi_{\alpha} +
d\theta_{\alpha\beta}$ and the arguments of the family $\psi_{\alpha}$
define a connection on 
$\surface\setminus \vset\cup\avset$ which we call $d\chi$.
%
Let $(\hf, \gp): I \subset \reals \to
\confSp$ be a curve on the space of solutions to the Bogomolny
equations, for each $t \in I$, we denote the core positions of $(\hf(t, \cdot), 
\gp(t, \cdot))$ 
 as $p_j(t) \in \vset \cup \avset$, $j = 1, \ldots, k_+ + k_-$. 
 We assume the cores are not intersecting and each curve $p_j(t)$ is 
 differentiable. Given $t \in I$, we choose a gauge such that 
 $(\dot \hf, \dot \gp)$ is perpendicular to the gauge orbit,   
%
by~\eqref{eq:perp-condition} choosing this gauge is equivalent to 
\begin{align}\label{eq:fake-gauss-law}
\lproduct{\dot\hf, X_{\hf}} + d^{*}\dot\gp = 0.
\end{align}
By~\eqref{eq:chi-beta-alpha}, $(\hf, A)$ defines a  function 
\begin{align}
\dot\chi: \surface_{I}\to \reals, 
\end{align}
where
\begin{align}
\surface_{I} = (I \times
\surface)\setminus \set{(t, p_j(t)) \mid t \in I, \; j = 1, \ldots 
k_+ + k_- }.
\end{align}

Let $\eta =
\frac{\dot \fnu}{2} + \dot \chi i: \surface_I\to\cpx$, then in any
holomorphic trivialisation, $\dot \psi_{\alpha} = \eta\, \psi_{\alpha}$, 
moreover, by 
\eqref{eq:fake-gauss-law} and the Taubes equation, on each 
time slice
%
\begin{align}
 \surface_t = \surface\setminus\set{p_j(t) \mid j = 1, \ldots, k_+ + k_-}, 
\end{align}
$\eta$ is a solution to 
\begin{align}\label{eq:eta1-simple}
  -\laplacian \eta = \frac{4e^\fnu}{(1 + e^\fnu)^2}\,\eta.
\end{align}

Now we will extend~\eqref{eq:eta1-simple} to an equation 
valid on all of $I\times\surface$, not just $\surface_I$.
%
%
Let us assume $U_{\alpha}$ is dense and $p_j(t) \in U_\alpha$ for 
all $t \in I$ and $j \in\set{1, \ldots, k_+ + k_-}$. For any given 
$t \in I$, 
let $z_j(t) = \varphi_{\alpha}(p_j(t)) \in \cpx$ and to 
simplify notation, let us write $z_j(t)$ as $z_j$ since time will play no role 
in the following. We define the signature $\sgn_j \in \set{\pm 1}$ as,
\begin{align}
\sgn_j = \begin{cases}
1, & p_j \in \vset,\\
-1, & p_j \in \avset.
\end{cases}
\end{align}

By proposition~\ref{prop:degree}, there is a smooth function 
$R_\alpha:\cpx \to \cpx$ such that, 
\begin{align}
 \psi_{\alpha}(z) = \prod_{j=1}^{k_+ + k_-}(z - 
 z_j)^{\sgn_j}R_{\alpha}(z), 
 \qquad z \in
 \cpx \setminus\varphi_\alpha{(\vset \cup \avset)}, 
\end{align}
where the remainder also satisfies
$R_{\alpha}(z_j)\neq 0$. Whence,
\begin{align}
  \fnu(\varphi^{-1}_\alpha(z)) = \sum_{j=1}^{k_+ + k_-}\sgn_j \log\,\abs{z - 
  z_j}^2 + h_{\alpha}(z),
\end{align}
where $h_{\alpha}: \cpx\setminus\varphi_\alpha{(\vset \cup \avset)} \to \reals$ 
is smooth. Since
the chart is holomorphic, the metric can be written as
$e^{\Lambda(z)}|dz|^2$ and the Laplacian as
$\laplacian = -4 e^{-\Lambda}\del_{z}\bar\del_{z}$. Therefore, as
distributions, 
\begin{align}\label{eq:laplacian-log-z-zj-exp-delta-zj}
  \laplacian \log\,\abs{z - z_j}^2 = -4\pi\,e^{-\Lambda}\delta_{z_j}.
\end{align}
%
 Recall the volume form of the surface in holomorphic coordinates is $\vol =
i/2\,e^{\Lambda}\,dz\wedge d\conj z$, 
equation~\eqref{eq:laplacian-log-z-zj-exp-delta-zj} means that for any test 
function $\varphi:
\cpx \to \reals$,
\begin{align}
  \int_{\cpx} \log\, \abs{z - z_j}^2\, \laplacian\varphi\, \vol
  = -4\pi\,\varphi(z_j).
\end{align}
%
 Let $z_j = z_j^1 + z_j^2\,i$, 
we denote by $D_\epsilon(z_j)$ the holomorphic disk $\abs{z - z_j} < \epsilon$ 
and by $(r_j,
\theta_j)$ polar coordinates centred at $z_j$. Now, we let $t$ vary 
and compute the following time derivative, 
\begin{align}
  \del_t\int_{\cpx} \log\,\abs{z - z_j(t)}^2\,
  \laplacian\varphi\, \vol
  &= \int_{\cpx} -2\,\brk(\frac{\dot z^1_j \cos(\theta_j) + \dot z^2_j
    \sin(\theta_j)}
  {r_j}) \laplacian\varphi \,\vol \nonumber\\
  &= \lim_{\epsilon \to 0}\int_{\cpx\setminus
    D_{\epsilon}(z_j)} 
    -2\,\brk(\frac{\dot z^1_j \cos(\theta_j) + \dot z^2_j 
    \sin(\theta_j)}
  {r_j}) \laplacian\varphi \,\vol \nonumber\\
  &= - 4\pi\,(\dot z^1_j\, \del_1\varphi(z_j) 
  + \dot z^2_j\, \del_2\varphi(z_j)) \nonumber\\
  &= -8\pi \Re (\dot z_j\,\del_{z}\varphi(z_j)).
\end{align}
where we applied the divergence theorem 
to compute the limit, hence, in the sense of distributions,
%
\begin{align}
\laplacian (\del_t\log\,\abs{z - z_j}^2) = 8\pi \Re\brk(\dot 
z_j\,\del_z\delta_{p_j})
= -8\pi \Re\brk(\dot 
z_j\,\del_{z_j}\delta_{p_j}).  
\end{align}

For a given trajectory of the cores, the right side of this equation defines 
a distribution on $\surface$, on the other hand, on the left side is the 
time 
derivative of the singular part of $h\circ\varphi^{-1}$, from this observation 
we state formally (i.e. without considering details about convergence in 
function spaces) that $\dot h$ must be a solution to the equation,
%
\begin{align}
-\laplacian \dot \fnu = \frac{4e^\fnu \dot \fnu}{(1 + e^\fnu)^2} 
+ 8\pi \sum_{j}\sgn_j \Re (\dot z_j\,\del_{z_j}\delta_{p_j})
\end{align}

Similarly, for any $z_j \in \varphi_\alpha(\vset \cup \avset)$ there is a small 
neighbourhood $D\subset \cpx$ such that,
\begin{align}
\dot\chi = \Im\left(\frac{\dot\psi_{\alpha}}{\psi_{\alpha}}\right)
= -\sgn_j
\brk (
\frac{-\dot z^1_j\, \sin(\theta_j) + \dot z^2_j\,\cos(\theta_j)}{r_j}
) + \tilde \chi_{\alpha},
\end{align}
for some smooth function $\tilde\chi_{\alpha}: D \to
\reals$. For the singular part of this equation we have,
\begin{align}
  \laplacian \brk(
\frac{-\dot z_j^1\, \sin(\theta_j) + \dot z_j^2\,\cos(\theta_j)}{r_j}
) &= - 2\pi \brk(
-\dot z_j^1 \del_2\delta_{p_j} + \dot z_j^2 \del_1\delta_{p_j}
) \nonumber\\
&= -4\pi\,\Im (\dot z_j\del_z\delta_{p_j}) \nonumber\\
&=  4\pi\,\Im (\dot z_j\del_{z_j}\delta_{p_j}).
\end{align}

Hence, $\dot{\chi}$ is a solution to the equation,
 %
 \begin{align}
-\laplacian \dot{\chi} = \frac{4e^h}{(1 + e^h)^2}\,\dot{\chi} 
+ 4\pi \sum_{j}\sgn_j \Im (\dot z_j\del_{z_j}\delta_{p_j})
.
\end{align}
 
We conclude that $\eta = \dot h / 2 + \dot \chi\,i$ is a solution to the 
equation,
 \begin{align}\label{eq:eta}
 -\laplacian \eta = \frac{4e^h}{(1 + e^h)^2}\,\eta +
 4\pi \sum_{j} \sgn_j \dot z_j\,\del_{z_j} \delta_{p_j}. 
 \end{align}
 
 Equation \eqref{eq:eta} is formal, in order to make sense of it, 
 we have to supplement it with analytical properties of the solution $h$ to the 
 Taubes equation and in the case of the plane 
 with proper limiting behaviour at infinity. 
 In the successive chapters we will address these issues. We assume however the 
 existence of exactly one solution to~\eqref{eq:eta}. Under this assumption, 
 the solution is given by the function
 %
 \begin{align}\label{eq:eta-formula}
\eta = \sum_{j} \dot z_j \, \del_{z_j} \fnu.
\end{align}
 
Note that although each core position $z_j$ is defined up to holomorphic 
coordinates, the right hand side is well defined independently of the chart 
chosen, provided the cores are contained in it. With this initial setup, 
we compute the localization
 formula~\eqref{eq:loc-formula}. 

\begin{lemma}\label{lem:coef-sym}
    Let $\varphi: U \subset \surface \to \cpx$ be a holomorphic chart, $U$ 
    open and dense, such that $\vset\cup\avset \subset U$. Assume the cores 
    are simple, for each $p_j 
    \in\vset\cup\avset$ define,
 \begin{align}\label{eq:pre-bcoef}
 b_j = 2 \eval{\conj\del}{z = z_j} \brk(\sgn_j \fnu(\varphi^{-1}(z)) 
 - \log\,\abs{z - z_j}^2),
 \end{align}
 where $z = \varphi(x)$, $z_j = \varphi(p_j)$. Then the coefficients $b_j$ 
 have the symmetries,
 \begin{align}
     \del_{z_i} b_j = \conj\del_{z_j} \conj b_i, &&
     \del_{z_i} \conj b_j = \del_{z_j} \conj b_i.
 \end{align}
\end{lemma}

\begin{proof}
\newcommand{\Eop}{\operatorname{K}}

For the proof we generalise the argument of Manton and 
Sutcliffe~\cite[pg.~209]{manton_topological_2004} given for vortices of the 
Ginzburg-Landau functional on the Euclidean plane. 
Let $\Eop = -(\laplacian + 4e^\fnu(e^\fnu + 1)^{-2})$, 
by the Taubes equation, $\sgn_i\del_{z_i}\fnu$ (no summation) is a
fundamental solution of $\Eop$, 
\begin{align}
    \Eop (\sgn_i\del_{z_i}\fnu) = -4\pi\del \delta_{p_i}.
\end{align}

If $i\neq j$, $\del_{z_i}\fnu$ and $\del_{z_j}\fnu$ have different
 singularities and we can integrate by parts to obtain,
\begin{align}
  \int_\surface \brk(
  \sgn_j\del_{z_j}\fnu\Eop (\sgn_i\del_{z_i}\fnu) 
  -  \sgn_i\del_{z_i}\fnu\Eop (\sgn_j\del_{z_j}\fnu)
  )\,\vol = 0,
\end{align}
where the integration by parts involves computing a limit at each
singularity, we omit the details for clarity of the argument.

On the other hand,
\begin{align}
    \int_\surface \brk(
    \sgn_j\del_{z_j}\fnu \,(-4\pi \del\delta_{p_i})
    -  \sgn_i\del_{z_i}\fnu \,(-4\pi \del\delta_{p_j})
    )\,\vol &= 4\pi \brk(
    \sgn_j\del(\del_{z_j}\fnu)(p_i) - 
    \sgn_i\del(\del_{z_i}\fnu)(p_j)
    )\nonumber\\
    &= 2\pi \sgn_j\sgn_i \brk(\del_{z_j}\conj b_i - \del_{z_i}\conj b_j).
\end{align}

Therefore, $\del_{z_i}\conj b_j = \del_{z_j}\conj b_i$. Since $\Eop$ is a real
operator, 
\begin{align}
    \Eop (\sgn_i\conj\del\fnu_{z_i}) = -4\pi\conj\del \delta_{z_i},
\end{align}
hence,
\begin{equation}
\begin{aligned}
 \int_\surface \brk( \sgn_j\conj\del_{z_j}\fnu \Eop (\sgn_i\del_{z_i}\fnu)
 - \sgn_i\del_{z_i}\fnu \Eop (\sgn_j\conj\del_{z_j}\fnu)) \vol 
 &= 
 -4\pi \int_\surface \brk( \sgn_j\conj\del_{z_j}\fnu\,\del\delta_{p_i}
 - \sgn_i\del_{z_i}\fnu \,\conj\del\delta_{p_j}) \vol\\
 &= 2\pi\sgn_j\sgn_i \, (\conj\del_{z_j}\conj b_i - \del_{z_i} b_j).
 \end{aligned}
\end{equation}

As in the previous case, we can apply integration by parts to 
prove that the first integral is zero. Therefore 
$\del_{z_i}b_j = \conj\del_{z_j} \conj b_i$.

\let \Eop\undefined
\end{proof}

We denote by $\moduli^{k_+,k_-}$ the moduli space of solutions to the 
Bogomolny equations with $k_+$ vortices and $k_-$ antivortices. 

 \begin{theorem}\label{thm:bps-loc-formula}
 If $\varphi: U \subset \surface \to \cpx$ is an open and dense holomorphic
  chart, containing the cores of a time varying trajectory
  $(\hf, \gp): I \subset \reals \to \moduli^{k_+, k_-}$, such that the
  variation $(\dot\hf, \dot \gp)$ satisfies Gauss's equation and each core $p_i 
  \in \vset \cup \avset$ is simple, then the
  kinetic energy of the trajectory can be computed as,
 \begin{align}
 \label{eq:loc-formula}
 \Energy &= \pi \sum_{i,j=1}^{k_+ + k_-} 
            \brk( e^{\Lambda(z_i)}(1 - \sgn_{i} \tau)\delta_{ij}
             + \del_{z_i} b_{j}
             )\,\dot z_i \, \conj{\dot z_j},
 \end{align}
 where $z_j = \varphi(p_j)$. Moreover, the quadratic form,
 \begin{align}\label{eq:kmetric}
 \kmetric = 2\pi \sum_{i,j = 1}^{k_+ + k_-}
            \brk( e^{\Lambda(z_i)}(1 - \sgn_{i} \tau)\delta_{ij}
             + \del_{z_i} b_{j}
             )\, d z_i\, d \conj z_j,
 \end{align}
 determines a K\"ahler metric in the open and dense set 
 of non intersecting vortices and antivortices.
\end{theorem}

Theorem~\ref{thm:bps-loc-formula} was proved for $\tau = 0$
in~\cite{romao2018}. We follow the authors ideas and extend them to
the remaining cases.

\begin{proof}
Let $D_\epsilon$ be a collection of disjoint holomorphic 
$\epsilon$-disks, each one 
 centred at one of the cores in $\varphi(\vset \cup \avset)$ and let 
 $U_\epsilon = 
 U\setminus D_\epsilon$. We will make a calculation similar to the
 one done in~\cite{manton_topological_2004} for the Ginzburg-Landau
functional. The energy of the trajectory can be
 computed as,
\begin{align}
\Energy &= \half \brk(\norm{\dot\psi}^2 + \norm{\dot\gp}^2) \nonumber\\
        &= \lim_{\epsilon\to 0} \half\int_{U_\epsilon} \brk(
          \frac{4e^{\fnu}\brk(\frac{1}{4}\dot\fnu^2 + \dot\chi^2)}{(1 +
          e^{\fnu})^2} 
          + \abs{\dot\gp}^2)\, \vol,
\end{align}
by the first of the Bogomolny equation, $\conj\del\psi = \half
(-\gp_2 + \gp_1 i)\psi$, on the other hand, $\conj\del\psi = \psi
\,\conj\del(\half\fnu + \chi i)$, hence,
\begin{align}\label{eq:gp-chi-fnu}
  \gp = d\chi - \half *d\fnu,
\end{align}
which implies,
\begin{align}
  \abs{\dot\gp}^2 = \abs{d\dot\chi}^2 
  - \lproduct{d\dot\chi, *d\dot\fnu}
  + \frac{1}{4}\abs{*d\dot\fnu}^2.
\end{align}

Integrating by parts,
\begin{align}
  \int_{U_{\epsilon}} \abs{d\dot\chi}^2\,\vol
  &= \int_{U_{\epsilon}} d\dot\chi \wedge *d\dot\chi\nonumber\\
  &= \int_{\del U_{\epsilon}} \dot\chi\,*d\dot\chi +
      \int_{U_{\epsilon}}
      \dot\chi\,\laplacian \dot\chi\,\vol\nonumber\\
  &= -\int_{\del D_{\epsilon}}\dot\chi\,*d\dot\chi
    - \int_{U_{\epsilon}}\frac{4e^{\fnu}\dot\chi^2}{(1 + e^{\fnu})^2} \vol.
\end{align}

Proceeding in a similar way, we obtain a second pair of
 equations,
\begin{align}
  \int_{U_{\epsilon}} \lproduct{d\dot\chi, *d\dot\fnu}\,\vol
  &= \int_{\del D_{\epsilon}} \dot\chi\,d\dot\fnu, \\
  \int_{U_{\epsilon}} \abs{*d\dot\fnu}^2\,\vol
  &= -\int_{\del D_{\epsilon}} 
     \dot\fnu\,*d\dot\fnu
    - \int_{U_{\epsilon}}\frac{4 e^{\fnu}\dot\fnu^2}{(1 + e^{\fnu})^2} \vol.
\end{align}
Substituting into the equation for the energy, we obtain,
\begin{align}
  \Energy &= -\half \lim_{\epsilon\to 0}\int_{\del D_{\epsilon}}\brk(
  \dot\chi\,*d\dot\chi + \dot\chi\,d\dot\fnu
  + \frac{1}{4} \dot\fnu\,*d\dot\fnu
  ) \nonumber\\
  &= -\half \lim_{\epsilon\to 0}\int_{\del D_{\epsilon}}\brk(
  \dot\chi\,*\dot\gp - \half \,\dot\fnu\,\dot\gp),
\end{align}
where we have used the time derivative of equation~\eqref{eq:gp-chi-fnu} 
to simplify the energy. Since $\epsilon\to 0$, the only terms that
contribute to the 
energy are the singular terms. We will compute each of these terms at the 
respective core. For any  
$z_j \in \varphi(\vset\cup\avset)$, let $D_{\epsilon}(z_j)$ be the
$\epsilon$ holomorphic disk centred at this point. 
If $\epsilon$ is small, for $z \in \disk_{\epsilon}(z_j)$ we have the 
approximations,
\begin{align}
    \dot\fnu &= -2\sgn_j \brk(
    \frac{\dot z_j^1 \cos(\theta_j) + \dot z_j^2\sin(\theta_j)}{r_j}) + 
    R_\fnu(z),\\
    \dot\chi &= -\sgn_j \brk(
    \frac{-\dot z_j^1 \sin(\theta_j) + \dot z_j^2 \cos(\theta_j)}{r_j})
    + R_\chi(z),
\end{align}
for some residual smooth functions $R_\fnu$ and $R_\chi$. We also
expand $\dot\gp$ in  
polar coordinates centred at $z_j$,
\begin{align}
    \dot\gp = \dot\gp_r\,dr_j + \dot\gp_\theta\,r_j\,d\theta_j,
\end{align}
where,
\begin{align}
    \dot\gp_r      &= \dot\gp_1\,\cos(\theta_j) + \dot\gp_2\,\sin(\theta_j), &
    \dot\gp_\theta &= -\dot\gp_1\,\sin(\theta_j) + \dot\gp_2\,\cos(\theta_j).
\end{align}

The singular terms in the energy integral contribute as,
\begin{align}
    \lim_{\epsilon\to 0}\int_{\del D_\epsilon(z_j)} \dot\chi\,*\dot\gp 
    &= -\sgn_j \lim_{\epsilon\to 0}\int_{D_\epsilon(z_j)} \brk(
      \frac{-\dot z_j^1 \sin(\theta_j) + \dot z_j^2 \cos(\theta_j)}{\epsilon}
    )\,\dot\gp_r\,\epsilon d\theta_j\nonumber\\
    &= \pi\sgn_j (\dot z_j^1\,\dot\gp_2(z_j) - \dot 
    z_j^2\,\dot\gp_1(z_j))\nonumber\\
    &= -\pi\sgn_j \Im\brk( \dot z_j \, \brk(\dot\gp_1(z_j) - \dot\gp_2(z_j)\,i))
\end{align}
and
\begin{align}
    \lim_{\epsilon\to 0}\int_{\del D_\epsilon(z_j)} \dot\fnu\,\dot\gp 
    &= -2\sgn_j \lim_{\epsilon\to 0}\int_{\del D_\epsilon(z_j)} \brk(
    \frac{\dot z_j^1 \cos(\theta_j) + \dot z_j^2\sin(\theta_j)}{\epsilon})
    \,\dot\gp_\theta\,\epsilon\, d\theta_j\nonumber\\
    &= -2\pi\sgn_j\,\brk(\dot z_j^1\,\dot\gp_2(z_j) - \dot 
    z_j^2\,\dot\gp_1(z_j))\nonumber\\
    &= 2\pi\sgn_j\,\Im\brk( \dot z_j \, \brk(\dot\gp_1(z_j) - 
    \dot\gp_2(z_j)\,i)),
\end{align}

Therefore, the energy of a moving pair is,
\begin{align}
    \Energy = \pi \sum_{j=1}^{k_+ + k_-} \sgn_j 
    \Im\brk( \dot z_j \, \brk(\dot\gp_1(z_j) - \dot\gp_2(z_j)\,i)). 
\end{align}

By equations~\eqref{eq:gp-chi-fnu} and~\eqref{eq:eta-formula},
\begin{align}
    \dot\gp_1 - \dot\gp_2\,i 
    &= \brk(\del_1\dot\chi + \half\,\del_2\dot\fnu)
      - \brk(\del_2\dot\chi - \half \del_1\dot\fnu) i \nonumber\\
    &= \del_1\brk( \half\dot\fnu - \dot\chi\,i)\,i 
      + \del_2\brk(\half\dot\fnu - \dot\chi\,i) \nonumber\\
    &= 2i\,\del_z\conj\eta \nonumber\\
    &= 2i\sum_{j} \conj{\dot z}_j
    \,\del_z \conj\del_{z_j} \fnu.
\end{align}

In a small neighbourhood of any $z_j$, we have the 
asymptotic expansion,
\begin{equation}
  \begin{aligned}
    \sgn_{j} \fnu(\varphi^{-1}(z)) = &\log\,\abs{z - z_j}^2 + a_{j} + 
    \half \conj b_{j}\,(z - z_j)
    + \half b_j\,{(\conj z - \conj z_j)}\\
    &+ \conj c_j\,(z - z_j)^2
    + d_j\,\abs{z - z_j}^2 + c_j\,{(\conj z - \conj z_j)^2}\\
    &+ \mathcal{O}(|z_j|^3).
  \end{aligned}\label{eq:loc-h-coeffs-p}
\end{equation}

Hence, 
\begin{align}
 d_j &= \lim_{z\to z_j} \del_z\conj\del_z(\sgn_j \fnu(\varphi^{-1}(z)) 
 - \log \abs{z -
   z_j}^2)\nonumber\\
   &= \frac{1}{4} \lim_{z \to z_j}\,(-e^{\Lambda(z)}\laplacian) 
   (\sgn_{j} \fnu(\varphi^{-1}(z)) - \log\,\abs{z - z_j}^2)\nonumber\\
   &= \half \sgn_{j} \,e^{\Lambda(z_j)} \lim_{z \to z_j}
   \brk(\frac{e^\fnu -
   1}{e^\fnu + 1} + \tau)\nonumber\\
   &= -\half e^{\Lambda(z_j)}\brk(1 -\sgn_j \tau)
\end{align}
and since $\del_z\conj \del_{z_k} \log\,\abs{z - z_j}^2 = 0$ for any $z \neq 
z_j$, 
\begin{align}
\del_z(\conj\del_{z_k} \fnu)(z_j) 
= \sgn_j \brk(\half \conj\del_{z_k} \conj b_{j} - 
d_{j}\,\delta_{jk}).
\end{align}

Hence,
\begin{align}
  \Energy &= \pi \sum_j
            \sgn_j \Im\brk(\dot p_j \cdot
            2i\,\sum_k\conj{\dot z_k}\eval{\del_z(\conj\del_{z_k}\fnu)}{z = z_j}
            )
            \nonumber\\
          &= 2 \pi\sum_{j,k}\Re\brk(
            \dot z_j \, \conj{\dot z_k}\,
            \brk(\half\conj\del_{z_k}\conj b_j - d_j\delta_{jk})
            )
            \nonumber\\
          &= \pi \sum_{j,k} \Re\brk(
            \brk( e^{\Lambda(z_j)}(1 - \sgn_j \tau)\delta_{jk}
             + \conj\del_{z_k} \conj b_j
             )\,\dot z_j \, \conj{\dot z_k}
            )
            \nonumber\\
          &= \pi \sum_{j,k} \Re\brk(
            \brk( e^{\Lambda(z_j)}(1 - \sgn_{j} \tau)\delta_{jk}
             + \del_{z_j} b_k
             )\,\dot z_j \, \conj{\dot z_k}
            ).
            \label{eq:pre-energy-real-part}
\end{align}

The last equation is a consequence of the symmetry 
$\del_{z_j}b_k = \conj\del_{z_k} \conj b_j$. Also by this
symmetry,~\eqref{eq:loc-formula} 
is a real quantity and therefore coincides
with~\eqref{eq:pre-energy-real-part} as expected, since $\Energy$
represents the kinetic energy of a trajectory on the moduli space.  
To prove that the metric is K\"ahler we must prove that the induced form,
\begin{align}
  \omega = \pi i \sum_{j,k}
  \brk( e^{\Lambda(z_j)}(1 - \sgn_j \tau)\delta_{jk}
  + \del_{z_j}b_k)\,dz_j\wedge d\conj z_k,
\end{align}
is closed. For the following computation, we employ lemma~\ref{lem:coef-sym} 
and the fact that each term $e^{\Lambda(z_j)}(1 - \sgn_j \tau)\,dz_j \wedge 
d\conj z_k$ 
is closed,
\begin{align}
  d\omega &= \pi i \sum_{r,s,t} \brk(
  \del_{z_t}\del_{z_r}b_s\, dz_t\wedge dz_r \wedge d\conj z_s
  + \conj\del_{z_t}\del_{z_r}z_s\, d\conj z_t\wedge dz_r \wedge d\conj z_s
  )\nonumber\\
  &= \pi i \sum_{r,s,t} \brk(
  \del_{z_t}\del_{z_r}b_s\, dz_t\wedge dz_r \wedge d\conj z_s
  - \conj\del_{z_t}\del_{z_r}b_s\,d\conj z_t \wedge d\conj z_s \wedge 
  dz_r)\nonumber\\
  &= \pi i \sum_{r,s,t} \brk(
  \del_{z_t}\del_{z_r}b_s\, dz_t\wedge dz_r \wedge d\conj z_s
  - \conj\del_{z_t}\del_{z_s}b_r\, d\conj z_t \wedge d\conj z_r \wedge d 
  z_s)\nonumber\\
  &= \pi i \sum_{r,s,t} \brk(
  \del_{z_t}\del_{z_r}b_s\, dz_t\wedge dz_r \wedge d\conj z_s
  - \conj\del_{z_t}\conj\del_{z_r}\conj b_s\, d\conj z_t \wedge d\conj z_r 
  \wedge 
  d z_s)\nonumber\\
  &= -2\pi \Im\brk( \sum_{r,s,t}
  \del_{z_t}\del_{z_r}b_s\, dz_t\wedge dz_r \wedge d\conj z_s)\nonumber\\ 
  &= 0,
\end{align}
where the last sum is zero by the commutativity of the mixed derivatives.
\end{proof}



\newcommand*{\fnF}{F}
\newcommand*{\fnV}{V}
\newcommand*{\cmu}{\mu}
\newcommand*{\fnh}{h}
\newcommand*{\spW}{\mathrm{W}}
\newcommand*{\spL}{\mathrm{L}}
\newcommand*{\spH}{\mathrm{H}}
\newcommand*{\spS}{\mathrm{S}}
\newcommand*{\spHilbert}{\mathcal{H}}
\newcommand*{\avg}[1]{\overline{#1}}
\newcommand*{\manifold}{M}
\newcommand*{\wto}{\rightharpoonup}
\newcommand*{\opT}{\mathrm{T}}
\newcommand*{\cf}{\Lambda}
\newcommand*{\sign}{s}

\section{The governing elliptic problem}\label{c:intro-gov-elliptic}

Equation \eqref{eq:taubes} is the governing elliptic problem. Once $h$ is 
determined, the Bogomolny equations determine $\Bform$ and then $\gp$ and $\hf$
up to gauge equivalence. We let, 
\begin{align}
\begin{gathered}
    \fnF : \reals \to \reals,\\
    \fnF(t) = 2\brk(\frac{e^t - 1}{e^t + 1} + \tau),
\end{gathered}  &&
    \begin{gathered}
    \fnV: \reals \to \reals^+, \\
    \fnV(t) = \frac{4e^t}{(1 + e^t)^2}.
    \end{gathered}
\end{align}
Note that $\fnV = \fnF'$ and that $\fnF$ and all
of its derivatives are bounded functions, moreover, if 
$\cmu = \log \brk( (1 - \tau)(1 + \tau)^{-1})$, then $\fnF$ satisfies the
following properties,
\begin{align}
\fnF(\cmu) &= 0,  \\
\fnF'(\cmu) &> 0,\\
\fnF(t) &< 0, \qquad  t <  \cmu,\\
\fnF(t) &> 0, \qquad  t > \cmu,
\end{align}
and,
\begin{equation}
  \norm {\fnF}_{\mathrm{L}^{\infty}(\reals)} +
  \norm{(1 + e^{-t})\fnV}_{\mathrm{L}^{\infty}(\reals)} +
  \norm{e^{-t}(e^t - 1)^{-1}\fnV'}_{\mathrm{L}^{\infty}(\reals)} < \infty.
\end{equation}

If $\vset$ or $\avset$ is non-empty, there exists exactly one function 
$\fnh \in C^{\infty}(\plane \setminus \vset \cup \avset)$
 \cite{han2000existence}, such that,
\begin{align}
-\laplacian \fnh = \fnF(\fnh) + 4\pi\sum_{p\in\vset}
\delta(x - p) - 4\pi\sum_{q \in \avset} \delta(x -q), &&
\lim_{\abs x \to \infty} \fnh = \cmu,
\end{align}
moreover, for any $\epsilon \in (0, 1)$, there exist positive constants 
$C = C(\epsilon)$ and $R = R(\epsilon)$ such that
\begin{equation}
\abs{\fnh(x) - \cmu} \leq 
  C \exp \brk(- \half\sqrt{(1 - \tau^2)(1 - \epsilon)} \abs x),
  \qquad \abs x \geq R.
\end{equation}

Therefore, in the euclidean plane, there exists a unique solution to the 
Taubes 
equation. For a compact surface, existence of a solution to the Taubes 
equation was proved for $\tau = 0$ in~\cite{sibner2010}. We will prove
that this is also the case for $\tau \neq 0$ in
chapter~\ref{c:vav-compact}.

 In this section we prove that solutions to the Taubes equation depend
 smoothly on vortex positions. Recall Sobolev's space $\spW^{k,p}$ is
 the completion of the space of $C_0^\infty$ functions compactly
 supported with respect to  Sobolev's norm, 
 \begin{align}
     \norm{\varphi}_{\spW^{k,p}} 
     = \brk(\sum_{j = 0}^k \norm{\nabla^k \varphi}_{\spL^p}^p)^{1/p},
 \end{align}
 where $\nabla^j\varphi \in (T^*\surface)^{\otimes{j}}$ is the jth exterior
 covariant derivative. We denote the space $\spW^{k,2}$ as
 $\spH^k$. This is a Hilbert space with the product, 
 \begin{align}
   \lproduct{\varphi, \psi}_{\spH^k} = 
   \sum_{j = 0}^k \lproduct{\nabla^j\varphi, \nabla^j\psi}_{\spL^2}.
 \end{align}

 For the inner product in $\Lsp^2$ we omit the
 subindex if is clear from the context that we refer to $\Lsp^2$
 functions. 
 
 In the sequel, we will use some results of analysis that we quote
 here for further  
 reference. The proofs are standard and can be found in the
 literature, for example in~\cite{gilbarg2015elliptic} 
 and~\cite{evans2010partial}.
 
 \begin{theorem}[Banach-Alaoglu]
 Let $\mathrm{X}$ be a Banach space, then the closed unit ball of the 
 dual $\mathrm{X}^*$ is compact with respect to the weak-* topology.
 \end{theorem}
 
 \begin{theorem}[Rellich-Kondrachov] 
 If $\Omega$ is a 
 an open bounded Lipschitz domain of $\reals^n$, $1 \leq p < n$, 
 $p^* = \frac{np}{n - p}$, then $\spW^{1, p}(\Omega)$ is continuously embedded in
 $\spL^{p^*}(\Omega)$ and compactly embedded in $\spL^q(\Omega)$ for any 
 $1 \leq q < p^*$.
 
 If $\Omega$ is a compact manifold of dimension $n$, $k > l$, $k - n/p > l - n/q$,
  then the embedding $\spW^{k,p} \subset \spW^{l, q}$ is completely continuous.
 \end{theorem}

 That the embedding $\spW^{k,p} \subset \spW^{l, q}$ is completely
 continuous is equivalent to claiming that any bounded sequence of
 functions in $\spW^{k,p}$ has a subsequence converging in
 $\spW^{l,q}$. In practice, we will use the Rellich-Kondrachov theorem
 to guarantee that given a bounded sequence of $\spW^{1,p}$ functions
 either on a compact surface or on an open bounded subset of $\plane$,
 we can find a subsequence convergent in $\spL^p$. 
 
 \begin{theorem}[Sobolev's embedding]
 If $\Omega$ is either $\reals^n$ or a bounded domain of with Lipschitz boundary 
 of a compact Riemannian manifold of dimension $n$, and if $k > l$, 
 $1 \leq p < q < \infty$ and $\alpha \in (0, 1]$ are such that,
 \begin{align}
     \frac{1}{p} - \frac{k}{n} = - \frac{r + \alpha}{n},
 \end{align}
 then we have the continuous embedding 
 $\spW^{k,p}(\Omega) \subset C^{r,\alpha}(\Omega)$.
 \end{theorem}
 
 \begin{theorem}[Lax-Milgram]
 If $B: \spHilbert \times \spHilbert \to \reals$ is a continuous bilinear form in a 
 Hilbert space $\spHilbert$ and there is a positive constant $\alpha$ such that,
 \begin{align}
     \abs{B(u,u)} \geq \alpha\norm{u}^2,
 \end{align}
 then, for any $u \in \spHilbert$ there is a unique $v \in \spHilbert$, such 
 that,
 \begin{align}
     B(v, x) = \lproduct{u, x}
     \qquad \forall x \in \spHilbert.
 \end{align}
 Moreover, 
 \begin{align}
     \norm{v} \leq \frac{1}{\alpha} \norm{u}. 
 \end{align}
\end{theorem}

The proof can be found in \cite[p.~83]{gilbarg2015elliptic}.

 \begin{theorem}[Schauder's estimates]
 If $\Omega' \Subset \Omega$ are open sets of any manifold $\manifold$, $f \in
\spH^k(\Omega)$ and  $u \in \spH^1(\Omega)$ is a weak solution to the equation
\begin{equation}
  \laplacian_\manifold u = f,
\end{equation}
then $u \in \spH^{k + 2}(\Omega')$ and
\begin{equation}
  \norm{u}_{\spH^{k + 2}(\Omega')} \leq C \brk(
  \norm{f}_{\spH^k(\Omega)} + \norm{u}_{\spL^2(\Omega)}),
\end{equation}
for some constant $C = C(k, \Omega, \Omega')$. In a compact manifold $\manifold$, 
we also have the estimate,
\begin{align}
\norm{u - \avg u}_{\spH^{k + 2}} \leq C\, \norm{f}_{\spH^k},    
\end{align}
for some constant $C = C(k, \Omega, \Omega')$, where 
$\avg u = \frac{1}{\abs{\manifold}}\,\int_\manifold u\,\vol$ is the trace of $u$.
\end{theorem}

Given a pair of open sets $\Omega'$, $\Omega$ on a topological space,
the notation $\Omega' \Subset \Omega$ means $\overline{\Omega'}
\subset \Omega$.

\subsection{Smooth parametric dependence of $h$}\label{subsec:smooth-par-dep}

The moduli space can be identified with $\brk(\surface^{k_+} \times 
\surface^{k_-}
\setminus \Delta_{k_+,k_-}) / \spS_{k_+} \times \spS_{k_-}$, where 
$\Delta_{k_+,k_-}$ is the big fat diagonal of intersecting 
vortices and antivortices and the product of symmetric groups 
$\spS_{k_+} \times \spS_{k_-}$ act permuting the components of 
$\surface^{k_+} \times \surface^{k_-}$. Let us focus in the open and
dense subset of 
non-overlapping cores. We can identify this space with
 $\surface^{k_+}\times\surface^{k_-} \setminus \Delta_{k_+,k_-}$. We
 aim to prove  
 that in this subspace, $\fnh$ depends smoothly on the positions of
 the cores. 
 

 \begin{lemma}\label{lem:h2-l2-iso}
   Let $\surface$ be either the plane or a compact surface.
   If $V \in C^\infty(\surface)$, is a non-negative smooth function
  with only finite zeros, such that if $\surface$ is the plane,
  $\lim_{\abs x \to \infty}V(x) \in (0, 1]$, and all the derivatives
  $\grad^kV$ are  bounded, then for any $r \geq 0$, Schrodinger's
  operator, 
  \begin{align}
   \laplacian + V : \spH^{r + 2}(\surface) \to \spH^r(\surface),
  \end{align}
  is a Hilbert space isomorphism.
\end{lemma}
 
\begin{proof}
  By the hypothesis on the potential function $V$, the operator
  $\laplacian + V$ is continuous. Let us define
  the bilinear form $B: \spH^1 \times \spH^1 \to \reals$ and the 
  linear functional $A: \spH^1 \to \reals$ such that,
  \begin{equation}
    \begin{aligned}
      B(u, v) &= \lproduct{\grad u, \grad v} + \lproduct{Vu,
        v},\\
      A(u) &= \lproduct{b, u},
    \end{aligned}
  \end{equation}
  where $b \in \spH^r$. By the Cauchy-Schwarz inequality, $A$ and $B$ are
  continuous.

  Firstly, we claim $B$ is coercive, i.e., there is a positive
  constant $\alpha$ such that, 
  \begin{equation}
    \alpha\,\norm{u}^2_{\spH^1} \leq B(u, u).
  \end{equation}
  
  Let $\Omega$ be either the compact surface $\surface$ or 
  an open disk $\disk_R(0)\subset \plane$ such that 
  there is a constant $a \in (0, 1]$ for which $V(x) \geq a > 0$ if $x
  \in \plane\setminus \Omega$. In the latter case,
  \begin{equation}\label{eq:pre-sch-pot-euc-bound}
    \norm{\grad u}^2_{\spL^2(\plane\setminus\Omega)} + \norm{V^{1/2}
      u}^2_{\spL^2(\plane\setminus \Omega)} \geq a
    \norm{u}^2_{\spH^1(\plane\setminus\Omega)}. 
  \end{equation}
  
  Assume towards a contradiction the existence of a sequence
  $\set{u_n} \subset \spH^1(\Omega)$, such that
  \begin{align}
\norm{u_n}_{\spH^1(\Omega)} &= 1, &
B(u_n, u_n) &\leq \frac{1}{n}.
\end{align}

By the Banach-Aloglu theorem, we can assume $u_n \wto u_{*}$ in
$\spH^1(\Omega)$, and by the Rellich theorem, we can assume the strong
convergence $u_n\to u_{*}$ in $\Lsp^2(\Omega)$. Since $B(u_n,u_n)\to
0$, 
\begin{equation}
  \norm{\grad u_n}_{\Lsp^2(\Omega)} \to 0,
\end{equation}
hence $u_*$ is constant almost everywhere, because by the strong convergence 
in $\spL^2$ and the convergence,
\begin{align}
\lproduct{u_n, u_*}_{\spH^1(\Omega)} \to \norm{u_*}^2_{\spH^1(\Omega)},
\end{align}
we deduce $\lproduct{\grad u_n, \grad u_*}_{\spL^2(\Omega)} \to \norm{\grad 
u_*}^2_{\spL^2(\Omega)}$, but $\lproduct{\grad u_n, \grad 
u_*}_{\spL^2(\Omega)}\to 0$, hence $\grad u_* = 0$ almost everywhere. On the 
other hand,
\begin{equation}
  \norm{V^{1/2} u_n}_{\spL^2(\Omega)} \to 0,
\end{equation}
and $V$ is positive except for a finite set, thence $u_{*} = 0$. We
conclude $u_n \to 0$ in $\spH^1(\Omega)$, but this is a
contradiction because each $u_n$ has unit norm. Therefore, there is a
positive 
constant $a'$ such that if $u \in \spH^1(\Omega)$, then 
$B(u, u) \geq a' \norm{u}^2_{\spH^1(\Omega)}$. If $\surface$ is
compact we conclude $B$ is coercive. If
$\surface$ is the plane, let us 
take $\alpha = \min(a, a')$. If $u \in \spH^1(\plane)$,
\begin{equation}
  B(u, u) \geq a \norm{u}^2_{\spH^1(\plane\setminus \Omega)} + a'
  \norm{u}^2_{\spH^1(\Omega)}  \geq \alpha \norm{u}^2_{\spH^1(\plane)}.
\end{equation}

Secondly, we prove the basic inequality,
\begin{align}\label{eq:energy-estimate-potential}
  \norm{u}_{\spH^{r + 2}} \leq C\,\norm{(\laplacian + V)\,u}_{\spH^r},
\end{align}
where $u \in \spH^{r + 2}$ is arbitrary. If $\surface$ is compact this
is by Schauder's estimates and coercivity of $B$. If $\surface$ is
the plane, we first prove the inequality for $\varphi \in
C^{\infty}_0$. Assume $r = 0$, by coercivity,
\begin{align}
  \norm{\varphi}_{\spH^1} \leq C\,\norm{(\laplacian + V)\varphi}_{\spL^2}.
\end{align}

We know in this case $\norm{\grad^2\varphi}_{\spL^2} = \norm{\laplacian
  \varphi}_{\spL^2}$~\cite[Thm. 9.9]{gilbarg2015elliptic}, hence,
\begin{align}
  \norm{\varphi}_{\spH^2} &\leq C\,(\norm{(\laplacian + V)\varphi}_{\spL^2} +
  \norm{\grad^2\varphi}_{\spL^2})\nonumber\\
  &= C\,(\norm{(\laplacian + V)\varphi}_{\spL^2} +
  \norm{\laplacian\varphi}_{\spL^2})\nonumber\\
  &\leq C\,(\norm{(\laplacian + V)\varphi}_{\spL^2} +
  \norm{\varphi}_{\spL^2})\nonumber\\
  &\leq C\,\norm{(\laplacian + V)\varphi}_{\spL^2}.
\end{align}

Let $\psi = (\laplacian + V)\,\varphi \in C^{\infty}_0$. Given the
test function $\varphi$,
$\del_j\varphi$ is a solution to the problem,
\begin{align}
  (\laplacian + V)\,\del_j\varphi = \del_j\psi - \del_jV\,\varphi.
\end{align}

By hypothesis, the derivatives of the potential are bounded. Applying
 the previous bound to $\grad\varphi$,
\begin{align}
  \norm{\grad\varphi}_{\spH^2} &\leq 
  C\,\norm{(\laplacian + V)\,\grad\varphi}_{\spL^2} 
  \nonumber\\
  &\leq C\,(\norm{\grad\psi}_{\spL^2} 
  + \norm{\varphi\,\grad V}_{\spL^2}) \nonumber\\
  &\leq C\,(\norm{\psi}_{\spH^1} + \norm{\varphi}_{\spL^2}) \nonumber\\
  &\leq C\,(\norm{(\laplacian + V)\,\varphi}_{\spH^1}).
\end{align}

We apply this argument recursively. Having found
bounds for $\varphi$ and $\grad \varphi$ up to some $r$,
\begin{align}
  \norm{\varphi}_{\spH^{r + 3}} &\leq \norm{\varphi}_{\spH^{r + 2}} +
  \norm{\grad\varphi}_{\spH^{r + 2}}\nonumber\\
  &\leq C(\norm{(\laplacian + V)\,\varphi}_{\spH^{r}} + 
  \norm{(\laplacian + V)\,\varphi}_{\spH^{r + 1}})\nonumber\\
  &\leq C\,\norm{(\laplacian + V)\,\varphi}_{\spH^{r + 1}}.
\end{align}

Thus, for all $r \geq 0$ there is a constant $C$ such that for any
$\varphi \in C^{\infty}_0$,
\begin{align}\label{eq:varphi-psi-hr-bound}
  \norm{\varphi}_{\spH^{r + 2}} 
  \leq C\,\norm{(\laplacian + V)\,\varphi}_{\spH^r}.
\end{align}

Since $C^{\infty}_0$ is dense in $\spH^r$ and $(\laplacian + V)$ is
continuous, we conclude~\eqref{eq:energy-estimate-potential} is also valid
on the plane. 

Thirdly, we prove $(\laplacian + V)$ is surjective. By the Lax-Milgram
theorem, for any $b \in \spH^r$ there is a 
unique  $u \in \spH^1$, such that $B(u, v) = A(v)$ for all 
$v \in \spH^1$. This function is a weak solution of the
equation, 
\begin{equation}\label{eq:linear-elliptic-problem-existence-lemma}
  (\laplacian + V)\,u = b.
\end{equation}

If $\surface$ is compact, elliptic regularity
implies $u$ is a strong solution in  $\spH^{r + 2}(\surface)$. We
prove this is also the case on
the plane. Let $\psi \in C^{\infty}_0$ and
denote by $\varphi$ the weak solution to the equation,
\begin{align}
  (\laplacian + V)\,\varphi = \psi.
\end{align}

Elliptic regularity and Sobolev's embedding imply $\varphi$ is a
strong solution in $C^{\infty}$. Notice $\varphi \in \spH^{r +
  2}$  $\forall r \geq 0$ because our previous argument can
still be applied to show~\eqref{eq:varphi-psi-hr-bound} holds. Let
$\left\{\psi_n\right\} \subset C^{\infty}_0 $ be a sequence of test
functions converging to $b$ in $\spH^r$. For each $\psi_n$ let
$\varphi_n \in C^{\infty}$ be a strong solution of the elliptic
problem. By~\eqref{eq:varphi-psi-hr-bound} $\left\{\varphi_n\right\}$ is a
Cauchy sequence in $\spH^{r +2}$, thus there is $u \in \spH^{r
+ 2}$ such that $\varphi_n \to u$. By continuity of $\laplacian + V$,
$u \in \spH^{r + 2}$ is a strong solution
of~\eqref{eq:linear-elliptic-problem-existence-lemma}.

Finally, \eqref{eq:energy-estimate-potential} implies $\laplacian +
V$ is injective. By the open mapping
theorem, the inverse is also continuous and the operator is an
isomorphism. 
\end{proof}

For a compact manifold in general, we can estimate the norm of
solutions to linear problems,

\newcommand \afn {a}
\newcommand \bfn {b}
\newcommand \ueps {u}
\newcommand \Hsp {\mathrm{H}}
\newcommand \Xsp{\mathcal{X}}
\newcommand \aForm {\mathrm{A}}
\newcommand \fFunc {\mathrm{B}}
\newcommand \intManifold {\int_{\manifold}}

\begin{proposition}
  \label{lem:linear-prob-cmp}
If $\manifold$ is a compact Riemannian manifold of dimension
$n$, $-\laplacian$  is the Laplace-Beltrami operator of the
metric, $\afn, \bfn \in \Lsp^2(\manifold)$ are functions such that
$\afn$ is non-negative and bounded with positive integral $\int_M \afn\,
\vol$, then the problem, 
\begin{equation}
  \label{eq:linear-problem}
  -\laplacian \ueps = \afn\ueps + \bfn,
\end{equation}
has exactly one solution $\ueps \in \Hsp^2(\manifold)$. Moreover,
there is a positive constant $C(a)$ such that,
\begin{equation}
  \norm{\ueps}_{\spH^2} \leq K\,\norm{\bfn}_{\spL^2},
\end{equation}
where the constant $K(\afn)$ depends on the bound for $\afn$ and
$\int_{\manifold}\afn \,\vol$. 
\end{proposition}

If $n = 2, 3$, by Sobolev's embedding, $\ueps \in C^0(\manifold)$, in
general, we only have $\ueps \in {\Hsp}^2(\manifold)$ unless we know
$\afn$ and $\bfn$ have more regularity. This problem has been
studied for different conditions on the coefficients in the references 
\cite{mitrea1999boundary,mitrea2000potential}.

\begin{proof}
  We will prove the existence of solutions to the linear problem 
  and continuity on the datum as an
    application of the Lax-Milgram theorem.
    
Let $\Xsp$ be the subspace of $\Hsp^1 (\manifold)$ of
functions of zero average, 

\begin{equation}
  \Xsp = \set{ u \in \Hsp^1(\manifold) \st \int_M u\,\vol = 0}.
\end{equation}

$\Hsp^1 (\manifold)$ can be decomposed as
\begin{equation}
  \Hsp^1 (\manifold) = \Xsp \oplus \reals.
\end{equation}

Finding a solution to equation~\eqref{eq:linear-problem} is equivalent
to find $(\ueps_0, c) \in\Xsp \oplus \reals$, such that
\begin{equation}
  -\laplacian \ueps_0 = \afn\ueps_0 + \afn\,c + \bfn.
\end{equation}

By the divergence theorem, the constant is
\begin{equation}
  \label{eq:c-value}
  c = \frac{
   -\int_{\manifold} (\afn\ueps_0 + \bfn) \,\vol
  }{
   \int_{\manifold}\afn \,\vol
  }.
\end{equation}

\eqref{eq:linear-problem} is equivalent to finding $\ueps_0 \in \Xsp$
such that 
\begin{equation}
  \label{eq:h10-equiv-problem}
  -\laplacian\, \ueps_0 = \afn\ueps_0 + \bfn - \frac{
    \afn \cdot \int_{\manifold} (\afn\ueps_0 + \bfn)\, \vol
  }{
    \int_{\manifold}\afn \,\vol
  }.
\end{equation}

Let us define the operators $\aForm:\Xsp\times\Xsp \to \reals$, $\fFunc:
\Xsp \to \reals$, as

\begin{align}
\aForm(u, v) &= \lproduct{ d\ueps, dv} + \lproduct{\afn\ueps, v}
  - \frac{1}{\intManifold\afn\,\vol}\,
   \intManifold \afn \ueps\,\vol \cdot\intManifold \afn v\,\vol,\\
\fFunc(v) &=
\frac{
    1
  }{
    \intManifold\afn \,\vol
  }
  \intManifold
  \left(
    \afn\cdot\intManifold\bfn\,\vol - \bfn \cdot\intManifold\afn\,\vol
  \right)\cdot v\,\vol. 
\end{align}

Equation (\ref{eq:h10-equiv-problem}) can be rewritten in variational
form as the problem of finding $\ueps_0 \in\Xsp$, such that for any $v
\in \Xsp$,
\begin{equation}
  \aForm(\ueps_0, v) = \fFunc(v).
\end{equation}

$\fFunc$ is bounded and $\aForm$ continuous because
$\afn, \bfn \in \Lsp^2(\manifold)$. By Cauchy-Schwartz's inequality,
\begin{equation}
  \abs*{\intManifold \afn \ueps\, \vol} \leq
  \norm{\sqrt\afn}_{\spL^2}\,\norm{\sqrt\afn\ueps}_{\spL^2},
\end{equation}
hence,
\begin{equation}
  \begin{aligned}[b]
      \aForm(u, u) &= \norm{d\ueps}^2_{\spL^2}
  + \frac{
    1
  }{
    \intManifold a \,\vol
  } \brk(
    \intManifold a \ueps^2 \,\vol \cdot \intManifold a \,\vol
    -  \brk( \intManifold a\ueps \,\vol)^2
    )\\
   &= \norm{d\ueps}^2_{\spL^2} + \frac{
    1
  }{
    \intManifold a \,\vol
  } \brk( 
    \norm{\sqrt{\afn} u}^2_{\spL^2}\,\norm{\sqrt\afn}^2_{\spL^2} - 
    \lproduct{\afn, u}^2_{\spL^2}
  )\\
   &\geq \norm{d\ueps}^2_{\spL^2}.
  \end{aligned}
  \end{equation}

  By Poincar\'e's inequality, there is a positive
  constant $\alpha$, such that
  \begin{equation}
    \alpha\,\norm{\ueps}_{\spH^1}^2 \leq 
    \norm{d\ueps}^2_{\spL^2} \leq  \aForm(\ueps, \ueps). 
  \end{equation}

 Therefore, there exists a unique solution $\ueps \in
\Hsp^1 (\manifold)$ to \eqref{eq:linear-problem}. By standard elliptic
regularity estimates, $u \in \Hsp^2(\manifold)$. By the Lax-Milgram theorem,

\begin{equation}
  \norm{\ueps_0}_{\spH^1} \leq \frac{
    1
  }{
    \alpha\intManifold a \,\vol
  }\,\norm*{
    a\cdot\intManifold b\,\vol - b\cdot\intManifold a \,\vol
  }_{\spL^2}. 
\end{equation}

By equation~(\ref{eq:c-value}),
\begin{equation}
  \abs{c} \leq \frac{
    1
  }{
    \intManifold \afn \,\vol
  } \brk(
    \norm{\afn}_{\spL^2}\cdot \norm{\ueps_0}_{\spL^2} +
    \abs*{\intManifold b \,\vol}
  ).
\end{equation}

Therefore, $\ueps$ is bounded in $\Hsp^1 (\manifold)$ by,
\begin{equation}
  \norm{\ueps}_{\spH^1} \leq
  \frac{
    C
  }{
    \intManifold \afn \,\vol
  } \brk(
    \frac{
      \norm{\afn}_{\spL^2}^2
    }{
      \intManifold \afn \,\vol
    }
    + \norm{\afn}_{\spL^2}
    + 1
  )
  \norm{\bfn}_{\spL^2},
\end{equation}
for some suitable constant $C$. By the elliptic estimate  we conclude
$\ueps \in \Hsp^2(\manifold)$ and since $\afn$ is bounded,
\begin{align}
  \norm{u}_{\spH^2} &\leq K\,\left({\norm{\laplacian u}_{\spL^2} + 
  \norm{u}_{\spL^2}}\right)\nonumber\\
  &\leq K\,\left(\norm{a \ueps_0}_{\spL^2} + \norm{b}_{\spL^2} +
    \norm{\ueps_0}_{\spL^2}\right)\nonumber\\
  &\leq K\,\left(\norm{\ueps_0}_{\spL^2} + \norm{b}_{\spL^2}\right)\nonumber\\
  &\leq K\,\norm{b}_{\spL^2}.
\end{align}
where the constant $K$ was renamed from one inequality to the following.
\end{proof}

\let \afn \undefined
\let \bfn \undefined
\let \ueps \undefined
\let \Hsp \undefined
\let \Xsp \undefined
\let \aForm \undefined
\let \fFunc \undefined
\let \intManifold \undefined


We prove smooth dependence on parameters by the implicit function
theorem. If $\surface = \plane$, we define,
\begin{align}\label{eq:v-g-def}
v_c &=  - \log \brk(1 + \frac{1}{\abs{x -
c}^2}), &
g_c &= - \frac{4}{ \brk(1 + \abs{x - c}^2)^2},
\end{align}
then 
\begin{equation}
  -\laplacian v_c = g_c + 4\pi\delta(x - c). 
\end{equation}

If $\surface$ is compact, we rely on the existence of Green's 
function~\cite{aubin2013some}. This is a smooth symmetric function 
$G:\surface\times \surface\setminus \Delta \to \reals$, such that,
\begin{align}
    -\laplacian_x G(x,y) &= \delta_y - \frac{1}{\abs\surface}, &
    \int_\surface G(x,y) \vol_x &= 0.
\end{align}

Notice that we have chosen the oposite sign for $G(x,y)$ with respect
to~\cite{aubin2013some}. In this case, we define,  
\begin{align}\label{eq:intro-vc-compact-case}
    v_c &=  4\pi\,G(x, c).
\end{align}

Given tuples $\vb{p} = (p_1,\ldots, p_{k_+})$, $\vb{q} =
(q_1,\ldots, q_{k_-})$ of non intersecting vortices and antivortices, let 
\begin{align}
  \vb c &= (p_1, \ldots, p_{k+}, q_1, \ldots, q_{k_-}) \in \surface^{k_+ + 
  k_-},\\
  v &= \sum_j\sign_jv_{c_j}.\label{eq:fn-v}\\
  g &= \begin{cases}
  \sum_j \sign_jg_{c_j}, & \surface = \plane,\\
  -\frac{4\pi}{|\surface|} (k_+ - k_-), & \surface\,\text{ compact}.
  \end{cases}
\end{align}

Let $\tilde \fnh = \fnh - v - \cmu$, then the Taubes equation is equivalent to 
its regularised counterpart,
\begin{equation}
  \label{eq:taubes-reg}
    -\laplacian \tilde \fnh = \fnF(\tilde \fnh + v + \cmu) - g.
\end{equation}

If $\surface$ is the euclidean plane, we add the boundary condition,
\begin{equation}
    \lim_{\abs x \to \infty} \tilde h = 0.
\end{equation}


\begin{theorem}\label{thm:h-reg-param-dep}
Let $\vb p = (p_1, \ldots, p_{k_+})$, $\vb q = (q_1, \ldots, q_{k_-})$
be sequences of non-intersecting simple cores in $\surface$, either a compact
surface or the Euclidean plane. 
Let us denote by $h(x; \vb p, \vb q)$ the solution to the Taubes equation
for this configuration. For any families $U_r \subset\surface$, $r =
1,\ldots, k_+$, $V_s \subset\surface$, $s = 1, \ldots, k_-$,
of open neighbourhoods on $\surface$ such that
$U_r \cap V_s = \emptyset$, let $W = (\cup_{r}U_r)
\bigcup (\cup_sV_s)$, then the restriction
\begin{align}
  h: \left(\surface \setminus 
  \overline{W}\right) \times_r U_r \times_s
V_s \to \reals, 
\end{align}
is smooth.
\end{theorem}

\begin{proof}
Consider the function

\begin{equation}
f(r) = \frac{1}{(1 + r^2)^2}, \qquad r \in  \reals,
\end{equation}
this function has the property that it and all of its derivatives are
dominated by $r^{-4}$ as $r \to \infty$. This guarantees that,
\begin{equation}
g \in \spH^r(\plane), \qquad \forall r \geq 0.
\end{equation}
and that as a function $\reals^{2n} \to \spH^r(\plane)$,
$g$ varies smoothly. We note that in the plane, the function 
\begin{equation}
e^{v_c} = \frac{\abs{x - c}^2}{1 + \abs{x - c}^2},
\end{equation}
and all of its derivatives are bounded, and that $e^{v_{\vb p}}$ and
$e^{v_{\vb q}}$ have no common zeros if $\vb{p}$ and $\vb{q}$ have 
no common elements. In the compact case, it is known that for fixed $y$, 
$G(x,y)$ has a singularity at $y$, however, locally in any open disk 
$\disk_r(y)$ of smaller radius than the injectivity radius, $G(x,y)$ has the 
asymptotic expansion,
\begin{align}\label{eq:green-logd-smoothpart}
    G(x, y) = \frac{1}{2\pi}\,\log \brk(d(x, y)) + \tilde{G}(x,y),
    \qquad \forall x \in \disk_r(y), 
\end{align}
where $d(x,y)$ is the Riemannian distance and $\tilde G(x, y)$ is a smooth 
function defined on the disk. Hence $e^{v_c}$ is also
smooth and well defined on $\surface$.

In any case, $F(u + v + \cmu) \in \spH^r(\surface)$ for any 
$u \in \spH^r(\surface)$. Let 
$\Delta\surface = \surface^{k_+}\times\surface^{k_-}\setminus \Delta_{k_+,k_-}$, 
then the function,
\begin{align}
\Delta\surface \times \spH^r(\surface) \to \spH^r(\surface), &&
(\vb p, \vb q, u) \mapsto F(u + v + \cmu) - g,
\end{align}
is smooth. Therefore, the operator
\begin{equation}
\opT\, u = \laplacian u + F(u + v + \cmu) - g, \qquad u \in
\spH^{r + 2}(\surface),  
\end{equation}
is a well defined, smooth operator $\Delta\surface \times \spH^{r +
  2}(\surface) \to \spH^r(\surface)$. If $\tilde \fnh$ is a solution
to 
the regularised Taubes equation, then $\del_{\tilde h}\opT:
\spH^{r+2} \to \spH^r$ is the operator,
\begin{equation}
\brk(\del_{\tilde \fnh}\opT) \delta u 
= (\laplacian + \fnV(\tilde\fnh + v + \cmu))\, \delta u,
\end{equation}
where as a function of $\surface$,
\begin{equation}
  V(x) = \fnV\brk(\tilde \fnh + v + \cmu) \in C^r
\end{equation}
is a positive function whose zero set is $\vset \cup \avset$ and 
if $\surface$ is $\plane$, has
the property that $\lim_{\abs x \to \infty} \fnV(x) = (1 - \tau^2)$.
  By lemma~\ref{lem:h2-l2-iso}, $\del_{\tilde h}\opT$ is an
  isomorphism, by the implicit  
function theorem, the mapping $(\vb p, \vb q) \mapsto \tilde \fnh$ is
smooth as a map $\Delta \surface \to \spH^r(\surface)$. By Sobolev's embedding,
it is also smooth as a map 
$\Delta\surface \to C^{r -2}(\surface)$ for all $r \geq 2$. Hence, it 
depends smoothly on $(\vb p, \vb q)$. 
Finally, since the solution to the Taubes equation is 
$u = \tilde \fnh + v + \cmu$, we have that for any
neighbourhood $W$ of $\vset\cup\avset$, the restriction $u:
\surface\setminus \conj W \to \reals$ depends smoothly on the cores.
\end{proof} 

\begin{corollary}
    Let $U$ be either  an open and dense subset of the compact surface 
    $\surface$ or the euclidean plane. In any holomorphic chart $\varphi: U \to 
    \cpx$ containing the cores, the localization formula~\eqref{eq:loc-formula} 
    can be extended continuously to the coincidence set.
\end{corollary}

\begin{proof}
%
For any given core $p_j \in U$, let $z_j = \varphi(p_j) \in \cpx$. 
We assume each $p_j$ is simple and that all the $z_j$ are contained in a 
bounded domain $D \subset \cpx$. This assumption is superfluous for the
 Euclidean 
plane but for a compact surface is necessary for the existence of a 
smooth function $H: D \times D \to \reals$ such that for any $z, w \in D$,
\begin{align}
    G(\varphi^{-1}(z), \varphi^{-1}(w)) = 
    \frac{1}{2\pi}\,\log\,\abs{z - w} + H(z, w).
\end{align}
%
Assume without loss of generality $\sgn_1 = \sgn_2$, to prove the result it is 
enough to show that for any pair of indices $i,j \in \set{1, \ldots, k_+ + 
k_-}$, $\lim_{z_1 \to z_2} \partial_{z_j}b_i(\vb z)$ exists, 
where $\vb{z} = (z_1, \ldots, z_{k_+ + k_-})$.
%
 In the following computation, we denote by $h_\varphi(z) = 
 h(\varphi^{-1}(z))$, 
 $G_\varphi(z, w) = G(\varphi^{-1}(z), \varphi^{-1}(w))$, $\tilde h_\varphi(z) 
 = \tilde h(\varphi^{-1}(z))$ the local representation of the  
 functions, 
\begin{align}
b_i &= 2\,\conj\del_{z = z_i} \brk(
s_i\,\fnh_\varphi(z) - \log\abs{z - z_i}^2)\nonumber\\
  &= 2\,\conj\del_{z = z_i}\brk(4\pi\,G_\varphi(z, z_i) - \log\,\abs{z - 
  z_i}^2)
   + 8\pi\,\sum_{k \neq i} \sgn_i\sgn_k\conj\del_z G_\varphi(z_i, z_k)
   + 2\sgn_i\,\conj\del_z\tilde\fnh(z_i,\vb z)\nonumber\\
  &= 8\pi\,\conj\del_z H_\varphi(z_i, z_i) 
  + 8\pi\sum_{k \neq i} \sgn_i\sgn_k\conj\del_z G_\varphi(z_i, z_k)
   + 2\sgn_i\,\conj\del_z\tilde\fnh_\varphi(z_i,\vb z),
\end{align}
where $\conj\del_z$ refers to the derivative with respect to the first 
variable in each term. Hence,
\begin{align}
    \del_{z_j}b_i &= 8\pi\,\del_{z_j}\conj\del_z H_\varphi(z_i, z_i) 
  + 8\pi\,\sum_{k \neq i} \sgn_i\sgn_k\del_{z_j}\conj\del_z G_\varphi(z_i, z_k)
   + 2\sgn_i\,\del_{z_j}\conj\del_z\tilde\fnh_\varphi(z_i,\vb z).
\end{align}

The functions $\del_{z_j}\conj\del_z H_\varphi(z_i, z_i)$ and 
$\del_{z_j}\conj\del_z\tilde\fnh_\varphi(z_i,\vb z)$ vary continuously with 
$(z_i,\vb z)$, whence, the limits,
\begin{align}
\lim_{z_1 \to z_2} \del_{z_j}\conj\del_z H_\varphi(z_i, z_i), &&
\lim_{z_1 \to z_2} \del_{z_j}\tilde\fnh_\varphi(z_i,\vb z),
\end{align} 
both exist. In the above sum, if $i \neq 1$ and $k \neq 1$, or if 
either $i = 1$ or $k = 1$ 
and the other is not index 2, the limit
\begin{align}
\lim_{z_1 \to z_2} \del_{z_j}\conj\del_z G_\varphi(z_i, z_k)  
\end{align}
exists because $G$ is smooth away of the diagonal set of $\surface \times 
\surface$. Finally, if $\set{i, k} = \set{1, 
2}$, we can assume without loss of generality $i = 1$, $k = 2$, to compute,
\begin{align}
\lim_{z_1 \to z_2} \del_{z_j}\conj\del_z G_\varphi(z_1, z_2)   
&= \lim_{z_1 \to z_2} \del_{z_j}\conj\del_{z = z_1} \brk(
 \frac{1}{2\pi}\,\log\,\abs{z - z_2} + H_\varphi(z, z_2))\nonumber\\
 &= \lim_{z_1 \to z_2} \del_{z_j}\conj\del_z H_\varphi(z_1, z_2)\nonumber\\
 &= \del_{z_j}\conj\del_z H_\varphi(z_2, z_2).
\end{align}

Therefore, $\lim_{z_1 \to z_2}\del_{z_j} b_i(\vb z)$ exists, implying the 
localization formula can be extended to the coincidence set. 
\end{proof}

In later applications we will also focus on vortices of the
Ginzburg-Landau functional, in this case, the governing elliptic
problem is the orginal Taubes equation,
\begin{align}
  -\laplacian h = e^h - 1 + 4\pi\, \sum_{i} \delta_{p_i}.
\end{align}

If $\surface$ is the euclidean plane, we add the condition $\lim_{\abs
x \to \infty} h = 0$.  In both cases, we know that
there exists a solution  $h$ to the Taubes equation for any configuration
$\vb p$ of points. On the plane this is proved in
\cite{taubes1980} whereas in a compact surface
$\surface$ the proof can be found in \cite{yang_system_2000}. 
As for the $O(3)$ Sigma model, given a configuration $\vb p = (p_1, \ldots,
p_n)$ of cores, if we define $\tilde h$ such that $h = \tilde h + v_{\vb
  p}$, then $\tilde h$ is the unique solution of the regularized
Taubes equation for the Ginzburg-Landau functional, 
\begin{align}\label{eq:reg-taubes-gl}
  -\laplacian \tilde h = e^{\tilde h + v_{\vb p}} -  1 - g_{\vb p},
\end{align}
where the functions $v_{\vb p}$, $g_{\vb p}$ are defined either as in
equation~\eqref{eq:v-g-def} if $\surface$ is the euclidean plane or
$v_{\vb p}$ is defined as in equation~\eqref{eq:v-g-def} and $g_{\vb
  p}$ is the constant function $4\pi n {\abs \surface}^{-1}$ if
$\surface$ is compact. 

Mimicking the proof of theorem~\ref{thm:h-reg-param-dep}, we prove the
following proposition.

\begin{proposition}\label{prop:gl-reg-taubes-param-dep}
  Let $\vb p \in \surface^n$ be a sequence of points on a Riemann 
  surface $\surface$ either compact or the euclidean plane. If
  $\surface$ is compact assume 
  $n \in \mathbb{Z}^+$ satisfies Bradlow's bound for vortices of the
  Ginzburg-Landau functional,
  \begin{align}
    4\pi\,n \leq \abs \surface.
  \end{align}

  Let $\tilde h : \surface \times \surface^n \to \reals$ be such that
  $\tilde h(x; \vb p)$ is the unique solution
  to equation~\eqref{eq:reg-taubes-gl} with data $\vb p$, then $\tilde h$
  is a smooth function of $x$ and the data.
\end{proposition}

\begin{proof}
 As in the proof of theorem~\ref{thm:h-reg-param-dep}, we define an
 operator $\opT: \surface^n \times \mathrm{H}^{r + 2} \to \mathrm{H}^r$,
 such that,
 \begin{align}
  \opT\,(\vb p, u) = \laplacian u + e^{u + v_{\vb p}} - 1 - g_{\vb p},
\end{align}
and observe that as in  the proof of the theorem, this operator is
smooth. Moreover, the derivative $\del_u\opT: \mathrm{H}^{r+2} \to
\mathrm{H}^r$ at $(\vb p, \tilde h)$ is,
\begin{align}
  \del_u\opT(\delta u) = (\laplacian + e^{\tilde h + v_{\vb p}})\,\delta u. 
\end{align}

We notice the potential $V(x) = e^{\tilde h + v_{\vb p}}$ and all the
derivatives are bounded functions. If $\surface$ is compact, this is
because $V$ is smooth and if the surface is $\plane$, this is
becuase $e^{v_{\vb p}}$ has this property, as shown in the proof of
theorem~\ref{thm:h-reg-param-dep} and because $h$, the solution
to the Taubes equation, and all of their derivatives decay exponentially
as $\abs x \to \infty$. Hence, $\tilde h = h - v_{\vb p}$ and the
derivatives are continuous bounded functions. By lemma 
\ref{lem:h2-l2-iso}, $\del_u\opT$ is an isomorphism. By the implicit
function theorem, for any $\vb p \in \surface^n$ and any $r \geq 2$,
there is a neighbourhood $U \subset \surface^n$ of $\vb p$, such that
the map $\vb p \mapsto \tilde h(x;\vb p)$ is smooth as a function $U
\to \mathrm{H}^{r+2}$. By Sobolev's embedding, $\tilde h(x)$ is of class
$C^r$ as a function $\surface \times U \to \reals$. Since
differentiability is a local property, this implies $\tilde h$ is of
class $C^r$ on $\surface \times \surface^n$ for any $r \geq
0$. Therefore $\tilde h$ is smooth.
\end{proof}

\section{Topological methods}\label{sec:top-meth}

We finalize this chapter with a brief exposition of some results of analysis 
that we will use to prove existence of solutions to the
elliptic problem on compact manifolds in chapters~\ref{c:vav-compact}
and~\ref{ch:cs-moduli}. Both methods are attributable to Leray and
Schauder. Our exposition will be short and will focus on the results
we need. Details can be found in the books~\cite{chipot2011}
and~\cite{deimling2010}. Recall a subset of a topological space is
precompact if the closure is compact.

\begin{definition}
  Let $X, Y$ be Banach spaces and $\Omega \subset X$. A continuous map $T:
  \Omega \to Y$ is compact if it maps bounded subsets of $\Omega$ to
  precompact subsets of $Y$.
\end{definition}
\noindent  
As a caveat, in~\cite{deimling2010} compact operators are called completely 
continuous.

\begin{theoremdef}
  Let $\Omega \subset X$ be an open and bounded subset of a real Banach
  space, $T: \Omega \to X$ compact and $y \not\in (I -
  T)(\del\Omega)$. For each admisible triple $(T, \Omega, y)$, there is a 
  unique integer
  $\deg(I - T, \Omega, y) \in \mathbb{Z}$, with the following
  properties:
  \begin{enumerate}
  \item $\deg(I, \Omega, y) = 1$ for $y \in \Omega$. 
  \item $\deg(I - T, \Omega, y) = \deg(I - T, \Omega_1, y) + \deg(I -
    T, \Omega_2, y)$ whenever $\Omega_1$, $\Omega_2$ are disjoint open
    subsets of $\Omega$ such that $y \not\in (I - T)(\conj\Omega
    \setminus (\Omega_1 \cup \Omega_2))$.
  \item Homotopy invariance: $\deg(I - H(t,\cdot), \Omega, y(t))$ is
    independent of $t \in 
    [0, 1]$ whenever $H: [0, 1] \times \conj\Omega \to X$ is compact,
    $y: [0, 1] \to X$ is continuous and $y(t) \not\in (I - H(t,
    \cdot))(\del \Omega)$ on $[0, 1]$.
    \item General homotopy invariance: Let $\Theta \subset [0, 1]
      \times X$ be bounded and open in $[0, 1] \times X$ with
      $\Theta_t = \set{x \in X : (t, x) \in \Theta}$. If $H:
      \conj\Theta \to X$ is compact and $y: [0, 1] \to X$ is
      continuous with $y(t) \not\in (I - H(t, \cdot))(\del \Theta_t)$
      for all $t \in [0, 1]$, then $\deg(I - H(t, \cdot), \Theta_t,
      y(t))$ is independent of t.
  \end{enumerate}
\end{theoremdef}

$\deg(I - T, \Omega, y)$ is Leray-Schauder's
degree. It can be proved~\cite[Thm. 8.2]{deimling2010} that $\deg(I
- T, \Omega, y) \neq 0$ implies $(I - T)^{-1}(y) \neq \emptyset$. As
an application of this concept, there is the following
result of Sch\"afer~\cite{schaefer1955methode},

\begin{theorem}\label{thm:fixed-point-alternative}
  Let $T: X \to X$ be compact. Then the following alternative holds:
  \begin{enumerate}
  \item $x - \lambda\,T(x) = 0$ has a solution for every $\lambda \in
    [0, 1]$, or
  \item $S = \set{x \in X : \exists \lambda \in [0, 1]\,s.t.\, x -
      \lambda \,T(x) = 0}$ is unbounded.
  \end{enumerate}
\end{theorem}
\noindent 
For a linear operator, alternative 1 always holds by choosing the  
solution $x = 0$, however, 
for non-linear operators this is not always the case. A proof of the theorem 
can be found in~\cite[Cor. I.1.18]{chipot2011}. In
general, computing the degree is a difficult task. Suppose $x_0 \in (I
- T)^{-1}(y)$ isolated, then $x_0$ is the only solution of the
equation $x - T(x) = y$ in some disk $\disk_{\epsilon_0}(x_0)$. By
homotopy invariance, $\deg(I - T, \disk_{\epsilon}(x_0), y)$ is
independent of $\epsilon$ for $0 < \epsilon < \epsilon_0$.

\begin{definition}
  With the previous assumptions, the index of an isolated solution
  $x_0$ to the equation $x - T(x) = y$ is
  \begin{align}
    \mathrm{ind}(I - T, x_0, y) = \deg(I - T, \disk_{\epsilon}(x_0), y),
  \end{align}
  where $\epsilon > 0$ is any sufficiently small radius. 
\end{definition}

If $T$ is compact and differentiable at $x_0$, then
$T'(x_0)$ is  a compact linear operator. We state the following theorems,

\begin{theorem}[Leray-Schauder]\label{thm:leray-schauder-index}
  If $T: \Omega \subset X \to X$ is compact and differentiable at
  $x_0$ and if $I - T'(x_0)$ is injective, then $\mathrm{ind}(I - T,
  x_0, y) = \pm 1$. More precisely,
  \begin{align}
    \mathrm{ind}(I - T, x_0, y) &= \mathrm{ind}(I - T'(x_0), x_0, y)\nonumber
    \\
    &= (-1)^{\beta}, \qquad \beta = \sum_{\lambda > 1} m(\lambda).
  \end{align}

  The sum is taken over all eigenvalues $\lambda > 1$ of $T'(x_0)$ and
  $m(\lambda)$ is the algebraic multiplicity of $\lambda$. 
\end{theorem}

\begin{definition}\label{def:cont-lambda-unif-balls}
    An operator $H(\lambda, x)$, $H: \reals \times X \to X$,
    is continuous in $\lambda$ uniformly with respect to $x$ in balls in $X$
    if for any given ball $B \subset X$ and for any $\epsilon > 0$, 
    there is a $\delta > 0$ such that   
    $|\lambda_2 - \lambda_1| < \delta$ implies  $|H(\lambda_2, x) - 
    H(\lambda_1,x)| < \epsilon$ for all $x \in B$.
\end{definition}
The following theorem can
be found in~\cite[Thm. I.3.3]{chipot2011}. 
\begin{theorem}\label{thm:ls-deg-continuum}
  Let $H: \reals \times X \to X$ be such that for all $\lambda \in \reals$
  the map $H(\lambda, \cdot): X \to X$ is compact and $H(\lambda, x)$
  is continuous in $\lambda$ uniformly with respect to $x$ in balls in
  $X$~(definition~\ref{def:cont-lambda-unif-balls}). Let $(\lambda_0, x_0)$ be 
  a solution of the equation
  \begin{align}\label{eq:implicit-equation}
    x - H(\lambda, x) = 0.
  \end{align}
  
  Suppose $\mathcal{U} \subset X$ is an open, bounded set such that
  $x_0 \in \mathcal{U}$ and,
  \begin{enumerate}
  \item for fixed $\lambda_0$ there is no other solution in $\overline{\mathcal{U}}$,
    
  \item $\deg(I - H(\lambda_0,\cdot), \mathcal{U}, 0) \neq 0$.
  \end{enumerate}

    Then there exist two connected and closed sets (=continua)
    $\mathcal{C}^+ \subset [\lambda_0,\infty) \times X$ and
    $\mathcal{C}^- \subset (-\infty, \lambda_0] \times X$ of solutions
    of~\eqref{eq:implicit-equation} with $(\lambda_0, x_0) \in
    \mathcal{C}^+ \cap\mathcal{C}^-$. For $\mathcal{C}^+$ one of the
    following two alternatives hold:
    \begin{enumerate}
    \item $\mathcal{C}^+$ is unbounded or,
    \item $\mathcal{C}^+ \cap (\set{\lambda_0} \times (X\setminus \overline{\mathcal{U}})) \neq \emptyset$. 
    \end{enumerate}

    The same alternatives hold for $\mathcal{C}^-$.
\end{theorem}

The hypotesis on $H$ implies the restriction to bounded subsets of
$\reals \times X$ is compact. The definition of compact
operator on the reference is slightly different, however, it is not 
difficult to go through the proof and adapt it to our current
definition.

We conclude the section mentioning that several results related to 
theorem~\ref{thm:ls-deg-continuum} can be found in 
the literature. A good survey of related applications is~\cite{Mawhin1999}.

\let \sign      \undefined
\let \Bform     \undefined
\let \fnF       \undefined
\let \fnV       \undefined
\let \cmu       \undefined
\let \fnh       \undefined
\let \spW       \undefined
\let \spL       \undefined
\let \spH       \undefined
\let \spS       \undefined
\let \spHilbert \undefined
\let \avg       \undefined
\let \manifold  \undefined
\let \wto       \undefined 
\let \opT       \undefined
\let \cf        \undefined


\let\confSp      \undefined
\let\gaugeGp     \undefined
\let\hf          \undefined
\let\connSp      \undefined
\let\fieldSp     \undefined
\let\bigTangent  \undefined
\let\gp          \undefined
\let\vertSp      \undefined
\let\Energy      \undefined
\let\stgp        \undefined
\let\Diff        \undefined
\let\mf          \undefined
\let\ef          \undefined
\let\vset        \undefined
\let\avset       \undefined
\let\moduliSp    \undefined
\let\pair        \undefined
\let\Lagrangian  \undefined           
\let\vb          \undefined
\let\dvr         \undefined
\let\north       \undefined
\let\south       \undefined
\let\cDiff       \undefined
\let\eform       \undefined
\let\Bform       \undefined
\let\Jop         \undefined
\let\vol         \undefined
\let\fnu         \undefined
\let\sgn         \undefined
\let\kmetric     \undefined

\newcommand*{\modulic}{\moduli_0^{1,1}}
\newcommand*{\qty}{\brk}
\newcommand*{\lagrangian}{\mathrm{L}}
\newcommand*{\bog}{Bogomolny}
\newcommand*{\va}[1]{\mathbf{#1}}
\newcommand*{\vu}[1]{\mathbf{#1}}
\newcommand*{\dd}{d}
\newcommand*{\lgrdensity}{\mathcal{L}}
\newcommand*{\delzbar}{\conj \del_x}


\chapter{
  Asymmetric vortex-antivortex systems in the euclidean plane
}\label{c:vav-euclidean}

\newcommand{\hf}{\phi} 
\newcommand{\gp}{A}
\newcommand{\eucGroup}{\mathbb{E}(2)}
\newcommand{\fnh}{h}
\newcommand{\vset}{P}
\newcommand{\avset}{Q}
\newcommand{\vb}{\mathbf}
\newcommand{\sgn}{s}
\newcommand{\kmetric}{\mathrm{K}}
\newcommand{\cf}{\Lambda}
\newcommand{\kform}{\omega}
\newcommand{\Energy}{\mathrm{E}}

In this chapter we study the moduli space of vortex-antivortex pairs
on the euclidean plane in detail. Our approach will be analytical and 
numerical. To understand the geometry of the moduli space, we need to
analyse the properties of the Taubes equation in the critical case when a
vortex and an antivortex collide.

In section~\ref{sec:moduli-vort-antiv}, we study the space of vortex-antivortex 
pairs, the main result is that it is incomplete. To prove this theorem, we find 
bounds for $h_\epsilon$, the solution to the Taubes equation in 
several lemmas in subsection~\ref{sec:eps-0}.

In section~\ref{s:scattering-theory}, we develop an asymptotic approximation 
for 
the $\Lsp^2$ metric of vortex-antivortex pairs in the centre of 
mass frame, and complement it with the point source formalism in 
subsection~\ref{sec:point-source-formalism}, in which we approximate the 
fields linearising the field equations. The main result is 
the Lagrangian~\eqref{eq:lag-point-source-formalism} which confirms the 
asymptotic formula obtained previously. In 
subsection~\ref{sec:short-range-approx}, we find another 
asymptotic aproximation for the metric, this time for small $\epsilon$, 
the main result is equation~\eqref{eq:lambda-eps}. 

In section~\ref{sec:num-approx-metric} we approximate the $\Lsp^2$ metric 
numerically, using the data found by numerical methods to study 
the scattering of vortex-antivortex pairs in 
subsection~\ref{sec:vav-scattering} and in this way testing our 
approximations of the previous section.

Finally, in section~\ref{sec:ricci-magnetic-geod-motion} we study 
Ricci magnetic geodesics. These curves are of mathematical interest, 
there are a few results about the relation between extensibility of them and 
completeness of 
the underlying space, as the moduli space of vortex-antivortex pairs 
is incomplete, the question of whether or not it is complete in the Ricci 
magnetic sense is interesting in its own.

In order to start, let us note that on $\plane$ any fibre 
bundle is trivial, therefore we can 
consider sections on the target manifold as pairs $(\hf, \gp)$ of
a function $\hf: \plane \to \sphere$ and a 1-form 
$\gp \in \Omega^1(\plane)$. Since the Lagrangian is isometrically invariant,
by Noether's theorem there will be
conserved currents. In the Euclidean case, the conserved
quantities are the total energy, \(\Energy\), the linear 
 and angular momenta. We already know how to compute the energy.
 For the remaining constants of motion note that 
 the Laplacian is invariant under the action of the group of 
 isometries of the plane, 
 $\eucGroup \cong \plane \rtimes \mathbb{O}(2)$, which is a Lie 
 group of dimension three. If $\fnh(x; \vb p, \vb q): \plane\setminus 
 (\vset \cup \avset) \to \reals$ is the solution to the Taubes equation,
  and $\gamma \in \eucGroup$ acts in $\vb p$ and $\vb q$
  component-wise, this implies,
  \begin{align}
      \fnh(x; \gamma \vb p, \gamma \vb q) = 
      \fnh(\gamma^{-1}x; \vb p, \vb q).
  \end{align}

\begin{lemma}\label{lem:isom-inv-coefs}
  Let $\vb c = (c_1, \ldots, c_{k_+ + k_-})$ be a sequence of 
  cores, ordered such that the first $k_+$ are the vortices. Let 
  $b_j$ be the coefficients defined in lemma~\ref{lem:coef-sym}. 
  If $\gamma x = \alpha x + \beta$, $\alpha, x, \beta \in \plane$, 
  $\abs{\alpha} = 1$ ,
  is an orientation preserving isometry, then
  \begin{align}
      b_j(\gamma \vb c) = \alpha\,b_j(\vb c).
  \end{align}
  
  If $\gamma$ is the orientation reversing generator, 
  $\gamma x = \conj x$, we have,
  \begin{align}
      b_j(\conj{\vb c}) = \conj b_j(\vb c).
  \end{align}
\end{lemma}

\begin{proof}
  After some algebraic manipulation and the chain rule,
  \begin{align}
      b_j(\gamma\vb c) &= 2\,\conj\del_{x = \gamma c_j} 
      \brk(\sgn_j\fnh(x; \gamma \vb c) - \log\,\abs{x - \gamma 
      c_j}^2)\nonumber\\
      &= 2\,\conj\del_{x = \gamma c_j}  \brk(
       \sgn_j\fnh(\gamma^{-1}x; \vb c) - \log\,\abs{\gamma^{-1} x 
       - c_j}^2)\nonumber\\
       &= 2\,\del_{x =c_j}  \brk(
       \sgn_j\fnh(x; \vb c) - \log\,\abs{x - c_j}^2)
       \conj\del_{x = \gamma c} (\gamma^{-1}x)\nonumber\\
       & \quad + 2\,\conj\del_{x =c_j}  \brk(
       \sgn_j\fnh(x; \vb c) - \log\,\abs{x - c_j}^2)
       \conj\del_{x = \gamma c}\conj{(\gamma^{-1}x)}\nonumber\\
       &= 2\,\conj\del_{x =c_j}  \brk(
       \sgn_j\fnh(x; \vb c) - \log\,\abs{x - c_j}^2)
       \frac{1}{\conj \alpha}\nonumber\\
       &= \alpha \,b_j(\vb c),
  \end{align}
  where we used the fact that $\gamma^{-1}x$ is holomorphic and $\alpha$
  unitary to simplify the result of the chain rule. 
  For the second identity we proceed analogously,
  \begin{align}
      b_j(\conj{\vb c}) &= 2\,\conj\del_{x = \conj c_j} \brk(
      \sgn_j \fnh(x; \conj{\vb c}) - \log\,\abs{x - \conj c_j}^2)\nonumber\\
      &= 2\,\conj\del_{x = \conj c_j} \brk(
      \sgn_j\fnh(\conj x; \vb c) - \log\,\abs{\conj x - c_j}^2)\nonumber\\
      &= 2\,\del_{x = c_j} \brk(
      \sgn_j\fnh(x; \vb c) - \log\,\abs{ x - c_j}^2)\nonumber\\
      &= \conj b_j(\vb c).
  \end{align}
\end{proof}

As a consequence of the lemma, the coefficients $b_j$ are translation
invariants: if $X$ is the Killing field generated by a one parameter
family of isometries $\gamma_s$, then 
\begin{align}
\mathcal{L}_X b_j(\vb c) = \dot a_0\,b_j(\vb c),
\end{align}
where $\mathcal{L}_X$ denotes the Lie derivative. 
If $\gamma_s x = x + sb$ is a one parameter family of translations, then 
$\dot a_0 = 0$, on the other hand, 
\begin{align}
X = \sum_k \brk(b\,\del_{c_k} + \conj b\, \conj\del_{c_k}),
\end{align}

Letting $b = 1$ and $b = i$, we obtain,
\begin{align}
\sum_k\brk(\del_{c_k} + \conj\del_{c_k}) b_j(\vb c) = 0, &&
\sum_k\brk(\del_{c_k} - \conj\del_{c_k}) b_j(\vb c) = 0. 
\end{align}

Hence, $\sum_k\del_{c_k}\,b_j = \sum_k\conj\del_{c_k} b_j = 0$. Applying 
the symmetries of the coefficients,
\begin{gather}
\conj\del_{c_j} \,\sum_k b_k = \sum_k \conj\del_{c_j} b_k 
= \sum_k \conj\del_{c_k} b_j = 0,\\
\del_{c_j} \,\sum_k b_k =\sum_k \del_{c_j} b_k 
= \sum_k \conj\del_{c_k} \conj b_j = 
\conj{\sum_k \del_{c_k} b_j} = 0.
\end{gather}

Therefore, $\sum_k b_k$ is constant. Repeating this argument with the 
one-parameter family of rotations $\gamma_s\,x = e^{s\,i}\,x$, we find that 
$X = i\sum_k \brk(c_k\,\del_{c_k} - \conj c_k\,\conj \del_{c_k})$, thence,
\begin{align}
\sum_k \brk(c_k\,\del_{c_k} - \conj c_k\,\conj \del_{c_k})\,b_j = b_j.
\end{align}

Summing over $j$, we find,
\begin{align}
\sum_j b_j = \sum_j \sum_k \brk(c_k\,\del_{c_k} - \conj c_k\,\conj 
\del_{c_k})\,b_j
 = \sum_k \brk(c_k\,\del_{c_k} - \conj c_k\,\conj \del_{c_k})\cdot\sum_j b_j = 
 0,
\end{align}
since $\sum_j b_j$ is constant. This result is analogous to the similar 
result obtained by Samols for vortices of the Ginzburg-Landau functional 
in~\cite{samols1992}. As a consequence of this symmetry we have the following 
proposition about conservation of momentum.


\begin{proposition}
  The total conserved momentum of a vortex-antivortex system with cores at 
  position $\vb c$ is,
  \begin{align}
      P_1 + P_2\,i = 2\pi\,\sum_j (1 - \sgn_j\tau)\,\dot c_j,
  \end{align}
  where $s_j = s_{c_j}$ is the sign function determining the type of the 
  core.
\end{proposition}

\begin{proof}
  By lemma~\ref{lem:isom-inv-coefs} the translation group acts
  isometrically on the moduli space. Hence for any $b \in \cpx$ the fields,
  \(X = \sum_k(b \,\del_{c_k} + \conj b\,
  \conj\del_{c_k})\) are Killing fields and the 
  product
  \begin{align}
    P_b = \lproduct{\dot{\vb c}, X}
  \end{align}
  is constant along geodesic trajectories and corresponds to the
  projection of  
  momentum on the $b$ direction. If $\kmetric$ denotes the K\"ahler 
  metric, equation~\eqref{eq:kmetric}, 
  \begin{align}
      P_b &= \Re \brk(\kmetric(\dot{\vb c}, X))\nonumber \\
      &= \half \brk(
      \kmetric(\dot{\vb c}, X) + \conj{\kmetric(\dot{\vb c}, X)})\nonumber\\
      &= \half \brk(
      \kmetric(\dot{\vb c}, X) +\kmetric(X, \dot{\vb c}))\nonumber\\
      &= \half \sum_{i,j}\kmetric_{i \conj j}
      \,\brk(\dot c_i \conj X_j + X_i\conj{\dot c_j}).
  \end{align}
  
  On the other hand, 
  $\kmetric_{i \conj j} = 2\pi \brk( (
  1 - \sgn_i\tau) \delta_{ij} + \del_{c_i}b_j)$. Note that by the invariance of 
  the coefficients $b_j$, we have,
  \begin{align}
    \sum_{i,j}
      \del_{c_i} b_j\,\brk(\dot c_i \conj X_j + X_i\conj{\dot c_j}) &=  
      \conj b\,\sum_i \dot c_i\del_{c_i} \brk(\sum_j b_j)
      + b\sum_j \brk(\dot c_j \sum_i \del_{c_i} b_j) = 0.
  \end{align}
  
  Hence,
  \begin{align}
      P_b = \pi \sum_i (1 - \sgn_i\tau) 
      (\dot c_i\,\conj b + b \, \conj{\dot c_i}).
  \end{align}
  
  If we let $b = 1$ and $b = i$ we get the momentum in the direction of the 
  real an imaginary axes are the real and imaginary parts of the vector,
  \begin{align}
  2\pi\,\sum_j (1 - \sgn_j\tau)\,\dot c_j.
  \end{align}
\end{proof}

\begin{proposition}
  The angular momentum of a vortex-antivortex system is,
\begin{equation}
\ell = \kform(\vb c, \dot{\vb c}),
\end{equation}
where $\kform \in \Omega^2(\plane)$ is the K\"ahler form of the metric.
\end{proposition}


\begin{proof}
  Conservation of angular momentum corresponds to the action of
\(SO(2)\) on the moduli space. Let 
$X = i\,\sum_k \brk(c_k\del_{c_k} - \conj c_k \,\conj\del_{c_k})$ be the 
Killing field generating the action of $SO(2)$, the
conserved angular momentum is, 
\begin{align}
  \ell &= \lproduct{\dot{\vb c}, X}\nonumber\\
  &= \half \brk(\kmetric(\dot{\vb c}, X) + \kmetric(X, \dot{\vb c}))\nonumber\\
    &= \pi i\,\sum_{j,k} \brk(
      \brk( (1 - \sgn_j)\,\delta_{jk} + \del_{c_j} b_k ) 
      ( -\dot c_j \, \conj c_k + c_j \, \conj{\dot c_k}))\nonumber\\
      &= \kform(\vb c, \dot{\vb c}).
\end{align}
\end{proof}

It is convenient to express the dynamics of vortex-antivortex systems in the 
centre of mass frame. Let us define,
\begin{align}
  C &= \frac{1}{(1 - \tau) k^+ + (1 + \tau) k^{-}} 
  \sum_j (1 - \sgn_j \tau) c_j,\\
  M &= 2\pi\,(1 - \tau) k^+ + 2\pi\,(1 + \tau) k^{-},
\end{align}

$M$ is the total mass and $C$ the centre of mass of the system as determined 
by conservation of momentum and energy. Let us define the variables  
\(\xi_j \in\plane \) such that, 
\begin{align}
c_j = C + \xi_j.
\end{align}

Let \(m_j = 2 \pi (1 - \sgn_j \tau) \) be the effective mass of a core, then
 \(\sum_j m_j\,\xi_j = 0 \). Note that this linear combination is invariant 
 under the action of $S_{k_+}\times S_{k_-}$ on the moduli space, where $S_n$ 
 is the symmetric group of order $n$, hence it determines a well defined 
 subspace  $\moduli_0^{k_+,k_-} \subset
 \moduli^{k_+,k_-}$ where $C = 0$.

\begin{proposition}\label{prop:metric-cm-decomp}
  Let $\kmetric_0$ be the restriction of the K\"ahler metric 
  to $\moduli_0^{k_+,k_-}$, then,
\[
  \kmetric = M \abs{dC}^2 + \kmetric_0.
\]
\end{proposition}

\begin{proof}
This is a consequence of translation invariance, let 
$m_i = 2\pi\,(1 - \sgn_i\,\tau)$ 
be the mass of the core at $c_i$, 
\begin{align}
\kmetric &= \sum_i m_i\,\abs{dC + d\xi_i}^2 
+ 2\pi\sum_{i,j} \del_{c_i} b_j\,(dC + d\xi_i)\,\conj{(dC + d\xi_j)}.
\end{align}

The first terms can be split into
\begin{equation}
M\abs{dC}^2 + 2\Re\,\brk(\conj{dC}\,\sum_i m_i d\xi_i) + \sum_i m_i 
\abs{d\xi_i}^2
= M\abs{dC}^2 + \sum_i m_i \abs{d\xi_i}^2,
\end{equation}
and the second terms can be split as,
\begin{multline}
2\pi\brk( \sum_j (\abs{dC}^2 + \conj{d\xi}_j dC)\cdot \sum_i \,\del_{c_i} b_j
 + \conj{dC}\,\sum_i d\xi_i\cdot \sum_j \del_{c_i} b_j
 + \sum_{i, j} \del_{c_i} b_j \conj{d\xi}_j d\xi_{c_i}) \\
 = 2\pi\,\sum_{i, j} \del_{c_i} b_j\, \conj{d\xi}_j d\xi_{i},
\end{multline}
where the first two terms cancelled because the coefficients $b_j$ are
 translation invariant. 
Substituting back into the formula for \(\kmetric\) we conclude the
claim of the proposition.
\end{proof}

As a consequence, the moduli space decomposes in a product of K\"ahler
manifolds,  
\begin{equation}
\moduli^{k_+, k_-} \cong \plane \times \moduli_0^{k_+, k_-},
\end{equation}
such that the metric splits in a trivial flat metric in $\plane$ and 
the nontrivial restriction to $\moduli_0^{k_+,k_-}$. This
splitting was first observed by Samols \cite{samols1992} for vortices in
the Abelian Higgs model. \(\moduli_0^{k_+,k_-} \) is the space of
vortices and antivortices with fixed centre of mass. Given the decomposition of 
the metric in the moduli space, the energy and angular momentum in the centre 
of 
mass frame are,
\begin{align}
     E &= \half \kmetric(\dot\xi, \dot\xi),\\
  \ell &= \kform(\xi, \dot{\xi}),
\end{align}
where $\kmetric$ and $\kform$ are the K\"ahler metric and K\"ahler form of
 ambient space at $\xi = (\xi_1, \ldots, \xi_{k_+ + k_-})$.

\let  \Energy   \undefined
\let  \hf       \undefined 
\let  \gp       \undefined 
\let  \eucGroup \undefined
\let  \fnh      \undefined
\let  \sgn      \undefined
\let  \kmetric  \undefined
\let  \cf       \undefined
\let  \kform    \undefined 


\newcommand{\hf}{\phi} 
\newcommand{\gp}{A}
\newcommand{\fnh}{h}
\newcommand{\sgn}{s}
\newcommand{\kmetric}{\mathrm{K}}
\newcommand{\cf}{\Omega}
\newcommand{\kform}{\omega}
\newcommand{\gmetric}{g}

\newcommand*{\fnv}{v}
\newcommand*{\fng}{g}
\newcommand*{\fnV}{V}
\newcommand*{\cmu}{\mu}
\newcommand*{\fnu}{u}
\newcommand*{\spL}{\mathrm L}
\newcommand*{\spH}{\mathrm H}
\newcommand*{\spW}{\mathrm W}
\newcommand*{\fnF}{F}
\newcommand*{\domain}{\mathcal D}
\newcommand*{\fnFmu}{\fnF_\cmu}
\newcommand*{\fnVmu}{\fnV_\cmu}
\newcommand*{\fnth}{\tilde\fnh}
\newcommand*{\vol}{\mathrm{Vol}}
\newcommand*{\avg}{\overline}
\newcommand*{\energyDens}{\mathcal{E}}

\section{The moduli space of vortex-antivortex pairs}
\label{sec:moduli-vort-antiv}

In this section we focus on the moduli space of 
vortex-antivortex pairs on Euclidean space and extend the analysis done by 
Rom\~ao-Speight in~\cite{romao2018} for $\tau = 0$. 
We focus on the non trivial part of the metric in the submanifold  
$\moduli_0^{1,1} 
\cong \plane\setminus\set{0}$, of 
pairs with centre of mass at the origin.
Let
\begin{align}
  b(x) &= b_1(x, -x), \qquad x \in \reals^+.
\end{align}

By the invariance of the coefficient $b_1$ with respect to conjugation, 
$b$ is a real function. Let us assume $c_1$ is the vortex position, 
introducing $(\epsilon, \theta)$ 
coordinates such that $c_1 - c_2 = 2\epsilon e^{i\theta}$, we have,
\begin{align}
b_1(c_1, c_2) &= e^{i\theta}b(\epsilon).
\end{align}

Recall $b_1 + b_2 = 0$ and $\del_1 b_1 + \del_2 b_1 = 0$, then the restriction 
of the metric to $\moduli_0^{1,1}$ is,
\begin{align}
  g_0 = \cf(\epsilon)\,(d\epsilon^2 + \epsilon^2d\theta^2),
\end{align}
where the conformal factor is,
\begin{align}
  \cf(\epsilon) = 2\pi \brk(
2 (1 - \tau^2) + \frac{1}{\epsilon}\,\dv{\epsilon}\brk(\epsilon b(\epsilon))).
\end{align}

\subsection{The singularity at $\epsilon = 0$}\label{sec:eps-0}

In this section we study the limiting behaviour of solutions
to the Taubes equation for vortex-antivortex pairs as $\epsilon\to 0$. We 
aim to prove bounds for $h_\epsilon$ in order to estimate the 
length of radial geodesics and finalize proving that the moduli space 
of vortex-antivortex pairs is incomplete. We start 
defining the following constant and 
functions,
\begin{align}
\mu &= \frac{1 - \tau}{1 + \tau}, &
    \fnFmu(t) &= 2\,\frac{e^t - 1}{\cmu\,e^t + 1}, &
    \fnVmu(t) &= \frac{2(\cmu + 1)\,e^t}{(\cmu\,e^t + 1)^2}.
\end{align}

If $h_T(x,\epsilon, -\epsilon)$ is the solution to the Taubes 
equation with a vortex at position $(\epsilon, 0)$ and an antivortex at 
$(-\epsilon, 0)$, let us define the function $\fnh_\epsilon$ such that 
$h_T = \fnh_\epsilon + \mu$. To express the Taubes equation in a convenient 
way, 
we make the change of variable, 
\begin{align}
x' =
(1 - \tau^2)^{-1/2}\,x,
\end{align}
under this change of variable, the position of a vortex or antivortex 
is $(\pm\epsilon',0) = (\pm(1 - \tau^2)^{-1/2}\,\epsilon, 0)$. By an abuse 
of notation, we still denote by $x$ coordinates in the rescaled 
Euclidean plane and by $(\pm\epsilon, 0)$ the positions of the cores. 
With these definitions, the Taubes equation is equivalent to,
\begin{align}
  -\laplacian \fnh_{\epsilon} = \fnFmu(\fnh_\epsilon) + 
  4\pi\delta_{\epsilon} - 4\pi\delta_{-\epsilon},
\end{align}
together with the constraint,
\begin{align}
\lim_{\abs x \to \infty} \fnh_\epsilon = 0.
\end{align}

Let $\fnu$ be the solution to the Taubes equation for the Ginzburg-Landau
functional \cite{taubes1980},
\begin{align}\label{eq:taubes-gl}
-\laplacian \fnu = e^u - 1 + 4\pi\delta_0,
\end{align}
Yang proves in~\cite{yang_strings_1999} that $u < 0$. For the following 
results, 
we will assume $\tau \in [0, 1)$, the case $\tau < 0$ being similar.
Repeating the argument of Yang, the function 
$u_{\epsilon}(x) = u(x - \epsilon)$ is a sub-solution of
$\fnh_{\epsilon}$, i.e. $u_\epsilon < 0$ and 
\begin{align}
-\laplacian u_\epsilon \geq \fnFmu(u_\epsilon),\qquad x \in \plane\setminus 
\set{\epsilon, -\epsilon}.
\end{align}

 On the other hand, the function 
$-u_{-\epsilon} = -\fnu(x + \epsilon)$ is a super-solution: it is positive and 
\begin{align}
-\laplacian (-u_{-\epsilon}) \leq \fnFmu(-u_{-\epsilon}),\qquad
x \in \plane\setminus \set{\epsilon, -\epsilon}.
\end{align}

 By the maximum principle,
\begin{align}
u_{\epsilon}(x) < h_{\epsilon}(x) < -u_{-\epsilon}(x), \qquad x \in
  \plane \setminus \set{\epsilon, -\epsilon}.
\end{align}

\begin{lemma}\label{lem:vav-euc-u-du-estimates}
For any $\delta \in (0, 1)$ there exist two constants $C(\delta)$ and 
$R(\delta)$ such that 
\begin{align}
    \abs{u(x)} \leq C\,e^{-(1 - \delta)\,\abs{x}}, &&
    \abs{\grad u(x)} \leq C\,e^{-(1 - \delta)\,\abs{x}}, &&
    \abs{x} > R.
\end{align}
In particular, $\norm{u}_{\spL^p} < \infty$ for any $p > 0$.
\end{lemma}

\begin{proof}
That $u$ and its derivatives decay exponentially fast at infinity
can be found in the literature, for example 
in~\cite{jaffe1980monopoles,taubes1980}, here we adapt a proof 
of Yang for solutions of the elliptic problem of the $O(3)$ Sigma 
model in the symmetric case~\cite[Lemma 8.3]{yang1999strings}. Since 
$\lim_{\abs{x} \to \infty} u = 0$, we linearise~\eqref{eq:taubes-gl} 
about  $u = 0$ in a neighbourhood of infinity to obtain,
\begin{align}\label{eq:linear-taubes-at-infty}
    -\laplacian u = f(x)\,u, \qquad \abs{x} \geq R.
\end{align}
where $f(x)$ is a function such that $f(x) \to 1$ as $\abs{x} \to \infty$. 
 Let us introduce the comparison function 
 \begin{align}
     w(x) = C\,e^{-(1 - \delta)\,\abs{x}}, \qquad \abs{x} \geq R,
 \end{align}
 where $C(\delta)$ and $R(\delta)$ are positive constants yet to be determined. 
  The Laplacian of this function is 
 $-\laplacian w = (1 - \delta)\brk(1 - \delta - \abs{x}^{-1})w$. 
 Choosing $R$ sufficiently large, we can guarantee that 
 \begin{align}
     f(x) > (1 - \delta)\brk(1 - \delta - \frac{1}{\abs x}),
 \end{align}
 hence,
 \begin{align}
     -\laplacian (u - w) > f(x) (u - w), 
 \end{align}
 for $\abs{x} > R$. Let us choose $C$ big enough for the continuous function 
 $u - w$ to be negative at the boundary $\abs{x} = R$. Since $u(x) - w(x) \to 0$
  as $\abs{x} \to \infty$, by the maximum principle $u(x) < w(x)$ for all 
  $\abs{x} \geq R$. Since \eqref{eq:linear-taubes-at-infty} is linear, we 
  can apply the same argument to $-u$. Choosing the bigger of each pair of 
  constants $(C, R)$ the decay rate of $u$ is proved. 
  
  For the decay rate of $\grad u$, we know that $u \in \spH^r$ for all $r
  \geq 2$  
  \cite{taubes1980}, 
  in particular, $\grad u \to 0$ as $\abs{x} \to \infty$. Linearising
  in a neighbourhood  
  of infinity, $\grad u$ is a solution to the equation,
  \begin{align}
      -\laplacian\,(\grad\fnu) = f(x)\,\grad \fnu,
  \end{align}
  for some function $f(x)$ such that $f(x) \to 1$ as $\abs{x} \to
  0$. We can apply the same argument as before to obtain the
  exponential decay estimate of $\grad\fnu$. To prove the assertion
  about the $\spL^p$ norm of $u$, note that  
  $\abs{u}^p$ also decays exponentially fast at infinity for any $p > 0$ and 
  since the singularity at $x = 0$ is logarithmic and $\lim_{\abs{x} \to 0} \,\abs{x}\,(\log\,\abs{x})^p = 0$, the integral 
  \begin{align}
      \int_{\plane} \abs u^p \, dx 
      = \int_{0}^{2\pi}\int_0^\infty \abs u^p \, r\,dr\,d\theta 
  \end{align}
  is convergent.
 \end{proof}

For any $R > 0$ and $\epsilon_0 > 0$, if $\abs 
{x} > R$ and $\epsilon < 
\epsilon_0$, by the triangle inequality  $\abs{x \pm \epsilon} > R - 
\epsilon_0$. As a consequence of this observation and 
Lemma~\ref{lem:vav-euc-u-du-estimates}, we have the following corollary,

\begin{corollary}\label{lem:vav-euc-ueps-bound}
    For any $\delta \in (0, 1)$ and $\epsilon_0 > 0$, there exists constants 
    $C(\delta)$, and $R(\delta, \epsilon_0)$, such that if 
    $\epsilon < \epsilon_0$ and $\abs x > R$,
    \begin{align}
    \abs{u_{\pm \epsilon}(x)} \leq C\,e^{-(1- \delta)\abs x}, &&
    \abs{\nabla u_{\pm \epsilon}(x)} \leq C\,e^{-(1- \delta)\abs x}.
    \end{align}
\end{corollary}

We also have the following uniform bounds, valid for any $p > 0$,

\begin{align}
\norm{\fnh_{\epsilon}}_{\spL^p} 
&\leq \norm{ -u_{-\epsilon} - u_{\epsilon}}_{\spL^p}\nonumber\\
&\leq \norm{u_{-\epsilon}}_{\spL^p} +  \norm{u_{\epsilon}}_{\spL^p}\\
  &= 2\,\norm{u}_{\spL^p}.
\end{align}

Let us introduce the functions,
\begin{align}
  v(t) &= -\log\,(1 + t^{-2}), &
  g(t) &= \frac{4}{(1 + t^2)^2}, &
  t > 0.
\end{align}
and let $v_{\epsilon} = v(\abs{x - \epsilon}) 
- v(\abs{x + \epsilon})$, 
$g_{\epsilon} = g(\abs{x - \epsilon}) - g(\abs{x + \epsilon})$. 
We have the norm estimates,
\begin{align}
  \norm{g_{\epsilon}}_{\spL^p} &\leq 2\,\norm{g(\abs{x})}_{\spL^p},\\
  \norm{v_{\epsilon}}_{\spL^p} &\leq 2\,\norm{v(\abs{x})}_{\spL^p}.
\end{align}

Each of the functions $\abs{\fnv_\epsilon}^p$, $\abs{\fng_\epsilon}^p$ 
is pointwise convergent to zero. Therefore,
\begin{align}
    \lim_{\epsilon\to 0}\, \norm{g_\epsilon}_{\spL^p} &= 0, & p > \half,\\
    \lim_{\epsilon\to 0}\, \norm{v_\epsilon}_{\spL^p} &= 0, & p > 1.
\end{align}

Let us define 
$\fnth_{\epsilon} = \fnh_{\epsilon} - v_{\epsilon}$.
Then $\fnth_{\epsilon}$ is a solution to the regularised Taubes equation,
\begin{align}
  -\laplacian \tilde\fnh_{\epsilon} =
  \fnFmu (v_{\epsilon} + \tilde\fnh_{\epsilon}) - g_{\epsilon}.
\end{align}

From now onwards, we will use the same variable $C$ to denote a positive 
constant, independent of $\epsilon$,  
that can change from one inequality to the following. By our estimates for 
the $p$ norm of $\fnh_\epsilon$ and $\fnv_\epsilon$, 
$\fnth_{\epsilon}$ is uniformly bounded in $\spL^p$ for $p > 1$.

\begin{lemma}\label{lem:vav-plane-htilde-h1-bound}
    Let $\epsilon_0 > 0$ be an arbitrary positive constant, $||\tilde 
    h_\epsilon||_{\spH^1} \leq C$ for $\epsilon 
    < \epsilon_0$.
\end{lemma}

\begin{proof}
    Since $\tilde \fnh_\epsilon$ is uniformly bounded on $\spL^2$, we aim to 
    show that $\norm{\nabla \tilde \fnh_\epsilon}_{\spL^2}$ is also bounded
    if $\epsilon < \epsilon_0$. We have,
\begin{align}
\norm{\grad \tilde\fnh_{\epsilon}}_{\spL^2}^2 
&= -\lproduct{\fnF_{\cmu}(v_{\epsilon}
    + \tilde\fnh_{\epsilon}),
    \tilde\fnh_{\epsilon}} +
\lproduct{\fng_{\epsilon},
    \tilde\fnh_{\epsilon}},\nonumber\\
&= -\lproduct{\fnF_{\cmu}(\fnh_\epsilon),
    \tilde\fnh_{\epsilon}} +
\lproduct{\fng_{\epsilon},
    \tilde\fnh_{\epsilon}},\nonumber\\
&\leq \abs{\lproduct{\fnF_{\cmu}(\fnh_\epsilon),
        \tilde\fnh_{\epsilon}}} + 
        \norm{\fng_\epsilon}_{\spL^2}
        \,\norm{\tilde \fnh_\epsilon}_{\spL^2}\nonumber\\
    &\leq \abs{\lproduct{\fnF_{\cmu}(\fnh_\epsilon),
            \tilde\fnh_{\epsilon}}} + C,
\end{align}   
where $\lproduct{\cdot,\cdot}$ is the $\spL^2$ product. It remains to show 
$\lproduct{\fnFmu(\fnh_\epsilon), \tilde \fnh_\epsilon}$ is uniformly bounded. 
Let $\delta \in (0 , 
1)$ be any given number, by corollary~\ref{lem:vav-euc-ueps-bound}, there are 
positive constants $R$, $C$ such that if $\abs x > R$ and $\epsilon < 
\epsilon_0$, 
\begin{align}
\abs {\fnh_\epsilon(x)} &\leq \abs{u(x - \epsilon) - u(x + \epsilon)}\\
&\leq C\,e^{-(1-\delta)\abs{x}}.
\end{align}

Hence, there is another constant, such that, 
\begin{align}
\fnFmu(\fnh_{\epsilon}) \leq C\,e^{-(1-\delta)\abs{x}}, \qquad \abs x \geq R.
\end{align}

Let $U$ be the exterior of the disk $\disk_R(0)$, by the previous bound,
\begin{align}
\abs{\lproduct{\fnF_{\cmu}(\fnh_\epsilon),
        \tilde\fnh_{\epsilon}}}_{\spL^2(U)}
   \leq \norm{\fnF_{\cmu}(\fnh_\epsilon)}_{\spL^2(U)}\,
    \norm{\tilde\fnh_{\epsilon}}_{\spL^2(U)}
    \leq C\,\norm{e^{-(1-\delta)\abs{x}}}_{\spL^2(U)},
\end{align}
since $\tilde \fnh_\epsilon$ is uniformly bounded on $\spL^2$.
%
On the other hand, 
$\fnFmu$ is a bounded function, hence, 
\begin{align}
\abs{\lproduct{\fnF_{\cmu}(\fnh_\epsilon),
        \tilde\fnh_{\epsilon}}}_{\spL^2(\disk_R(0))}
\leq \norm{\fnF_{\cmu}(\fnh_\epsilon)}_{\spL^2(\disk_R(0))}\,
\norm{\tilde\fnh_{\epsilon}}_{\spL^2(\disk_R(0))}
\leq C.
\end{align}

This concludes the proof that  $\lproduct{\fnFmu(\fnh_\epsilon), \tilde 
\fnh_\epsilon}$ is bounded on $\spL^2$. 
\end{proof}

\begin{proposition}\label{prop:tilde-h-eps-conv-0}
    $\lim_{\epsilon \to 0} ||\tilde h_\epsilon||_{\spL^2} = 0$.
\end{proposition}

\begin{proof}
Let $\tilde\fnh_n = \fnth_{\epsilon_n}$
be any sequence such that $\epsilon_n\to 0$. By 
lemma~\ref{lem:vav-plane-htilde-h1-bound} $\set{\tilde \fnh_n}$ is bounded on 
$\spH^1$, hence,  
by the Banach-Alaoglu theorem, after passing to a subsequence if necessary, 
there is a function
$\tilde\fnh_{*} \in \spH^1$ such that 
$\tilde\fnh_n \rightharpoonup \tilde\fnh_{*}$ weakly on 
$\spH^1$ and by the 
Rellich-Kondrashov theorem, after passing to another subsequence if 
necessary, 
we can 
assume that for any bounded
domain $\domain$, $\tilde\fnh_n
\to \tilde\fnh_{*}$ strongly on $\spL^2(\domain)$. 
We will assume without further notice that domains are bounded 
and their boundaries have at least Lipschitz regularity. 
Let $\varphi \in C_0^1(\plane)$ and let $\domain$ be a domain containing the 
support of $\varphi$, 

\begin{align}
  \lproduct{\varphi, \tilde\fnh_{*}}_{\spH^1} 
  &=  \lim\,\lproduct{\varphi, \tilde\fnh_n}_{\spH^1}\nonumber\\
  &= \lim\, \lproduct{\grad\varphi, \grad\tilde\fnh_n}_{\spL^2} + \lim\,
    \lproduct{\varphi, \tilde\fnh_n}_{\spL^2}\nonumber\\
  &= \lim\,\lproduct{\varphi, \laplacian\tilde\fnh_n}_{\spL^2} + \lim\,
    \lproduct{\varphi, \tilde\fnh_n}_{\spL^2}\nonumber\\
  &= -\lim\,\lproduct{\varphi,
\fnF_{\cmu}(v_n + \tilde\fnh_n) - g_n
    }_{\spL^2} + 
    \lproduct{\varphi, \tilde\fnh_{*}}_{\spL^2}.
\end{align}

The last equation because the convergence $\tilde\fnh_n\to
\tilde\fnh_{*}$ is strong on bounded domains
and $\varphi$ is compactly supported. Consequently,
\begin{align}
  \lproduct{\grad \varphi, \grad \tilde\fnh_{*}}_{\spL^2} =
  -\lim\,\lproduct{\varphi, \fnF_{\cmu}(v_n + \tilde\fnh_n) - g_n
    }_{\spL^2}.
\end{align}

By the mean value theorem, we have the estimate,
\begin{align}
\norm{\fnF_{\cmu}(v_n + \tilde\fnh_n) - \fnF_{\cmu}(
  \tilde\fnh_{*})}_{\spL^2(\domain)}
  &\leq C\,\brk(
  \norm{v_n}_{\spL^2(\domain)} 
  + \norm{\tilde\fnh_n - \tilde\fnh_{*}}_{\spL^2(\domain)}
  ).
\end{align}

Therefore, $\fnF_{\cmu}(v_n + \tilde\fnh_n) \to \fnF_{\cmu}(\tilde\fnh_{*})$ 
and $g_n \to 0$ strongly on $\spL^2(\domain)$, thence 
$\tilde\fnh_{*}$ is a weak solution of the equation,
\begin{align}
-\laplacian\tilde\fnh_{*} = \fnF_{\cmu}(\tilde\fnh_{*}).
\end{align}

By elliptic regularity $\tilde\fnh_{*}$ is a strong solution
and by the maximum principle $\tilde\fnh_{*} = 0$. 

Let $\domain$ be any other domain, our previous argument shows that any 
sequence
$\fnth_{n}$ has a 
convergent subsequence 
$\fnth_{n_j} \to 0$ on $\spL^2(\domain)$. Therefore for any domain
$\domain$, $\lim_{\epsilon \to 0} \norm{\fnth_\epsilon}_{\spL^2(\domain)}
 = 0$. Now we prove that 
 $\lim_{\epsilon \to 0} \norm{\fnth}_{\spL^2} = 0$, to this 
 end, let $\rho > 0$ and let us take $R > 0$ such that 
 \begin{align}
     \norm{u}_{\spL^2(\plane\setminus \disk_R(0))} < \frac{\rho}{2}.
 \end{align}
 Let $\epsilon_0$ be small enough such that $\abs{x \pm \epsilon} > R$ 
 for all $\epsilon < \epsilon_0$ and $\abs{x} > 2R$. 
  In this situation we have,
  \begin{align}
      \norm{\fnth_\epsilon}_{\spL^2(\plane\setminus \disk_{2R}(0))} 
      &< \norm{u_\epsilon}_{\spL^2(\plane\setminus \disk_{2R}(0))}
      + \norm{u_{-\epsilon}}_{\spL^2(\plane\setminus \disk_{2R}(0))}\nonumber\\
      &\leq 2\norm{u}_{\spL^2(\plane\setminus \disk_{R}(0))}\nonumber\\
      &< \rho.
  \end{align}
  
  On the other hand, there exists $\epsilon_1$ such that if $\epsilon < \epsilon_1$, 
  then,
  \begin{align}
    \norm{\fnth_\epsilon}_{\spL^2(\disk_{2R}(0))} < \rho,
  \end{align}
  taking $\epsilon' = \min(\epsilon_0, \epsilon_1)$, we conclude that 
  \begin{align}
    \norm{\fnth_\epsilon}_{\spL^2} < 2\rho,  \qquad \forall \epsilon < \epsilon'
  \end{align}
  and the limit $\lim_{\epsilon \to 0} \norm{\fnth_\epsilon}_{\spL^2} = 0$ 
  holds. 
\end{proof}  

Since $\fnth_\epsilon \to 0$ strongly as $\epsilon \to 0$, by the mean
value theorem as in the proof of the proposition, we have, 
 \begin{align}
     \lim_{\epsilon\to 0}\, \norm{\fnFmu(\fnv_\epsilon +
       \fnth_\epsilon)}_{\spL^2} 
      = 0.
 \end{align}

Moreover,
\begin{align}
    \norm{\laplacian \fnth_\epsilon}_{\spL^2} &\leq 
    \norm{\fnFmu( \fnv_\epsilon + \fnth_\epsilon)}_{\spL^2} 
    + \norm{\fng_\epsilon}_{\spL^2},
\end{align}
since both terms on the right side of the inequality converge 
to $0$, we have the limit 
\begin{align}
 \lim_{\epsilon\to 0}\,\norm{\laplacian\fnth_\epsilon}_{\spL^2} = 0.
\end{align}

\begin{lemma}\label{lem:vav-euc-fnth-uniform-conv-domains}
Let $\domain$ be any domain on the plane, the restrictions 
$\fnth_\epsilon|_\domain$ and  $\grad\fnth_\epsilon|_\domain$ converge 
uniformly to 0.
\end{lemma}

\begin{proof}
If we take any pair of domains $\domain \Subset \domain'$, by 
Schauder's estimates,
\begin{align}\label{eq:schauder-bound-domain-fnth}
    \norm{\fnth_\epsilon}_{\spH^2(\domain)} \leq C (
    \norm{\laplacian \fnth_\epsilon}_{\spL^2(\domain')} 
    + \norm{\fnth_\epsilon}_{\spL^2(\domain')}),
\end{align}
which implies $\fnth_\epsilon \to 0$ in $\spH^2(\domain)$ as 
$\epsilon \to 0$. By 
Sobolev's embedding, we have that for any domain, 
$\lim_{\epsilon\to 0} \fnth_\epsilon = 0$ uniformly. Let 
$p > 2$ be any real number and let $\fnth_n$ be any sequence of functions 
such that $\epsilon_n \to 0$. Since the convergence is uniform on $\domain$, we 
can 
apply the dominated convergence theorem to obtain,
\begin{align}
    \norm{\tilde\fnh_n}_{\spL^p(\domain)} \to 0, &&
    \norm{\fnF(v_n + \tilde\fnh_n)}_{\spL^p(\domain)} \to 0
\end{align}
and since the sequence is arbitrary, we conclude the limits,
\begin{align}
    \lim_{\epsilon \to 0}\,\norm{\tilde\fnh_\epsilon}_{\spL^p(\domain)} = 0, &&
    \lim_{\epsilon \to 0}\,
    \norm{\fnF(v_\epsilon + \fnth_\epsilon)}_{\spL^p(\domain)} = 0,
\end{align}
are valid for any domain. In particular, both limits are valid for the
domain $\domain'$ of equation~\eqref{eq:schauder-bound-domain-fnth}. 
By Schauder's estimates 
$\norm{\fnth_\epsilon}_{\spW^{2,p}(\domain)} \to 0$ as $\epsilon \to 0$ 
and by Sobolev's embedding,
\begin{align}
    \lim_{\epsilon\to 0}\,\norm{\fnth_\epsilon}_{C^1(\domain)} = 0.
\end{align}
\end{proof}

\begin{proposition}\label{prop:vav-euc-fnth-unif-conv}
The convergence $\fnth_\epsilon \to 0$ is uniform on $\plane$.
\end{proposition}

\begin{proof}
Recall 
\begin{align}
 \abs{\fnth_\epsilon} 
 \leq 
 \abs{\fnh_\epsilon} + \abs{\fnv_\epsilon} 
 \leq \abs*{u(x - \epsilon) - u(x + \epsilon)}
 + \abs*{v(\abs{x - \epsilon}) - v(\abs{x + \epsilon})}.
\end{align}

Let $R > 0$ be any large positive constant, such that the estimates of 
lemma~\ref{lem:vav-euc-u-du-estimates} hold for $\delta = \half$. 
If $\abs{x} > 2R$ and 
$\epsilon < R$, then $\abs{x \pm \epsilon} > R$. We can apply the mean value 
theorem to obtain the estimate
\begin{align}
    \abs*{v(\abs{x - \epsilon}) - v(\abs{x + \epsilon})}
    &= 
    \abs*{\log\,(1 + \abs{x - \epsilon}^{-2}) - \log\,(1 + \abs{x + 
    \epsilon}^{-2})}\nonumber\\
    &\leq \frac{1}{R^2}\abs*{\abs{x - \epsilon}^{-2} - \abs{x + 
    \epsilon}^{-2}}\nonumber\\
    &= \frac{4\epsilon\,\abs{x_1}}{
    R^2\,\abs{x - \epsilon}^2\abs{x + \epsilon}^2
    }\nonumber\\
    &= \abs*{
    \frac{4\epsilon\,(x_1 + \epsilon)}{
    R^2\,\abs{x - \epsilon}^2\abs{x + \epsilon}^2}
    - \frac{4\epsilon^2}{
    R^2\,\abs{x - \epsilon}^2\abs{x + \epsilon}^2}
    }\nonumber\\
    &\leq \frac{4\epsilon}{R^5} + \frac{4\epsilon^2}{R^6}.
\end{align}

Likewise, there is some $\xi$ in the linear segment joining $x - \epsilon$ to 
$x + \epsilon$ such that,
\begin{align}
    \abs{u(x - \epsilon) - u(x + \epsilon)} &= 2\,\abs{
    \del_1\fnu(\xi)\,\epsilon
    } \leq 2\,Ce^{-\half\abs{\xi}}\,\epsilon
      \leq 2\,C\epsilon,
\end{align}
where we have used lemma~\ref{lem:vav-euc-u-du-estimates}. We conclude that 
$\fnth_\epsilon \to 0$ uniformly on $\plane\setminus\disk_R(0)$, but by  
lemma~\ref{lem:vav-euc-fnth-uniform-conv-domains}, $\fnth_\epsilon$ also 
converges uniformly on $\disk_R(0)$. 
\end{proof}

Recall Poincare's constant of a domain $\domain$ is the best constant 
$C_p(\domain)$ such that for any zero average function $u: \domain \to \reals$, 
\begin{align}
\norm{u}_{\spL^2(\domain)} \leq C_p\,\norm{\nabla u}_{\spL^2(\domain)}.
\end{align}

\begin{lemma}\label{lem:vav-euc-coerc-cnst}
Let $a: \plane \to [0, M)$ be a continuous function, 
such that
\begin{enumerate}
    \item For some convex domain $\domain$ with diameter $d < \pi/M$,  
    $\int_\domain a\,\vol > 0$,
    
    \item $a$ is positive on $\Omega = \plane \setminus \domain$.
\end{enumerate}

If $m = \inf_\Omega a > 0$, 
the bilinear form 
\begin{align}
    B: \spH^1\times\spH^1 \to \reals, &&
    B(u, v) = \lproduct{\grad u, \grad v}_{\spL^2} + \lproduct{u,\,a
      v}_{\spL^2}, 
\end{align}
is coercive with coercivity constant,
\begin{align}
    0 < \alpha < \min \brk(m, 1, \frac{\int_\domain a\,\vol}{\vol(\domain)}, 
    \frac{1 - \frac{M\,d}{\pi}}{1 + Cp(\domain)}, 
    \frac{1}{\vol(\domain)}\,\int_\domain a \brk (1 
    - \frac{a}{M})\,\vol).
\end{align}
\end{lemma}

\begin{proof}
We aim to prove the existence of a positive constant $\alpha$ such that 
for any $u \in \spH^1$,
\begin{align}
\norm{u}^2_{\spH^1} \, \alpha \leq B(u, u),
\end{align}
 
Let $\alpha_1 = \min(m, 1)$, in the exterior $\Omega$ of the given domain,
 \begin{align}
     \norm{u}^2_{\spH^1(\Omega)}\,\alpha_1 \leq 
     \norm{\grad u}^2_{\spL^2(\Omega)} + \lproduct{u,\, au}_{\spL^2(\Omega)}.
 \end{align}
 
 On the other hand, any $u \in \spH^1(\domain)$ can be decomposed as 
 $u_0 + \avg u$,
  where $u_0$ is of zero average on $\domain$ and $\avg u \in \reals$, 
  hence, 
  $u_0$ is orthogonal to $\avg u$ in $\spH^1(\domain)$. 
  Coercivity in $\domain$ is equivalent to find a positive constant 
  $\alpha_2$ such that, 
\begin{multline}
    \brk(
    \norm{u_0}_{\spH^1(\domain)}^2 + \avg u^2\,\vol(\domain))\alpha_2 
    \leq 
    \norm{\grad u_0}^2_{\spL^2(\domain)} 
    + \lproduct{a,\,u_0^2}_{\spL^2(\domain)}\\
    + 2\,\avg u\, \lproduct{a,\,u_0}_{\spL^2(\domain)}
    + \avg u^2\,\lproduct{a, 1}_{\spL^2(\domain)},
\end{multline}
or equivalently,
\begin{multline}
\brk(\lproduct{a, 1}_{\spL^2(\domain)} - \alpha_2\,\vol(\domain))\,\avg u^2
+ 2\,\lproduct{a, u_0}_{\spL^2(\domain)}\,\avg u 
+ (1 - \alpha_2)\,\norm{\grad u_0}^2_{\spL^2(\domain)} \\
+ \lproduct{a,\,u_0^2}_{\spL^2(\domain)} 
- \alpha_2\norm{u_0}^2_{\spL^2(\domain)} \geq 0.
\end{multline}

For this is quadratic inequality on $\avg u$ to hold regardless of $\avg u$, 
the leading coefficient with respect to $\avg u$ must be positive and 
the discriminant of the quadratic must be non-positive, from these two 
conditions 
we deduce the following restrictions:
\begin{gather}
  \alpha_2 <  \frac{\lproduct{a, 1}_{\spL^2(\domain)}}{\vol(\domain)}
  = \frac{\int_\domain\,a\,\vol }{\vol(\domain)}.
\end{gather}

\begin{multline}
  \lproduct{a, u_0}^2 _{\spL^2(\domain)}
    \leq 
    \brk(\lproduct{a, 1}_{\spL^2(\domain)} - \alpha_2\vol(\domain))\\
    \brk(
    (1 - \alpha_2)\,\norm{\grad u_0}^2_{\spL^2(\domain)}
+ \lproduct{a,\,u_0^2}_{\spL^2(\domain)} 
- \alpha_2\norm{u_0}^2_{\spL^2(\domain)}
).
\end{multline}

We claim the existence of a positive constant $\alpha_2$ such that the 
second restriction is independent of $\fnu_0$. To this end, let us divide this 
inequality by $M$,
\begin{multline}
\lproduct*{\frac{a}{M}, u_0}^2 _{\spL^2(\domain)}
\leq 
\brk(\lproduct*{\frac{a}{M}, 1}_{\spL^2(\domain)} - 
\frac{\alpha_2}{M}\,\vol(\domain))\\
\brk(
\frac{1 - \alpha_2}{M}\,\norm{\grad u_0}^2_{\spL^2(\domain)}
+ \lproduct*{\frac{a}{M},\,u_0^2}_{\spL^2(\domain)} 
- \frac{\alpha_2}{M}\norm{u_0}^2_{\spL^2(\domain)}
).
\end{multline}

By Cauchy-Schwarz,
\begin{align}
    \lproduct*{\frac{a}{M},\, u_0}^2 _{\spL^2(\domain)}
    \leq \norm*{\frac{a}{M}}^2_{\spL^2(\domain)}
    \,\norm*{u_0}^2_{\spL^2(\domain)}.
\end{align}

Notice that,
\begin{align}
    \norm*{\frac{a}{M}}^2_{\spL^2(\domain)} \leq 
    \lproduct*{\frac{a}{M}, 1}_{\spL^2(\domain)} 
    - \frac{\alpha_2}{M}\vol(\domain),
\end{align}
if and only if
\begin{align}
    \alpha_2 \leq \frac{1}{\vol(\domain)}\,\int_\domain a \brk (1 
    - \frac{a}{M})\,\vol.
\end{align}

On the other hand, the inequality 
\begin{align}
    \norm{u_0}^2_{\spL^2(\domain)}
    \leq 
    \brk(
    \frac{1 - \alpha_2}{M}\,\norm{\grad u_0}^2_{\spL^2(\domain)}
    + \lproduct*{\frac{a}{M},\,u_0^2}_{\spL^2(\domain)} 
    - \frac{\alpha_2}{M}\,\norm{u_0}^2_{\spL^2(\domain)}
    )
\end{align}
is equivalent to,
\begin{align}
    \norm{u_0}^2_{\spL^2(\domain)}
    - \lproduct*{\frac{a}{M},\,u_0^2}_{\spL^2(\domain)} 
    + \frac{\alpha_2}{M}\,\norm{u_0}^2_{\spL^2(\domain)}
    \leq 
    \frac{1 - \alpha_2}{M}\,\norm{\grad u_0}^2_{\spL^2(\domain)}.
\end{align}

By Poincare's inequality and the bound $0 \leq a/M < 1$,
\begin{align}
    \norm{u_0}^2_{\spL^2(\domain)}
    - \lproduct*{\frac{a}{M},\,u_0^2}_{\spL^2(\domain)} 
    + \frac{\alpha_2}{M}\,\norm{u_0}^2_{\spL^2(\domain)}
    &\leq 
    \brk(1 + \frac{\alpha_2}{M})\,\norm{u_0}^2_{\spL^2(\domain)}\nonumber\\
    &\leq 
    \brk(1 + \frac{\alpha_2}{M})C_p\,\norm{\grad u_0}^2_{\spL^2(\domain)}.
\end{align}

For the right side of this inequality to be lesser than 
$(1 - \alpha_2)/M$, we require,
\begin{align}
    \alpha_2 < \frac{1 - M\,C_p}{1 + C_p}.
\end{align}

Since $\domain$ is convex, we know by a result of Payne and 
Weinberger~\cite{Payne1960} that 
 $C_p \leq d/\pi$. Since $d/\pi < 1/M$ implies $MC_p < 1$ and 
\begin{align}
    \frac{1 - \frac{Md}{\pi}}{1 + C_p}
    \leq \frac{1 - M\,C_p}{1 + C_p},
\end{align}
 it is enough to require 
 $\alpha_2 < \brk(1 - \frac{Md}{\pi})(1 + C_p)^{-1}$ to 
 obtain the final bound. Defining $\alpha \leq \min(\alpha_1, \alpha_2)$,
 we prove coercivity with a constant as stated in the lemma. 
\end{proof}

\begin{lemma}\label{lem:vav-euc-fnv-fng-bounds}
For any $\epsilon_0 > 0$, there is a positive constants $C(\epsilon_0)$, such 
that for all $\epsilon \leq \epsilon_0$,
\begin{align}
    \norm {\fng_\epsilon}_{\spL^p} &\leq
    C\epsilon,
    & 
    p &> \frac{2}{5},\\
    \norm {\fnv_\epsilon}_{\spL^2} &\leq
    C\epsilon\,\abs{\log \epsilon}, \\
    \norm {\fnv_\epsilon}_{\spL^p} &\leq
    C\epsilon^{2/p}, & 
    p &> 1,\, p \neq 2,
\end{align}
\end{lemma}

\begin{proof}
Let us rewrite $\fng_\epsilon$,
\begin{align}
    \fng_\epsilon(x) &= \frac{4}{\brk(1 + \abs{x - \epsilon}^2 )^2}
    - \frac{4}{\brk(1 + \abs{x + \epsilon}^2 )^2}\nonumber\\
    &= \frac{4\brk(
    \abs{x + \epsilon}^2 - \abs{x - \epsilon}^2
    )
    \brk(
    2 + \abs{x + \epsilon}^2 + \abs{x - \epsilon}^2
    )
    }{
    \brk( 1 + \abs{x + \epsilon}^2 )^2\brk(1 + \abs{x - \epsilon}^2 )^2
    }\nonumber\\
    &=
    \frac{16\,\epsilon x_1 \brk(
    2 + \abs{x + \epsilon}^2 + \abs{x - \epsilon}^2
    )
    }{
    \brk( 1 + \abs{x + \epsilon}^2 )^2\brk(1 + \abs{x - \epsilon}^2 )^2
    }
\end{align}
and let us take $R > \epsilon_0$. If $\Omega = \plane\setminus\disk_R(0)$, 
\begin{align}
    \norm{\fng_\epsilon}_{\spL^p(\Omega)} 
    \leq 
    16\,\epsilon \norm*{
    \frac{ 
    x_1 \brk(
    2 + 2(\abs x + R)^2
    )
    }{
    \brk( 1 + (\abs x - R)^2 )^4
    }}_{\spL^p(\Omega)}.
\end{align}

The norm on the right decay as $\abs{x}^{-5}$ as $\abs x \to \infty$, 
hence is convergent for $p > 2/5$. On the other hand, we have,
\begin{align}
    \norm{\fng_\epsilon}_{\spL^p(\disk_R(0))} 
    \leq 
    16\,\epsilon \norm*{
    \,x_1 \brk(
    2 + 2(\abs x + R)^2
    )}_{\spL^p(\disk_R(0))}.
\end{align}

Thence $\norm{\fng_\epsilon}_{\spL^p} \leq C\epsilon$ 
if $\epsilon < R$. For $\fnv_\epsilon$ we follow several steps, 
dividing the plane in subregions where we can have control of the 
logarithmic singularities. We start with an algebraic rearrangement,
\begin{align}
    \fnv_\epsilon(x) &= \log\brk(
    \frac{1 + \abs{x + \epsilon}^{-2}}{1 + \abs{x - \epsilon}^{-2}}
    )\nonumber\\
    &= \log\brk( 1 + 
    \frac{\abs{x + \epsilon}^{-2} - \abs{x - \epsilon}^{-2}}
    {1 + \abs{x - \epsilon}^{-2}})\nonumber\\
    &= \log\brk( 1 - 
    \frac{4\epsilon x_1}
    {\abs{x + \epsilon}^2(1 + \abs{x - \epsilon}^2)}).
\end{align}

Let $R > 2\epsilon_0$ be a large positive constant such that 
if $\abs{x} \geq R$ and $\epsilon < R/2$, we have the approximation,
\begin{align}
    \abs{\fnv_\epsilon(x)} &=  
    \frac{4\epsilon \abs{x_1}}
    {\abs{x + \epsilon}^2(1 + \abs{x - \epsilon}^2)}
    + \order(\epsilon^2)\nonumber\\
    &\leq \frac{4\epsilon \abs{x_1}}
    {\brk(\abs{x} - \frac{R}{2})^2\brk (1 + (\abs{x} - \frac{R}{2})^2)}.
\end{align}
$\fnv_\epsilon$ is bounded in 
$\Omega = \plane \setminus \disk_R(0)$ by a function of order $\abs {x}^{-3}$,
 hence,
 \begin{align}
     \norm{\fnv_\epsilon}_{\spL^p(\Omega)} \leq 4\epsilon\,
     \norm*{
     \frac{x_1}
    {\brk(\abs{x} - \frac{R}{2})^2\brk (1 + (\abs{x} - \frac{R}{2})^2)}
     }_{\spL^p(\Omega)},
 \end{align}
 for any $p > 1$. On the other hand,
 \begin{multline}
     \norm{
     \fnv(\abs{x - \epsilon}) - \fnv(\abs{x + \epsilon})}_{\spL^p(\disk_R(0))}
     \leq 
     \norm*{
       \log(1 + \abs{x + \epsilon}^2) 
     - \log(1 + \abs{x - \epsilon}^2)
     }_{\spL^p(\disk_R(0))}
     \\+
     \norm*{
       \log(\abs{x - \epsilon}^2) 
     - \log(\abs{x + \epsilon}^2)
     }_{\spL^p(\disk_R(0))}.
\end{multline}

For the first term, the difference can be bounded as,
\begin{align}
    \norm*{
       \log(1 + \abs{x + \epsilon}^2) 
     - \log(1 + \abs{x - \epsilon}^2)
     }_{\spL^p(\disk_R(0))} 
     &\leq
    \norm*{
      \,\abs{x + \epsilon}^2 - \abs{x - \epsilon}^2
     }_{\spL^p(\disk_R(0))}\nonumber\\
     &\leq 
    4 \epsilon \norm*{
      x_1
     }_{\spL^p(\disk_R(0))}\nonumber\\
    &\leq 
    4 \epsilon R \brk(\pi\,R^2)^{1/p}.
\end{align}

For the second term, we proceed in two steps. Firstly, let 
us consider the annulus $2\epsilon \leq \abs{x} \leq R$ 
and note that,
\begin{align}
    \log\abs{x - \epsilon}^2 - \log\abs{x + \epsilon}^2
    = \log\abs*{1 - \frac{\epsilon}{x}}^2 
    - \log\abs*{1 + \frac{\epsilon}{x}}^2.
\end{align}

Let $A(R, 2\epsilon) = \disk_R(0)\setminus
\disk_{2\epsilon}(0)$ be the given annulus, with 
$A(1/(2\epsilon), 1/R)$ defined accordingly. We make the
change of variables $x' = 1/x$ and compute,
\begin{align}
    \norm*{
    \log\,\abs{x - \epsilon}^2 - \log\,\abs{x + \epsilon}^2
    }_{\spL^p(A(R, 2\epsilon))}
    &= 
    \norm*{
    \brk
    (\log\,\abs{1 - \epsilon x'}^2 - \log\,\abs{1 + \epsilon 
    x'}^2)
    \abs{x'}^{-2}
    }_{\spL^p(A(1/(2\epsilon), 1/R))}\nonumber\\
    &\leq
    2\,
    \norm*{
    \brk
    (\abs{1 - \epsilon x'}^2 - \abs{1 + \epsilon 
    x'}^2)
    \abs{x'}^{-2}
    }_{\spL^p(A(1/(2\epsilon), 1/R))}\nonumber\\
    &\leq 
    8\epsilon\,
    \norm*{\,
    \abs{x'}^{-1}
    }_{\spL^p(A(1/(2\epsilon), 1/R))}.
\end{align}

The last norm can be computed exactly, we found that,
\begin{align}
    \norm*{\,
    \abs{x'}^{-1}
    }_{\spL^p(A(1/(2\epsilon), 1/R))}
    = 
    \begin{cases}
    \sqrt{2\pi} 
    \brk(
    \log\brk(
    \frac{R}{2\epsilon}
    )
    )^{1/2}, & p = 2,\\
    \frac{4}{\abs{p - 2}^{1/p}}\,
    \abs*{\frac{\epsilon^2}{2^{p - 2}} - 
    \frac{\epsilon^p}{R^{p - 2}}}^{1/p},
    & p \neq 2.
    \end{cases}
\end{align}

Secondly, we use the inequality $\abs{x} \leq \abs{x \pm \epsilon} +
\epsilon$, which can be obtained by an application of the triangle 
inequality. With this inequality at hand,
\begin{align}
    \norm*{
       \log(\abs{x - \epsilon}^2) 
     - \log(\abs{x + \epsilon}^2)
     }_{\spL^p(\disk_{2\epsilon}(0))} 
     &\leq 
     \norm*{
       \log(\abs{x - \epsilon}^2) 
     }_{\spL^p(\disk_{2\epsilon}(0))} 
     + 
     \norm*{
      \log(\abs{x + \epsilon}^2)
     }_{\spL^p(\disk_{2\epsilon}(0))}\nonumber\\
     &\leq
     2 \norm*{
      \log\,\abs{x}^2
     }_{\spL^p(\disk_{3\epsilon}(0))}.
\end{align}

The last norm can also be computed,
\begin{align}
    \norm*{
      \log\,\abs{x}^2
     }_{\spL^p(\disk_{3\epsilon}(0))}
    =
    \begin{cases}
    6\sqrt{\pi}\,\epsilon\,\brk(
     \log^2(3\epsilon) - \log(3\epsilon) + \half
     )^{1/2}, & p = 2,\\
     \frac{\pi^{1/p}}{2}\,\brk(
     \int_{-2\log (3\epsilon)}^\infty u^p\,e^{-u}\,du)^{1/p},
     & p \neq 2.
    \end{cases}
\end{align}

In the last integral, $e^{-u}$ dominates $u^p$, hence 

\begin{align}
\norm*{
    \log\,\abs{x}^2
}_{\spL^p(\disk_{3\epsilon}(0))}
\leq C\,
\begin{cases}
\epsilon\,\abs{\log\epsilon} & p = 2,\\
\epsilon^{2/p}
& p \neq 2,
\end{cases}
\end{align}
where the constant is independent of $\epsilon$. 

%

Taking into account all the regions in which we divided the plane, we 
find that the dominant term is $\epsilon\,\abs{\log \epsilon}$ for $p = 2$ and 
$\epsilon^{2/p}$ in other case. This concludes the proof of the lemma.
\end{proof}

\begin{proposition}\label{prop:vav-euc-del1-fnth-estimate-loc}
For any domain neighbourhood $\domain$ of 
the origin, there is an $\epsilon_0 > 0$ such that,
\begin{align}
    \max_\domain\,\abs{\del_1\fnth_\epsilon(x)} \leq C 
    \epsilon^{2/p},
\end{align}
for all $\epsilon < \epsilon_0$.
\end{proposition}

\begin{proof}
We start defining a family of potentials $a_\epsilon$ which 
approximate $\fnVmu(\fnv_\epsilon + \fnth_\epsilon)$ as 
$\epsilon \to 0$. If $x \neq \pm\epsilon$, there is a 
$\xi_\epsilon(x)$ such that $\abs{\xi_\epsilon(x)} \leq \abs*{\fnv_\epsilon(x) + \fnth_\epsilon(x)}$ 
and,
\begin{align}
    \fnFmu(\fnv_\epsilon(x) + \fnth_\epsilon(x))
    = \fnVmu(\xi_\epsilon(x))
    (\fnv_\epsilon(x) + \fnth_\epsilon(x)).
\end{align}

Let $a_\epsilon = \fnVmu(\xi_\epsilon)$, this is a positive 
function such that if $\fnv_\epsilon(x) + \fnth_\epsilon(x)\neq 0$,
\begin{align}
    a_\epsilon(x) = 
    \frac{
    \fnFmu(\fnv_\epsilon(x) + \fnth_\epsilon(x))
    }{
    \fnv_\epsilon(x) + \fnth_\epsilon(x)
    },
\end{align}
hence $a_\epsilon$ is continuous in the complement of the zeros 
of $\fnv_\epsilon + \fnth_\epsilon$. Moreover, if $x_0$ is 
in the set of zeros of $\fnv_\epsilon + \fnth_\epsilon$, 
\begin{align}
    \lim_{x \to x_0} a_\epsilon(x) = 
    \lim_{x \to x_0}
    \frac{
    \fnFmu(\fnv_\epsilon(x) + \fnth_\epsilon(x))
    }{
    \fnv_\epsilon(x) + \fnth_\epsilon(x)
    }
    = \fnVmu(0)
    = a_\epsilon(x_0),
\end{align}
since $\xi_\epsilon(x_0) = 0$ because $\xi_\epsilon$ is bounded by 
$\abs{\fnv_\epsilon + \fnth_\epsilon}$ and 
$\fnv_\epsilon + \fnth_\epsilon \to 0$ as $x \to x_0$.

Hence,
 $a_\epsilon$ is a continuous function on 
 $\plane\setminus\set{\pm\epsilon}$ which we can extend continuously  to 
 $\pm\epsilon$, because $\fnFmu$ and $\fnth_\epsilon$
  are bounded functions and $\fnv_\epsilon$ diverges to $\pm\infty$ 
  at the poles $\pm\epsilon$, hence, 
  $\lim_{x\to \pm\epsilon}a_\epsilon(x) = 0$. Redefining $a_\epsilon$ as this 
  extension, notice that it determines a family of bounded non-negative, 
  continuous functions, each of them with only two zeros at 
  the vortex-antivortex positions. Let $\domain'$ be a convex 
  domain neighbourhood of the origin, with diameter  $d < \pi/M$ 
  for some strict upper bound $M$ of $\fnVmu$. Pointwise, each 
  $\xi_\epsilon(x) \to 0$ as $\epsilon\to 0$, hence we also have 
  the convergence $a_\epsilon(x) \to 2(\cmu + 1)^{-1}$ as 
  $\epsilon \to 0$. By the dominated convergence theorem, 
  \begin{gather}
      \int_{\domain'}
      a_\epsilon\vol \to \frac{2}{1 + \cmu}\abs{\domain'}, \\
      \int_{\domain'}
      a_\epsilon\,\brk(
      1 - \frac{a_\epsilon}{M}
      )\vol \to 
      \frac{2}{1 + \cmu}\,\brk(
      1 - \frac{2}{M(1 + \cmu)})
      \abs{\domain'}.
  \end{gather}

Let $\Omega = \plane \setminus \domain'$, 
$m_\epsilon = \inf_\Omega a_\epsilon$ and let us 
assume $\epsilon_0$ is small enough for $\pm\epsilon \in \domain'$ provided 
$\epsilon \leq \epsilon_0$. 
 We know that $\fnv_\epsilon + \fnth_\epsilon \to 0$ uniformly 
 in $\Omega$, hence,
\begin{align}
    \lim_{\epsilon \to 0} m_\epsilon 
    =
    \lim_{\epsilon \to 0} \inf_\Omega 
    \fnVmu(\fnv_\epsilon + \fnth_\epsilon) = \frac{2}{\cmu + 1}.
\end{align}

By lemma~\ref{lem:vav-euc-coerc-cnst}, the potentials $a_\epsilon$ 
define coercive continuous bilinear functions 
$\spH^1\times\spH^1 \to \reals$, such that 
\begin{align}
    C_\epsilon \,\norm{\fnth_\epsilon}_{\spH^1}^2 
    \leq 
    \norm{\grad \fnth_\epsilon}^2_{\spL^2} 
    + \lproduct{a_\epsilon\,\fnth_\epsilon, \fnth_\epsilon}.
\end{align}

Let, 
\begin{align}
    m = \frac{1}{\cmu + 1}, &&
    \alpha_1 = \frac{1}{\cmu + 1}, &&
    \alpha_2 = \frac{1}{\cmu + 1}\,\brk(
      1 - \frac{2}{M(\cmu + 1)}),
\end{align}
if we select a positive constant $C > 0$ such that,
\begin{align}
    C < \min\brk(
    m, 1, \alpha_1, 
    \frac{1 - \frac{M\,d}{\pi}}{1 + Cp(\domain')}, 
    \alpha_2
    ),
\end{align}
then according to lemma~\ref{lem:vav-euc-coerc-cnst} we can use 
$C$ as a common coercivity constant for all 
the potential functions $a_\epsilon$ with $\epsilon \leq \epsilon_0$. 
Therefore,
\begin{align}
    \norm{\grad \fnth_\epsilon}^2_{\spL^2} 
    = -\lproduct{\fnF_{\cmu}(v_{\epsilon}
  + \tilde\fnh_{\epsilon}),
  \tilde\fnh_{\epsilon}} +
  \lproduct{\fng_{\epsilon},
  \tilde\fnh_{\epsilon}}
  = -\lproduct{a_\epsilon\cdot (\fnv_\epsilon + \fnth_\epsilon), 
  \fnth_\epsilon} + \lproduct{\fng_{\epsilon},
  \fnth_{\epsilon}}.
\end{align}

If we apply the uniform coercivity constant, we obtain
the bound,
\begin{align}
    C\,\norm{\fnth_\epsilon}^2_{\spH^1} 
    &\leq 
    \norm{\grad \fnth_\epsilon}^2_{\spL^2} 
    + \lproduct{a_\epsilon\fnth_\epsilon, \fnth_\epsilon}\nonumber\\
    &= 
    -\lproduct{a_\epsilon\fnv_\epsilon, 
  \fnth_\epsilon} + \lproduct{\fng_{\epsilon},
  \fnth_{\epsilon}}\nonumber\\
  &\leq 
  C_2\brk(
  \norm{\fnv_\epsilon}_{\spL^2}
  + 
  \norm{\fng_\epsilon}_{\spL^2}
  )
  \norm{\fnth_\epsilon}_{\spL^2},
\end{align}
where we have used Cauchy-Schwarz and the fact that the set $\set{a_\epsilon
    \,:\,\epsilon \leq \epsilon_0}$ 
is uniformly bounded. From this inequality, we deduce the 
existence of a positive constant $C$, such that,
\begin{align}
    \max\brk(\norm{\fnth_\epsilon}_{\spL^2},\,
    \norm{\grad \fnth_\epsilon}_{\spL^2}) \leq C \brk(
    \norm{\fnv_\epsilon}_{\spL^2} + 
    \norm{\fng_\epsilon}_{\spL^2}
    )
\end{align}

Applying lemma~\ref{lem:vav-euc-fnv-fng-bounds} 
we infer the existence of another constant, such that, 
\begin{align}
    \norm{\fnv_\epsilon}_{\spL^2} + 
    \norm{\fng_\epsilon}_{\spL^2}
    \leq 
    C\,\epsilon\,\abs{\log\epsilon},
\end{align}
for $\epsilon \leq \epsilon_0$. By the elliptic estimates and 
Sobolev's embedding, 
\begin{align}
    \norm{\fnth_\epsilon}_{C^0(\domain)} \leq 
    C_1 \norm{\fnth_\epsilon}_{\spH^2(\domain)}
    \leq 
    C_2 \brk(
    \norm{\laplacian \fnth_\epsilon}_{\spL^2(\domain)}
    +
    \norm{\fnth_\epsilon}_{\spL^2(\domain)}
    ).
\end{align}

Since,
\begin{align}
    \norm{\laplacian \fnth_\epsilon}_{\spL^2(\domain)}
    =
    \norm{a_\epsilon\, (\fnv_\epsilon +
    \fnth_\epsilon)}_{\spL^2(\domain)}
    \leq 
    C\brk(\norm{\fnv_\epsilon}_{\spL^2(\domain)} 
    + 
    \norm{\fnth_\epsilon}_{\spL^2(\domain)}),
\end{align}
we apply lemma~\ref{lem:vav-euc-fnv-fng-bounds} again and the 
estimate for the $\spL^2$ norm of $\fnth_\epsilon$ we have 
obtained to deduce that,
\begin{align}
    \norm{\fnth_\epsilon}_{C^0(\domain)} \leq 
    C \epsilon\abs{\log(\epsilon)}, \qquad \epsilon \leq \epsilon_0.
\end{align}

We use this estimate and Sobolev's 
embedding again to estimate the supremum of $\del_1\fnth_\epsilon$ 
at $\domain$. If $p > 2$, we have,
\begin{align}
    \norm{\fnth_\epsilon}_{C^1(\domain)} \leq 
    C_1 \norm{\fnth_\epsilon}_{\spW^{2,p}(\domain)}
    \leq 
    C_2 \brk(
    \norm{\laplacian \fnth_\epsilon}_{\spL^p(\domain)}
    +
    \norm{\fnth_\epsilon}_{\spL^p(\domain)}
    ).
\end{align}

Again by lemma~\ref{lem:vav-euc-fnv-fng-bounds} and the previous 
estimate on the $C^0$ norm of $\fnth_\epsilon$,
\begin{align}
 \norm{\laplacian \fnth_\epsilon}_{\spL^p(\domain)}
+
\norm{\fnth_\epsilon}_{\spL^p(\domain)}
&\leq 
\norm{a_\epsilon\,(\fnv_\epsilon +
    \fnth_\epsilon)}_{\spL^p(\domain)}
    + 
    \norm{\fng_\epsilon}_{\spL^p(\domain)}
    +
    \norm{\fnth_\epsilon}_{\spL^p(\domain)}\nonumber\\
    &\leq
    C\brk(
    \norm{\fnv_\epsilon}_{\spL^p(\domain)}
    +
    \norm{\fnth_\epsilon}_{\spL^p(\domain)}
    +
    \norm{\fng_\epsilon}_{\spL^p(\domain)})\nonumber\\
    &\leq 
    C\brk(
    \epsilon^{2/p}
    +
    \epsilon\,\abs{\log\epsilon}\cdot \abs{\domain}^{1/p}
    +
    \epsilon)\nonumber\\
    &\leq C\,\epsilon^{2/p}.
\end{align}

Since asymptotically $\epsilon\,\abs{\log\epsilon} \leq \epsilon^{2/p}$ as 
$\epsilon \to 0$. 
Therefore, $\norm{\del_1\fnth_\epsilon}_{C^0(\domain)} \leq 
C\,\epsilon^{2/p}$ if $\epsilon$ is small.
\end{proof}




Going back to the original, undilated coordinates $x \in \plane$, we can state 
the following theorem,

\begin{theorem}\label{thm:vav-euc-incompleteness}
The moduli space $\moduli_0^{1,1}$ is an incomplete metric space, 
such that geodesic discs centred at the singular point 
$\epsilon = 0$ have finite area.
\end{theorem}

In comparison, the moduli space of vortices for the Ginzburg-Landau functional 
is 
complete, as can be seen in the results of 
Strachan~\cite{strachan_lowvelocity_1992} who studied geodesic motion on 
hyperbolic space or 
Samols~\cite{samols1992} on the euclidean plane. Incompleteness of 
the moduli space was expected by previous results of Rom\~ao-Speight, 
who conjectured an asymptotic logarithmic approximation to $\cf(\epsilon)$ for 
small $\epsilon$ at $\tau = 0$~\cite{romao2018}.

\begin{proof}
We will prove that $\moduli_0^{1,1}$ is incomplete exhibiting a curve
of finite length reaching the singularity at $\epsilon = 0$. Let us take any 
radial geodesic parametrized as 
\begin{align}
 \gamma_\theta:(0, \epsilon_0] \to \moduli_0^{1,1}, &&
 \gamma_\theta(\epsilon) = \epsilon\,e^{i\theta}.
\end{align}
 
By Cauchy-Schwarz, the length of this curve is bounded since,
\begin{align}
    \ell = \int_0^{\epsilon_0} \cf(\epsilon)^{1/2}\,d\epsilon
    \leq \epsilon_0^{1/2}\,\brk(
    \int_0^{\epsilon_0} \cf(\epsilon)\,d\epsilon
    )^{1/2}.
\end{align}

We will prove that the energy, and therefore the length, is finite.
Recall the interaction coefficient is given by 
\begin{align}
 b(\epsilon)  &= 2\eval{\del_1}{x = \epsilon}\brk(\fnh_\epsilon(x) 
 - \log\,\abs{x - \epsilon}^2)\nonumber\\
 &= 2\eval{\del_1}{x = \epsilon}\brk(
 \fnth_\epsilon(x) + \fnv(\abs{x + \epsilon}) 
 - \log\brk(1 + \abs{x - \epsilon}^2) + \cmu
 )\nonumber\\
 &= 2\,\del_1\fnth_\epsilon(\epsilon) 
 + \frac{8\epsilon}{1 + 4\epsilon^2}
 - \frac{2}{\epsilon}.
\end{align}

Let $\tilde b(\epsilon) = 2\,\del_1\fnth_\epsilon(\epsilon) 
 + 8\,\epsilon(1 + 4\epsilon^2)^{-1}$, we have,
\begin{align}
    \int_0^{\epsilon_0} \cf(\epsilon)\,d\epsilon &= 
    2\pi 
    \int_0^{\epsilon_0}   
2 (1 - \tau^2) + \frac{1}{\epsilon}\,\dv{\epsilon}\brk(\epsilon b(\epsilon))
    \,d\epsilon\nonumber\\
    &= 
    4\pi(1 - \tau^2)\,\epsilon_0 
    + 2\pi\,\int_0^{\epsilon_0} 
    \frac{1}{\epsilon}\,\dv{\epsilon}\brk(\epsilon \tilde b(\epsilon))
    \,d\epsilon\nonumber\\
    &= 
    4\pi(1 - \tau^2)\,\epsilon_0 
    + 2\pi
    \brk(
    \tilde b(\epsilon_0) - \lim_{\epsilon\to 0} \tilde b(\epsilon)
    + \int_0^{\epsilon_0}
    \frac{\tilde b(\epsilon)}{\epsilon}
    \,d\epsilon
    ),
\end{align}
where we have used integration by parts in the last equation.
 Let us assume $\epsilon_0$ is so small we can use the 
 estimate in 
 proposition~\ref{prop:vav-euc-del1-fnth-estimate-loc},
 \begin{align}
    \int_0^{\epsilon_0} \cf(\epsilon)\,d\epsilon &=  
    4\pi(1 - \tau^2)\,\epsilon_0 
    + 2\pi \tilde b(\epsilon_0)
    + 8\pi \tan^{-1}\brk(2 \epsilon_0)
    + 4\pi \int_0^{\epsilon_0} 
    \frac{\del_1\fnth_\epsilon(\epsilon)}{\epsilon}
    \,d\epsilon\nonumber\\
    &\leq 
    4\pi(1 - \tau^2)\,\epsilon_0 
    + 8\pi \tan^{-1}\brk(2 \epsilon_0)
    + \frac{16\pi \epsilon_0}{1 + 4\epsilon_0^2}
    + C\brk(
    \epsilon_0^{2/p}
    + \int_0^{\epsilon_0} 
    \epsilon^{\frac{2}{p} - 1}
    \,d\epsilon
    )\nonumber\\
    &\leq 
    4\pi(1 - \tau^2)\,\epsilon_0 
    + 8\pi \tan^{-1}\brk(2 \epsilon_0)
    + \frac{16\pi \epsilon_0}{1 + 4\epsilon_0^2}
    + C\epsilon_0^{2/p}.
 \end{align}
 
 Therefore, the energy is finite, hence, the 
 length of the geodesic is also finite, moreover, 
 the length is bounded by,
 \begin{align}
     \ell \leq 2\pi\brk(
     5 - \tau^2 
     )^{1/2}\,\epsilon_0 
     + C\,\epsilon_0^{1/p}.
 \end{align}
 
 For the area of a disk, we have a similar calculation,
 \begin{align}\label{eq:vol-disk-bound}
     \vol(\disk_R(0)) &= 
     2\pi \int_0^R \cf(\epsilon)\,\epsilon\,d\epsilon\nonumber\\
     &=
     4\pi^2\,(1 - \tau^2)\,R^2 + 4\pi^2 R\,\tilde b(R)\nonumber \\
     &\leq 
     4\pi^2\,(1 - \tau^2)\,R^2 + \frac{32\pi^2 R^2}{
     1 + 4 R^2}+ 8\pi^2 R\,\del_1\fnth_\epsilon(\epsilon)\nonumber \\
     &\leq 
     4\pi^2\,(1 - \tau^2)\,R^2 + \frac{32\pi^2 R^2}{
     1 + 4 R^2}+ C\, R^{1 + \frac{2}{p}}.
 \end{align}
 \end{proof}

Samols compared the area of small disks on the moduli space for the 
Ginzburg-Landau 
functional with the area of a cone with deficit angle $\pi$. 
Recall each vortex/antivortex has effective mass 
$2\pi(1\mp \tau)$ respectively, in the centre of mass coordinates,  
the reduced mass of the vortex-antivortex system is 
$\pi (1 - \tau^2)$, hence, if we normalize~\eqref{eq:vol-disk-bound} 
dividing by the reduced mass, we find that the first term in the 
upper bound is $4\pi R^2$, the 
area of a right circular cone of radius $R$ and deficit 
angle $3\pi/2$. The second and third terms in the upper bound are far from 
 optimal, because they do not depend on $\tau$ and  
 the third term is of order smaller than $2$, however, the 
 conjectured asymptotics of the conformal factor for small 
 $\epsilon$~\eqref{eq:lambda-eps} leads us to also conjecture that the 
 first term in the upper bound is the first term of an approximation 
 to $\vol(\disk_{R}(0))$ for small $R$.

\let  \hf       \undefined 
\let  \gp       \undefined 
\let  \fnh      \undefined
\let  \sgn      \undefined
\let  \kmetric  \undefined
\let  \cf       \undefined
\let  \kform    \undefined 
\let  \gmetric  \undefined 
\let  \fnV        \undefined

\let \fnv        \undefined
\let \fng        \undefined
\let \cmu        \undefined
\let \fnu        \undefined
\let \spL        \undefined
\let \spH        \undefined
\let \spW        \undefined
\let \fnF        \undefined
\let \domain     \undefined
\let \fnFmu      \undefined
\let \fnVmu      \undefined
\let \fnth       \undefined
\let \vol        \undefined
\let \energyDens \undefined


\newcommand*{\hf}{\phi}
\newcommand*{\gp}{A}
\newcommand{\fnh}{h}
\newcommand*{\cp}{\times}
\newcommand*{\fnhh}{\hat h}
\newcommand*{\cf}{\Omega}
\newcommand*{\gmetric}{g}
\newcommand*{\fnv}{v}
\newcommand*{\fng}{g}
\newcommand*{\fnV}{V}
\newcommand*{\cmu}{\mu}
\newcommand*{\fnu}{u}
\newcommand*{\fnF}{F}
\newcommand*{\fnth}{\tilde\fnh}
\newcommand*{\energyDens}{\mathcal{E}}
\newcommand*{\pbrk}[1]{\left({#1}\right)}
\newcommand*{\mDiff}{\mathcal{D}}
\newcommand*{\bxi}{\boldsymbol{\xi}}
\newcommand*{\sign}{s}

\section{Asymptotic approximation at large
  separation}\label{s:scattering-theory} 

If the cores are separated by a large distance, it is plausible to
assume that the interactions are so weak, that in the neighbourhood of
any of them, they can be described by the solution corresponding to
one vortex plus a small perturbation term due to the
interactions. We use this idea to approximate dynamics in the moduli space
for well separated vortices. For Ginzburg-Landau
vortices this was done by Speight in~\cite{speight_static_1997} and 
Manton-Speight in~\cite{manton-speight-asymptotic}. We start finding Hedgehog
solutions to the Bogomolny equations. Let us assume that there are
exactly $N$ vortices at the origin. We will use the Ansatz, 
\begin{equation}
 \begin{aligned}
\hf &=  (\sin(f)\cos(N\theta),\,\sin(f)\sin(N\theta),\,\cos(f)), \\
\gp &= N a(r)\,d\theta,
\end{aligned}
\end{equation}
which assumes circular symmetry of the field equations. This Ansatz
was used  
before by Schroers to study solutions of the $U(1)$-gauged $O(3)$ Sigma 
model for $\tau = 1$ 
\cite{schroers_bogomolnyi_1995}. The energy density of this static
configuration 
 is,
 \begin{align}
    \energyDens = \pbrk{\frac{N(a - 1)\sin(f)}{2 r}}^2 
    + \pbrk{\tau - \cos(f)}^2.
 \end{align}

For these fields to represent $N$ vortices at the 
origin with finite energy, we add the boundary conditions,
\begin{align}
    f(0) &= 0, &
    a(0) &= 0, &
    \lim_{r\to \infty}f &= \cos^{-1}{\tau},  &
    \lim_{r\to \infty}a &= 1.
\end{align}

With this Ansatz, the Bogomolny equations reduce to the system of ODEs, 
\begin{align}
    f' &= \frac{N}{r}\,(a - 1)\sin(f), &
    a' &= \frac{r}{N}\,(\cos(f) - \tau).
\end{align}

Unfortunately, we cannot extend these equations to the origin, instead, we 
select a small initial value $\delta$ and perturb the Bogomolny equations to 
lowest order in $\delta$. We found that to lowest order,
\begin{align}
    f(\delta) &= \alpha\,\delta^N, &
    a(\delta) &= \frac{1 - \tau}{2N}\,\delta^2,
\end{align}
then we used $\alpha$ as a shooting parameter. In practice, we chose 
$\delta = 10^{-8}$ and for the boundary condition at infinity, we selected 
$r_\infty = 10$ except for the last $\tau$, for which $r_\infty = 20$. 
We took $r_\infty$ as infinity and shot until $(f(r_\infty), a(r_\infty))$ 
satisfied the boundary condition, as in the paper of 
Speight~\cite{speight_static_1997}. We used the solver \emph{solve\_ivp} of the 
scientific library \emph{SciPy} with default parameters. Internally, 
it uses the Runge-Kutta  method of order 5(4), which controls the error 
using a local extrapolation and uses a quartic interpolation polynomial 
to compute the solution at the preconfigured set of points shown in 
Figure~\ref{fig:hedgehogs}.
\begin{figure}
    \centering
    \includegraphics[width=.8\textwidth]{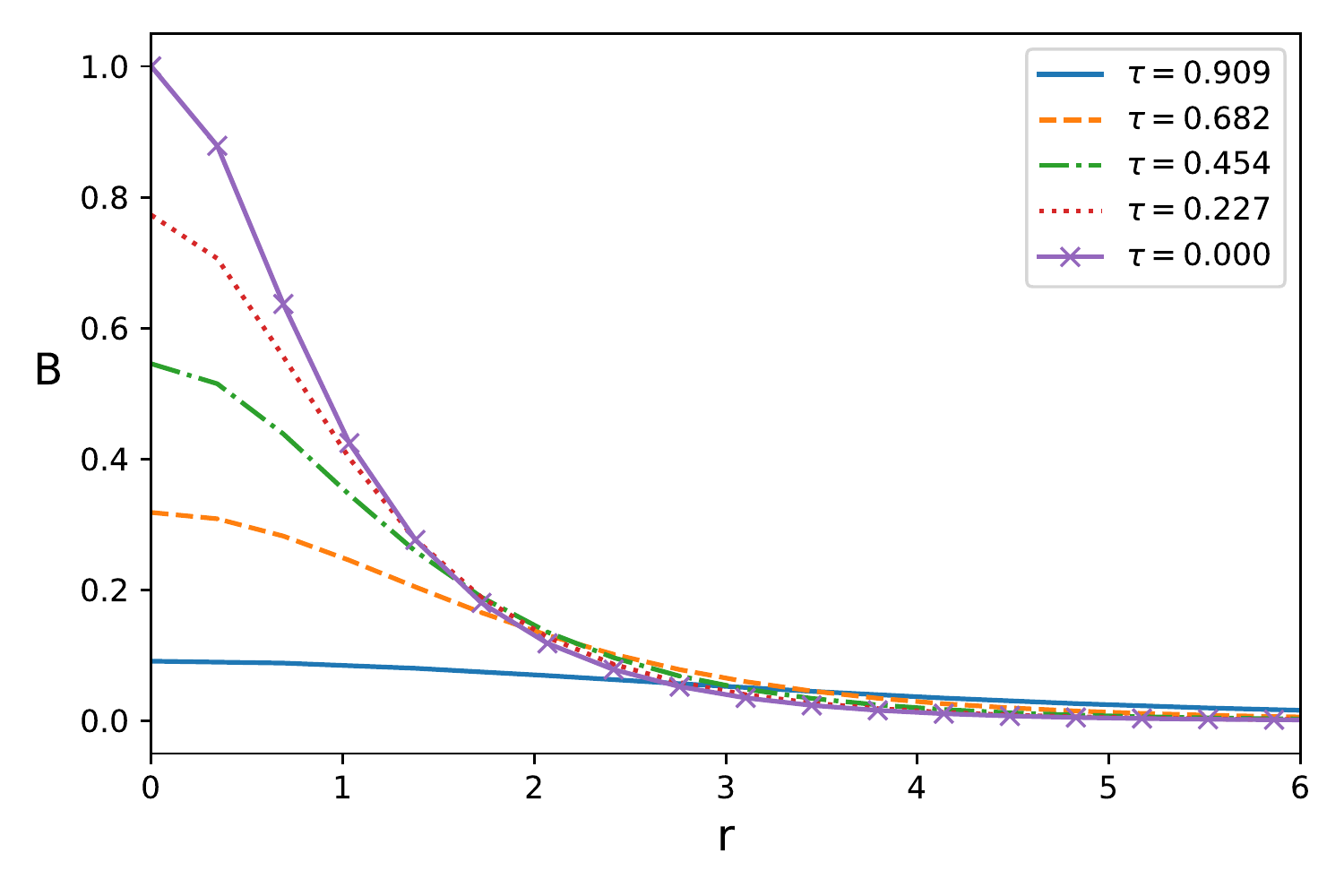}
    \includegraphics[width=.8\textwidth]{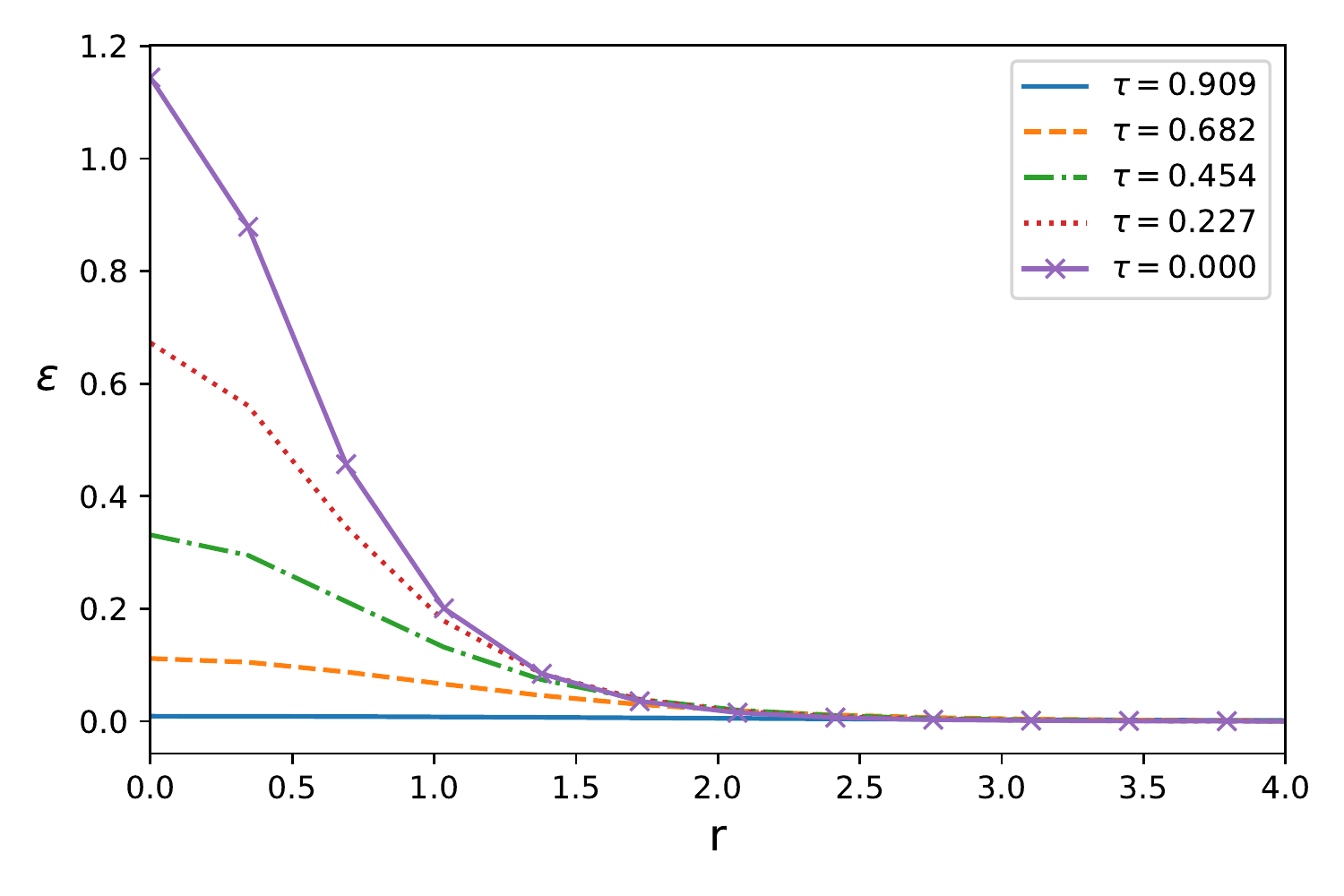}
    \caption{Magnetic field and energy density of hedgehog solutions for 
    positive values of $\tau$. The graphs show how as $\tau$ grows, 
    the energy and magnetic field weaken.}
    \label{fig:hedgehogs}
\end{figure}

If $(\phi, a)$ is the solution to the Bogomolny equations with $N$ vortices at 
the 
origin and parameter $\tau$ and we invert the orientation of the sphere, 
selecting $-n$ as the north pole, it is not difficult to see that $(\phi, -a)$ 
is also a solution to the Bogomolny equations, this time with parameter 
$-\tau$ 
and $N$ antivortices at the origin, hence the qualitative properties of an 
antivortex hedgehog are the same, except that to a $\tau$-vortex corresponds a 
$-\tau$ antivortex and to a $B$ (vortex) magnetic field corresponds a $-B$ 
(antivortex) magnetic field. 

Assuming there is only one core at the origin, 
the solution to the 
Taubes equation, $h$, is also radial, and away of the origin, is a
solution to the equation,
\begin{equation}
\label{eq:taubes-radial}
\frac{d^2h}{dr^2} + \frac{1}{r}\frac{dh}{dr} - 2\brk(
\frac{e^h - 1}{e^h + 1} + \tau) = 0.
\end{equation}
For small $r$, \(h \) has the asymptotic behaviour \(h = \pm \ln(r^2)
\) and for big \(r \), it approaches \(\log \left( \frac{1 - 
    \tau}{1 + \tau} \right)\). Linearizing about the limit at infinity, 
    we have the equation, 
\begin{align}
  \label{eq:linearised-taubes-infinity}
\frac{d^2\hat{h}}{dr^2} + \frac{1}{r}\frac{d\hat{h}}{dr} - (1 -
\tau^2)\hat{h} = 0, &&
\lim_{r \to \infty}\hat{h} = 0.
\end{align}
If we make the change of variables $r' = (1 - \tau^2)^{1/2}\,r$, then the function 
$\fnhh(r')$ is a solution to the modified Bessel equation,
\begin{align}
\frac{d^2\fnhh}{dr'^2} + \frac{1}{r'}\frac{d\fnhh}{dr'} - \fnhh = 0, &&
\lim_{r' \to \infty}\fnhh = 0.
\end{align}
whose general solution is a linear combination of modified Bessel's
function of first and second kind, \(J_0 \) and \(K_0 \). Since \(J_0 \)
diverges at infinity, we deduce the approximation,
\begin{equation}
\label{eq:assymptotic-expansion-h-infinity}
h(r) = \log \left( \frac{1 - \tau}{1 + \tau} \right) 
+ q K_0\brk((1 - \tau^2)^{1/2} r). 
\end{equation}
The constant $q$ has to be determined numerically, as in 
the approximation done for Ginzburg-Landau
vortices \cite{speight_static_1997}.
 We found this constant for several values of $\tau$ by solving the Bogomolny 
 equations as explained above, with this data, we computed 
 the pairs $\left(K_0((1 - \tau^2)^{1/2}\,r), h(r)\right)$ and fitted 
 a least squares line as a model, whose slope was $q$. We tested visually 
 and by means of the coefficient of determination $R^2$ the goodness of 
 fit of the model to the data, finding on average $R^2 = 0.9985$, 
 meaning the linear model explained $99.85\%$ of the data, hence the 
 fit was good.  
 The dependence of the constant $q$ on $\tau$ can 
 be seen in figure~\ref{fig:qvstau}. It is interesting to note that 
  the graph suggests $q$ depends linearly with $\tau$, this is unexpected 
  since $q$ is not well understood even for the Ginzbug-Landau 
  functional, where there is an argument by David Tong~\cite{tong2002ns5}   
  proposing an 
  explanation for the value of $q$ based on string theory, but otherwise, 
  the value of the constant is only known numerically and 
  it is not clear whether 
  such argument can be extended to the $O(3)$ Sigma model.  
  The computed values of $q$ are also 
 displayed in table~\ref{tab:qvstau}. For $\tau = 0$, the value of $\pi q$ 
 was computed by Rom\~ao-Speight~\cite[p.~23]{romao2018} as $-7.1388$, 
 as can be seen 
 in table~\ref{tab:qvstau}, we found a value of $\pi q = -7.1346$, in agreement 
 with the known data. 
 
 \begin{table}[]
     \centering
     \begin{equation*}
        \begin{array}{|c|l|l|l|l|l|l|l|l|l|}
        \hline
        \tau    &  
         -0.909 &
         -0.682 &
         -0.454 &
         -0.227 &
          \quad\! 0   &
          \quad\! 0.227 &
          \quad\! 0.454 &
          \quad\! 0.682 &
          \quad\! 0.909
        \\
        \hline
           q     & 
         -1.2457 &
         -1.5414 &
         -1.7921 &
         -2.0321 &
         -2.271  &
         -2.5134 &
         -2.7568 &
         -2.9784 &
         -3.2504 \\
         \hline
     \end{array}
     \end{equation*}
     \caption{Constant $q$ for different values of $\tau$ for a vortex at 
     origin in Euclidean space. For an antivortex,  $q$ has positive sign.} 
     \label{tab:qvstau}
 \end{table}
 
 \begin{figure}
     \centering
     \includegraphics[width=.8\textwidth]{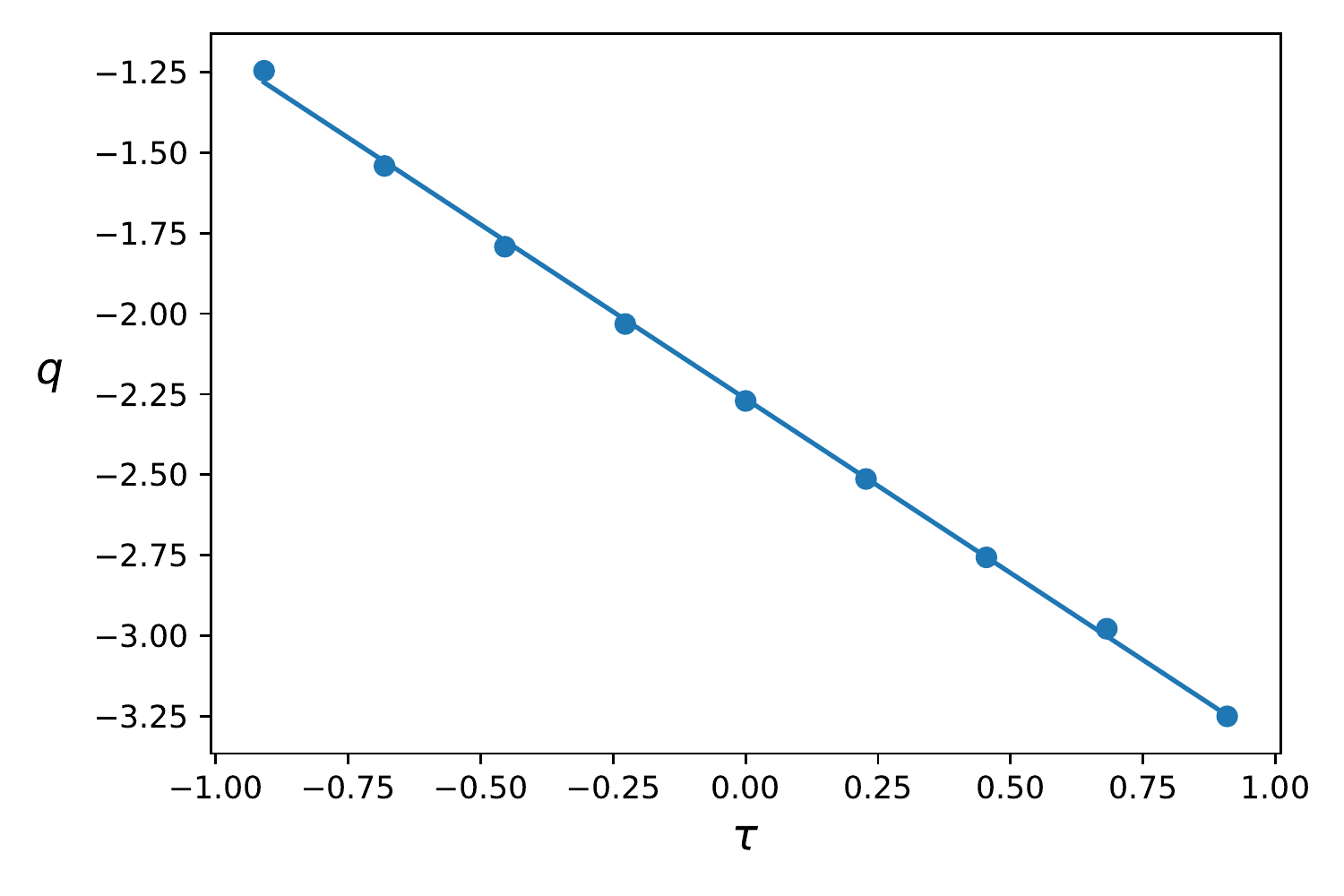}
     \caption{Dependency of the parameter $q$ on the asymmetry $\tau$ of the 
     vortex. For an antivortex $q$ is positive, the pattern is reversed 
     and $q$ increases with $\tau$.}
     \label{fig:qvstau}
 \end{figure}
 
 Let us consider an antivortex at position $-2\epsilon$ for big
 $\epsilon$. The antivortex perturbs $h$ in a  neighbourhood of the
 origin, since the separation is large, we can assume that this is a
 small perturbation of the Hedgehog solution. Let $\fnh_0$ be the
 single vortex solution at the origin. If $\fnh_1$ is a small
 perturbation of $\fnh_0$ caused by the antivortex in a neighbourhood
 of the origin, $\fnh_1$ is a solution to the linearization of the Taubes 
 equation,
 \begin{align}\label{eq:lin-h1}
    -\laplacian \fnh_1 = \frac{4e^{\fnh_0}\,\fnh_1}
    {(1 + e^{\fnh_0})^2}.
 \end{align}
 
 The singularity at origin is carried by $\fnh_0$ and 
 since the operator in equation~\eqref{eq:lin-h1} if free of 
 singularities, $\fnh_1$ extends smoothly to the origin. 
 Expanding in Fourier series $\fnh_1$, we find,
 \begin{align}
     \fnh = \fnh_0 + 
     \half f_0(r) + \sum_{n = 1}^\infty \pbrk{
     f_n(r)\cos(n\theta) + g_n(r)\sin(n\theta)}.
 \end{align}
 
 The functions $f_n(r)$ and $g_n(r)$ are solutions to the equation 
 \begin{align}
     f_n'' + \frac{1}{r}f_n' - \pbrk{
     \frac{4e^{\fnh_0}}{(1 + e^{\fnh_0})^2} +
     \frac{n^2}{r^2}}\,f_n = 0,
 \end{align}
and since $\fnh_1$ is well defined at $r = 0$, to lowest order we have 
$f_n(r) = \alpha_n\,r^n$, $g_n = \beta_n\,r^n$. 

To compute the coefficient $b_1$,
we note that $\fnh_0 = \log r^2 + \fnth_0(r)$, where the regular part $\fnth_0$
is a smooth function. Since $\log r^2$ 
 is the fundamental solution of Laplace's equation on the plane, 
 by~\eqref{eq:taubes-radial}, $\fnth_0$ is 
 a solution to the equation
 \begin{align}
 \frac{d^2\fnth_0}{dr^2} + \frac{1}{r}\frac{d\fnth_0}{dr} - 2\brk(
 \frac{e^{h_0} - 1}{e^{h_0} + 1} + \tau) = 0,
 \end{align}
hence,
 \begin{align}
     \conj\del_x \fnth_0(r) &= \fnth_0'(r)\,\conj\del_x\,r\nonumber \\
     &= \half \fnth_0'(r)\,e^{i\theta}\nonumber\\
     &= r\,\pbrk{
     2\pbrk{\frac{e^{\fnh_0} - 1}{e^{\fnh_0} + 1} + \tau} - \fnth_0''
     }.
 \end{align}
 
 The function $e^{\fnh_0}$ has no singularity at the origin, moreover it 
 is smooth, hence,
 \begin{align}
     \conj\del_x\fnth_0(0) = \lim_{r \to \infty} \conj\del_x\fnth_0(r) = 0.
 \end{align}
 $\fnh$ is symmetric with respect to the line joining the two cores. These 
 are located on the real axis, hence $\fnh(x) = \fnh(\conj x)$ which
 translates into 
 \begin{align}
 2\eval{\del_x}{x = 0}\pbrk{\fnh - \log r^2} 
 = \eval{\del_1}{x = 0}\pbrk{\fnh - \log r^2} = \alpha_1.
 \end{align}
 
 We conclude that $b_1 = \alpha_1$. To compute the nontrivial 
 coefficient in the metric of the moduli space, we note that for large
 $r$, $f_1$ is a  solution to the modified Bessel equation,
 \begin{align}
    f_1'' + \frac{1}{r}f_1' - \pbrk{
    1 - \tau^2 +
    \frac{1}{r^2}} f_1 = 0, 
 \end{align}
from here we can follow the computation done in
\cite{manton-speight-asymptotic} for Ginzburg-Landau vortices, the analysis 
is the same in the coordinate system $x'$ and the conclusion is that 
the coefficient $b_1$ for a pair of distant vortices is,
\begin{equation}
  \label{eq:b1-asymp-approx-pair}
  b_1(\epsilon) = \half q_1q_2\brk(1 - \tau^2)^{1/2}
  K_1\brk(2\brk(1 -\tau^2)^{1/2} \epsilon).
\end{equation}

By translation invariance, $b(\epsilon) = b_1(\epsilon)$, for $b$ the 
nontrivial term in the conformal factor of the metric in the reduced 
moduli space. Using the properties of Bessel's functions given in 
equation~\eqref{eq:bessel-props}, we find that at large separation
the conformal factor can be approximated as,
\begin{equation}
  \label{eq:conformal-factor-asymptotics-pair}
  \cf (\epsilon) = 2\pi(1 - \tau^2) \brk(
  2 - q_1q_2\,K_0\brk(2 
  \brk(1 - \tau^2)^{1/2} \epsilon)).
\end{equation}
From this formula we observe the conformal factor vanishes at 
$\tau = \pm 1$, this can be understood because the effective mass of a 
vortex or antivortex is $2\pi(1\mp\tau)$, hence as $\tau \to \pm 1$, 
most of the kinetic energy of a vortex-antivortex pair is concentrated at 
one of the cores which in the limit coincides with the 
centre of mass. Hence, by the decomposition of the $\Lsp^2$ metric in 
the centre of mass frame, proposition~\ref{prop:metric-cm-decomp}, 
one would expect this vanishing of the conformal factor.

\subsection{The point-source formalism}\label{sec:point-source-formalism}

Consider a single vortex or anti-vortex at origin, labelled 1, up to a local 
trivialization, the Higgs field is a map $\phi: U \subset \reals \times \plane 
\to \sphere$ 
with coordinates $\phi(x_0, x_1, x_2) = (X_1, X_2, X_3)$. In the south pole 
projection, this field is equivalent to $\psi(x_0, x_1, x_2) = (X_1/(1 + X_3), 
X_2/(1 + X_3))$. We can choose a local gauge, the real gauge, in which $\psi$ 
is real, or going back to the sphere, $\phi$ is constrained to the 
intersection circle of $\sphere$ with the plane \(X_2 = 0\). 
Since the field has nontrivial winding, this gauge choice can be 
made only with exception of the
core positions \cite{manton-speight-asymptotic}. We aim to calculate a
linear approximation to the field and vector potential far from the
core, in which case we can make this assumption. It will be convenient
to work in spherical coordinates, such that the Higgs field is
parameterised as \(\phi = (\sin(\varphi), 0,
\cos(\varphi)) \), with \(\varphi \) the azimuthal angle. In this
gauge, the spherical covariant derivatives are
\begin{equation}
  \label{eq:covariant-dv-real-gauge}
  D_{\mu}\phi = \del_{\mu}\phi - A_{\mu}\sin(\varphi)\,e_2.
\end{equation}

In this section, we aim to show that if we have a collection of cores, 
vortices and antivortices well separated among each other, we can approximate 
the dynamics of the system as if at each core position there were a scalar 
monopole point-source and a magnetic dipole. For large \(r \), the field 
approaches
the vacuum manifold, perturbatively 
we can approximate $\hf$ as \((\sin(\varphi + \varphi_{\infty}), 0, 
\cos(\varphi + \varphi_\infty))\), where 
\(\varphi_{\infty} =  \cos^{-1}(\tau) \) and \(\varphi\) is 
small. Keeping linear terms in \(\varphi \), we can make the approximation,
\begin{equation}
  \label{eq:linear-approx-covariant-dv-real-gauge}
  D_{\mu}\phi = \brk(\cos(\varphi_{\infty})\,\del_{\mu}\varphi,
  -A_{\mu}(\sin(\varphi_{\infty}) + \cos(\varphi_{\infty})\varphi),
  -\sin(\varphi_{\infty})\,\del_{\mu}\varphi). 
\end{equation}

Retaining terms up to quadratic order, far from the vortex position,
the Lagrangian density is approximately linear, corresponding to a 
non interacting field,
\begin{equation}
  \label{eq:bps-linear-lagrangian}
  \mathcal{L}_{free} = \half \del_{\mu}\varphi\,\del^{\mu}\varphi -
  \half \sin^2(\varphi_{\infty})\,\varphi^2 -
  \frac{1}{4}F_{\mu\nu}F^{\mu\nu} +
  \half \sin^2(\varphi_{\infty})\,A_{\mu}A^{\mu}.
\end{equation}

This is the Lagrangian density of two independent fields, whose extremals 
$(\varphi, \gp)$ satisfy the real Klein-Gordon and Proca equations,

\begin{align}
  (\Box + \sin^2(\varphi_{\infty}))\,\varphi &= 0,\\
  (\Box + \sin^2(\varphi_{\infty}))\,\gp_{\mu} &= \del_{\mu}\del^{\nu}\gp_{\nu},
\end{align}
where \(\Box = \del^2_t + \laplacian \) is the D'Alambertian
operator. We add a source term, 
\begin{align}
  \mathcal{L}_{source} = \rho\,\varphi - j_{\mu}A^{\mu},  
\end{align}
to the free Lagrangian density, in order to match the expected
behaviour at infinity of the fields, as in
\cite{speight_static_1997}. Therefore, the perturbed field 
equations are,
\begin{align}
(\Box + \sin^2(\varphi_{\infty}))\,\varphi &= \rho,
\label{eq:asymptotic-massive-field-equation}\\
(\Box + \sin^2(\varphi_{\infty}))\,\gp_{\mu} &= j_{\mu} +
\del_{\mu}\del^{\nu}\gp_{\nu}.
\end{align}
Taking the divergence of the second equation, we find that,
\begin{align}
(\Box + \sin^2(\varphi_{\infty}))\,\partial^\mu\gp_{\mu} &= 
\partial^\mu j_{\mu} +
\Box\del^{\nu}\gp_{\nu},
\end{align}
hence,
$\sin^2(\varphi_\infty)\,\partial^\mu \gp_\mu = \partial^\mu j_\mu$ and 
we infer,
\begin{align}
  (\Box + \sin^2(\varphi_{\infty}))\,\gp_{\mu} &= j_{\mu} +
  \frac{1}{\sin^2(\varphi_{\infty})}\del_{\mu}\del^{\nu}j_{\nu}.
  \label{eq:asymptotic-massive-potential-equation}
\end{align}

The sourced field equations of $\varphi$ and $\gp$ represent two
massive fields of mass \(\sin(\varphi_\infty) = \sqrt{1 - \tau^2} \). 
 In the real gauge, south pole stereographic projection of $\hf$ is 
\(\psi = \hf_1/(1 + \hf_3) \), hence, since $\varphi$ is small,
\begin{align}
  \psi = \frac{\sin(\varphi_{\infty}) +
    \cos(\varphi_{\infty})\,\varphi}{1 + \cos(\varphi_{\infty}) -
    \sin(\varphi_{\infty})\,\varphi},
\end{align}
moreover, to first order 
we have,
\begin{equation}
  \label{eq:psi-asymptotics-trigonometric}
  \psi = \frac{\sin(\varphi_{\infty})}{1 + \cos(\varphi_{\infty})} +
  \frac{1}{1 + \cos(\varphi_{\infty})}\,\varphi.
\end{equation}

On the other hand, if we fix one core and consider the field at a large 
distance from it but larger to the other cores, we have the approximation  
\(\psi = \exp\pbrk{\frac{1}{2}h_0} \), 
where \(h_0 \) is the solution to the radial Taubes equation, 
given by equation~\eqref{eq:assymptotic-expansion-h-infinity}. To first order we 
have, 
\begin{equation}
  \label{eq:psi-asymptotics-tau}
  \begin{aligned}
  \psi &= \pbrk{\frac{1 - \tau}{1 + \tau}}^{1/2}
  \pbrk{
  1 + \frac{q_1}{2} K_0\pbrk{(1 - \tau^2)^{1/2}\,r}}\\
  &= \frac{\sin(\varphi_{\infty})}{1 + \cos(\varphi_{\infty})}
  \pbrk{
  1 + \frac{1}{2} q_1 K_0\pbrk{\sin(\varphi_\infty)\,r}}.
  \end{aligned}
\end{equation}

Hence, the asymptotic expansion of \(\varphi\) is,
\begin{align}
\varphi = \frac{q_1}{2}\sin(\varphi_{\infty})\,K_0\qty(\sin(\varphi_{\infty})\,r). 
\end{align}

We are interested in static fields, in this case, the field equations
 reduce to the static Klein-Gordon equation with a mass term,

\begin{align}
  (\laplacian + \sin^2(\varphi_{\infty}))\,\varphi &= \rho, &
  (\laplacian + \sin^2(\varphi_{\infty}))\,\gp_{\mu} &= j_{\mu} + 
  \frac{1}{\sin^2(\varphi_{\infty})}\del_{\mu}\del^{\nu}j_{\nu}.
  \label{eq:asymptotic-static-massive-field-equation}
\end{align}

We have,
\begin{equation}
  \label{eq:sphere-dv-real-gauge}
\mDiff_k\hf = \del_k\phi - A_k\,n\times\phi = \del_k\varphi\cdot(\cos(\varphi)e_1
-\sin(\varphi)e_3) - A_k\sin(\varphi)\,e_2,
\end{equation}
and
\begin{align}
  \hf \times (\del_2\hf - A_2\,n\times\hf) &=
  \brk(\sin(\varphi)e_1 + \cos(\varphi)e_3) \\
  &\quad \times \brk(
     \del_2\varphi\cdot\cos(\varphi)e_1 -  \gp_2\sin(\varphi)\,e_2
     - \del_2\varphi\cdot\sin(\varphi)e_3)\\
&= \gp_2\sin(\varphi)\cos(\varphi)e_1 + \del_2\varphi\cdot e_2 
   - \gp_2\sin^2(\varphi)\,e_3.\label{eq:sphere-J-times-dv-real-gauge}
\end{align}

In the gauge \(A_0 = 0 \) the first \bog\, equation is,
\begin{align}
\brk(\del_1\varphi + \gp_2\sin(\varphi))\brk(\cos(\varphi)e_1 -
\sin(\varphi)e_3) + \brk(\del_2\varphi - \gp_1\sin(\varphi)) e_2 = 0,
\end{align}
which is equivalent to,
\begin{align}
  \del_1\varphi + \gp_2\sin(\varphi) = 0, &&
  \del_2\varphi - \gp_1\sin(\varphi) = 0.
\end{align}

In a region far from the core position, these equations can be
linearized as
\begin{align}
  \del_1\varphi + \sin(\varphi_{\infty})A_2 = 0, &&
  \del_2\varphi - \sin(\varphi_{\infty})A_1 = 0.
\end{align}

In the gauge \(A_0 = 0 \) if the fields are static we have,
\begin{align}
  j_0 + \frac{1}{\sin^2(\varphi_{\infty})}\del_0\del^{\nu}j_{\nu}  = 0.  
\end{align}

For the spatial components, note that
\begin{align}
(\gp_1, \gp_2) = \frac{1}{\sin(\varphi_{\infty})}\,
(\del_2\,\varphi, -\del_1\,\varphi).
\end{align}

Introducing a fictitious unit vector \(\vb{k} \) perpendicular to the
plane in the positive orientation of \(\reals^3 \) and defining
\(\vb{A} = (A_1, A_2) \), the spatial part of the linearized
potential can be related to the Higgs field with the vector equation,
\begin{equation}
  \label{eq:vector-potential-real-gauge}
\vb{A} = -\frac{1}{\sin(\varphi_{\infty})}\,\vb{k}\times\nabla\varphi.
\end{equation}

To make our deduction of the point-source approximation, we will 
work in space-time coordinates; to this end, in this section we denote 
space-time coordinates as $x$ and space coordinates as $\vb x$.

The static field equation of $\varphi$ is,
%
\begin{align}
  (\laplacian + \sin^2\varphi_\infty)\,\varphi &= \rho.
\end{align}


Green's function for the static Klein-Gordon equation is
\(K_0(\abs{\vb x}) \), 
\begin{equation}
\label{eq:static-klein-gordon-green-function}
\brk(\laplacian + 1)\,K_0(\abs{\vb x}) = 2\pi\delta(\vb x).
\end{equation}

Substituting the asymptotic approximation to \(\varphi \) we found, 
\begin{align}
\brk(\laplacian + \sin^2\varphi_\infty)\varphi &=
  \frac{q_1}{2}\sin(\varphi_{\infty})\,
  \pbrk{\laplacian + \sin^2\varphi_\infty} 
  K_0(\sin(\varphi_\infty)\,r)\nonumber\\
  &= q_1\pi\,\sin^3(\varphi_\infty)\,\delta(\sin(\varphi_\infty)\,\vb 
  x)\nonumber\\
  &= q_1\pi \sin(\varphi_{\infty})\,\delta(\vb x),
\end{align}
where in the last inequality we have used that for any constant $c$, 
$\delta(c\vb x) = c^{-2}\delta(\vb x)$. 
This suggests that the physics of a static vortex, seen
far from the core is equivalent to a particle with charge
\(q_1\pi \sin(\varphi_{\infty})\), therefore we define the one vortex source term,
\begin{equation}
\label{eq:asymptotic-approx-charge-density}
  \rho = q_1\pi \sin(\varphi_{\infty})\,\delta(\vb x).
\end{equation}



Applying the operator  \((\laplacian + \sin^2\varphi_\infty)\) to  
\(\vb\gp \), we find,
\begin{align}
\brk(\laplacian + \sin^2\varphi_\infty)\vb\gp &=
            -\frac{1}{\sin(\varphi_{\infty})}\,
            \vb{k}\times \nabla
            \brk(\laplacian + \sin^2\varphi_\infty)\,\varphi\nonumber\\
  &= -q_1\pi \,\vb{k} \times \nabla \delta(\vb x).
\end{align}


On the other hand, let us assume that the current is static, in the sense that
\(j_0 = 0\). From~\eqref{eq:asymptotic-static-massive-field-equation}, 
we have that $\vb \gp$ satisfies the equation,
\begin{align}
   (\laplacian + \sin^2\varphi_\infty) \vb\gp = 
   \vb{j}
   - \frac{1}{\sin^2(\varphi_\infty)}\,\nabla\brk(\nabla\cdot \vb j).
\end{align}

Thence,
%
\begin{equation}
\label{eq:current-equation-distribution}
\sin^2(\varphi_\infty)\,\vb{j} - \nabla\brk(\nabla\cdot \vb j)  =
-q_1\pi\sin^2(\varphi_\infty)\,\vb{k}\times\nabla\delta(\vb x).
\end{equation}

Taking the divergence of this equation we find that $\nabla\cdot \vb j$ is 
a solution in the sense of distributions, to the equation,
%
\begin{equation}
\label{eq:current-equation-laplace}
(\laplacian + \sin^2\varphi_\infty)\nabla\cdot{\vb{j}} = 0.
\end{equation}

We know that $\nabla\cdot\vb j$ is also a strong solution in 
 $\plane\setminus\set{0}$. It is sensible to assume that $\nabla\cdot\vb j$ is 
 an $\mathrm{L}^2$ solution to this equation. Under this assumption, by elliptic 
 regularity $\nabla\cdot\vb j$ is smooth in the plane and since 
 $\sin^2\varphi_\infty$ is in the 
 resolvent set of geometers' Laplacian, $\nabla\cdot\vb j = 0$. Therefore, 
the current is conserved and we have that the core behaves as a magnetic 
dipole generated by a point current,
%
\begin{align}
    \vb j = -q_1\pi \,\vb k \times \nabla\delta(\vb x).
\end{align}

We will need later space-time coordinates, we define,
\begin{equation}
\label{eq:asymp-approx-static-dipole}
  j_{static} = \brk(0, \vb j),
\end{equation}
as the space-time point current in the lab frame.

Having calculated expressions for the charge and current of the point
particle approximation, we can 
calculate the interaction potential of a pair of vortices. For this, it  
is necessary to calculate the interaction Lagrangian, which is
obtained as 
\begin{align}
\lagrangian_{int} = \int \lgrdensity_{cross}\,dx, 
\end{align}
where \(\mathcal{L}_{cross}\) are the cross terms of \(\mathcal{L}_{free} +
\mathcal{L}_{source}\) in a superposition of two pairs of 
fields \((\varphi_k, \vb\gp_k)\), with sources \((\rho_k,
\vb j_k)\). For a pair of cores, the 
interaction Lagrangian reduces to~\cite{speight_static_1997}
\begin{equation}
    \label{eq:interaction-lagrangian}
    \lagrangian_{int} = \int \rho_1 \varphi_2 
    - {j}_{\mu}^{(1)} {\gp}_{(2)}^{\mu}\, dx.
\end{equation}

We aim to calculate the 
interaction Lagrangian for any number of separated 
moving cores whose separations are large. Let us consider a core 
moving slowly in the laboratory frame and 
let \(\xi \) be the coordinates on space-time with 
respect to this frame, which has coordinates \(x \). If the vortex is
moving at constant speed \(u \) in the direction of \(x_1 \) with respect to
the lab frame, the coordinate change on tangent space at \(x \) is
\cite{carroll2004spacetime}, 
\begin{equation}
\label{eq:lorentz-boost-rest-lab}
\left. 
\begin{aligned}
\xi_0 &= \gamma(u)\,(x_0 - u\,x_1),\\
\xi_1 &= \gamma(u)\,(-u\,x_0 + x_1),\\
\xi_2 &= x_2,
\end{aligned}
\quad\right\}
\end{equation}
where \(\gamma(u) = (1 - u^2)^{-1/2} \) is the Lorentz contraction
factor and the speed is relative to the speed of light, \(\abs{u} < 1 \). 
Our aim is to write the charge and magnetic dipole
of the moving core as seen in the laboratory frame. If the 
velocity with respect to the lab frame is not along the \(x_1 \) axes,
we can always rotate the coordinates before and then 
boost in the $x_1$ direction. 
In the rest frame, the core is static, and therefore the
charge density at large separation from their neighbours is
\(\rho(\bxi) = 
q\,\pi \sin(\varphi_{\infty})\,\delta(\bxi) \). Since we
are interested in the infinitesimal behaviour of the charge, we can
take \(x_0 = 0 \) in the Lorentz transformations relating rest and 
laboratory frames,
\begin{align}
\nonumber
  \rho(\bxi) &=
q\,\pi \sin(\varphi_{\infty})\,\delta(\gamma\,x_1\,e_1
+ x_2\,e_2)\\
\label{eq:asymp-approx-delta-rescaling}
&= \frac{1}{\gamma}\,q\,\pi\,\sin(\varphi_{\infty})\,\delta(\vb x).
\end{align}

If the speed is much slower than the speed
of light, \(\gamma^{-1} \) can be approximated as 
\begin{align}
    \gamma(u)^{-1} = 1 - \frac{1}{2}\,u^2 + \order(u^4).
\end{align}

Discarding higher order terms in \(u\), the instantaneous charge density of 
a slowly moving vortex is
\begin{equation}
\label{eq:asymp-approx-relativistic-charge}
  \rho(x) = q \pi \sin(\varphi_{\infty})\pbrk{1 -
\frac{u^2}{2}}\,\delta(\vb x).
\end{equation}

If the core is at an arbitrary position $y(t)$ and $u = \dot y(t)$ is the 
speed of the moving core, we conclude the charge density as seen in the 
laboratory frame is,
\begin{equation}\label{eq:vav-boosted-charge}
\rho = q\pi\,\sin(\varphi_{\infty})\pbrk{1 -
\frac{\dot{y}^2}{2}}\,\delta(\vb{x} - \vb{y}).
\end{equation}


For an observer in an inertial frame, a slowly moving core $y(t)$ 
has the four-current,
\begin{equation}
\begin{aligned}
j_0 &= q\pi\,\vb{k}\cp\vb{\dot{y}}\cdot\nabla\delta(\vb{x} -
\vb{y}),\\
\vb{j} &= q\pi\,\pbrk{-\vb{k}\cp\nabla +
    (\vb{k}\cp \vb{\dot{y}})\,\vb{\dot{y}}\cdot\nabla +
    \vb{k}\cp\vb{\ddot{y}}}\delta (\vb{x} - \vb{y}).
\end{aligned}\label{eq:boosted-charge-current-lab}
\end{equation}

\eqref{eq:boosted-charge-current-lab} was computed by Speight for 
Ginzburg-Landau 
vortices, details of the computation can be found 
in~\cite[eqs.~(3.20)~(3.21)]{manton-speight-asymptotic}, for the O(3) 
Sigma model, the calculation is the same, except for the 
factor of $\pi$ coming from our conventions on the constant $q$.

Since current is conserved, the components $\gp_\mu$ of the gauge potential 
are solutions to the equation,
\begin{align}
(\Box + \sin^2(\varphi_{\infty}))\,\gp_{\mu} 
&= j_{\mu}.
\end{align}

If we define the primed coordinate system,
\begin{align}
x' &= \sin(\varphi_\infty)\,x, \label{eq:large-asympt-approx-primed-system}
\end{align}
and fields,
\begin{align}
\varphi'(x') &= \varphi(\sin(\varphi_\infty)^{-1}\,x'), &
A'_\mu(x') &= A_\mu (\sin(\varphi_\infty)^{-1}\,x'),
\end{align}
with sources,
\begin{equation}
\begin{aligned}
\rho'(x') &= \sin(\varphi_\infty)^{-2}\,\rho(\sin(\varphi_\infty)^{-1}\,x'), 
\\
j'_\mu(x') &= \sin(\varphi_\infty)^{-2}\,j_\mu (\sin(\varphi_\infty)^{-1}\,x'),
\end{aligned}
\end{equation}
then $\varphi'$, $A'_\mu$ are solutions to the equations,
\begin{equation}
\begin{aligned}
(\Box' + 1)\, \varphi' 
&= \rho',\\
(\Box' + 1)\, A'_\mu 
&= j'_\mu.
\end{aligned}\label{eq:primed-trick-kg}
\end{equation}

Since $d\vb{y}/dt = d\vb{y}'/dt'$, defining $q' = q\pi \sin(\varphi_\infty)$,  
by~\eqref{eq:vav-boosted-charge},
\begin{align}
\rho' &= q'\pbrk{1 -
    \frac{\dot{y}'^2}{2}}\,\delta(\vb{x}' - \vb{y}'),
\end{align}
whereas by~\eqref{eq:boosted-charge-current-lab},
\begin{equation}
\begin{aligned}
j_0' &= q'\,
\vb{k}\cp\vb{\dot{y}}'\cdot\nabla'\delta(\vb{x}' -
\vb{y}'),\\
\vb{j}' &= q'\,\pbrk{-\vb{k}\cp\nabla' +
    (\vb{k}\cp \vb{\dot{y}}')\,\vb{\dot{y}}'\cdot\nabla' +
    \vb{k}\cp\vb{\ddot{y}}'}\delta (\vb{x}' - \vb{y}').
\end{aligned}\label{eq:primed-trick-jsource}
\end{equation}

In the primed coordinate system, 
equations~\eqref{eq:primed-trick-kg}-\eqref{eq:primed-trick-jsource} are the 
same as those found in the asymptotic approximation of Ginzburg-Landau vortices 
by Speight and Manton, with the only exception that vortices and antivortices 
carry constants $q$ of different values. Hence, 
by~\cite[Eq.~(3.46)]{speight_static_1997}, for a pair of cores at positions 
labelled $\vb{x}_1$, $\vb{x}_2$,
\begin{equation}
\begin{aligned}
\lagrangian_{int} &= - \frac{q_1'q_2'}{4\pi }\,|\dot{\vb{x}}_2' - 
\dot{\vb{x}}_1'|^2\,K_0(|\vb{x}_2' - \vb{x}_1'|)\\
&= - \frac{q_1q_2}{4}\,\pi\,\sin^2(\varphi_\infty)\,|\dot{\vb{x}}_2 - 
\dot{\vb{x}}_1|^2\,K_0(\sin(\varphi_\infty)\,|\vb{x}_2 - \vb{x}_1|).
\end{aligned}
\end{equation}

Recall \(m_r = 2 \pi(1 + \sign_r \tau) \) is the effective mass of a core at 
position \(\vb x_r \), where $\sign_r = \pm 1$ is the
sign of the core, we conclude that if the cores are at large separation and 
moving slowly, their dynamics can be approximated by the Lagrangian,
\begin{equation}
\label{eq:interaction-lagrangian-approx-full}
\lagrangian = \sum_r \frac{m_r}{2} \abs{\dot{\vb x}_r}^2 
- \sum_{r \neq s} \frac{q_rq_s}{4}\,\pi\,\sin^2(\varphi_{\infty})
\abs{\dot{\vb x}_r - \dot{\vb x}_s}^2
K_0\pbrk{\sin(\varphi_{\infty}) \abs{\vb{x}_r - \vb{x}_s}}. 
\end{equation}

For a vortex-antivortex pair at large separation, if $M = m_1 + m_2$,  
$\vb X = \frac{m_1}{M}\, \vb x_1 + \frac{m_2}{M}\, \vb x_2$ is the centre 
of mass of the pair and $\vb x_1 - \vb x_2 = 2\,\epsilon e^{i\theta}$, are
coordinates relative to the centre of mass, the Lagrangian becomes,
\begin{align}\label{eq:lag-point-source-formalism}
    \lagrangian &= \frac{M}{2}\,\abs{\vb {\dot X}}^2 
    + \pbrk{\frac{2 m_1m_2}{M} 
    - q_1\,q_2\,\pi\,\sin^2(\varphi_\infty)
    \,K_0(2\sin(\varphi_\infty)\,\epsilon)}
    ({\dot \epsilon}^2 + \epsilon^2{\dot \theta}^2)\nonumber\\
    &= \frac{M}{2}\,\abs{\vb {\dot X}}^2 +
    (1 - \tau^2)\,\pi\,\pbrk{
    2 - q_1\,q_2\,K_0(2(1 - \tau^2)^{1/2}\,\epsilon)
    } ({\dot \epsilon}^2 + \epsilon^2{\dot \theta}^2).
\end{align}

If we get rid of the centre of mass term, we find that the conformal 
factor in the reduced moduli space is again as in 
equation~\eqref{eq:conformal-factor-asymptotics-pair}.

\let  \hf          \undefined 
\let  \gp          \undefined 
\let  \fnh         \undefined
\let  \cp          \undefined
\let  \fnhh        \undefined
\let  \cf          \undefined
\let  \gmetric     \undefined
\let  \fnv         \undefined
\let  \fng         \undefined
\let  \fnV         \undefined
\let  \cmu         \undefined
\let  \fnu         \undefined
\let  \fnF         \undefined
\let  \fnth        \undefined
\let  \energyDens  \undefined
\let  \pbrk        \undefined 
\let  \mDiff       \undefined
\let  \gpd         \undefined 
\let  \bxi         \undefined
\let  \sign        \undefined


\newcommand*{\cf}{\Omega}
\newcommand*{\fnh}{h}
\newcommand*{\fnth}{\tilde h}
\newcommand*{\pbrk}[1]{\left( #1 \right)}
\newcommand*{\curvature}{K}
\newcommand*{\Energy}{\mathrm{E}}
\newcommand*{\vol}{\mathrm{Vol}}
\newcommand*{\spL}{\mathrm{L}}

\subsection{Approximating the conformal factor in a neighbourhood of
  the singularity}\label{sec:short-range-approx}

In this section we aim to derive an asymptotic approximation to the conformal 
factor for 
small $\epsilon$, we do so finding the limit of the regular part of 
$h_\epsilon/\epsilon$ as $\epsilon \to 0$, where $h_\epsilon$ is the 
solution to the Taubes equation with vortex at $\epsilon$ and antivortex 
at $-\epsilon$ and then we prove the convergence 
is uniform in disks centred at the origin. Let us consider 
$(\epsilon, \theta)$ coordinates, 
we know $h_\epsilon$ depends smoothly on 
$\epsilon$ 
and the function $\partial_\epsilon h$ is a solution of the equation,
\begin{align}
-(\Delta + V(h_\epsilon))\,\partial_\epsilon h_\epsilon = 
4\pi\partial_1\delta_\epsilon + 4\pi\partial_1\delta_{-\epsilon}.
\label{eq:partial-eps-h-eps}
\end{align}

If $\mu = \log((1 - \tau)(1 + \tau)^{-1})$ is the limit value of 
$h_\epsilon$ as $|z| \to \infty$, we know that as $\epsilon \to 0$,  the 
potential function $V(h_\epsilon)$ converges 
pointwise to $V(\mu) = 1 - \tau^2 \in (0, 1]$ and uniformly outside of 
any neighbourhood of the origin. We also know each $\partial_\epsilon 
h_\epsilon$ decays exponentially fast as $|z| \to \infty$. 
Without loss of generality we assume  $\tau = 0$ from now onwards. 
 As the fundamental 
solution of the screened Poisson equation
\begin{align}
-(\Delta + 1) G = \delta_0,
\end{align}
with convergence $G \to 0$ as $|z| \to \infty$, is 
$(2\pi)^{-1}\,K_0(|z|)$, if we denote by  $\exp(i\theta_{\pm\epsilon})$ the 
argument of $z \mp \epsilon$, the function,
\begin{equation}
\begin{aligned}
H_\epsilon &= 2\pbrk{\partial_1 K_0(|z - \epsilon|) + \partial_1 
K_0(|z + \epsilon|)},\\
&= -2\,(\cos(\theta_\epsilon)\,K_1(|z - 
\epsilon|) 
    + \,\cos(\theta_{-\epsilon})\,
    K_1(|z + \epsilon|)),
\end{aligned}
\end{equation}
is the fundamental solution of the equation,
\begin{align}
-(\Delta + 1)\,H_\epsilon = 4\pi\partial_1\delta_\epsilon + 
4\pi\partial_1\delta_{-\epsilon}.\label{eq:partial-eps-problems-fund-sol}
\end{align}

By~\eqref{eq:partial-eps-h-eps} and~\eqref{eq:partial-eps-problems-fund-sol},
\begin{align}
(\Delta + V(h_\epsilon))\,(\partial_\epsilon h_\epsilon - H_\epsilon) 
= (1 - V(h_\epsilon))\,H_\epsilon.
\end{align}

Denoting by $f * g$ convolution on the plane, 
\begin{align}
\partial_\epsilon h_\epsilon - H_\epsilon = -\pbrk{(1 - 
V(h_\epsilon))\, 
H_\epsilon} * G_\epsilon,\label{eq:potential-conv-plane}
\end{align}
where $G_\epsilon$ is Green's function of the operator $-(\Delta + 
V(h_\epsilon))$. We aim to prove $|\partial_\epsilon h_\epsilon - H_\epsilon| 
\to 0$ 
 uniformly on the plane. To do this, we will use the concept of 
a doubling measure and prove a few common properties for the family of 
potentials $V(h_\epsilon)$. 

A measure $\nu$ is called doubling if there exists a constant $C > 0$, such 
that for any $z \in \cpx$ and $r > 0$, 
\begin{align}
\nu(\disk_{2R}(z)) \leq C\,\nu(\disk_{R}(z)).\label{eq:doub-cond}
\end{align}

Suppose $D \subset \plane$ is a measurable set with 
respect to the euclidean 
metric, we define,
\begin{align}
\nu_\epsilon(D) = \int_D V(h_\epsilon)\vol.
\end{align}

Given $\epsilon_0 > 0$, we will prove the existence of a uniform 
constant $C_d$ such that~\eqref{eq:doub-cond} holds for any $\epsilon \in (0, 
\epsilon_0)$ and a uniform constant 
$\delta >0$, such that,
\begin{align}
\nu_\epsilon(\disk_1(z)) > \delta\label{eq:delta-lower-bount-nueps}
\end{align}
for any $\epsilon \in (0, \epsilon_0)$, then, by a result of 
Christ~\cite[Thm.~1.13]{christ1991}
 there are a function 
$\varrho: \plane \to \reals^+$, a distance function $\rho: \plane \times \plane 
\to \reals$ induced by a Riemannian metric $d\rho^2$, and constants $C$, 
$\gamma$ all of them depending only on $C_d$, such that,
\begin{align}
|G_\epsilon(z_1, z_2)| \leq C\,\begin{cases}
\log (2\varrho(z_1)/|z_1 - z_2|), & |z_1 - z_2| \leq \varrho(z),\\
\exp(-\gamma \rho(z_1, z_2)), & |z_1 - z_2| \geq \varrho(z),
\end{cases}\label{eq:christs-bounds}
\end{align}
and $\rho(z_1, z_2) \geq c\,|z_1 - z_2|$ for some constant $c$ depending on 
$\delta$ but not on $C_d$. 

If $\tilde h_\epsilon = h_\epsilon + \log \abs{z - 
\epsilon} - \log \abs {z + 
\epsilon}$, we know that for any $\epsilon_0 > 0$,  there are constants $C_1$, 
$C_2$, such that for any 
 $z \in \cpx$ and $\epsilon \in (0, \epsilon_0)$,
\begin{align}
C_1 \leq e^{\tilde h_\epsilon} \leq C_2,
\end{align}
hence,
\begin{align}
\frac{4 C_1 |z - \epsilon|^2 |z + \epsilon|^2}{(C_1 |z - \epsilon|^2 + |z + 
\epsilon|^2)} \leq V(h_\epsilon) \leq 
\frac{4 C_2 |z - \epsilon|^2 |z + \epsilon|^2}{(C_2 |z - \epsilon|^2 + |z + 
\epsilon|^2)^2}.
\end{align}

Hence, there is a constant $C >0$ independent of $\epsilon$ such that,
\begin{align}
\frac{1}{C}\,\frac{4 |z-\epsilon|^2 |z+ \epsilon|^2}{(|z - \epsilon|^2 + |z + 
\epsilon|^2)^2} \leq V(h_\epsilon) \leq
C\,\frac{4 |z-\epsilon|^2 |z+ \epsilon|^2}{(|z - \epsilon|^2 + |z + 
\epsilon|^2)^2},
\end{align}
this implies the potential $V(h_\epsilon)$ induces a doubling measure if and 
only if
\begin{align}
\tilde V_\epsilon = \frac{4 |z-\epsilon|^2 |z+ \epsilon|^2  }{
(|z - \epsilon|^2 + |z + \epsilon|^2)^{2}},
\end{align}
does. 

\begin{lemma}\label{lem:doub-pot-geom}
    Let $M \in (0, 1)$, then $\tilde V_\epsilon^{-1}([0, M])$ consists of 
    two connected components, whose boundaries are the circles centred at 
    $\pm (1 - M)^{-1/2}\epsilon$ of radii $M^{1/2}(1 - M)^{-1/2}\epsilon$.
\end{lemma}

\begin{proof}
    Let $w = (z - \epsilon)(z + \epsilon)^{-1}$, if $\tilde V_\epsilon(z) = M$, 
    then, 
    \begin{align}
    \frac{4 |w|^2}{(|w|^2 + 1)^2} = M,
    \end{align}
    this equality implies,
    \begin{align}
    |w|^2 - \frac{2}{M^{1/2}} |w| + 1 = 0.
    \end{align}
    
    The roots of this equation are,
    \begin{align}
    r_\pm = \frac{1}{M^{1/2}} (1 \pm (1 - M)^{1/2}).\label{eq:lm-radii-pm-1}
    \end{align}
    
    If $z = x + y\,i$, for each root, the equation 
    \begin{align}
    \left\lvert \frac{z - \epsilon}{z + \epsilon}\right\rvert = r_\pm, 
    \end{align}
    determines the circles
    \begin{align}
    |z|^2 - 2\epsilon \pbrk{\frac{1 + r_\pm^2}{1 - r_\pm^2}} x
    + \epsilon^2 = 0,
    \end{align}
    of centres
    \begin{align}
    c_\pm = \epsilon \pbrk{\frac{1 + r_\pm^2}{1 - 
    r_\pm^2}}
= \frac{\mp \epsilon}{(1 - M)^{1/2}}\label{eq:lm-centres-pm-1}
    \end{align}
    and squared radii 
    \begin{align}
    R_\pm^2 &= |c_\pm|^2 - \epsilon^2\nonumber\\
    &= \frac{\epsilon^2}{1 - M} - \epsilon^2\nonumber\\
    &= \frac{M \epsilon^2}{1 - M}.\label{eq:lm-radii-pm-2}
    \end{align} 
    
    Mobius transformations map circles onto circles and $z = \epsilon$ is 
    mapped to $w = 0$, while $z = -\epsilon$ is mapped to $w = \infty$, where 
    for both points $4 |w|^2 (|w|^2 + 1)^{-2} = 0$, hence the disks 
    $|z - c_\pm| \leq r_\pm$ are the connected components of $\tilde 
    V_\epsilon([0, M])$.
\end{proof}

\begin{lemma}\label{lem:v1-doub-measure}
    $\tilde V_1$ defines a doubling measure.
\end{lemma}

\begin{proof}
    Assume otherwise towards a contradiction, then there exists a sequence 
  $\set{(z_n, r_n)}$ such that,
  \begin{align}
  \frac{\int_{\disk_{2r_n}(z_n)} \tilde V_{1}\; |dz|^2}
  {\int_{\disk_{r_n}(z_n)} \tilde V_{1}\; |dz|^2} \to 
  \infty\label{eq:lem-doub-pot-1}
  \end{align}   
  
  After passing to a subsequence if necessary, we can assume 
  $(z_n, r_n) \to (z_*, r_*) \in \overline \cpx \times 
  [0, 
  \infty]$, where $z_*$ could be the point at infinity, 
  meaning 
  $|z_n| \to \infty$. Through the proof we will consider a fixed but arbitrary 
  constant $M \in (0, 1)$. If $z_* \in \cpx$ 
  we consider four cases: 
  
  Case I. If $0 < r_* < \infty$, by the dominated 
  convergence theorem,
  \begin{align}
  \frac{\int_{\disk_{2r_n}(z_n)} \tilde V_{1}\; |dz|^2}
  {\int_{\disk_{r_n}(z_n)} \tilde V_{1}\; |dz|^2} \to 
  \frac{\int_{\disk_{2r_*}(z_*)} \tilde V_{1}\; |dz|^2}
  {\int_{\disk_{r_*}(z_*)} \tilde V_{1}\; |dz|^2}, 
  \end{align}   
  hence~\eqref{eq:lem-doub-pot-1} is not possible.
  
  Case II. If $r_* = \infty$, by lemma~\ref{lem:doub-pot-geom}, there is an 
  $R>0$ such 
  that $\tilde V_1(z) \geq M$ for $|z| \geq R$. Let $\Omega_n = 
  \disk_{r_n}(z_n) \setminus \disk_R(0)$. For $n$ sufficiently large $\Omega_n 
  \neq \emptyset$, moreover,
  \begin{align}
  \frac{\int_{\disk_{2r_n}(z_n)} \tilde V_{1}\; |dz|^2}
  {\int_{\disk_{r_n}(z_n)} \tilde V_{1}\; |dz|^2}
  \leq \frac{4\pi r_n^2}{M |\Omega_n|}
  = \frac{4\pi r_n^2}{M (\pi r_n^2 - |\disk_{r_n}(z_n) \cap \disk_R(0)|)}
  \to \frac{4}{M}.
  \end{align}
  
  Case III. If $r_* = 0$ and $z_* \neq \pm 1$, by the mean value theorem for 
  integrals, 
  \begin{align}
  \frac{\int_{\disk_{2r_n}(z_n)} \tilde V_{1}\; |dz|^2}
  {\int_{\disk_{r_n}(z_n)} \tilde V_{1}\; |dz|^2}
  &= 
  4 \, \frac{\frac{1}{4\pi r_n^2}\int_{\disk_{2r_n}(z_n)} \tilde V_{1}\; |dz|^2}
  { \frac{1}{\pi r_n^2}\int_{\disk_{r_n}(z_n)} \tilde V_{1}\; |dz|^2}
  \to 
  4,
  \end{align}
  since each averaged integral converges to $\tilde V_1(z_*) \neq 0$.
  
  Case IV. If $z_* = \pm 1$, assume without loss of generality 
  $z_* = 1$, let $R \in (0, 1/2)$ be any constant, for $n$ large enough, 
  the disk $\disk_{2r_n}(z_n)$ is contained in $\disk_R(1)$, then there is a 
  constant $C(R)$, such that for any $z \in \disk_R(1)$,
  \begin{align}
  \frac{1}{C}\,|z - 1|^2 \leq \tilde V_1(z) \leq C\,|z - 1|^2,
  \end{align}    
  the function $|z - 1|^2$ defines a doubling measure because it is a 
  non negative polynomial~\cite{shen1995}, implying for large $n$ the quotient 
  \begin{align}
\frac{\int_{\disk_{2r_n}(z_n)} \tilde V_{1}\; |dz|^2}
{\int_{\disk_{r_n}(z_n)} \tilde V_{1}\; |dz|^2}
\end{align}
is bounded. 

Therefore, if~\eqref{eq:lem-doub-pot-1} holds, $|z_n| \to \infty$. 
If $r_n \to r_0$ 
for $r_0 \in [0, \infty)$, for $n$ large the disk $\disk_{2r_n}(z_n)$ is in 
the exterior of the disk $\disk_{R}(0)$, hence $\tilde V_1 \in [M, 1]$, and 
 \begin{align}
\frac{\int_{\disk_{2r_n}(z_n)} \tilde V_{1}\; |dz|^2}
{\int_{\disk_{r_n}(z_n)} \tilde V_{1}\; |dz|^2}
&\leq 
\frac{4}{M}.
\end{align}

Finally, if $r_n \to \infty$, we can apply the same argument as in Case II to 
deduce that~\eqref{eq:lem-doub-pot-1} is not possible. This concludes all the 
possibilities for the sequence and proves the lemma. 
\end{proof}

If we define the change of variable $z = \epsilon w$, by 
lemma~\ref{lem:v1-doub-measure} we have,
\begin{align}
\frac{\int_{\disk_{2r}(z)} \tilde V_{\epsilon}\; |dz|^2}
{\int_{\disk_{r}(z)} \tilde V_{\epsilon}\; |dz|^2}
&= 
\frac{\int_{\disk_{2r/\epsilon}(z/\epsilon)} \tilde V_{1}\; |dw|^2}
{\int_{\disk_{r/\epsilon}(z/\epsilon)} \tilde V_{1}\; |dw|^2} < C,
\end{align}
where $C$ is independent of $\epsilon$, proving the following corollary.

\begin{corollary}
    For any $\epsilon_0 > 0$, there is a constant $C_d$ such 
    that~\eqref{eq:doub-cond} holds for any $\epsilon \in (0, \epsilon_0)$.     
\end{corollary}

\begin{lemma}
    For any $\epsilon_0 > 0$, there is a constant $\delta >0$ such that 
    \begin{align}
    \int_{\disk_1(z)}\tilde V_\epsilon > \delta,
    \end{align}
    for all $z \in \cpx$, $\epsilon \in (0, \epsilon_0)$.
\end{lemma}

\begin{proof}
    Pick $M \in (0, 1)$ such that $r_0 = M^{1/2}(1 - M)^{-1/2}\epsilon_0$ 
    satisfies $2 r_0^2 < 1$. By lemma~\ref{lem:doub-pot-geom} there are two 
    disks $D_1$, $D_2$ of radius $r < r_0$ such that in the exterior of the 
    disks $\tilde V_\epsilon \geq M$. The complement $\Omega = 
    \disk_1(z)\setminus (D_1\cup D_2)$ is non empty for any $z \in 
    \cpx$ and it has bounded area,
    \begin{align}
    |\Omega| \geq \pi - 2\,|D_1| \geq \pi\,(1 - 2 r_0^2),
    \end{align}
    hence,
    \begin{align}
    \int_{\disk_1(z)} \tilde V_\epsilon\, |dz|^2
    &\geq \int_\Omega \tilde V_\epsilon\, |dz|^2 \nonumber\\
    &\geq M\,\pi\,(1 - 2r_0^2).
    \end{align}
    
    Selecting any $\delta < M\,\pi\,(1 - 2r_0^2)$ proves the lemma.
\end{proof}

Therefore, for any $p \geq 2$, there is a constant $C > 0$ such that, 
\begin{align}
|| G_\epsilon ||_{\spL^p} < C, \qquad \forall \epsilon \in (0, \epsilon_0),
\end{align}
let us choose any $p > 2$ and let $p_* = p\, (p - 1)^{-1}$ be H\"older's 
conjugate 
of $p$, by \eqref{eq:potential-conv-plane} and H\"older's inequality,
\begin{align}
|\partial_\epsilon h_\epsilon - H_\epsilon| &\leq || (1 - V(h_\epsilon)) 
H_\epsilon||_{\spL^{p_*}}\,|| G_\epsilon ||_{\spL^p}\nonumber \\
&\leq C\,(
|| (1 - V(h_\epsilon)) K_1(|z - \epsilon|)||_{\spL^{p_*}}\nonumber\\
&\quad + || (1 - V(h_\epsilon)) K_1(|z + 
\epsilon|)||_{\spL^{p_*}}).\label{eq:partial-eps-h-eps-H-eps-bound}
\end{align}

\begin{lemma}
    \begin{align}
    \lim_{\epsilon \to 0} || (1 - V(h_\epsilon))\,K_1(|z - 
    \epsilon|)||_{\spL_{p_*}} = 0,
    \end{align}
    and a similar statement holds for $K_1(|z + \epsilon|)$.
\end{lemma}

\begin{proof}
    Let $w = z - \epsilon$, then,
    \begin{align}
    || (1 - V(h_\epsilon))\,K_1(|z - 
    \epsilon|)||_{\spL_{p_*}}^{p_*} &= 
    \int_{\plane} (1 - V(h_\epsilon(w + 
    \epsilon)))^{p_*}\,K_1(|w|)^{p_*}\,|dw|^2,
    \end{align}
    the function $K_1(|w|)$ is in $\spL^{p_*}$ for any $p_* < 2$, and 
    $(1 - V(h_\epsilon(w + \epsilon)))$ is a bounded function converging 
    pointwise to 0, by the dominated convergence theorem,
    \begin{align}
    \int_{\plane} (1 - V(h_\epsilon(w + 
    \epsilon)))^{p_*}\,K_1(|w|)^{p_*}\,|dw|^2 \to 0,
    \end{align}
    this proves the lemma for $K_1(|z - \epsilon|)$, for $K_1(|z + \epsilon|)$
    the proof is analogous.
\end{proof}

By~\eqref{eq:partial-eps-h-eps-H-eps-bound}, $|\partial_\epsilon h_\epsilon - 
H_\epsilon| 
\to 0$ uniformly on the plane as $\epsilon \to 0$. Note that the function 
\begin{align}
H_\epsilon^\tau = 
(1 - \tau^2)^{1/2}\,H_{(1 - \tau^2)^{1/2}\epsilon}((1 - \tau^2)^{1/2}\,|z|),
\end{align}
is the fundamental solution to 
\begin{align}
-(\Delta + (1 - \tau^2))\,H_\epsilon^\tau = 4\pi\,\partial_1\delta_\epsilon 
+ 4\pi\,\partial_1\delta_{-\epsilon}.
\end{align}

All the previous lemmas extend straightforwardly to conclude for any 
$\tau \in (-1, 1)$, the convergence $|h_\epsilon - H_\epsilon^\tau| \to 0$ 
uniform on the plane. Let us define the function 
\begin{equation}
f_{\epsilon} = \frac{1}{\epsilon}\brk(h - \log\,\abs{z - \epsilon}^2 
+ \log\,\abs{z + \epsilon}^2 - \log 
\frac{1 - \tau}{1 +
    \tau}),
\end{equation}
For this function at $\tau = 0$, Rom\~ao and Speight conjectured 
in~\cite{romao2018}, the uniform limit $f_\epsilon \to f_*$
%
%
%
%
%
%
%
where for general $\tau$, $f_{*}$ is a solution to the problem,
  \begin{align}
    -(\laplacian + 1 - \tau^2)\, f_{*} &=  
    - 4(1 - \tau^2)\,\frac{z_1}{\abs{z}^2},  \\
    \lim_{\abs{z} \to \infty} f_{*} &= 0, \\
    \lim_{\abs{z} \to 0} f_{*} &= 0.
  \end{align}

The equation for $f_{*}$ can be solved exactly, for $\tau = 0$, they found,
\begin{align}
f_* &= 4\,\frac{z_1}{|z|^2}\,(1 - |z| K_1(|z|))\nonumber\\
&= 4\,\cos(\theta)\,\pbrk{\frac{1}{|z|} - K_1(|z|)}.
\end{align}

If we define,
\begin{align}
f_*^\tau(z) = (1 - \tau^2)^{1/2}\,f_*((1 - \tau^2)^{1/2}\,z),
\end{align}
then $f_*^\tau$ is the conjectured limit for general $\tau$. 
%

Numerical evidence suggests $C^1$ uniform convergence as can be 
seen in  
Figure~\ref{fig:asympt-approx-fstar}. In the next proposition, we prove that in 
fact, the convergence is uniform 
at least in $C^0(\disk_R(0))$ for any disk centred at the origin.
\begin{figure}[p]
    \centering
    \includegraphics[width=.90\textwidth]{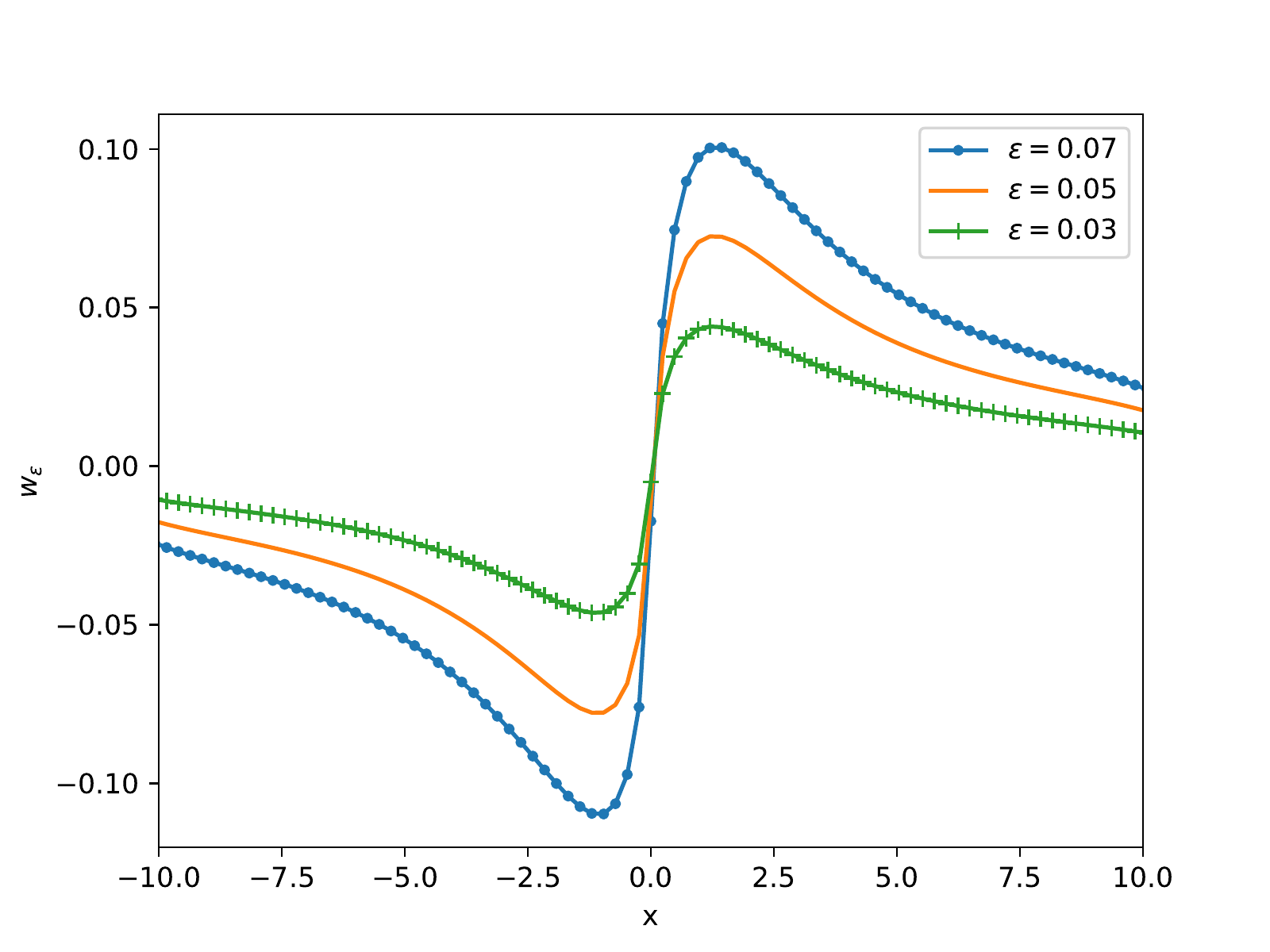}
    \includegraphics[width=.90\textwidth]{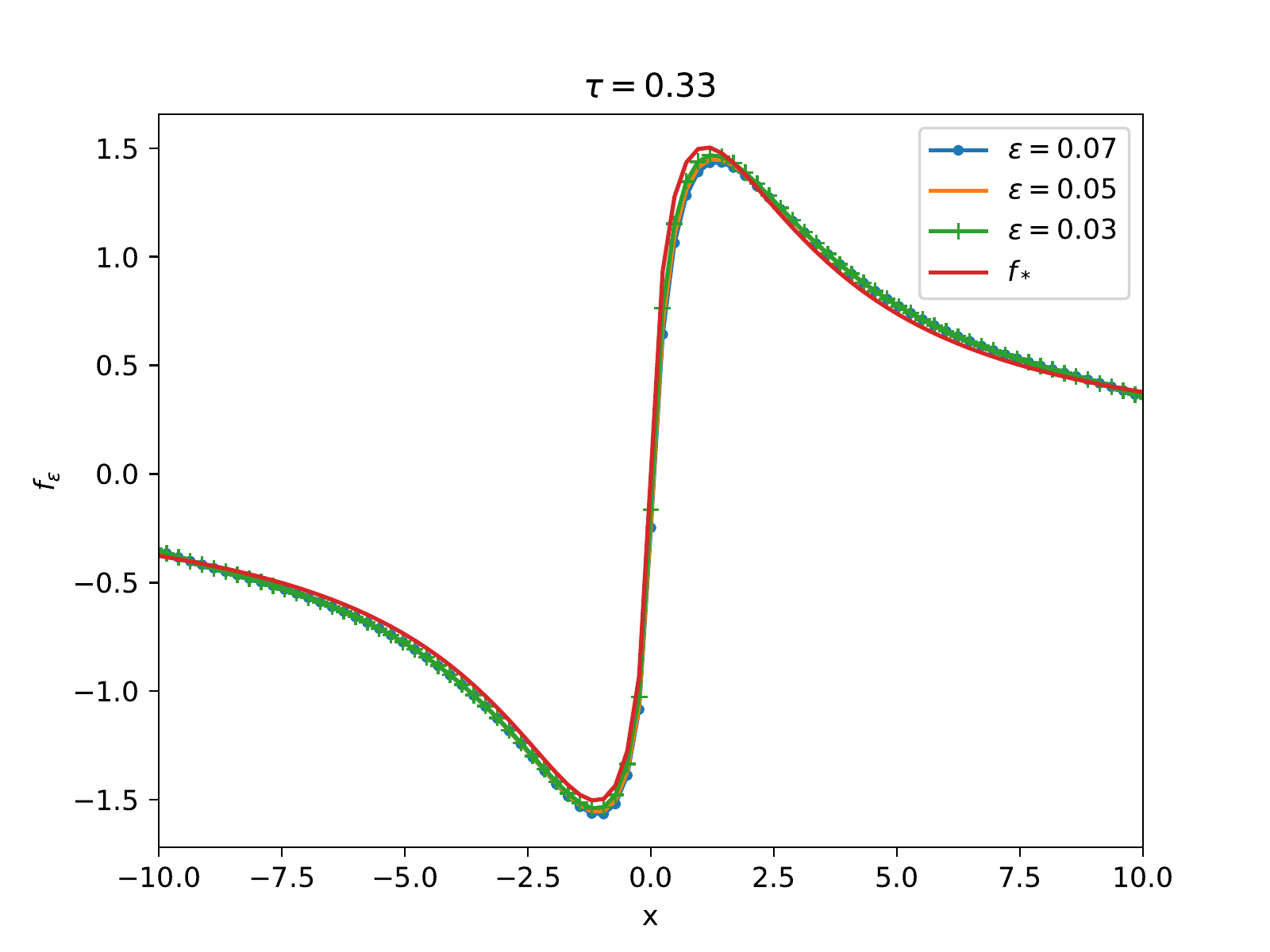}
    \caption{\textbf{Top.} Real profile of the functions \(\tilde h_{\epsilon} -
        \log(1-\tau)(1 + \tau)^{-1} \) converging uniformly to 0 on the real 
        axis. \textbf{Bottom.} Real profile of the functions \(f_{\epsilon}\)
        on the real axis and the conjectured asymptotic
        limit. In both cases, $\tau = 0.33$}\label{fig:asympt-approx-fstar}
\end{figure} 

\begin{proposition}\label{prop:f-eps-unif-conv}
    For any $R>0$, 
    $f_\epsilon \to f_*^\tau$ uniformly on $\overline{\disk_R(0)}$.
\end{proposition}

\begin{proof}
    Let
    \begin{align}
    \hat h_\epsilon = h - \log |z - \epsilon|^2 + \log |z + \epsilon|^2 - \log 
    \frac{1 - 
        \tau}{1 + \tau},
    \end{align}
    we know 
    $\hat h$ is smooth with respect to both $z$ and $\epsilon$ and 
    $\hat h_\epsilon \to 0$ in $C^1(\overline{\disk_R(0)})$ as $\epsilon \to 
    0$.   
    Let $\epsilon > 0$ be a given positive number, by the mean value theorem, 
    for any $z$ there is another 
    $\epsilon' \in (0, \epsilon)$ that may depend on $z$, such that, 
    \begin{align}
    f_\epsilon(z) = \partial_\epsilon \hat h_\epsilon|_{\epsilon'}(z),
    \end{align}
    hence, to prove the statement it is sufficient to show that 
    $\partial_\epsilon \hat h_\epsilon \to f_*^\tau$ uniformly on 
    $\overline{\disk_R(0)}$. Since,
    \begin{align}
    \partial_\epsilon \hat h_\epsilon = \partial_\epsilon h_\epsilon
    + \frac{2 \cos(\theta_\epsilon)}{|z - \epsilon|} 
    + \frac{2 \cos(\theta_{-\epsilon})}{|z + \epsilon|}, 
    \end{align}
    the convergence,
    \begin{align}
    \left\lvert
    \partial_\epsilon\hat h_\epsilon - H_\epsilon^\tau - 
    \frac{2 \cos(\theta_\epsilon)}{|z - \epsilon|} 
    - \frac{2 \cos(\theta_{-\epsilon})}{|z + \epsilon|}\right\rvert \to 
    0\label{eq:partial-eps-hat-h-eps-conv}
    \end{align}
    is uniform in $\overline{\disk_R(0)}$. Let,
    \begin{align}
    \tilde f^\tau(z) = 2 \cos(\theta) \pbrk{\frac{1}{|z|} - (1 - 
    \tau^2)^{1/2}\,K_1((1 - \tau^2)^{1/2}\,|z|)},
    \end{align}
    $\tilde f^\tau$ is a continuous function defined on the compact set 
    $\overline{\disk_R(0)}$, hence, it is equicontinuous, moreover,
    note that,
    \begin{align}
    H_\epsilon^\tau + 
    2\,\pbrk{\frac{\cos(\theta_\epsilon)}{|z - \epsilon|} 
    + \frac{\cos(\theta_{-\epsilon})}{|z + \epsilon|}} &= 
    \tilde f^\tau(z - \epsilon) + \tilde f^\tau(z + \epsilon),
    \end{align}
    since $\tilde f^\tau$ is equicontinuous, on $\overline{\disk_R(0)}$ we have 
    the 
    uniform convergence,
    \begin{align}
    \lim_{\epsilon \to 0}\, (\tilde f^\tau(z - \epsilon) 
    + \tilde f^\tau(z + \epsilon)) =  
    2\,\tilde f^\tau(z) = f^\tau_*(z).\label{eq:lim-tilde-f-eps}
    \end{align}
    
    By~\eqref{eq:partial-eps-hat-h-eps-conv} 
    and~\eqref{eq:lim-tilde-f-eps} $\partial \hat h_\epsilon \to f_*^\tau$ 
    uniformly 
    on $\overline{\disk_R(0)}$ as claimed. This concludes the proof of the 
    proposition. 
\end{proof}

Numerics together with proposition~\ref{prop:f-eps-unif-conv} suggest we can 
extend our claim about uniform convergence to higher order derivatives. In the 
following, we assume the asymptotic expansion, 
\begin{equation}
  \label{eq:gov-elliptic-approx-asympt}
  \fnh_{\epsilon}= \epsilon f_{*}^\tau + \log \abs{z - \epsilon}^2 - \log
  \abs{z + \epsilon}^2 + \log \frac{1 - \tau}{1 + \tau},
\end{equation}
is also valid for derivatives of $h_\epsilon$. With this expression, it is  
possible to derive
an asymptotic approximation to the conformal factor for small
$\epsilon$ as well.

Using the identities
\begin{align}
K_0' = -K_1, &&  K_1' = - K_0 - \frac{1}{x_1} K_1,
\label{eq:bessel-props}
\end{align}
and defining $\nu = (1 - \tau^2)^{1/2}$ to shorten the notation, 
we obtain the approximation
\begin{align}
  \cf (\epsilon) = 4\pi \nu^2
  \pbrk{
 1 + 2\,\pbrk{
 (2 - \nu)
 K_0( \nu \epsilon )
 - \nu \epsilon K_1 ( \nu \epsilon )}
    },\label{eq:lambda-eps}
\end{align}
valid for small \(\epsilon
\).

\section{Numerical approximation to the metric}\label{sec:num-approx-metric}

To approximate the conformal factor numerically, we define $\fnth =
\fnh - \log\, \abs{x - \epsilon}^2 + \log\, \abs{x + \epsilon}^2$. $\fnth$ is 
the  
solution of the regularised equation, 
\begin{equation}
\label{eq:vav-reg}
\begin{aligned}
-\laplacian \fnth &= 2 \brk(\frac{
    \abs {x - \epsilon}^2 e^{\fnth} - \abs{x + \epsilon}^2
}{
    \abs {x + \epsilon}^2 e^{\fnth} + \abs{x + \epsilon}^2
} + \tau), \qquad 
\lim_{\abs x \to \infty} \fnth &= \log \frac{1 - \tau}{1 + \tau}.
\end{aligned}
\end{equation}

Since $\fnth$ is symmetric with respect to the $x_1$ axis, the
regularised Taubes equation was solved with an over-relaxation method
on the domain $-10 \leq x_1 \leq 10$, $0 \leq x_2 \leq 10$. The domain
was discretized with a square grid of size 0.1 as
in~\cite{samols1992}. The initial condition was taken as a
superposition of an approximated vortex and an antivortex as,
\begin{equation}
\label{eq:vav-numeric-fnth0}
\fnth_{0} = \log(\rho^2(R_+)) - \log(R_+^2)
-\log(\rho^2(R_-)) + \log(R_+^2) - \log(\mu),
\end{equation}
where $\rho = \tanh(0.6 r)$, $R_{\pm}$ is the distance of a point
in the grid to $\pm\epsilon$ and $\mu = (1-\tau)(1+\tau)^{-1}$.

The non-trivial term in the metric was computed as,
\begin{align}
\dv{\epsilon}\pbrk{\epsilon\,b(\epsilon)} 
= \dv{\epsilon}\pbrk{\epsilon \del_1 \fnth(\epsilon)}.
\end{align}

Figure~\ref{fig:cf-tau} shows the 
conformal factor 
for various values of $\tau$. Motivated by the asymptotic approximations, 
conformal factor data was 
interpolated by a curve 
\begin{align}
\hat \Omega = A + B\,K_0(2\epsilon).\label{eq:hat-omega}
\end{align}

The 
interpolation showed to explain $99\%-96\%$ of the data, depending on the 
value of $\tau$. As can be seen in the figure, the metric 
flattens as $\tau \to 1$,
preserving the singularity at the
origin. 

\begin{figure}
    \centering
    \includegraphics[scale=.9]{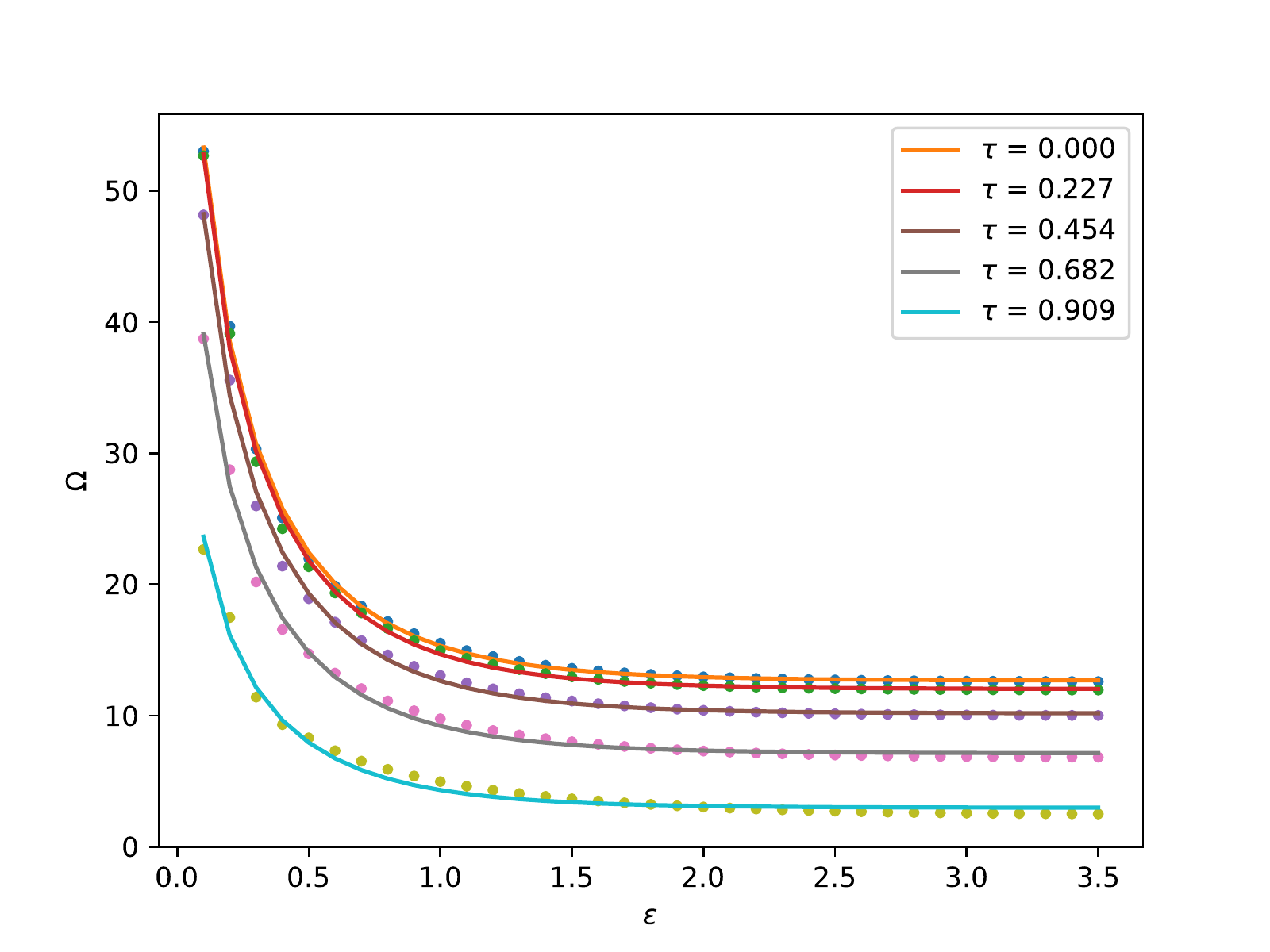}
    \caption{Conformal factor of the metric for some values of 
        $\tau$. The graph shows that as $\tau$ increases from 0, the metric 
        flattens, maintaining its singularity at the origin.}
    \label{fig:cf-tau}
\end{figure}

Figure~\ref{fig:conformal-factor-approx-short-long} shows the
short and long range approximations to the conformal factor for the
symmetric case and for $\tau = 0.909$. As can be seen in the figure, the 
approximations are consistent with the data, with the long range
approximation slightly better in the range of $\epsilon$ that the Taubes  
equation 
was solved. 

\begin{figure}[p]
    \centering
    \includegraphics[scale=.80]{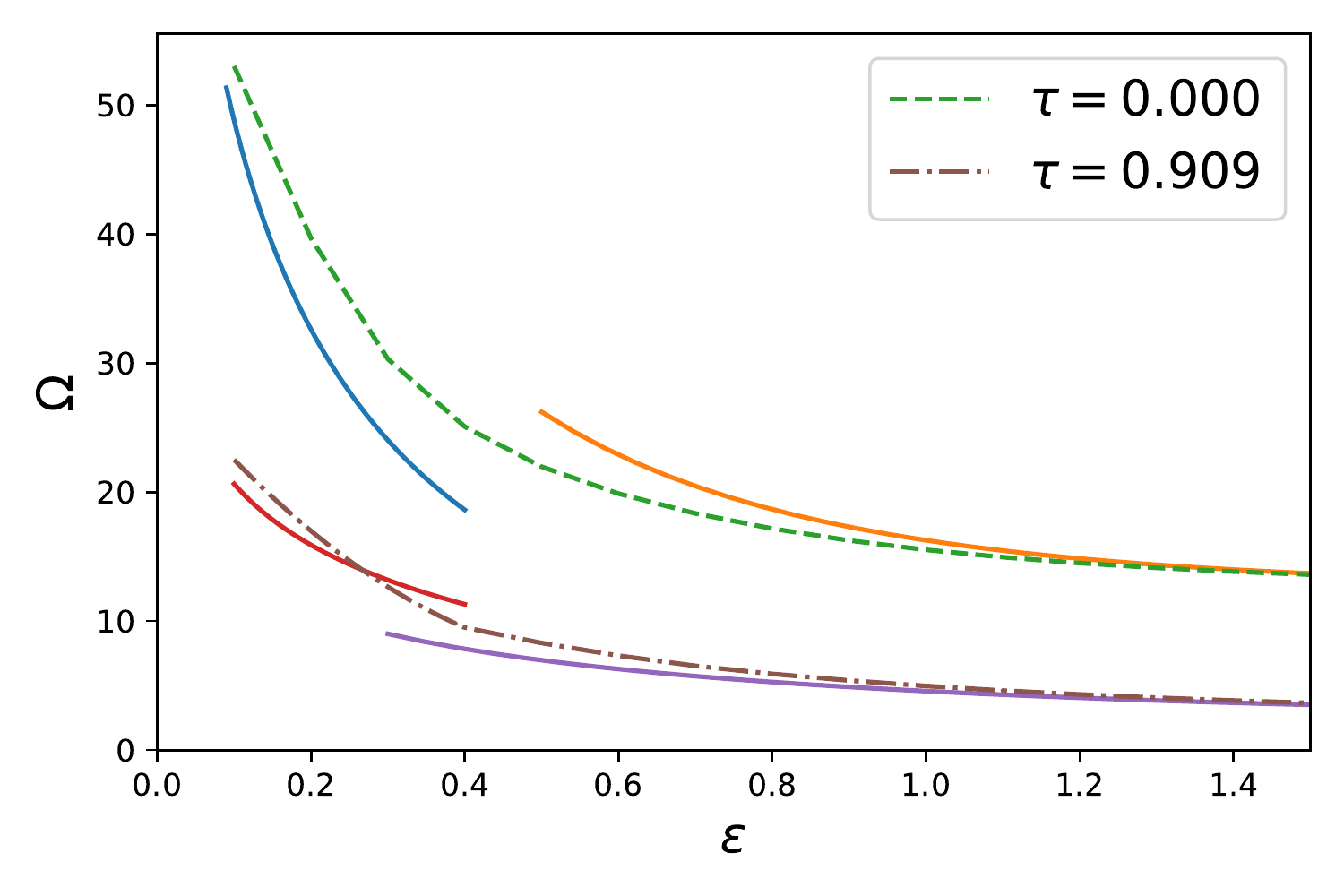}
    \caption{ Short and long range approximation to the conformal factor 
        for the symmetric case and a highly asymmetric configuration. The
        graph shows how the approximations fit the  
        numerical data in these 
        cases.}\label{fig:conformal-factor-approx-short-long}
\end{figure}

Figure~\ref{fig:curv-conf-factor} shows the Gaussian curvature
computed from the conformal factor,
\begin{equation}
\curvature = -\frac{1}{2\epsilon\,\cf}\dv{\epsilon}\pbrk{
    \epsilon \dv{\epsilon} \log \Omega
}.
\end{equation}
As can be seen in figure~\ref{fig:curv-conf-factor}, the curvature diverges to 
$\infty$ as
$\epsilon \to 0$, while on the other hand, for large $\epsilon$, it is
negative and decays exponentially fast to $0$ as $\epsilon\to\infty$.  
The moduli space can be realised as an embedded surface in $\reals^3$, 
we used proposition~2.3 of~\cite{hwang2003} to compute the embedding 
shown in figure~\ref{fig:embedding-r3}.

\begin{figure}
    \centering
    \includegraphics[width=.8\textwidth]{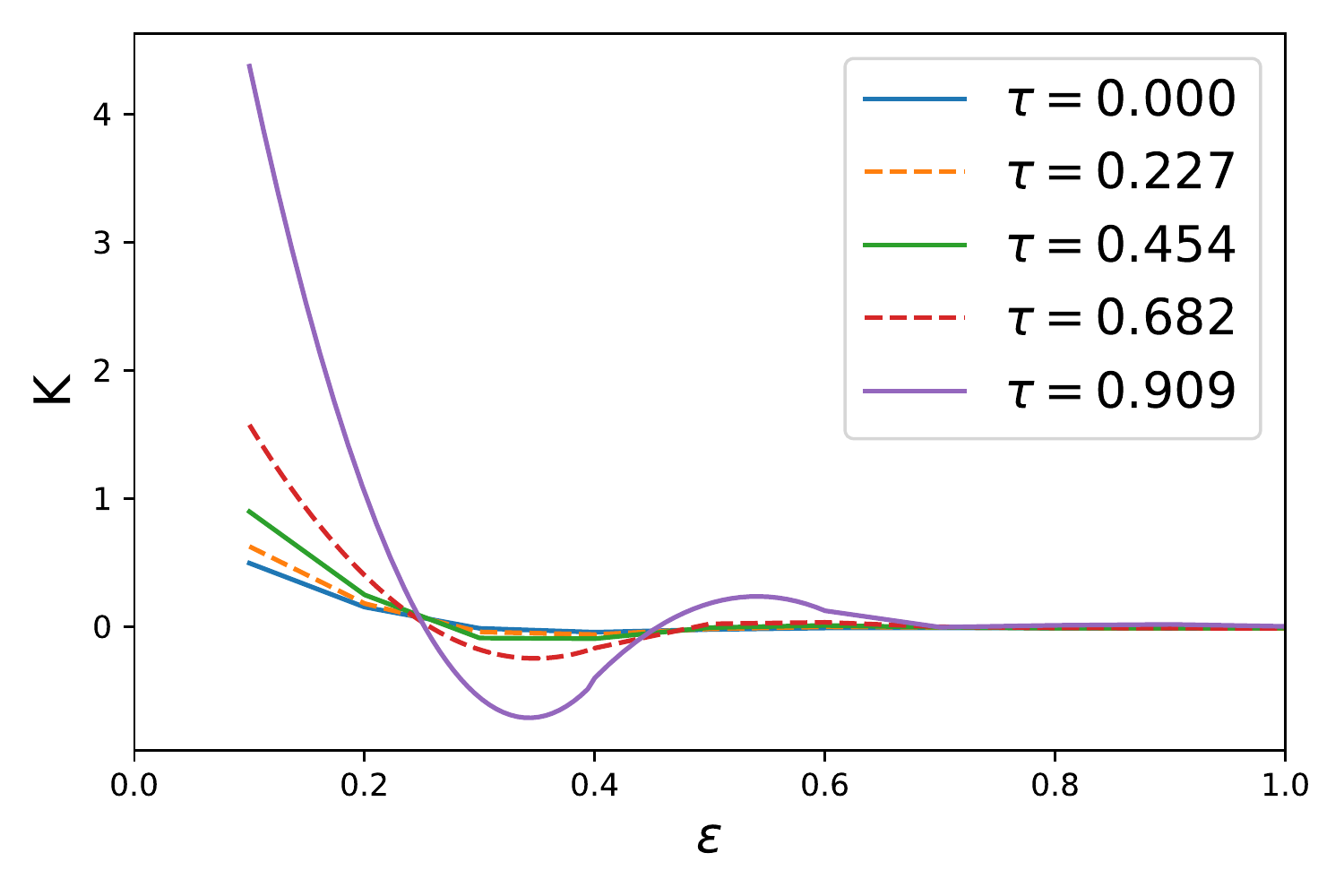}
    \caption{Curvature of the conformal factors}\label{fig:curv-conf-factor}
\end{figure}

\begin{figure}
    \centering
    \includegraphics[width=.8\textwidth]
    {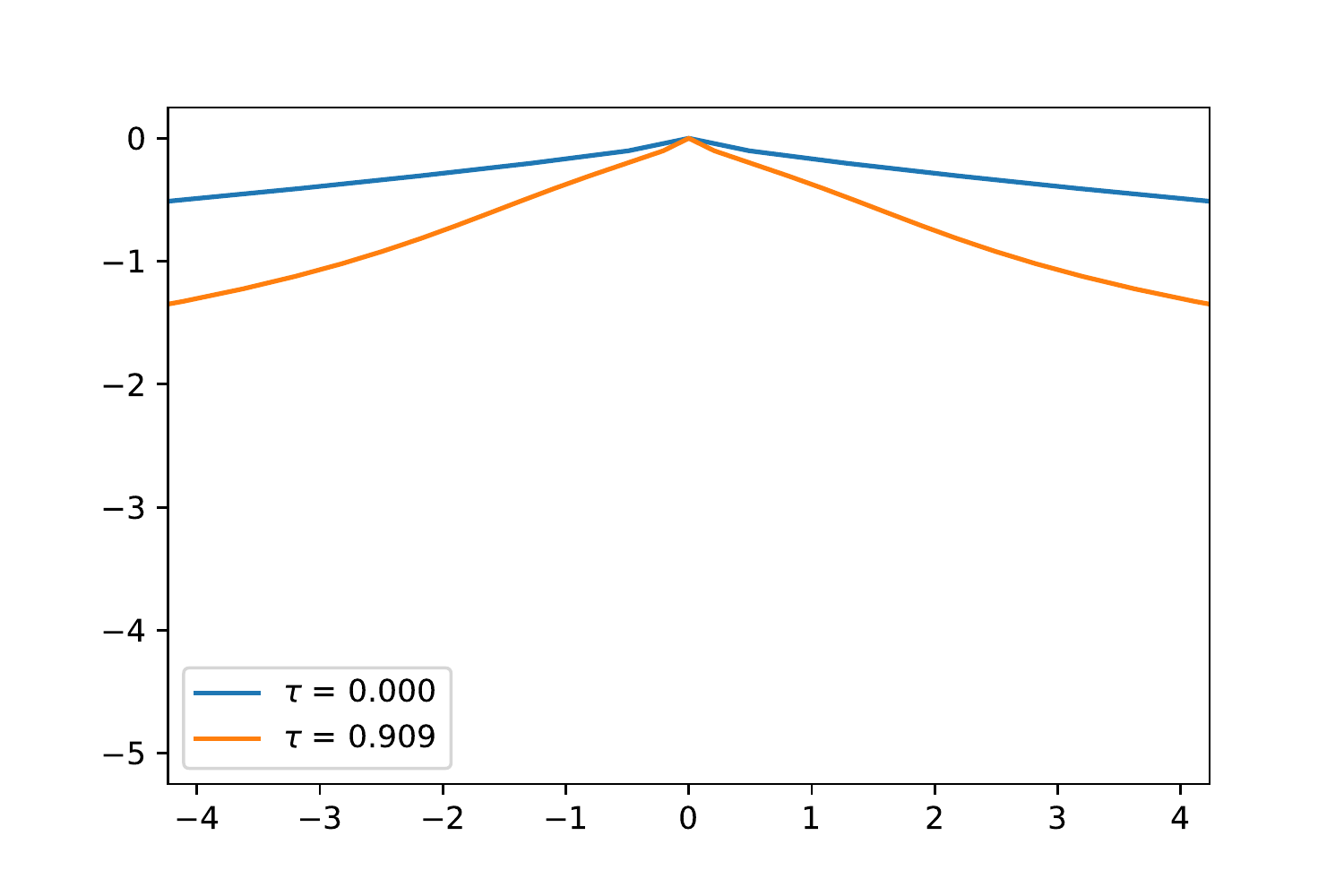}
    \caption{The image shows the profile of the moduli space as an embedded 
    revolution surface in $\reals^3$. 
    The data shows that the moduli spaces embed as 
    flat disks at infinity, with  
    infinite gaussian curvature at the origin.}\label{fig:embedding-r3}    
\end{figure}

Assuming the asymptotic approximations for large and 
small $\epsilon$, total Gaussian curvature can be shown to be zero, since 
total curvature is,
\begin{align}
2\pi\,\int_0^\infty \curvature(\epsilon)\,\epsilon\,\Omega(\epsilon)\,d\epsilon
= -\pi \pbrk{\epsilon 
\,\left.\frac{\Omega'(\epsilon)}{\Omega(\epsilon)}}\right\rvert_0^\infty.
\label{eq:total-curv-integral}
\end{align}

By~  \eqref{eq:conformal-factor-asymptotics-pair}, $\lim_{\epsilon \to 
\infty}\Omega(\epsilon) = 4\pi(1 - \tau^2)$, while 
 \begin{align}
 \Omega' = 4\pi (1 - \tau^2)^{3/2} q_1q_2 K_1(2(1-\tau^2)^{1/2}\epsilon),
 \end{align}
 hence $\lim_{\epsilon \to \infty} \epsilon\,\Omega'\,\Omega^{-1} = 0$ since 
 $K_1$ decays exponentially. For small $\epsilon$, we know $\Omega$ diverges 
 as $|\log \epsilon|$ according to~\eqref{eq:lambda-eps} while 
 $\epsilon\,\Omega'$ remains bounded, since $\Omega'$ diverges as  
 $\epsilon^{-1}$, then $\lim_{\epsilon \to 0} \epsilon\,\Omega'\,\Omega^{-1} = 
 0$. By~\eqref{eq:total-curv-integral} the total Gaussian curvature in 
 $\moduli_0^{1,1}$ is 0.

\subsection{Scattering}\label{sec:vav-scattering}

In this section we study the scattering of vortex-antivortex pairs in
the centre of mass frame. In the centre of mass frame, total momentum
is zero and the system preserves energy and angular momentum. 
For a trajectory on $(\epsilon, \theta)$ coordinates,
\begin{align}
  \Energy = \half\cf(\epsilon)\,(\dot\epsilon^2 + \epsilon^2\dot\theta^2), &&
  \ell = \cf(\epsilon)\,\epsilon^2\,\dot\theta.\label{eq:energy-ang-mom-conserv}
\end{align}

Hence, $\epsilon(t)$ is a solution to the autonomous system,
\begin{align}\label{eq:close-geod-condition}
  \dot\epsilon = \pbrk{\frac{2\Energy}{\cf} -
    \frac{\ell^2}{\cf^2\epsilon^2}}^{1/2}.
\end{align}
Equation~\eqref{eq:close-geod-condition} yields a necessary condition 
for the existence of closed geodesics, if  $\epsilon_0$ is the 
radial position of a closed geodesic,
\begin{align}\label{eq:close-geod-position}
2\Energy\,\cf(\epsilon_0)\,\epsilon_0^2 = \ell^2,
\end{align}
however, the right hand side of equation~\eqref{eq:close-geod-condition} 
is not differentiable at $\epsilon_0$ and therefore the  
fundamental theorem of existence and uniqueness of solutions of 
ordinary differential equations is not applicable  
and~\eqref{eq:close-geod-position} 
is not sufficient.
  
Based on our calculations, we assume for large separations
$2\epsilon$, the conformal factor is approximately constant,
\begin{align}
 \Omega_{\infty} =  4\pi (1 - \tau^2),
\end{align}

Suppose on the centre of mass frame a vortex moves from very far on
the left with initial
speed $v$ parallel to the $x$-axis towards an antivortex. Hence the
antivortex seems to move from far on the right towards the vortex with
initial speed $v' = (1 - \tau)(1 + \tau)^{-1}v$. We define
our impact parameter $a$ as the distance of the instantaneous initial
trajectory of the vortex to the $x$-axis as shown in the following diagram.

\begin{figure}[h]
    \centering
    \begin{tikzpicture}
    \draw[dashed] (-6,0) -- (6, 0) node [anchor=west] {$x$};
    \draw (-5, 2) -- (-5, 0) node [anchor=east, midway] {$a$};
    \draw [->] (-5, 2) -- (-4, 2) node [anchor=west] {$v$};
    \draw [->] (5, -1) -- (4, -1) node [anchor=east] {$v'$};
    \draw (-5, 2) -- (5, -1);
    \fill (5/3, 0) circle[radius=2pt, gray];
    \node at (5/3, 0) [anchor=north east] {$C.M.$};
    \end{tikzpicture}
    \caption{Scattering geometry with respect to the centre of mass.}
\end{figure}
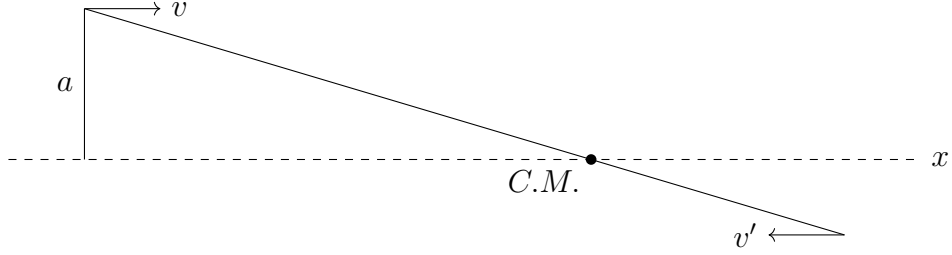

The total energy and angular 
momentum of the system are,
\begin{align}
  \Energy = \half\,\cf_{\infty}\,v^2, &&
  \ell = \cf_{\infty}\,a\,v.
\end{align}

Energy and angular momentum relate as,
\begin{align}
  \Energy = \half\frac{\ell^2}{\cf_{\infty}a^2},
\end{align}
and if we assume $\theta$ depends implicitly
on time as a function $\theta(\epsilon)$, we have,
\begin{align}
  \frac{d\theta}{d\epsilon} = \frac{\dot\theta}{\dot\epsilon} =
  \frac{\ell}{\epsilon\pbrk{2\Energy\,\cf\,\epsilon^2 - \ell^2}^{1/2}}.
\end{align}

The total deviation angle of the trajectory, $\Delta\theta$, from the
initial time to $\epsilon_{\min}$, at the moment of minimum approach
of the pair, therefore is,
\begin{align}
  \Delta\theta &= -\int_{\epsilon_{\min}}^{\infty} \frac{d\epsilon}
{\epsilon\,\left(\frac{\cf \, \epsilon^2}{\Omega_\infty a^2} -
    1\right)^{1/2}}.\label{eq:delta-theta-scattering}
\end{align}

As for a classical mechanical system, we define the deflection angle
as~\cite{goldstein2002classical},
\begin{align}
\Theta = \pi + 2\,\Delta\theta.
\end{align}

To compute $\epsilon_{\min}$ we used a secant method to solve the equation
\begin{align}
  \cf(\epsilon)\,\epsilon^2 - \cf_{\infty}\,a^2 = 0.
\end{align}

Then we used the numerical library \texttt{scipy} to compute the integral based 
on the approximation $\hat\Omega$. In practice, we chose a
small $\delta\epsilon$ and a value $\epsilon_{\max}$ for which our
data showed the conformal factor was almost constant. Then we computed
the integral,
\begin{align}
  \Delta\theta_1 = -\int_{\epsilon_{\min} +
    \delta\epsilon}^{\epsilon_{\max}} \frac{d\epsilon} 
{\epsilon\,\left(\frac{\cf \, \epsilon^2}{\Omega_\infty a^2} -
    1\right)^{1/2}},
\end{align}
and added the result to
\begin{align}
  \Delta\theta_2 = -\int_{\epsilon_{\max}}^{\infty} \frac{a\,d\epsilon}{
    \epsilon\,\pbrk{\epsilon^2 - a^2}^{1/2}
  } =
  -\frac{\pi}{2} + \arctan\pbrk{\frac{(\epsilon_{\max}^2 - a^2)^{1/2}}{a}}.
\end{align}

The result of our computations can be seen on
figure~\ref{fig:scattering}. The deflection angle is negative, hence a
vortex-antivortex pair behaves as a pair of attractive point
particles, however, we would not expect bound orbits because 
as the impact parameter decreased, the angle also decreased
until reaching a minimum, then is started growing again. The behaviour of the 
scattering angle can be explained based on the 
approximation~\eqref{eq:hat-omega}. We assume $\Omega$ is a monotonous, 
decreasing function,  such that, 
\begin{equation}
\begin{aligned}
\Omega(\epsilon) &\geq \Omega_\infty, \\
\Omega(\epsilon) &\approx -C \log 
\epsilon, \qquad \epsilon << 1.
\end{aligned}\label{eq:cf-asymp-props-scattering}
\end{equation}
where $C > 0$ is some constant, and such that there are positive constants 
$C_1$, $C_2$ such that,
\begin{align}
-C_1 \leq \Omega'(\epsilon)\,\epsilon < 0, &&
0 < \Omega''(\epsilon)\,\epsilon^2 
\leq C_2.
\end{align}  

Note that the approximation $\hat \Omega$ and the asymptotic approximations for 
small and large $\epsilon$ are consistent with these assumptions. Since 
for small $\epsilon$,
\begin{align}
\frac{d}{d\epsilon}\,(\Omega(\epsilon)\,\epsilon^2)
= (\Omega'(\epsilon)\,\epsilon + 2\,\Omega(\epsilon))\,\epsilon > 0,
\end{align}  
with these assumptions, there is a continuous bijection between small impact 
parameters 
$a$ and solutions $\epsilon_{\min}$ to the equation,
\begin{align}
\Omega(\epsilon)\,\epsilon^2 = \Omega_\infty\,a^2.\label{eq:eps-min}
\end{align}

If we use the approximation $\hat \Omega$ instead of $\Omega$, 
this is actually a global bijection valid for any $a > 0$. We aim to show that,
\begin{align}
\lim_{a \to 0}\,\Delta\theta  = -\frac{\pi}{2},
\end{align}
where $\Delta \theta$ is the integral \eqref{eq:delta-theta-scattering}. 
From now onwards we denote $\epsilon_{\min}$ as $m$ to shorten the following 
computations. 
With the change of variables $u = \epsilon / m$, the integral 
transforms into,
\begin{align}
\Delta\theta &= - \int_1^\infty \frac{d u}{u \,
    \pbrk{\frac{\Omega(m\cdot u)\,m^2\,u^2}{\Omega_\infty\,a^2} - 
    1}^{1/2}}\nonumber\\
&= - \int_1^\infty \frac{d u}{u \,
    \pbrk{\frac{\Omega(m\cdot u)}{\Omega(m)}\,u^2 - 
    1}^{1/2}},\label{eq:integrand-scattering}
\end{align}
where in the last step we used~\eqref{eq:delta-theta-scattering}. 
By~\eqref{eq:cf-asymp-props-scattering}, for any 
$u \geq 1$, we have pointwise convergence,
\begin{align}
\lim_{m \to 0} \frac{\cf(m\cdot u)}{\cf(m)} = 1.
\end{align}

To compute the integral by the dominated convergence theorem, we need to 
exhibit a function integrable in $[1, \infty)$ and bigger than each of the 
functions in the integrand of~\eqref{eq:integrand-scattering}. To this end, 
let us define the function
\begin{align}
f(u) = \frac{\cf(m\cdot u)}{\cf(m)}\,u^2 - 1,
\end{align}
as a short cumputation shows,
\begin{align}
f'(1) &= 2 + \frac{m\,\cf'(m)}{\cf(m)}\\
f''(u) &= \frac{1}{\cf(m)}\,\pbrk{
2\,\cf(m\cdot u) + 4 m\,\cf'(m\cdot u)\,u + m^2\,\cf''(m\cdot u)\,u^2} 
\label{eq:df2-scattering}.
\end{align}


Assume $f''(u) \geq 0$ for any $u \geq 
1$. 
By Taylor's theorem, for any $u > 1$, there is some $\xi \in (1, u)$, such that
\begin{align}
f(u) = f'(1)\,(u - 1) + \frac{1}{2}f''(\xi)(u - 1)^2 > f'(1)(u - 1).
\end{align}

Since $f'(1) > 2$ for any $m > 0$, we deduce, 
\begin{align}
\int_1^\infty \frac{d u}{u \,
    \pbrk{\frac{\Omega(m\cdot u)}{\Omega(m)}\,u^2 - 
        1}^{1/2}} & = \int_1^\infty \frac{d u}{u \,
    f(u)^{1/2}}\nonumber\\
 &< \frac{1}{\sqrt 2}\,\int_1^\infty \frac{du}{u\,(u - 1)^{1/2}}\nonumber\\
 &= \frac{\pi}{\sqrt 2}.\label{eq:deflection-integral-bound}
\end{align}

Hence, by the dominated convergence theorem,
\begin{align}
\lim_{a \to 0}\Delta\theta &= 
-\lim_{m \to 0} \int_1^\infty \frac{du}{u f(u)^{1/2}}\nonumber\\
&= -\int_1^\infty \frac{du}{u\,(u^2 - 1)^{1/2}}\nonumber\\
&=  -\frac{\pi}{2},
\end{align}
provided $f''(u)$ is non-negative, or equivalently 
by~\eqref{eq:df2-scattering}, if 
\begin{align}
2\cf(x) + 4 x \cf'(x) + x^2\cf''(x) \geq 0,\label{eq:cf-second-der-bounds}
\end{align}
for all $x > 0$. By the asymptotic properties of $\Omega$, we know this is 
the case for small and large $x$, which shows it is sensible to assume this 
is the case, at least for not very large $\tau$, as 
figure~\ref{fig:cf-second-der-bounds} shows.  

Therefore, the total deflection satisfies,
\begin{align}
\lim_{a \to 0} \Theta = \pi + 2 \lim_{a \to 0}\Delta\theta = 
0,\label{eq:lim-defl-angle-scattering}
\end{align}
as shown in figure~\ref{fig:scattering}. Finally,  
equation~\eqref{eq:deflection-integral-bound} shows $\Delta \theta > -\pi/\sqrt 
2$ at least up to some $\tau$, hence, the data suggests the lower bound,
\begin{align}
\Theta > -(1 - \sqrt{2})\pi \approx -74.5^\circ,
\end{align}
as can be seen in the figure.
\begin{figure}
    \centering
    \includegraphics[width=\textwidth]
    {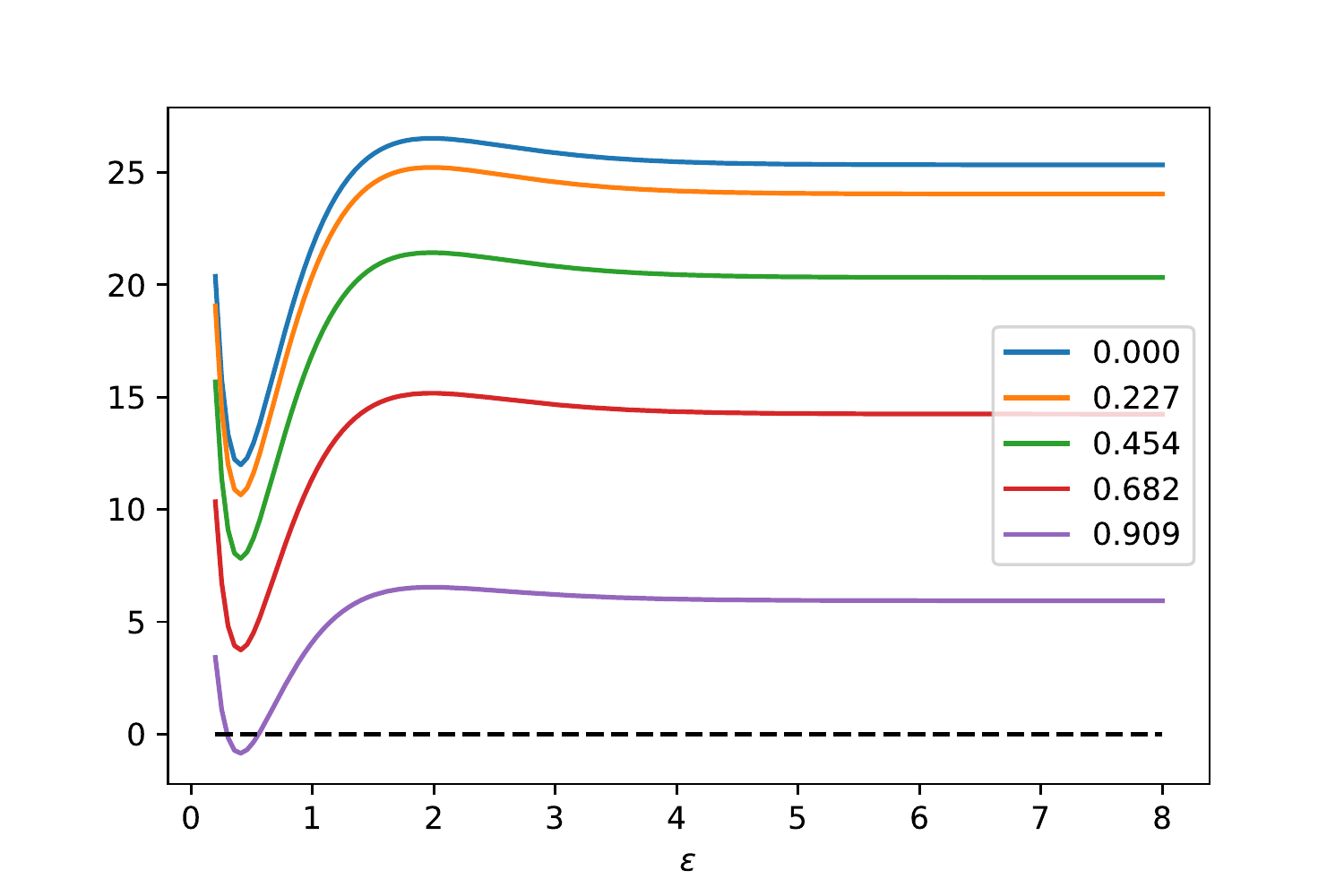}
    \caption{
    The graph shows the function $2\hat\cf(x) + 4 x \hat\cf'(x) + 
    x^2\hat\cf''(x)$ for various values of $\tau$, where $\hat\cf(x) = 
    A\,K_0(2x) + B$, and the coefficients $A$, $B$ are chosen such that 
    $\hat\cf$ interpolates the 
    values 
    of $\cf$ computed solving Taube's equation. The data 
    shows equation~\eqref{eq:cf-second-der-bounds} is expected to hold for 
    $\tau$ up 
    to some value $\tau_{\max}$, implying the deflection angle converges to 
    0 as the impact parameter decreases.
}\label{fig:cf-second-der-bounds}
\end{figure}

\subsubsection{Scattering at large separation}

We also approximated the scattering angle of a vortex-antivortex pair 
at large separation with the method Manton and 
Speight~\cite{manton-speight-asymptotic}. Suppose \(x(s) \) is a geodesic in
Cartesian coordinates, with initial position \(x(0)\), such that
\(a =x_2(0)\) is very big, \(x_1(0) \ll 0\) and the initial velocity is
\(\dot{x}(0) = v\,\del_1\). The geodesic equation for $x_2$ is
\begin{equation}\label{eq:geodesic-equations}
    \ddot{x}_2 + \frac{\cf'}{\cf} \pbrk{
    \dot\epsilon \dot x_2 - \frac{x_2}{2\epsilon} 
    (\dot x_1^2 + \dot x_2^2)} = 0.
\end{equation} 

Since  \(a\) is big, the metric is
almost flat across the 
trajectory of the geodesic, the small deflection in the \(x_2 \)
axis is caused by the small correction on \(\dot x_2 \) due to the
conformal factor derivative. To leading order, \(\cf \) is constant but
we take $\cf'$ varying as in the long range approximation. 
\begin{equation}\label{eq:uprime-approx}
  \frac{\cf'}{2\cf} = 
  \half (1 - \tau^2)^{1/2}q_1 q_2 K_1(2 (1-\tau^2)^{1/2}\epsilon).
\end{equation}

We approximate \(x_2 \) and \(\dot{x}_1 \) as constants, \(\dot{x}_2
\) as a small number, such that the leading order term for \(\ddot{x}_2 \) is,

\begin{equation}\label{eq:vav-scattering-y-deflection}
  \ddot{x}_2 = \frac{\cf'}{2\cf}\,\frac{a v^2}{\epsilon}.
\end{equation}

For big \(a \) the deflection is small, the deviation angle can be
approximated as

\begin{equation}
  \label{eq:dev-angle-formula}
  \Theta = \frac{\Delta\dot{x}_2}{v}.
\end{equation}

The difference in $\dot x_2$ is,

\begin{equation}
  \label{eq:angle-difference}
  \Delta \dot{x}_2 = \int_{-\infty}^{\infty}
  \frac{\cf'}{2\cf}\,\frac{av^2}{\epsilon} d s 
  = a v\,\int_{-\infty}^{\infty}
  \frac{\cf'}{2\cf\,\epsilon} dx_1.
\end{equation}

Hence,

\begin{align}
  \Theta = \frac{a}{2} (1 - \tau^2)^{1/2}
  q_1q_2\, \int_{-\infty}^{\infty}\frac{K_1(2 (1 - \tau^2)^{1/2}\,\epsilon) }{\epsilon} d{x}_1.
\end{align}

Recall \(\epsilon = ( a^2 + x^2_1)^{1/2} \) and let us make 
the change of variables
\begin{align}
u = (1 - \tau^2)^{1/2}x_1, &&
a_\tau = (1 - \tau^2)^{1/2} a.
\end{align}

The deflection angle is
\begin{align}\label{eq:deflection-angle-asympt}
  \Theta &= \frac{a_\tau}{2}
  q_1q_2\, \int_{-\infty}^{\infty}
  \frac{K_1(2\,(a_\tau^2 + u^2)^{1/2})}{
  (a_\tau^2 + u^2)^{1/2}
  } du\\
  &= -\frac{q_1q_2}{4} \dv{a_\tau} 
  \int_{-\infty}^{\infty}
  {K_0(2\,(a_\tau^2 + u^2)^{1/2})}du
\end{align}

The last integral was calculated in \cite{manton-speight-asymptotic}, 
using their result, the deflection angle is,
\begin{align}
    \Theta = -\frac{q_1q_2}{4} \dv{a_\tau}\pbrk{\frac{\pi}{2}\exp(-2a_\tau)}
    = q_1q_2\,\frac{\pi}{4}\,\exp(-2a_\tau).
\end{align}

The constant $q_1 q_2$  is negative, hence, the
geodesic are slightly deflected towards the origin, which indicates a 
vortex-antivortex pair behaves as a pair of attractive particles in the
long distance approximation. On figure~\ref{fig:scattering} we can see
the large distance approximation fits the scattering data for the
symmetric case. Since $\Theta \to 0$ 
as $a \to 0$, the fact that for large $a$, $\Theta$ is negative explains the 
existence of a minimum negative deflection as seen in 
figure~\ref{fig:scattering}.

\begin{figure}
  \centering
  \includegraphics[width=.9\textwidth]{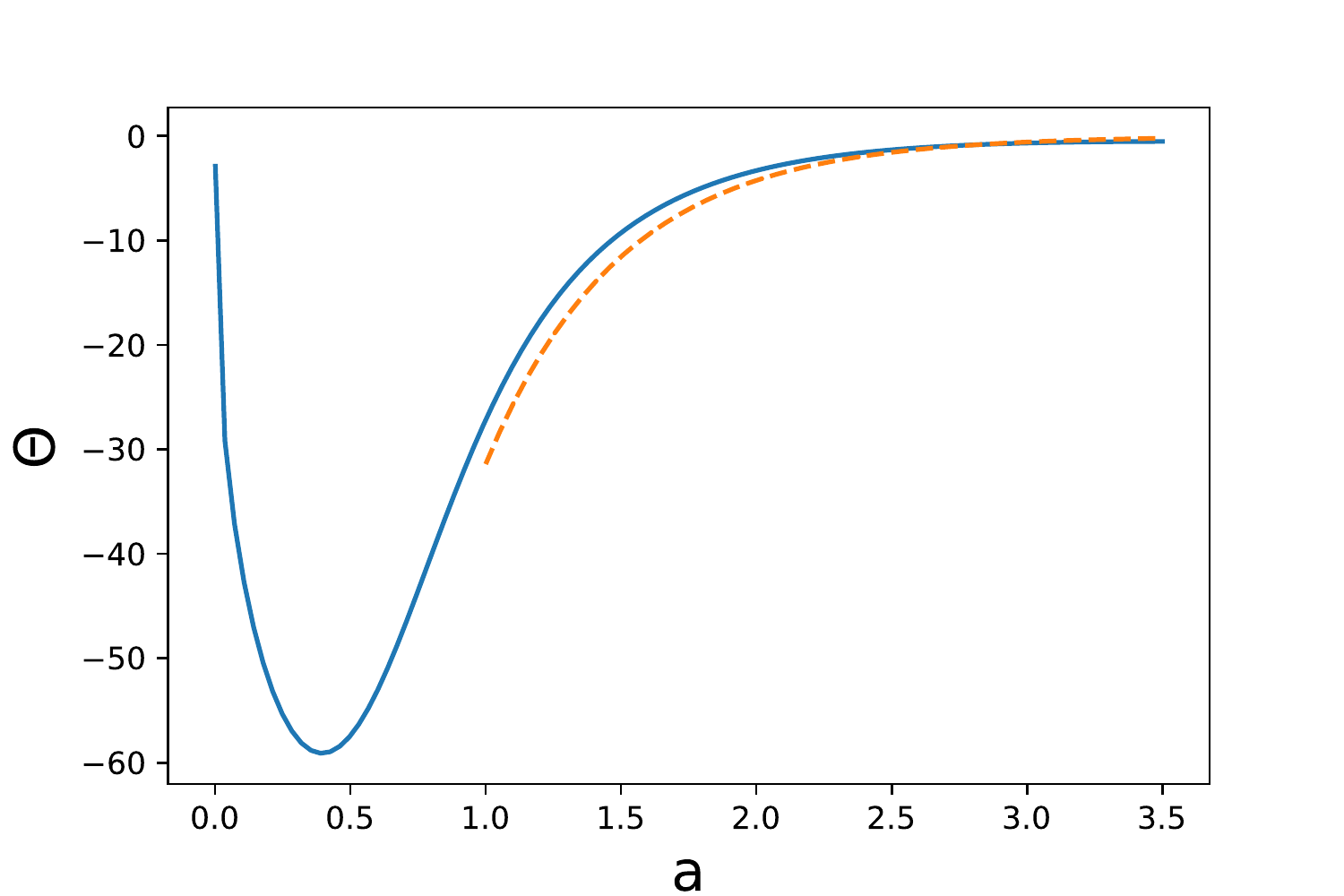}
  \includegraphics[width=.9\textwidth]{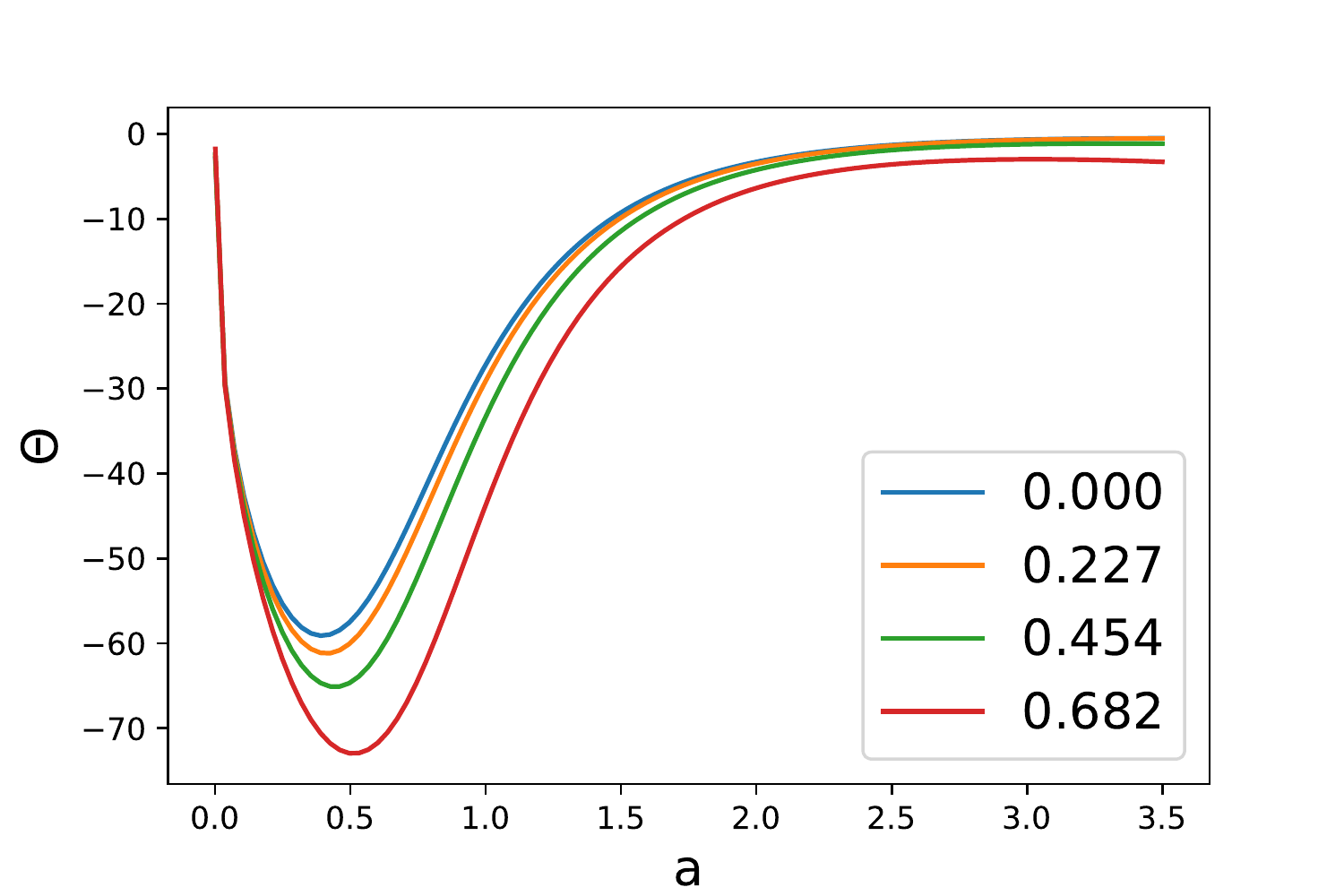}
  \caption{Above. Deflection angle at $\tau = 0$ and asymptotic
    approximation. Below. Comparison of the deflection angle for 
    different values of $\tau$.}\label{fig:scattering} 
\end{figure}

\let\spL\undefined 
\let\vol\undefined
\let \Energy\undefined
\let\cf\undefined
\let\curvature\undefined


\newcommand*{\angMomentum}{\ell}
\newcommand*{\ricci}{Ric}
\newcommand*{\ricciForm}{\rho}
\newcommand*{\cf}{\Omega}
\newcommand*{\gaussK}{K}
\newcommand*{\interiorProd}{\iota}
\newcommand*{\Lagrangian}{\mathrm{L}}
\newcommand*{\Energy}{\mathrm{E}}
\newcommand*{\veff}{\mathrm{V}_{eff}}

\section{
  Ricci magnetic geodesic motion
}\label{sec:ricci-magnetic-geod-motion}

The metric on $\moduli^{1,1}(\plane)$ can be split isometrically in a
product with one flat term isometric to $\plane$, the centre of mass 
coordinate. Since this term is flat, in the reduced moduli space we
have that the global Ricci tensor coincides with the Ricci tensor as a
Riemann surface. Therefore, the Ricci form in
$\moduli_0^{1,1}(\plane)$ is the restriction of the global Ricci form
to the centre of mass frame, 
\begin{equation}
\label{eq:ricci}
\ricciForm = \gaussK\,\epsilon\,d\epsilon\wedge d\theta,
\end{equation}
where $\gaussK$ is the Gauss curvature of the reduced moduli
space. Interaction of vortices with a magnetic field can be modelled
by means of Ricci magnetic geodesics, abbreviated RMGs. RMGs on the moduli 
space 
were introduced for the Ginzburg-Landau
model with a Chern-Simons term by Collie and
Tong~\cite{collie_dynamics_2008}, who proposed that the Ricci form 
was the magnetic form of the Chern-Simons term. Later, mathematical properties 
of RMGs were investigated
by Krusch-Speight on hyperbolic space~\cite{krusch2010exact}. 
Although in our case 
RMG dynamics is not physically motivated, these curves are of mathematical 
interest: Krusch-Speight conjectured that
geodesic completeness and RMG completeness were equivalent until 
Alqahtani-Speight found 
examples of incomplete surfaces which are RMG 
complete~\cite{alqahtani2015ricci}.  
 A curve
$\gamma$ is a Ricci magnetic 
geodesic if there is a constant scalar $\lambda$ such that,
\begin{equation}
\label{eq:rmg}
\nabla_{\dot\gamma} \dot\gamma = \lambda\,
(\interiorProd_{\dot\gamma}\ricciForm)^{\sharp},
\end{equation}
where $\interiorProd_{\dot\gamma}\ricciForm = \ricciForm(\dot\gamma,
\cdot)$ is the interior 
product. Unlike geodesic flow, RMG trajectories are speed dependent,
with changes in initial speed being reflected in the constant
$\lambda$. On a surface of
revolution, RMG equations are determined by the Lagrangian,
\begin{equation}
\label{eq:rmg-lagrangian}
\Lagrangian = \half \cf (\dot\epsilon^2 + \epsilon^2\dot\theta^2) +
\frac{\lambda}{2} \pbrk{\frac{\epsilon\,\cf'}{\cf}}\dot\theta.
\end{equation}
This is a conservative Lagrangian symmetric with respect to
translations in time and rotations of space, therefore, RMG
trajectories on the reduced moduli space preserve energy and angular
momentum,
\begin{align}
  \Energy &= \half \cf\,(\dot\epsilon^2 + \epsilon^2\dot\theta^2), &
  \ell &= \cf\,\epsilon^2\dot\theta + \frac{\lambda}{2} \,
         \frac{\epsilon\,\cf'}{\cf}.\label{eq:energy-ell}
\end{align}
Eliminating $\dot\theta$ from these equations, a RMG is a solution to
the first order equation,
\begin{align}
\Energy = \half\,\cf\dot\epsilon^2 + \veff,
\end{align}
where the effective potential is defined as,
\begin{align}
\veff = \frac{1}{2\epsilon^2\,\cf}\pbrk{ \ell - \lambda\,\frac{\epsilon
  \,\cf'}{2\,\cf}}^2.\label{eq:veff}
\end{align}

Figure~\ref{fig:vav-veff} shows $\veff$ for several values of $\tau = 0$. 
 Data confirms $\veff \to \infty$ as $\epsilon \to 0$,  
consistently  
with the asymptotic approximation to the conformal factor, 
likewise, for $\epsilon \to 
\infty$, $\veff \to 0$ since $\cf \to \cf_\infty$ and $\cf' \to 0$. 
 The effective potential can be seen in figure~\ref{fig:vav-veff}, the  
 shape depends on the relative value of $\ell / \lambda$. A large computation 
 reveals
 \begin{align}
 \veff' = \frac{-1}{2\epsilon^3\cf}
 \pbrk{\ell - \lambda\frac{\epsilon\,\cf'}{2\cf}}
 \pbrk{\frac{\lambda}{2}\,
     \pbrk{\frac{\epsilon^2\cf''}{\cf} - 
 3\frac{\epsilon^2\,\cf'^2}{\cf^2}} 
 + 2\ell \pbrk{1 + \frac{\epsilon\,\cf'}{2\cf }}},
 \end{align}
 by virtue of the asymptotic approximations, both $\epsilon\cf'$ and 
 $\epsilon^2\cf''$ are bounded functions, while $\cf$ is positive and bounded 
 below, hence for given $\lambda$ if $|\ell|$ is large, 
 $\veff$ is a positive decreasing function.
 \begin{figure}
     \begin{center}
         \includegraphics[width=.4\textwidth]
         {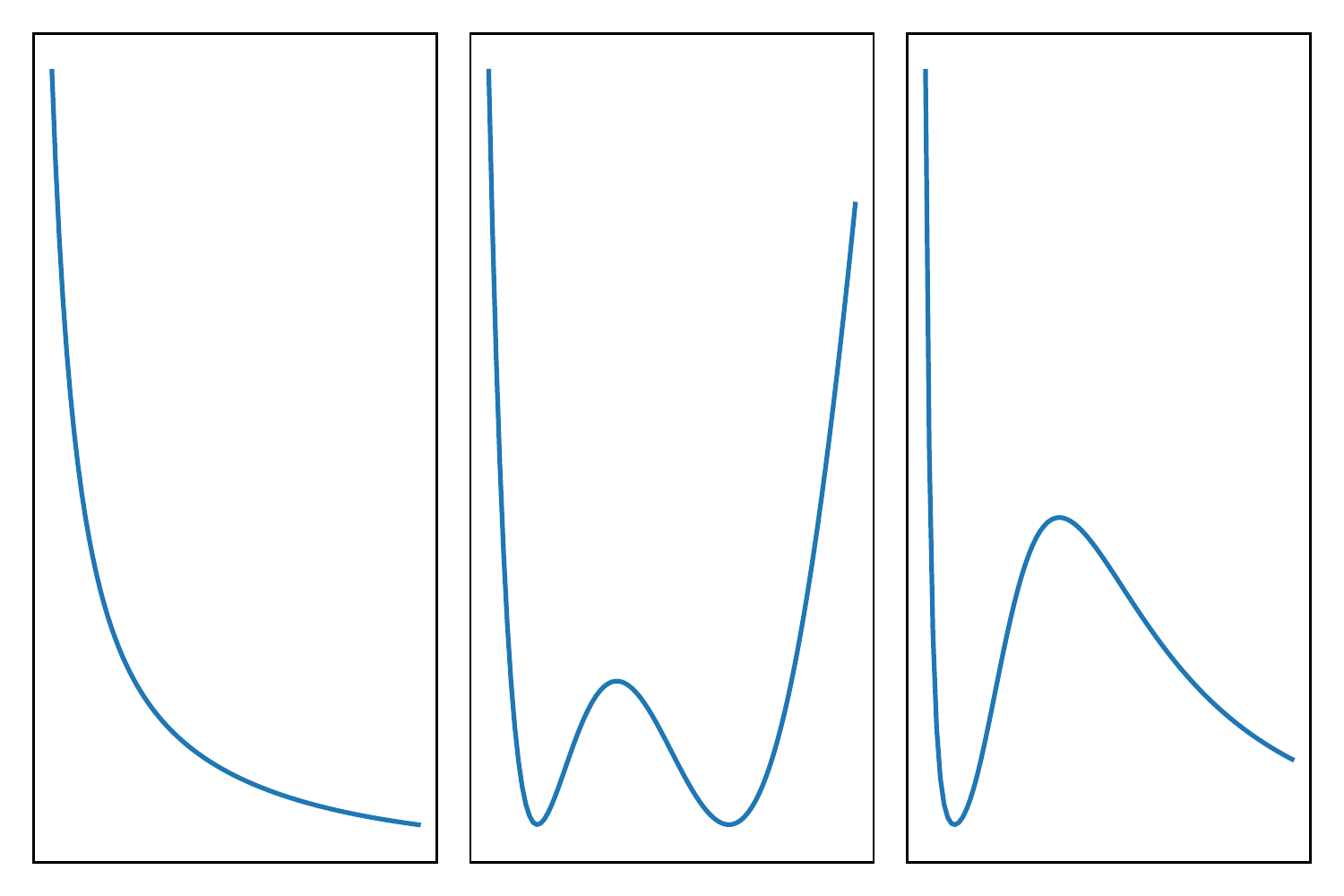}
     \end{center}
     \caption{
         Typical types of effective potentials for $\tau = 0$ (description in 
         text). In the three 
         cases, $\lambda = 1$, in the second case, $\veff \to 0$ as $\epsilon 
         \to \infty$ 
         although is not apparent in the figure because of the scale.
     }\label{fig:vav-veff}
 \end{figure}
 
  In this case RMGs are all unbounded curves. If $\ell$ is not very large, 
  $\veff$ has relative extrema, giving rise to both unbounded and bounded 
  trajectories orbiting 
  around the singularity at $\epsilon = 0$. By 
  equations~\eqref{eq:energy-ell} and~\eqref{eq:veff}, trajectories for which 
  $\Energy = \veff$ at constant $\epsilon_0$ are circular if 
  $\veff(\epsilon_0) \neq 0$ 
  or constant if $\veff(\epsilon_0) = 0$. If the perturbation is around a zero 
  of $\veff$, the angular velocity alternates sign, the pattern  
  is as seen on  the bounded curves on  the first row of figure~\ref{fig:rmgs}. 
  If the perturbation is around a local minimum of $\veff$ which is not a zero, 
  the angular velocity keeps the same sign and gives rise to the patterns seen 
  on the second row of the figure.
 \begin{figure}
     \centering
     \includegraphics[width=0.6\textwidth]
     {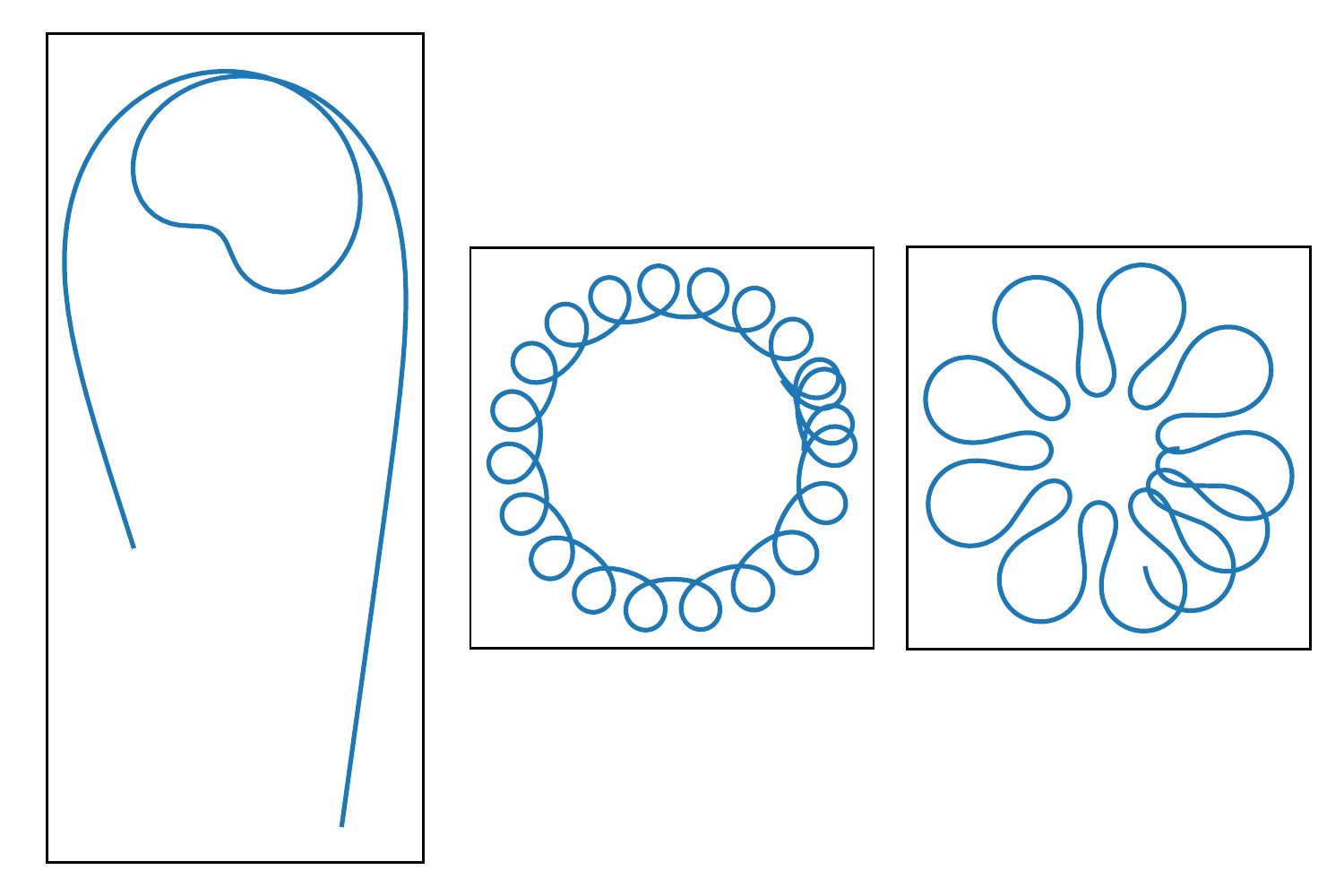}
     \includegraphics[width=0.6\textwidth]
    {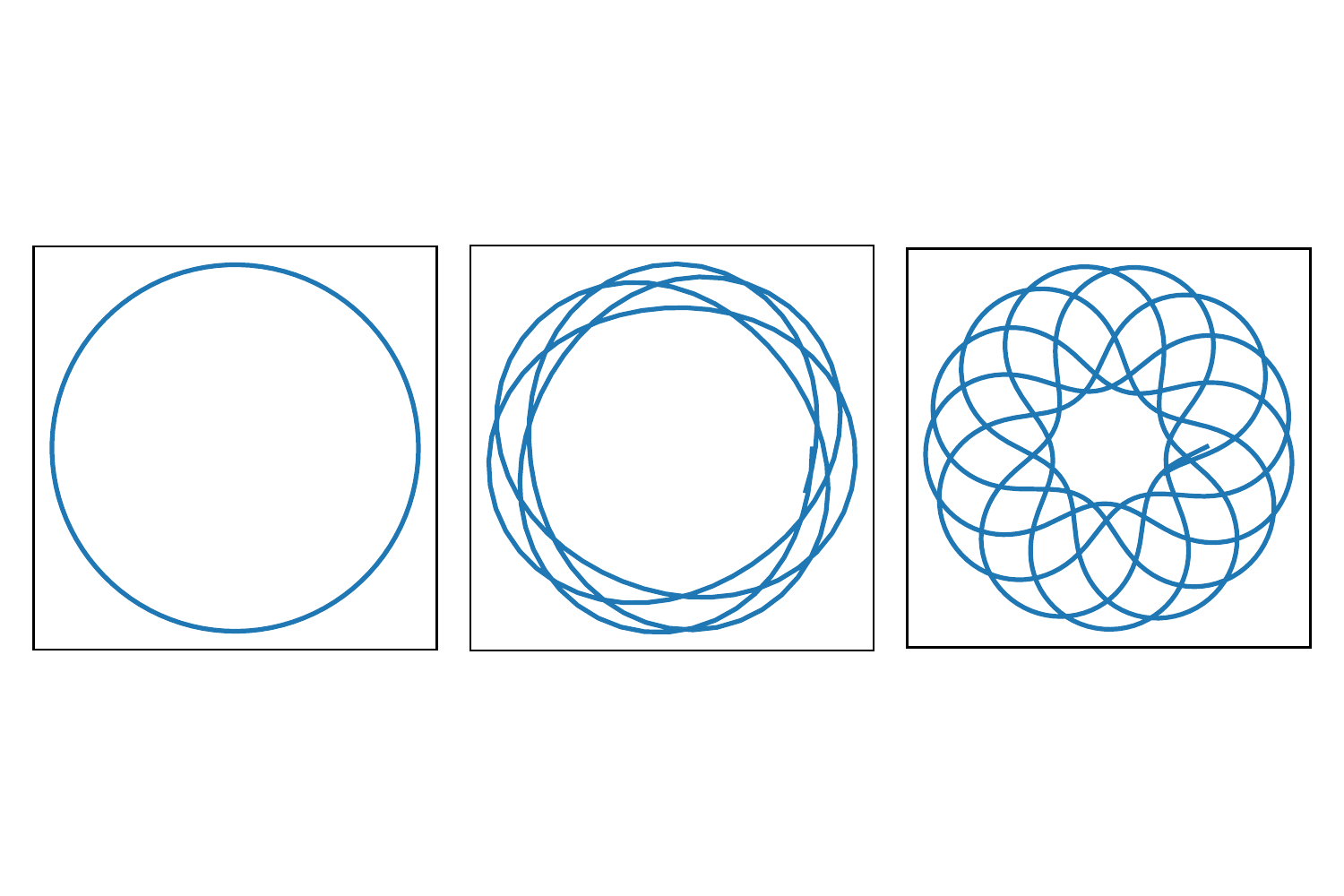}
     \caption{RMGs for $\tau = 0$ (Description in text).}\label{fig:rmgs}
 \end{figure}


%
%
%
%
As numerics show, the moduli space is RMG complete, even though 
it is geodesically incomplete, because the divergence of $\veff$ at 
the origin prevents RMGs of hitting the singularity.

\let \angMomentum \undefined
\let \ricci \undefined
\let \ricciForm \undefined
\let \cf \undefined
\let \gaussK \undefined
\let \interiorProd \undefined
\let \Lagrangian \undefined
\let \Energy \undefined
\let \veff \undefined


\let\modulic    \undefined
\let\qty        \undefined
\let\lagrangian \undefined
\let\bog        \undefined
\let\va         \undefined
\let\vu         \undefined
\let\dd         \undefined
\let\lgrdensity \undefined 
\let\delzbar    \undefined 


\label{ch:vav-euc}

\newcommand*\inc{\iota}
\newcommand*\kform{\omega}
\newcommand*\proj{\Pi}
\newcommand*\sign{s}
\newcommand*\wto{\rightharpoonup}
\newcommand*\diag{\Delta}
\newcommand*\vform{\Volume}
\newcommand*\Volume{\mathrm{Vol}}
\newcommand*\Wsp{\mathrm{W}}
\newcommand*\Bop{\mathrm{B}}
\newcommand*\fnF{F}
\newcommand*\diff{D}
\newcommand*\fnV{V}
\newcommand*\fnv{v}
\newcommand*{\gp}{A}
\newcommand*{\Jop}{\mathrm{J}}
\newcommand*{\fbundle}{\mathcal{F}}
\newcommand*{\pbundle}{\mathcal{P}}
\newcommand*{\Diff}{\mathcal{D}}
\newcommand*{\hf}{\phi}
\newcommand*{\potentialE}{\mathrm{V}}
\newcommand*{\Fstable}{F}
\newcommand*{\vq}{ {v_{\mathbf{q}}} }
\newcommand*{\vp}{ {v_{\mathbf{p}}} }
\newcommand*{\cf}{\Omega}
\newcommand*{\vol}{\mathrm{Vol}}
\newcommand*{\Hsp}{\mathrm{H}}
\newcommand*{\htilde}{\tilde h}
\newcommand*{\ctilde}{\tilde c}
\newcommand*{\pconnection}{\omega}
\newcommand*{\Energy}{\mathrm{E}}
\newcommand*{\Xsp}{\mathcal{X}}
\newcommand*{\suchthat}{\mid}
\newcommand*{\dist}{d}
\newcommand*{\moduliMetric}{g}
\newcommand*{\domain}{\mathcal{D}}
\newcommand*{\isometry}{\mathcal{I}}
\newcommand*{\pSpace}{\mathbb{P}}

\chapter[Vortex-antivortex on a compact surface]{Asymmetric vortex-antivortex pairs on a compact surface}\label{c:vav-compact}

In this chapter we study vortex-antivortex systems on a compact surface. 
 We aim to prove that the moduli space is 
 incomplete and to compute the volume of the moduli space for the 
 round sphere and flat tori. On a general compact domain,  the problem
of the statistical mechanics of Ginzburg-Landau vortices was
addressed by Manton~\cite{manton1993statisticalmec} and 
by Manton-Nasir~\cite{manton1999volume}. 
As shown 
in \cite{manton1993statisticalmec}, it can be described if we know the volume 
of the moduli space. For
the abelian $O(3)$ Sigma model however, the problem of the volume of the moduli
space is constrained by the fact that vortices and antivortices cannot
coalesce, however, computing the volume is necessary for the partition function
of a gas of BPS vortices  
\cite{romao2018,manton1999volume,manton1993statisticalmec}. There is a 
conjectured formula for the volume by Speight and Rom\~ao that depends on 
topological data, the 
volume of the domain, 
$\tau$ and the size of the sets $\vset$, $\avset$ of
core positions~\cite{romao2018}. The content of the chapter is as follows.

In section~\ref{sec:main-thm}, we prove that the Taubes equation has
 exactly one solution for any $\tau \in (-1, 1)$.

The main result of section~\ref{sec:compact-vav-eqn} 
is theorem~\ref{thm:moduli-space-incomplete-compact-surface} 
which asserts that the moduli 
space of vortex-antivortex pairs is incomplete. We prove the theorem 
after proving several lemmas necessary to bound the derivatives of solutions 
to the Taubes equation.

In section~\ref{sec:volume-moduli-space} we compute the volume of 
the moduli space of vortex-antivortex pairs for the round sphere and 
flat tori and compare our results with the conjecture.

\section{Existence of vortices}\label{sec:main-thm}

In this section we will prove the existence of solutions to the 
Taubes equation
 on a compact surface. In \cite{sibner2010} 
Sibner-Signer-Yang proved  existence and 
uniqueness of solutions of the gauged $O(3)$ Sigma model on a compact manifold 
for $\tau = 0$.
We prove the following generalisation of their results. 
\begin{theorem}
  \label{thm:main}
On any compact
  Riemann surface there exists exactly one
 solution $u$ to the 
Taubes equation~\eqref{eq:taubes}, provided the
condition
\begin{equation}
  \label{eq:vav-size-constraint}
 - \frac{1 + \tau}{2\pi}\abs{\surface}
  <
    k_+ - k_-  
  <  
   \frac{1 - \tau}{2\pi}\abs{\surface}
\end{equation}
holds. Moreover, 
$u$ is of class $C^2$ except for the core positions.
\end{theorem}

We prove the theorem at the end of the section. The inequality 
\eqref{eq:vav-size-constraint} is a Bradlow type restriction
\cite{bradlow1990}, constraining the relative number of vortices and
antivortices on a compact surface.
It arises naturally from the second Bogomolny 
equation~\eqref{eq:bog2}, since the total magnetic flux is,
\begin{align}
2\pi(k_+ - k_-) &= 
\int_{\surface} B \nonumber\\
&= \int_\surface \lproduct{N,\hf}\,\vol - 
\tau\,|\surface|,
\end{align}
where $N$ is the north pole section on the target sphere 
and hence $\lproduct{N,\phi}\in [-1, 1]$, it follows 
that~\eqref{eq:vav-size-constraint} is a necessary condition for a pair 
$(\hf, \gp)$ of a field and a connection  to be a solution to the Bogomolny 
equations.

We will define 
 the function $\Fstable: \reals \to \reals$, 
\begin{equation}
\Fstable(t) = 2 \left(
\frac{e^t - 1}{e^t + 1} + \tau
\right),
\end{equation}
and the constant,
\begin{equation}
\label{eq:fstable-pm-infty}
\Fstable^{\pm\infty} = 2(\pm 1 + \tau),
\end{equation}
in order to simplify notation in the proof of theorem~\ref{thm:main}. 
 Let us  define 
$\Fstable_0: \reals \to \reals$ as the function,

\begin{equation}
\label{eq:param-fstable-sphere}
\begin{aligned}[b]
  \Fstable_0(t) &= 2\brk(\frac{e^t - 1}{e^t + 1} + \tau) + \frac{4\pi(k_+ -
    k_-)}{\abs\surface }\\
  &= \frac{4 e^t}{e^t + 1} - C_0, 
\end{aligned}
\end{equation}
where the constant $C_0$ is,
\begin{align}
C_0 = 2(1 - \tau) - \frac{4\pi}{\abs\surface}(k_+ - k_-).
\end{align}

For a given configuration of non-coalescent vortices, recall the function 
$v : \surface \to \reals\cup\set{\pm\infty}$, defined on 
equation~\eqref{eq:fn-v}, if $u$ is the solution of the Taubes equation, and 
we define $\tilde h = u - v$, then  
the regularized Taubes equation on a compact surface, 
equation~\eqref{eq:taubes-reg}, is equivalent to,
%
\begin{align}
\label{eq:regular-taubes-bps}
-\laplacian \tilde h = \Fstable_0(v + \tilde h).
\end{align}

 Equation~\eqref{eq:regular-taubes-bps} shows why Bradlow's bound is 
necessary: If a smooth solution exists, by the divergence 
theorem a necessary condition for $C_0$ is,
\begin{equation}
\label{eq:abstract-c0-cond}
    C_0 = \frac{1}{|\surface|}\,\int_\surface\frac{4\,e^{v + \tilde h}}
    {e^{v + \tilde h} + 1}\,\vol \in [0, 4],
\end{equation}
%
Bradlow's bound is equivalent to
\eqref{eq:abstract-c0-cond}. Let 
\begin{equation}
  \label{eq:Xsp-def}
  \Xsp = \left\{
    u \in \Hsp^1(\surface) \,:\, \int_\surface u \,\vform = 0
  \right\}
\end{equation}
be the subspace of Sobolev's space $\Hsp^1(\surface)$ of functions of zero
average. Since $\surface$ is compact, $\Hsp^1(\surface)$ can be
decomposed as 
\begin{equation}
  \label{eq:cann-decomp}
  \Hsp^1(\surface) = \Xsp \oplus \reals.
\end{equation}

Any $h \in\Hsp^1(\surface)$ can be decomposed as a pair
$(u, \ctilde) \in \Xsp\times\reals$, such that $h = u +
\ctilde$. Hence, $u$ is a solution to the equation,
\begin{equation}
  \label{eq:htilde}
  -\laplacian u = \Fstable_0(v + u + \ctilde). 
\end{equation}

We will use
Leray-Schauder theory to prove existence of solutions to the Taubes 
equation as in the proof of Sibner et 
al.~\cite{sibner2010} for $\tau = 0$. Given $\htilde \in \Xsp$, the
function
%
%
\begin{equation}
\label{eq:integral-eqn}
c \mapsto \int_{\surface} \Fstable_0(v + \htilde + c) \,\vol,
\end{equation}
is a well defined, monotonous, continuous function. By Bradlow's
bound, there exists a unique number  
$\ctilde$ such that
%
%
\begin{equation}
\label{eq:ctilde-defs}
\int_{\surface} \Fstable_0(v + \htilde + \ctilde) \,\vol = 0.
\end{equation}

\begin{lemma}\label{lem:ctilde-weak-cont}
  The function $ \mathcal{C}: \Xsp \to \reals$, $\mathcal{C}(\tilde h) = \tilde 
  c$ 
  is 
  weakly 
  sequentially continuous in
  $\Xsp$. 
\end{lemma}

\begin{proof}
We will highlight the steps different from~\cite{sibner2010} in the
general case. If $\htilde_n \wto \htilde_0$ in $\Xsp$, then 
$\htilde_n$ is a bounded sequence in $\Xsp$, and by the Rellich lemma,
after passing to a sub-sequence if necessary, we can assume 
 $\htilde_n \to \htilde_0$ in $\Lsp^p$ for $p \geq 1$. Let $\ctilde_n =
 \ctilde(\htilde_n)$, $\ctilde_0 = \ctilde(\htilde_0)$ and assume
 towards a contradiction that $\ctilde_n$ does not converge to
 $\ctilde_0$. In this case we can assume the existence of a constant
 $\epsilon_0$ such that,
 \begin{equation}
   \label{eq:cn-conv-contradiction}
   \abs{\ctilde_n - \ctilde_0} \geq \epsilon_0,
 \end{equation}
 for all $n$. We claim the sequence $\left\{\ctilde_n\right\}$ 
 is bounded. Assume the contrary, after passing to a sub-sequence if
 necessary, we can assume the limit
 $\ctilde_n \to \infty$. Let $K$ be any bound  
 for $\Fstable_0$. By Egorov's theorem~\cite{liebanalysis} 
 and the strong convergence in $\Lsp^p$,
 there exists a measurable set $\surface_{\epsilon}$ and a constant
 $K_{\epsilon}$, such that $\abs{\surface_{\epsilon}} < \epsilon
 K^{-1}$, the sequence $\htilde_n$ converges uniformly to $\htilde_0$
 in $\surface \setminus \surface_{\epsilon}$ and $\abs{\htilde_n} \leq
 K_{\epsilon}$ in $\surface\setminus\surface_{\epsilon}$.

 On the one hand,
 the equality
 \begin{equation}
   \int_{\surface\setminus\surface_{\epsilon}}\Fstable_0(v +
     \ctilde_n + \htilde_n) \,\vol =
   -\int_{\surface_{\epsilon}}\Fstable_0(v + \ctilde_n +
     \htilde_n) \,\vol,
 \end{equation}
 implies,
 \begin{equation}
   \label{eq:egorov-integral-impl}
   \abs*{\int_{\surface\setminus\surface_{\epsilon}}\Fstable_0(v +
     \ctilde_n + \htilde_n)  \,\vol } \leq \epsilon,
 \end{equation}
 and on the other hand, by monotony of $\Fstable_0$, 
 \begin{equation}
   \label{eq:monotony-fstable0-integral}
   \int_{\surface\setminus\surface_{\epsilon}}\Fstable_0(v + \ctilde_n
   - K_{\epsilon}) \,\vol \leq 
   \int_{\surface\setminus\surface_{\epsilon}}\Fstable_0(v + \ctilde_n
   + \htilde_n) \,\vol.
 \end{equation}
 
Taking the limit as $n \to \infty$, from these two equations 
we have,
\begin{equation}
\label{eq:int-finty-surface-minus-eps}
(\Fstable^{\infty} - C_0)(\abs\surface - \abs{\surface_{\epsilon}}) \leq \epsilon.
\end{equation}
Hence,
\begin{equation}
   \label{eq:fstable-inf-bound-contradiction}
(\Fstable^{\infty} - C_0) \abs\surface \leq \epsilon + 
K\abs{\surface_{\epsilon}} < 2\epsilon,
\end{equation}
%
%
%
%
%
a contradiction since $\epsilon$ is arbitrary. A similar
argument shows $\ctilde_n$ is bounded below. Therefore, 
$\ctilde_n$ is a bounded sequence of real numbers. By the 
Bolzano-Weierstrass theorem, we can assume towards a contradiction
$\ctilde_n \to \ctilde$, but $\ctilde \neq \ctilde_0$
by~\eqref{eq:cn-conv-contradiction}. Let  
\begin{equation}
  \label{eq:alpha-def}
  \alpha = \abs*{
    \int_{\surface}\Fstable_0(v + \htilde_0 + \ctilde) \,\vol
  } > 0,
\end{equation}
bearing in mind the definition of $\htilde_n$, 
\begin{multline}
  \alpha = \abs*{
    \int_{\surface}\Fstable_0(v + \htilde_0 + \ctilde)
      - \Fstable_0(v + \htilde_n + \ctilde_n)
      \,\vol
    }\\
  \leq
    \sup_{t \in \reals}{\set{\Fstable'(t)}}\cdot\left(
     \abs{\ctilde - \ctilde_n}\cdot\abs{\surface} + C\,\norm{\htilde_0 -
       \htilde_n}_0 \cdot \abs{\surface}^{1/2}
    \right) \to 0.\label{eq:alpha-bound}
\end{multline}
%

Hence $\alpha = 0$, a
contradiction. Therefore \eqref{eq:cn-conv-contradiction} is 
false and $\ctilde_n \to \ctilde_0$. This proves the lemma. 
\end{proof}

Let us consider the operator $T: \Xsp \to \Xsp$, mapping each
$\htilde\in\Xsp$ to the weak solution $H \in \Xsp$ of the equation
\begin{equation}
  \label{eq:T-op}
  -\laplacian H = \Fstable_0(v + \ctilde + \htilde).
\end{equation}

Given that $\int_\surface \Fstable_0(v + \ctilde + \htilde)\,\vol = 0$, 
existence of a weak $\Hsp^1$ solution to~\eqref{eq:T-op} is a well established 
analysis fact~\cite[Thm.~4.7]{aubin2013some}, moreover, any two weak 
solutions to the equation differ by a constant, by taking $H \in \Xsp$ we 
guarantee it is unique. 

Recall a compact operator is an operator that maps bounded sequences
to sequences with convergent subsequences. We aim to use Sch\"afer's
alternative, theorem~\ref{thm:fixed-point-alternative}, to prove $T$
has a fixed point.

\begin{lemma}
  The operator $T:\Xsp \to \Xsp$ is compact in the strong topology of 
  $\Xsp$ as a subspace of $\Hsp^1(\surface)$.
\end{lemma}

\begin{proof}
Let $\{\tilde h_n\} \subset \Xsp$ be a bounded sequence, after
passing to a subsequence if necessary, we can assume $\htilde_n \wto
\htilde_0$ in $\Xsp$ and strongly in $\Lsp^2$. Let $H_n = T\htilde_n$,
$n \geq 0$, by lemma~\ref{lem:ctilde-weak-cont} $\ctilde_n \to
\ctilde_0$. Moreover, 
\begin{align}
  \nonumber
  \norm{\grad H_n - \grad H_0}_{\Lsp^2}^2
  &= \int_{\surface} (H_n - H_0)\, \laplacian (H_n  - H_0)\,\vol\\
  \nonumber
  &= \int_{\surface} (H_n - H_0)\left(\Fstable(v + \ctilde_n + \htilde_n)
- \Fstable(v + \ctilde_0 + \htilde_0)\right) \,\vol\\
\nonumber
&\leq \sup_{t \in \reals} \set{F'(t)}\,\int_{\surface} \left(
  \abs{\ctilde_n - \ctilde_0} + \abs{\htilde_n - \htilde_0}\right)
  \abs{H_n - H_0}\,\vol\\
&\leq \sup_{t \in \reals} \set{F'(t)}\, \left(
  \abs{\ctilde_n - \ctilde_0}\cdot\abs{\surface}^{1/2} + \norm{\htilde_n -
\htilde_0}_{\Lsp^2}
\right) \norm{H_n - H_0}_{\Lsp^2}. 
\end{align}
The last inequality is a consequence of the Cauchy-Schwarz
inequality. By the Poincar\'e inequality, there are constants $C_1$,
$C_2$ such that 
\begin{equation}
  \label{eq:Hn-H0-X-convergence}
  \norm{H_n - H_0}_{\Hsp^1} \leq C_1\abs{\ctilde_n - \ctilde_0}
  + C_2\norm{\htilde_n - \htilde_0}_{\Lsp^2} \to 0.
\end{equation}
This proves compactness of $T$.
\end{proof}

 Let us consider the set
\begin{equation}
  \label{eq:wt-def}
S = \set{\htilde \in \Xsp\, : \, 
    \exists\, t \in [0, 1]\;s.t.\;
    \htilde = t\cdot
  T\,\htilde}. 
\end{equation}

If $\htilde \in S$, then it is a solution of the equation,
\begin{equation}
  \label{eq:htilde-t-sol}
  \laplacian \htilde = t \Fstable_0(v + \ctilde + \htilde),
\end{equation}
where $\tilde c = \mathcal C(\tilde h)$ was defined on 
lemma~\ref{lem:ctilde-weak-cont}.

By the Cauchy-Schwarz inequality,
\begin{equation}
  \begin{aligned}
    \norm{\grad\htilde_t}_{\Lsp^2}^2 &= \lproduct{\tilde h_t, \laplacian
      \tilde h_t}
    \leq C
    \int_{\surface}\abs{\htilde_t} \,\vol
    \leq C\,
  \abs{\surface}^{1/2}\,\norm{\htilde_t}_{\Lsp^2}.
  \end{aligned}
\end{equation}
By the Poincar\'e 
inequality we conclude the existence of a constant  $C$ such that
\begin{equation}
  \label{eq:ht-uniform-bound}
  \norm{\htilde_t}_{\Hsp^1} \leq C.
\end{equation}


\begin{proof}[Proof of Theorem~\ref{thm:main}]
Since $S$ is bounded, by Sch\"afer's alternative there is a fixed point
$\htilde$ of $T$. Let $h = \htilde + \ctilde$, where $\ctilde = \mathcal 
C(\tilde h)$, then $h$ is a weak
solution to the regularised Taubes equation. By the elliptic
estimates $h$ is also a strong solution in $\Hsp^2$. We follow a
bootstrap argument to prove $h \in C^2$: By Sobolev's embedding we
know $h$ is continuous, hence $h \in \Lsp^p$ for any $p \geq
1$. By~\eqref{eq:regular-taubes-bps} and the elliptic estimates $h 
\in \Wsp^{2,p}$ for some $p > 2$, once more by Sobolev's embedding $h
\in C^1$. Let $u = h + v$, the derivative $d h \in
\Gamma(T^{*}\surface)$ is a weak solution of the linearized equation,
\begin{align}
  -\laplacian\,d h = \frac{4\,e^u}{(e^u + 1)^2}\,d h +
  \frac{4\,e^u}{(e^u + 1)^2}\,d v.
\end{align}

The potential function $e^u(e^u + 1)^{-2}$ is continuous and with
zeros of the same order than the singularities of $d v$ at the cores,
hence $\laplacian(d h) \in \Lsp^p$, $p > 2$. Since $d h$ is 
continuous, it is also an $\Lsp^p$ form. By the elliptic estimates and
Sobolev's embedding we conclude $h \in C^2$. Since $\Fstable$
is monotonous, $h$ is unique by the strong maximum
principle. Finally, $u$ is the necessarily unique solution to the 
Taubes equation.
\end{proof}






\section{Incompleteness of the moduli space}
\label{sec:compact-vav-eqn}

In \cite{romao2018} Rom\~ao and Speight prove that the moduli space of 
symmetric 
vortex-antivortex pairs on the sphere is incomplete. In this section we extend 
their result to general $\tau$ on a compact manifold. In order to prove this, 
we find bounds for 
the derivatives $\partial_{z_j}\nabla h_\epsilon$ on a holomorphic chart, 
where the cores are at positions $z_1, z_2$.
Let $\mu = \log\, (1 - \tau) - \log\, (1 + \tau)$, first we prove a pair of 
technical lemmas. 

\begin{lemma}
      Let $\diag$ be the diagonal set of $\surface \times \surface$ and 
      let $\set{\vb x_n} \subset \moduli^{1,1}(\surface)$ be a
      sequence such that $\vb x_n \to \vb x \in \diag$ in the product metric. 
      Let $\tilde h_n$ be the
      solution of the regular Taubes equation corresponding to each
      $\vb x_n$, then
      $\tilde h_n \wto \mu$ in $\Hsp^1$ and $\tilde h_n \to \mu$
      strongly in $\Lsp^2$.
    \end{lemma}
    
    \begin{proof}
        Let $v_n = v_{\vb x_n}$ for  
      each point $\vb x_n$ in the given sequence. Let us
      decompose each 
      solution to the regular Taubes equation as $\tilde h_n = u_n +
      \tilde c_n \in \Xsp \oplus \reals$. We claim the sequence
      $\set{ \tilde c_n}$ is bounded. Assume towards a contradiction
      $\tilde c_n \to \infty$. Notice that in the vortex-antivortex
      case the functions $\fnF$ and $\fnF_0$ coincide. We know that,
      \begin{align}
        -\laplacian u_n = \fnF( u_n + \tilde c_n +
        \fnv_n).\label{eq:lap-h0n}
      \end{align}

      By the standard elliptic estimates, there is a constant $C$ such
      that
      \begin{align}
       \norm{u_n}_{\Hsp^2} \leq C \norm{\laplacian u_n}_{\Lsp^2}.
      \end{align}

      Since $\fnF$ is a 
      bounded function, $\set{u_n}$ is bounded in $\Hsp^2$ and by
      Sobolev's embedding also in $C^0$. 
      
      Assume $\vb x = (x_*, x_*)$ and notice that,
      \begin{align}
      |v_n(x)| = 4\pi |G(x, x_{1n}) - G(x, x_{2n})|,
      \end{align}
      where $\vb x_n = (x_{1n}, x_{2n})$, 
      since $G(x,y)$ is continuous away of the diagonal set,       
       $\fnv_n(x) \to 0$ for $x \neq x_*$, whence,
      we also have the convergence,
      \begin{align}
        \fnF( u_n + \tilde c_n +
        \fnv_n) \to 2(1 + \tau),
      \end{align}
      pointwise almost everywhere. Applying the dominated convergence theorem 
      and
      equation~\eqref{eq:lap-h0n},
      \begin{align}
        \int_{\surface} \fnF( u_n + \tilde c_n +
        \fnv_n) \,\vform = 0 \to 2(1 +
        \tau)\abs\surface,
      \end{align}
      a contradiction. If $\tilde c_n \to -\infty$ a similar argument
      holds. Therefore the sequence of averages $\tilde c_n$ is bounded, 
      implying
      $\{\tilde h_n\}$ is bounded in $C^0$.
      Hence, the sequence is also bounded in $\Lsp^p$ for any
      positive $p$. By the elliptic estimate
      \begin{align}
        \norm{\tilde h_n}_{\Hsp^2} \leq C
        \pbrk{\norm{\laplacian \tilde 
        h_n}_{\Lsp^2} + \norm{\tilde h_n}_{\Lsp^2}},
      \end{align}
      $\{\tilde h_n\}$ is also bounded in $\Hsp^1$. By the Alaoglu 
      and Rellich theorems, after passing to a subsequence if
      necessary, we can assume $\tilde h_n \wto h_{*} \in \Hsp^1$ and
      strongly in $\Lsp^2$. We claim that $h_*$ is the constant 
      function $\mu$. To see this, let $\varphi \in \Hsp^1$. From the 
      regularized Taubes equation we have,
      \begin{align}
        \lproduct{\tilde h_n, \varphi}_{\Hsp^1}
        &= \lproduct{\tilde h_n,
          \varphi}_{\Lsp^2} + \lproduct{\grad \tilde h_n, \grad
          \varphi}_{\Lsp^2},\nonumber \\
        &= \lproduct{\tilde h_n,
          \varphi}_{\Lsp^2} + \lproduct{\laplacian \tilde h_n,
          \varphi}_{\Lsp^2}\nonumber \\
        &= \lproduct{\tilde h_n,
          \varphi}_{\Lsp^2} - \lproduct{\fnF(\tilde h_n + \fnv_n),
          \varphi}_{\Lsp^2}. 
      \end{align}

      Since $\tilde h_n \to h_{*}$ strongly in $\Lsp^2$, after passing
      to a subsequence if necessary, we can assume
      $\tilde h_n\to h_{*}$ pointwise almost everywhere. By the weak
      convergence of $\tilde h_n$ in $\Hsp^1$, together with the
      strong convergence in $\Lsp^2$ and the dominated convergence
      theorem,
      \begin{align}
        \lproduct{h_{*}, \varphi}_{\Hsp^1} &= \lim\,\lproduct{\tilde h_n,
        \varphi}_{\Hsp^1}\nonumber\\
    &= \lim \lproduct{\tilde h_n,
        \varphi}_{\Lsp^2} - \lim \lproduct{\fnF(\tilde h_n + \fnv_n),
        \varphi}_{\Lsp^2}\nonumber\\
        &= \lproduct{h_{*},\varphi}_{\Lsp^2} - \lproduct{\fnF(h_{*}),
        \varphi}_{\Lsp^2}.
      \end{align}

      From this equation, we infer
      \begin{align}
        \lproduct{\grad h_{*}, \grad\varphi}_{\Lsp^2}
        = -\lproduct{\fnF(h_{*}), \varphi}_{\Lsp^2}.
      \end{align}

      Therefore, $h_{*}$ is a weak solution to the equation
      \begin{align}
        -\laplacian h_{*} = \fnF(h_{*}).
      \end{align}
      
      By elliptic regularity, $h_{*}$ is also a strong solution, and
      by the maximum principle, $h_{*}$ is constant since $\fnF$ is an
      increasing function. Since the only zero of $\fnF$ is at
      $t = \mu$, we conclude $h_{*} = \mu$. If $\tilde h_{n_k}$ is any
      subsequence of $\tilde h_{n}$, this argument shows it
      has a subsequence weakly converging to $\mu$ in $\Hsp^1$ and
      strongly in $\Lsp^2$, the
      claim of the lemma follows.
    \end{proof}

    \begin{lemma}
      $\tilde h_n \to \mu$ strongly in $\Wsp^{2,p}$ for any positive $p$.
    \end{lemma}

    \begin{proof}
      We will prove that any subsequence of $\tilde h_n$ has another
      subsequence converging to $\mu$ in $\Wsp^{2,p}$, implying the 
      lemma. To simplify
      notation, we denote subsequences of $\tilde h_n$ by the
      same symbol. 
      From the previous lemma, $\tilde h_n \to \mu$ strongly in
      $\Lsp^2$. After passing to a subsequence if necessary, we can
      assume that $\tilde h_n \to \mu$ pointwise almost
      everywhere. We apply the
      dominated convergence theorem to deduce the limit,
      \begin{align}
        \norm{\laplacian \tilde h_n}_{\Lsp^{p}}
        = \norm{\fnF(\tilde h_n + \fnv_n)}_{\Lsp^p}
        \to \norm{\fnF(\mu)}_{\Lsp^p} = 0.
      \end{align}

      If $p = 2$, by the standard elliptic estimates, there is a
      constant $C$, such that,
      \begin{align}
        \norm{\tilde h_n - \mu}_{\Hsp^2} \leq C \pbrk{
        \norm{\laplacian \tilde h_n}_{\Lsp^2} +
        \norm{\tilde h_n - \mu}_{\Lsp^2}} \to 0.
      \end{align}

      By Sobolev's embedding, $\tilde h_n \to \mu$ uniformly in
      $C^0$, hence also in $\Lsp^p$ for any positive $p$. We apply one
      more time the elliptic estimate,
      \begin{align}
        \norm{\tilde h_n - \mu}_{\Wsp^{2,p}} \leq C \pbrk{
        \norm{\laplacian \tilde h_n}_{\Lsp^p} +
        \norm{\tilde h_n - \mu}_{\Lsp^p}} \to 0.
      \end{align}
    \end{proof}

    As a consequence of this lemma and Sobolev's embedding, we have
    the convergence,
    \begin{align}
      \norm{\tilde h_n - \mu}_{C^1} \to 0,
    \end{align}
    for any arbitrary sequence $\set{\vb x_n} \subset \moduli^{1,1}(\surface)$,
    such that $\vb x_n \to \vb x \in \Delta$. This proves the following 
    corollary,
    
    \begin{corollary}\label{cor:lim-tilde-h}
        The limit,
        \begin{align}
        \lim_{d(x_1, x_2) \to 0}\norm*{\tilde h(x; x_1, x_2) - \mu
        }_{C^1(\surface)}
        = 0, 
        \end{align}
        holds, where $d(x_1, x_2)$ is the Riemannian distance in $\surface$.
    \end{corollary}

    Let $\surface^2_{\diag} = (\surface\times\surface)\setminus\diag$
    endowed with the product metric. As differentiable manifolds, 
    $\moduli^{1,1}$ and 
    $\surface^2_\diag$ are equivalent. In what follows, we will consider
    $\tilde h$ as a 
    function $\surface \times \surface^2_{\diag} \to \reals$. Let $U \subset 
    \surface$ be an open and dense subset and let 
    $\varphi:U \to V \subset \cpx$ be a 
    holomorphic chart. In what follows we denote points on the surface as $x$
    and points on $\cpx$ as $z$, so $z = \varphi(x)$ for $x \in U$. We also 
    assume vortices 
    and antivortices are both located in $U$, such that up to a holomorphic 
    chart, 
    $\tilde h: \surface \times V^2_\diag \to \reals$, where $V^2_\diag 
    = V^2\setminus \diag_V$ and  
    $\diag_V \subset \cpx^2$ is the diagonal set. On this chart partial 
    derivatives $\partial_{z_j}\tilde h$ are well defined functions
    \begin{align}
    \partial_{z_j}\tilde h : \surface \times V^2_\diag \to \cpx.
    \end{align}
    
    We denote the covariant derivative and Laplacian with respect to the first 
    variable by $\grad$ and $\laplacian$ and emphasize that 
    the metric on $V^2_\diag$ is the push forward of the metric induced 
    by the surface. Our aim is 
    %
    to estimate the rate at which the second 
    derivatives $\grad \partial_{z_j}\tilde h$ grow as a sequence $\vb z_n \in
    V^2_\diag$ diverges to the diagonal set. This will
    allow us to prove that the moduli space is incomplete. 
    %
    Since, $\laplacian$ and $\partial_{z_j}$ commute, $\partial_{z_j}\tilde h$ 
    is the solution to the elliptic problem,
    %
\begin{align}
-\laplacian\partial_{z_j}\tilde h = \fnV(h)
\partial_{z_j} \tilde h + \sign_j\,\fnV(h)\,\partial_{z_j} \fnv_j,
\end{align}
    %
    where $v_j(x) = 4\pi \,G(x, \varphi^{-1}(z_j))$. 
    Let $\dist_j(x) = \dist(x, x_j)$, $x_j = \varphi^{-1}(z_j)$, we know  
    there is a uniform constant $C$,
    such that the derivative of Green's function is
    bounded~\cite{aubin2013some},
    \begin{align}
        \abs{\grad G(x, x_j)} < \frac{C}{\dist_j}, &&
        \abs{\grad_2 G(x, x_j)} < \frac{C}{\dist_j},
      \end{align}
      where $\grad_2 G$ is the covariant derivative with respect to the 
      second variable. 
    Recall in holomorphic coordinates the metric is 
    $e^{\Lambda(z)}\abs{dz}^2$, hence, if $z_j$ is restricted to a 
    bounded domain,
    \begin{align}
    \abs{\del_{z_j}\fnv_j} \leq  4\pi e^{-\Lambda(z_j)}\,\abs{\grad_2 
    G(x,\varphi^{-1}(z_j))}
    < \frac{C}{d_j}.
    \end{align}
    

    \begin{lemma}\label{lem:vol-dist-bounds}
      For any positive constant
      $C_1$, there is another constant $C$, such that, for all 
      $x, x_1, x_2 \in U$,
      \begin{align}
        \frac{\dist_{12}^2}{C_1\dist_1^2 + \dist_2^2}
        & \leq C,\label{eq:vol-1}
        \\
          \frac{\dist_j\dist_k^2}{(C_1\dist_1^2 + \dist_2^2)^2}
        &\leq \frac{C}{\dist_{12}},\label{eq:vol-2}
      \end{align}
      where $\set{\dist_j, \dist_k} = \set{\dist_1, \dist_2}$ 
      and $\dist_{12}=\dist(x_1, x_2)$.
    \end{lemma}

    \begin{proof}
      By the triangle inequality and Cauchy-Schwarz,
      \begin{align}
        \dist_{12} \leq \dist_1 + \dist_2 \leq C\,(\dist_1^2 + \dist_2^2)^{1/2},
      \end{align}
      on the other hand, any two norms in a finite dimensional vector
      space are equivalents, hence, there is another constant such
      that,
      \begin{align}
        (\dist_1^2 + \dist_2^2)^{1/2} \leq C\,(C_1\dist_1^2 + \dist_2^2)^{1/2},
      \end{align}
      from these two inequalities we obtain the first claim of the
      lemma. For the second claim, it is enough to prove that the
      inequality
      \begin{align}
        \frac{\dist_1\dist_2^2}{(C_1\dist_1^2 + \dist_2^2)^2}
        \leq \frac{C}{\dist_{12}}, 
      \end{align}
      holds, the remaining case being equivalent to this one after
      relabelling $\dist_1$ and $\dist_2$. Let us note that since,
      \begin{align}
        \dist_1\dist_2 \leq \half (\dist_1^2+\dist_2^2)
        \leq C\,(C_1\dist_1^2 + \dist_2^2),
      \end{align}
      is sufficient to prove that,
      \begin{align}
        \frac{\dist_2}{C_1\dist_1^2 + \dist_2^2}
        \leq \frac{C}{\dist_{12}}.\label{eq:vol-dist2-c1norm-bound}
      \end{align}
      
      If
      $\dist_2 \leq \dist_1$, by the triangle inequality we have,
      \begin{align}
        \dist_2\dist_{12}
        &\leq \dist_1\dist_2 + \dist_2^2\nonumber\\
        &\leq \dist_1^2 + \dist_2^2\nonumber\\
        &\leq C\,(C_1\dist_1^2 + \dist_2^2),
      \end{align}
      hence~\eqref{eq:vol-dist2-c1norm-bound}. On the other hand, if
      $\dist_1\leq \dist_2$, repeating the 
      previous step, we find that
      \begin{align}
        \dist_1\dist_{12} \leq C\,(C_1\dist_1^2 + \dist_2^2),
      \end{align}
      this inequality, together with~\eqref{eq:vol-1} and the triangle
      inequality, implies,
      \begin{align}
        \frac{\dist_2}{C_1\dist_1^2+\dist_2^2}
        \leq \frac{\dist_1}{C_1\dist_1^2+\dist_2^2}
        + \frac{\dist_{12}}{C_1\dist_1^2+\dist_2^2}
        \leq \frac{C}{\dist_{12}}.
      \end{align}

      In any case, we conclude that equation~\eqref{eq:vol-2} holds.
    \end{proof}

    \begin{lemma}\label{lem:vav-sph-green-log-asympt}
      There is a constant $C$ such that for any pair of distinct
      points $x_1, x_2 \in \surface$,
      \begin{align}
        \abs*{G(x_1,x_2) - \frac{1}{2\pi}\log\,\dist(x_1,x_2)} \leq C.
      \end{align}
    \end{lemma}

    \begin{proof}
        We cover $\surface$ with a finite cover of metric disks 
        $\disk_{R_j/2}(p_j)$ such that $R_j < \delta$, where $\delta$ is the 
        injectivity radius of the metric and for each disk there is a 
        holomorphic 
        chart $\varphi_j: U_j \to \cpx$, $\disk_{R_j}(p_j) \subset 
        U_j$. Let $R = \min\set{R_j}$, for any pair of distinct points $x_1, 
        x_2 
        \in \surface$,
        such that $d(x_1, x_2) < R/2$,  
        there is a disk such that $x_1, x_2 \in \disk_{R_j}(p_j)$. 
        For any disk in the cover, let $R_j'$ be a positive radius, such that,
        \begin{align}
        |\varphi_j(x) - \varphi_j(p_j)| < R_j', \qquad
        \forall x \in \disk_{R_j}(p_j).
        \end{align}
        
        Let $z_j = \varphi_j(p_j)$ and let us denote by $D_{R'_j}(z_j) \subset 
        \cpx$ the 
        holomorphic disk of radius $R'_j$ centred at $z_j$. For any small 
        $\epsilon > 0$ there are continuous functions $\tilde G_j: D_{R'_j + 
        \epsilon}(z_j) \times D_{R'_j + \epsilon}(z_j) \to \reals$ such that 
        if $x_1, x_2 \in \disk_{R_j}(p_j)$,
        \begin{align}
        G(x_1, x_2) = \frac{1}{2\pi} \log\,\abs{\varphi_j(x_1) - 
        \varphi_j(x_2)} 
        + \tilde G_j(\varphi_j(x_1), \varphi_j(x_2)).
        \end{align}  
        
        If $\exp \Lambda_j(z)$ is the conformal factor of the metric in the 
        chart $\varphi_j$, let 
        \begin{equation}
        \begin{aligned}
        M_j &= \max \set{e^{\Lambda_j(z)/2}\,:\, z \in 
            \overline{D_{R'_j}(z_j)}},
        \\
        m_j &= \min \set{e^{\Lambda_j(z)/2}\,:\, z \in 
            \overline{D_{R'_j}(z_j)}},
        \end{aligned}
        \end{equation}
        and $M = \max_j \set{M_j}$, $m = \min_j \set{m_j}$. Since each 
        $\disk_{R_j}(p_j)$ is geodesically convex, for any $x_1, x_2 \in 
        \disk_{R_j}(p_j)$,
        \begin{align}
        m\,\abs{\varphi_j(x_1) - \varphi_j(x_2)} \leq d(x_1, x_2) 
        \leq M\,\abs{\varphi_j(x_1) - \varphi_j(x_2)}.
        \end{align}

      Taking the log of this inequality we find a positive constant such that,
      \begin{align}
      \abs{\,d(x_1, x_2) - \log\,\abs{\varphi_j(x_1) - \varphi_j(x_2)}\,} 
      \leq C,
      \end{align}
      whenever $x_1, x_2 \in \disk_{R_j}(p_j)$. Since each function $\tilde 
      G_j$ is continuous in the compact set $\overline{D_{R'_j}(z_j)}$, we find 
      another constant such that,
      \begin{multline}
      \abs*{G(x_1, x_2) - \frac{1}{2\pi}\log\,d(x_1, x_2)} =\\
      \abs*{\frac{1}{2\pi}\pbrk{\log\,\abs{\varphi_j(x_1) - \varphi_j(x_2)} - 
      \log\,d(x_1,
              x_2)} + \tilde G_j(\varphi_j(x_1), \varphi_j(x_2))}
      \leq C.      
      \end{multline}
      
      This proves the inequality whenever $d(x_1, x_2) < R/2$. Since $G$ and 
      the distance function are continuous on the compact set,
      \begin{align}
      \set{(x_1, x_2) \in \surface \times \surface \,:\, d(x_1, x_2) \geq 
          \frac{R}{2}},\
      \end{align}
      we can find a second constant satisfying the inequality whenever 
      $d(x_1, x_2) \geq R/2$. Taking the maximum of both constants 
      concludes the lemma.
    \end{proof}

    \begin{lemma}\label{lem:vol-unif-bound-pot}
        Let $D$ be any bounded domain on $\cpx$. For any $p > 0$, there is a 
        constant $C$, independent of $z_1, z_2 \in D$, $z_1 \neq z_2$ such 
        that, if $x_j = \varphi^{-1}(z_j)$,
    \begin{align}
    \norm{\fnV(h)\partial_{z_j} v_j}_{\Lsp^p} \leq \frac{C}{\dist(x_1, x_2)}.
    \end{align}
\end{lemma}

    \begin{proof}
      %
    By lemma~\ref{lem:vav-sph-green-log-asympt}, there is
    a constant, such that for all $x, y \in \surface$, $x \neq y$,
    \begin{align}
    \abs*{G(x, y) - \frac{1}{2\pi} \log \dist(x,y)} \leq C.
    \end{align}
    
      %
      Hence,
      %
\begin{align}
\abs{\fnV(h)\partial_{z_j}\fnv_j} = \abs*{
    \frac{4e^{\fnv_1}e^{\fnv_2}e^{\tilde h}}{ 
        (e^{\fnv_1}e^{\tilde h} + e^{\fnv_2})^2
    } \partial_{z_j}\fnv_j} \leq
C \abs*{\frac{4 \dist_1^2\dist_2^2\,e^{\tilde h}}{ 
        (\dist_1^2e^{\tilde h} + \dist_2^2)^2
    } \frac{1}{\dist_j}},
\end{align}
where the constant depends on $D$. 
    Since $\tilde h$ is uniformly
      bounded on $\surface$, there are constants $C$, $C_1$, such that
       by lemma~\ref{lem:vol-dist-bounds}. 
      %
      \begin{align}
\abs{\fnV(h)\partial_{z_j}\fnv_j} \leq
C \abs*{\frac{\dist_1^2\dist_2^2}{ 
        (\dist_1^2C_1 + \dist_2^2)^2
} \frac{1}{\dist_j}}
\leq \frac{C}{\dist(x_1,x_2)},
\end{align}
this 
 inequality implies the claim.
      %
    \end{proof}

    The proof of the lemma depends only on properties of Green's
    function, we could repeat the
    proof of lemma~\ref{lem:vol-unif-bound-pot} using $\grad \fnv_j$
    instead of $\partial_{z_j}v_j$ to prove for 
    any given domain $D\subset \cpx$ the existence of a constant, 
    independent of $z_1, z_2 \in D$ , such
    that,
    \begin{align}
      \norm{V(h)\grad v_j}_{\Lsp^p} \leq \frac{C}{d(x_1, x_2)}.
    \end{align}

    In the next lemmas we prove that the bilinear form, 
    \begin{align}
      \Bop:\Hsp^1\times \Hsp^1 \to \reals,
      &&
         \Bop(\phi,\psi) = \lproduct{\grad\phi, \grad\psi}^2_{\Lsp^2}
         + \lproduct{\fnV(h)\phi, \psi}_{\Lsp^2},
    \end{align}
    is coercive with a uniform coercivity constant.

    \begin{lemma}
      If $\fnV_n: \surface \to \reals$ is a sequence of continuous,
      uniformly bounded functions converging pointwise to the
      continuous function $\fnV_*$, and $\phi_n\to \phi_*$ in $\Lsp^2$,
      then
      \begin{align}
        \lproduct{\fnV_n, \phi_n^2}_{\Lsp^2} \to \lproduct{\fnV_*,
        \phi_*^2}_{\Lsp^2}.
      \end{align}
    \end{lemma}

    \begin{proof}
      We have,
      \begin{align}
        \abs{\lproduct{\fnV_n,\phi_n^2}_{\Lsp^2} -
        \lproduct{\fnV_*,\phi_*^2}_{\Lsp^2}}
        \leq
        \abs{\lproduct{\fnV_n,\phi_n^2 - \phi_*^2}_{\Lsp^2}}
        +
        \abs{\lproduct{\fnV_n - \fnV_*, \phi_*^2}_{\Lsp^2}}.
      \end{align}

      Since the functions $\fnV_n$ are uniformly bounded, there is a
      constant $C$ such that,
      \begin{align}
        \abs{\lproduct{\fnV_n,\phi_n^2 - \phi_*^2}_{\Lsp^2}}
        &\leq C\,\lproduct{
          \abs{\phi_n - \phi_*},
          \abs{\phi_n + \phi_*}}_{\Lsp^2}\nonumber\\
        &\leq C\,\norm{\phi_n - \phi_*}_{\Lsp^2}\,\norm{\phi_n +
          \phi_*}_{\Lsp^2}, 
      \end{align}
      by the convergence $\phi_n \to \phi_*$ in $\Lsp^2$, we obtain
      the limit 
      \begin{align}
        \abs{\lproduct{\fnV_n,\phi_n^2 - \phi_*^2}_{\Lsp^2}} \to 0.
      \end{align}

      Since there is a constant $C$ such that the functions
      $(\fnV_n - \fnV_*)\phi_*^2$ are bounded by the measurable
      function $C\phi_*^2$ and $\fnV_n - \fnV_* \to 0$ pointwise, by
      the dominated convergence theorem,
      \begin{align}
        \abs{\lproduct{\fnV_n - \fnV_*, \phi_*^2}_{\Lsp^2}} \to 0.
      \end{align}
      
      Therefore,
      \begin{align}
        \abs{\lproduct{\fnV_n,\phi_n^2}_{\Lsp^2} -
        \lproduct{\fnV_*,\phi_*^2}_{\Lsp^2}} \to 0,
      \end{align}
      concluding the proof of the lemma.
    \end{proof}

    \begin{lemma}\label{lem:vol-unif-coercive}
      There is a positive constant $C$, independent of $(x_1, x_2) \in 
      \surface^2_\diag$, such
      that for any $\phi \in \Hsp^1$, 
      \begin{align}
        C\norm{\phi}^2_{\Hsp^1} \leq \Bop(\phi, \phi).
      \end{align}
    \end{lemma}

    \begin{proof}
      By the bilinearity of $\Bop$, it is sufficient to prove the
      lemma assuming $\norm{\phi}_{\Hsp^1} = 1$.     
      Let us assume towards a contradiction the statement is false, in
      this case there is a sequence
      $(\phi_n,\vb{x}_n) \subset \Hsp^1\times \surface^2_{\diag}$,
      with $\norm{\phi_n}_{\Hsp^1} = 1$, such that,
      \begin{align}
        \Bop(\phi_n, \phi_n) = \norm{\grad\phi_n}_{\Lsp^2}^2 +
        \lproduct{\fnV_n, \phi_n^2}_{\Lsp^2} \to 0,
      \end{align}
      where $\fnV_n = \fnV(h_n)$ is the potential function determined
      by $h_n$, the solution to the Taubes equation with data
      $\vb{x}_n$. Since the functions $\fnV_n$ are non-negative,
      \begin{align}
        \norm{\grad\phi_n}_{\Lsp^2} \to 0,
        &&
          \lproduct{\fnV_n, \phi_n^2}_{\Lsp^2} \to 0.
      \end{align}
      
      Passing to a subsequence if necessary, we can assume
      $\phi_n \wto \phi_*$ in $\Hsp^1$ and strongly in $\Lsp^2$ and
      $\vb{x}_n \to \vb{x}_{*}$ in $\surface\times \surface$. Since
      the functions
      \begin{align}
        e^{\fnv_j}: \surface\times\surface \to \reals
      \end{align}
      are continuous and $\tilde h_n$ varies continuously with the
      initial data, if $\vb{x}_{*} \not\in \diag$, we have the uniform
      convergence $\fnV_n \to \fnV_* = \fnV(h_{*})$, where $h_{*}$ is
      the solution to the Taubes equation determined by $\vb{x}_{*}$. On
      the other hand, if $\vb{x}_{*} \in \diag$, we know that
      $\tilde h_n \to \mu$ in $C^1$, hence, we have pointwise
      convergence
      $\fnV_n \to \fnV_{*} \equiv 4\exp(\mu) (\exp(\mu) + 1)^{-2}$. In
      any case, by our previous lemma,
      \begin{align}
        \lproduct{\fnV_n, \phi_n^2}_{\Lsp^2} \to
        \lproduct{\fnV_*, \phi_*^2}_{\Lsp^2},
      \end{align}
      but this limit is zero, hence $\phi_{*} = 0$ almost everywhere
      and $\phi_n \to 0$ in $\Hsp^1$ strongly, a contradiction.
    \end{proof}

    \begin{proposition}\label{prop:sph-ddh-bound}
      Let $D \subset V$ be any bounded domain. There is a positive constant
      $C(D)$, such that 
      %
      \begin{align}
        \norm{\partial_{z_j}\tilde h}_{C^1} \leq \frac{C}{\dist_{12}},
        &&
           \text{and}
        &&
           \norm{\grad \tilde h}_{C^1} \leq \frac{C}{\dist_{12}},
      \end{align}
      for all $z_1, z_2 \in D$ 
      with $z_1 \neq z_2$, where $\tilde h(x) = \tilde h(x; \varphi^{-1}(z_1), 
      \varphi^{-1}(z_2))$ and  
      $\dist_{12} = \dist(x_1, x_2)$.
    \end{proposition}

    \begin{proof}
      $\partial_{z_j}\tilde h$ is a solution to the equation
      %
        \begin{align}
        -\laplacian \partial_{z_j}\tilde h
        = \fnV(h)\partial_{z_j}\tilde h + \sign_j\fnV(h)\partial_{z_j}\fnv_j.
        \end{align}
      
      By lemma~\ref{lem:vol-unif-coercive}, there is a positive
      constant $C_1$ 
      independent of $z_1, z_2$, such that
      \begin{align}
        C_1\,\norm{\phi}_{\Hsp^1}^2 \leq \norm{\grad\phi}_{\Lsp^2}^2 +
        \lproduct{\fnV(h)\,\phi, \phi}_{\Lsp^2},
      \end{align}
      for all $\phi \in \Hsp^1$. As in the proof of 
      lemma~\ref{lem:vol-unif-bound-pot}, a
      second uniform constant, dependent on $D$, can be found such that,
      %
      \begin{align}
        \norm{\fnV(h)\partial_{z_j}\fnv_j}_{\Lsp^2} \leq \frac{C_2}{\dist_{12}}.
      \end{align}

      By the Lax-Milgram theorem, we obtain the bound,
      %
      \begin{align}
\norm{\partial_{z_j}\tilde h}_{\Hsp^1} \leq \frac{C}{\dist_{12}},
\end{align}

      where $C = C_2/C_1$. Now we follow a recursive argument: by
      Schauder's estimates, 
      $\norm{\partial_{z_j}\tilde h}_{\Hsp^2}$ is also bounded
      by $C\,d_{12}^{-1}$ for some constant $C$. By Sobolev's
      embedding, there is another constant such that
      $\norm{\partial_{z_j}\tilde h}_{C^0}$ is also bounded by
      $C\,d_{12}^{-1}$. Thence, for any given $p > 2$, 
      %
      %
            \begin{align}
      \norm{\partial_{z_j}\tilde h}_{\Lsp^p} \leq \frac{C}{\dist_{12}}.
      \end{align}
      By the elliptic estimates, 
      %
\begin{align}
\norm{\partial_{z_j}\tilde h}_{\Wsp^{2, p}}
&\leq C\,(\norm{\laplacian \partial_{z_j}\tilde h}_{\Lsp^p}
+ \norm{\tilde h}_{\Lsp^p})\nonumber\\
&\leq C\,(\norm{\fnV(h)\partial_{z_j}\tilde h}_{\Lsp^p}
+ \norm{\fnV(h)\partial_{z_j}\fnv_j}_{\Lsp^p} + \norm{\partial_{z_j}\tilde
    h}_{\Lsp^p})\nonumber\\
&\leq \frac{C}{\dist_{12}},
\end{align}
      for the last inequality we have used that the function $\fnV(t)$
      is bounded. Sobolev's embedding implies the claimed
      bound,
      %
      \begin{align}
        \norm{\partial_{z_j}\tilde h}_{C^1} \leq \frac{C}{\dist_{12}}. 
      \end{align}

      This argument is also valid for $\grad \tilde h$, because it is a
      solution to the elliptic problem,
      \begin{align}
        -\laplacian (\grad \tilde h) = \fnV(h)\grad \tilde h + \fnV(h)(
        \grad\fnv_{1} - \grad\fnv_{2}),
      \end{align}
      and the upper bound
      \begin{align}
        \norm{\fnV(h)\grad \fnv_{j}} \leq \frac{C}{\dist_{12}}
      \end{align}
      also holds. 
    \end{proof}


    For latter application, we need to translate this estimate to a
    holomorphic chart.

    \begin{lemma}\label{lem:sph-dist-dist-comp}
    Let $\varphi:U \subset \surface \to V \subset \cpx$ be a holomorphic 
    chart and let $D$ be a geodesically convex neighbourhood such that 
    $\overline D \subset U$, there is a positive constant 
    $C$, such that for all $z_1, z_2 \in \varphi(D)$,
    \begin{align}
    C\,\abs{z_1 - z_2} \leq d(\varphi^{-1}(z_1), \varphi^{-1}(z_2)). 
    \end{align}
    \end{lemma}

    \begin{proof}
The conformal factor is 
      a 
      continuous positive function on 
      $V$ and $\varphi(\overline D)$ is compact, hence there is a constant 
      $C > 0$, such that for all $z \in \varphi(D)$,
      %
%
\begin{align}
C^2 \leq e^{\Lambda(z)}.
\end{align}

    Since $D$ is geodesically convex, for any pair $z_1, z_2 \in 
    \varphi(D)$, there is a curve $\gamma: [0, 1] \to \varphi(D)$ joining $z_1$ 
    to $z_2$ 
    such that $\varphi^{-1}\circ \gamma$ is a minimizing geodesic joining 
    $\varphi^{-1}(z_1)$ to $\varphi^{-1}(z_2)$, hence, 
      %
      \begin{align}
        C\int_0^1\abs{\dot \gamma}\,ds \leq \int_0^1 e^{\Lambda/2}\abs{\dot
    \gamma}\,ds = \dist(\varphi^{-1}(z_1), \varphi^{-1}(z_2)).
\end{align}

      By the triangle inequality,
      %
      \begin{align}
\abs{z_1 - z_2} = \abs*{\int_0^1\dot\gamma}\leq \int_0^1 \abs {\dot
    \gamma}\,ds,
\end{align}
      yielding the result.
    \end{proof}

    The advantage of the holomorphic chart is that it makes computations 
     possible, on the other hand, the Riemannian distance 
    is a geometric invariant defined globally on the surface and better suited 
    to prove analytical properties of the solutions to the Taubes equation. 
    For the next lemma, notice that if $\surface_1 \times \surface_2$ is a 
    product of 
    Riemmann surfaces, for any function $f: \surface_1 \times \surface_2 
    \to \cpx$ in local coordinates $\varphi_j:U_j \to \cpx$, 
    $\varphi_j(x_j) = z_j$,
    \begin{align}
    \del_{x_1}\del_{x_2}f = \partial_{z_1z_2}f\,dz^1\otimes dz^2
    \in \Omega^{(2,0)}(\surface_1\times\surface_2).
    \end{align} 

    In the product metric, $dz^1$ and $dz^2$ are orthogonal, hence,
    \begin{align}
    \abs{\partial_{x_1}\partial_{x_2}f} = \abs{\partial_{z_1, z_2}f}\,
        \abs{dz^1}\,\abs{dz^2}.
    \end{align}    
    
%
%

    \begin{lemma}\label{lem:sph-del1-b1-bound}
    For any holomorphic chart $\varphi: U\subset \surface \to V \subset \cpx$ 
    and any geodesically convex neighbourhood $D$ such that $\overline D 
    \subset U$, there is a constant $C > 0$  
      such that, for all $z_1, z_2 \in \varphi(D)$, $z_1 \neq z_2$,
\begin{align}
\abs{\del_{z_1}b_1(z_1, z_2)} \leq \frac{C}{\abs{z_1 - z_2}},
\end{align}
where the coefficient $b_1$ appearing in the metric 
of $\moduli^{1,1}(\surface)$ is defined as in~\eqref{eq:pre-bcoef}.
    \end{lemma}

    \begin{proof}
If  $z_1, z_2 \in \varphi(D)$, there is a smooth
function $\tilde v: \varphi(\overline D) \times \varphi(\overline D) \times 
\varphi(\overline D) \to \reals$, such that 
for
all triples $z, z_1, z_2$ of points in the domain with $z_1 \neq z_2$,
\begin{align}
v(\varphi^{-1}(z)) = \log\,\abs{z - z_1}^2 - \log\,\abs{z - z_2}^2+ \tilde
v(z, z_1, z_2).
\end{align}

Hence,
      \begin{align}
        b_1(z_1, z_2) &= 2\eval{\bar\del}{z=z_1}(h(\varphi^{-1}(z)) 
        - \log\,\abs{z - z_1}^2)\nonumber\\
        &= 2\eval{\bar\del}{z = z_1}(\tilde h(\varphi^{-1}(z)) - \log\,\abs{z - 
        z_2}^2 +
        \tilde v(z,z_1,z_2))\nonumber\\
        &= 2\,\bar\del_z\tilde h(\varphi^{-1}(z_1); \varphi^{-1}(z_1), 
        \varphi^{-1}(z_2)) - \frac{2}{\bar z_1 - \bar
        z_2} + 2\,\bar\del_z\tilde v(z_1, z_1, z_2),
    \end{align}
    where $\bar\del_z$ refers to complex derivatives with respect to the first 
    entry. In the following calculation we denote $\tilde h(\varphi^{-1}(z_1); 
    \varphi^{-1}(z_1), \varphi^{-1}(z_2))$ by $\tilde h$ and $\tilde v(z_1, 
    z_1, z_2)$ by $\tilde v$, whence,
    %
\begin{align}
\del_{z_1}b_1 &= 2 \pbrk{
    \del_z\bar\del_z\tilde h
    +{\del_{z_1}}\bar\del_z\tilde h 
    +\del_{z}\bar\del_z \tilde v
    +\del_{z_1}\bar\del_z \tilde v}\nonumber\\
&= 2 \pbrk{ -\frac{e^{\Lambda(z_1)}}{2}\laplacian_{\surface}\tilde h
+ \bar\del_z \del_{z_1} \tilde h 
+\del_{z}\bar\del_z \tilde v
+\del_{z_1}\bar\del_z \tilde v}\nonumber\\
&= 2\pbrk{\frac{e^{\Lambda(z_1)}}{2}\Fstable(h)
+ \bar\del_z \del_{z_1} \tilde h 
+\del_{z}\bar\del_z \tilde v
+\del_{z_1}\bar\del_z \tilde v}.
\end{align}

    Since $\varphi(\overline D)$ is compact, $\Lambda(z_1)$ and the last two 
    terms are 
    bounded functions on $\varphi(D)$ by continuity. Since function $F(t)$ is 
    also 
    bounded, we conclude the same statement for the first term. For the second 
    term, if $x = \varphi^{-1}(z)$ and $x_j = 
    \varphi^{-1}(z_j)$, we
    have by lemma~\ref{lem:sph-dist-dist-comp} and
    proposition~\ref{prop:sph-ddh-bound},
    %
    \begin{align}
\abs*{\bar\del_z \del_{z_1} \tilde h} &=
e^{\Lambda(z_1)/2} \abs*{\bar\del_z \del_{z_1} \tilde 
    h} \abs*{dz}\nonumber\\
&= \abs{\bar \partial_x\partial_{z_1}\tilde h(x_1, \varphi^{-1}(z_1), 
\varphi^{-1}(z_2))}\nonumber\\
&\leq \frac{C}{\dist(x_1,x_2)}\nonumber\\
&\leq \frac{C}{\abs{z_1 - z_2}}.
\end{align}

Therefore the lemma is proved.
    \end{proof}

    \begin{theorem}\label{thm:moduli-space-incomplete-compact-surface}
      The moduli space is incomplete. There is a
      Cauchy sequence $\set{\vb x_n} \subset \moduli^{1,1}(\surface)$
      such that $\vb x_n \to \vb x \in \diag$ as a sequence in
      $\surface\times \surface$.
    \end{theorem}

    \begin{proof}
Let $\varphi: U\subset \surface \to \cpx$ be an holomorphic chart defined on 
an open and dense neighbourhood $U$. 
    Let $z_1 \in \cpx$ be chosen such that $\varphi^{-1}(s z_1)$, $0 \leq s 
    \leq 
    1$ is contained in a geodesically convex neighbourhood of 
    $\varphi^{-1}(0)$. 
    Let us define  the curve, 
      %
       \begin{align}
 \gamma: (0, 1] \to \cpx^2_\diag, &&
 \gamma(s) = (s\,z_1, 0).
 \end{align}
      
      Let $z(s) = s z_1$ and let 
      $\varphi^{-1}_*\gamma(s) = (\varphi^{-1}(z(s)), \varphi^{-1}(0))$, 
      be the push forward of the curve $\gamma$ to the moduli space, hence,
      %
         \begin{align}
   \abs{\varphi^{-1}_*\dot\gamma}_\moduli^2 = (e^{\Lambda(z)}(1-\tau) +
   \del_{z_1}b_1)\,\abs{z_1}^2,
   \end{align}
   where we denote by $\abs{\cdot}_\moduli$ 
   the norm of vectors in $T_{\varphi^{-1}_*\gamma(s)}\moduli^{1,1}$.

      By Lemma~\ref{lem:sph-del1-b1-bound} there is a constant $C$,
      such that,
      \begin{align}
        \abs{\del_{z_1}b_1} \leq \frac{C}{\abs z} = \frac{C}{s\,\abs {z_1}}.
      \end{align}

      Since the conformal factor is a continuous positive function
      defined on the whole plane, there is another constant, also
      denoted $C$, such that,
      \begin{align}
        \abs {\varphi^{-1}_*\dot \gamma}_\moduli \leq \frac{C}{s^{1/2}}.
      \end{align}
      Let $\ell[\gamma, a, b]$ be the arc-length of the segment
      $\gamma|_{[a, b]}$, $a, b \in (0, 1)$, there is another constant,
       also denoted by $C$, such that,
      \begin{align}
        \ell[\gamma, a, b] = \int_a^b\abs{\varphi^{-1}_*\dot \gamma}_\moduli\,ds
        \leq C (b^{1/2} - a^{1/2}),
      \end{align}
whence,
      \begin{align}
        \dist(\varphi^{-1}_* \gamma(b), \varphi^{-1}_* \gamma(a)) \leq 
        C\,(b^{1/2} - a^{1/2}).
      \end{align}

      This inequality shows if $\set{s_n} \subset (0, 1]$ is any
      converging sequence $s_n \to 0$, the new sequence,
      \begin{align}
        \vb x_n = \varphi^{-1}_*\gamma(s_n) \in \moduli^{1,1}(\surface),
      \end{align}
      is Cauchy, however $\gamma$ is continuous which
      implies $\vb x_n \to (\varphi^{-1}(0), \varphi^{-1}(0)) \in 
      \diag_\surface$. 
      Therefore, the moduli space
      is incomplete.
    \end{proof}




\section{The volume of the moduli space}\label{sec:volume-moduli-space}

We conclude this chapter computing the volume of the moduli space 
$\moduli^{1,1}(\surface)$ for the 
round sphere and flat tori. As it will turn out, the existence of a Lie 
group of isometries will play an important role in the 
calculations. Symmetries were studied for their relation to conservation laws 
in a Schrodinger-Chern-Simons model by Manton and Nasir 
in~\cite{mantonnasir}, for the Riemann sphere, symmetries of the coefficients 
of 
the $\Lsp^2$ metric for vortices of a non-relativistic Chern-Simons model 
were treated by Rom\~ao~\cite{romao2001quantum}. We follow similar  
ideas for asymmetric vortices of the $O(3)$ Sigma model. There is a 
general conjecture for the volume of the moduli space by 
Rom\~ao-Speight~\cite{romao2018}, which can be stated as follows,

\begin{conjecture}[The volume conjecture]\label{conj:volume-conjecture}
    Given a compact Riemann surface $\surface$ of genus $g$ and total area 
    $\abs\surface$, let,
    \begin{align*}
    J_\pm &= 2\pi (1 \mp \tau)\abs\surface - 4\pi^2(k_\pm - k_\mp),\\
    K_\pm &= \mp 2\pi^2,
    \end{align*}
    then the total volume of the moduli space $\moduli^{k_+, k_-}(\surface)$ is,
    \begin{equation*}
    \vol(\moduli^{k_+, k_-}(\surface)) = 
    \sum_{l = 0}^g \frac{g!(g - l)!}{(-1)^l l!} \prod_{\sigma = \pm}
    \sum_{j_\sigma = l}^g 
    \frac{
        (2\pi)^{2l} J_\sigma^{k_\sigma - j_\sigma} K_\sigma^{j_\sigma - l}
    }{
        (j_\sigma - l)! (g - j_\sigma)! (k_\sigma - j_\sigma)!
    }.
    \end{equation*}
\end{conjecture}

For $\surface = \sphere_{round}$, they corroborated 
it for a 
vortex-antivortex pair and $\tau = 0$. We aim to confirm the
conjecture on the round sphere and flat tori for vortex-antivortex 
pairs and general $\tau$.

\subsection{The Riemann sphere}
\label{sec:sph-cp1-volume}


On the round sphere, the three dimensional Lie group of orthogonal
transformations, $O(3)$, acts by
isometries. The vortex equations are invariant under
isometric actions on the domain,  if $\isometry: \surface \to
\surface$ is an  isometry and $u$ is the solution of the Taubes equation 
with vortex set $\vset$ and antivortex set $\avset$, then $u \circ
\isometry$ is the solution with data $\isometry^{-1}(\vset)$,
$\isometry^{-1}(\avset)$.  We will make use of this symmetry to obtain
conservation laws for the non-trivial  coefficients $b_j$ and an explicit 
formula in the subspace of vortices and
antivortices located at antipodal positions. This formula will lead us
to the volume formula. We will prove the following theorem,

Recall the conformal factor of the sphere of
radius $R$ in a stereographic projection chart with coordinate $z$ is, 
\begin{align}
\cf=\frac{4R^{2}}{\left(1+\abs z^{2}\right)^{2}}.
\end{align}


We can give an explicit description of the coefficients in the metric
in the case of only $k_{+}$ coincident vortices or $k_{-}$ coincident
antivortices. By rotational symmetry, 
the function $u$ depends only on the chordal distance to
either the vortex or antivortex \cite{manton1999},
the coefficients $b_{\pm}$ in this case are, 
\begin{equation}
b_{\pm} = -\frac{ 2 k_{\pm}z_{\pm} }{ 1 +
  \abs{z_{\pm}}^{2}}.\label{eq:bpm-antipodal-position}  
\end{equation}

The proof relies on the rotational symmetry
of the configuration and is analogous to the proof for $n$ coincident
Ginzburg-Landau vortices on the sphere that can be found in
\cite{manton1993statisticalmec}. With this identity at hand, we prove the 
following theorem,
\begin{theorem}
    The volume of the moduli space $\moduli^{k_{+},0}\left( \sphere
    \right)$ is,
    
    \begin{equation}
    \label{eq:vol-moduli-k-0}
    \vol \left( \moduli^{k_{+},0}(\sphere) \right)
    =\frac{\left(4\pi^{2}R^{2}\left(
        2\,(1-\tau)-\frac{k_{+}}{R^{2}}\right)\right)^{k_{+}}}{k_{+}!}, 
    \end{equation}
    and the volume of $\moduli^{0, k_-}(\sphere)$ can be obtained from
    equation \eqref{eq:vol-moduli-k-0} by changing $\tau$ into
    $-\tau$. For a vortex-antivortex pair, the volume of $\moduli^{1,1}\left( 
    \sphere\right)$ is  
    \begin{align}
    \vol \left( \moduli^{1,1}(\sphere)  \right)
    =\left(8\pi^{2}R^{2}\right)^{2}\left(1-\tau^{2}\right).
    \end{align}
    \label{thm:vol}
\end{theorem}

For $k_+ = 0$ or $k_- = 0$ we follow ideas of 
Manton-Nasir~\cite{manton1999volume}, as their proof relies on the topology 
of the symmetric product $(\sphere)^N/S_N$, $S_N$ being the $N$ symmetric 
group, and can be adapted easily to vortices of the $O(3)$ Sigma model of the
same type. For the case $k_+ = k_- = 1$, we extend the proof given 
by Rom\~ao-Speight~\cite[Thm.~5.2]{romao2018} for the symmetric case. For 
general $\tau$ we no longer have the symmetry $(z_1, z_2) \mapsto (z_2, z_1)$, 
instead, we complement the symmetries induced by $SO(3)$ in the moduli space 
with 
the symmetry $(z_1, z_2) \mapsto (\conj z_1, \conj z_2)$ to deduce a suitable 
formula for the volume of a general K\"ahler metric on $\sphere\!_{\diag}$.

\subsubsection{$k_+$ vortices of the same type} 

If there are $k_+$ vortices on $\sphere$ and no antivortices, the moduli
space is isomorphic to $\pSpace^{n}$, the complex projective space of
dimension $k_+$ \cite{manton1999}. The subspace
$\moduli_0^{k_+,0}(\sphere) \subset \moduli^{k_+,0}(\sphere)$ of $k_+$
coincident vortices on the other hand is isomorphic to $\pSpace^1$,
and can be parametrized with the coordinate $z_+$ of the coincident
vortices.  By equation (\ref{eq:bpm-antipodal-position}) we know how
to compute the coefficient $b_{+}$ in $\moduli_0^{k_+,0}(\sphere)$, 
\begin{equation}
b_{+}=-\frac{2 k_+ z_+}{1+\abs {z_+}^2}.\label{eq:b-formula-k-0}
\end{equation}
The metric in $\moduli^{k_+,0}_0(\sphere)$ therefore is,
\begin{align}
ds^{2} & =2 k_+\pi\left((1-\tau)\cf + \frac{\del b_{+}}{\del
z_+}\right)\abs{dz_+}^{2}\nonumber \\ 
 & = k_+
 \pi\left(2(1-\tau) - \frac{k_+}{R^2}
 \right)\Omega\abs{dz_+}^{2},\label{eq:metric-moduli-kp-0}  
\end{align}
as can be seen, the metric is a multiple of the round metric,
hence, the volume of $\moduli^{k_+,0}_0(\sphere)$ is,
\begin{equation}
4\pi^{2} R^{2}k_+ \left( 2 (1-\tau)-\frac{k_+}{R^{2}}
\right),\label{eq:vol-moduli-N}  
\end{equation}
this volume is $k_+$ times the volume of the generating cycle in
$\pSpace^1$, 
\begin{equation}
4\pi^{2}R^{2}\left( 2 (1-\tau)-\frac{k_+}{R^{2}}
\right).\label{eq:generating-cycle-vol} 
\end{equation}

The total volume of the moduli space therefore is,
\begin{equation}
\vol\left(\moduli^{k_+,0}(\sphere)\right) =
\frac{\left(8\pi^{2}R^{2}(1-\tau)-4\pi^{2} k_+
  \right)^{k_+}}{k_+!},\label{eq:total-volume-moduli-space} 
\end{equation}
the proof of the volume formula in $\moduli^{0,k_-}(\sphere)$ is analogous,
\begin{equation}
\vol\left(\moduli^{0,k_-}(\sphere)\right) =
\frac{\left(8\pi^{2}R^{2}(1+\tau)-4\pi^{2} k_-
  \right)^{k_-}}{k_-!}.
\end{equation}

\subsubsection{The moduli space of vortex-antivortex pairs
}





In general, there is no explicit expression for the
coefficients $b_j$ of the metric if the cores are at general position, however, 
we
can deduce from the invariance of  
the Taubes equation under the action of $O(3)$ several constraints on the 
coefficients due to symmetry. Before doing so, we need a general lemma that 
will also be necessary for flat tori in the next section.

\begin{lemma}\label{prop:b-sing}
    Let $\varphi: U \subset \surface \to V\subset \cpx$ be a 
    holomorphic chart,  containing the core set $\mathcal{Z}$ of a 
    point in the 
    moduli space $\moduli^{1,1}(\surface)$.  
    For any bounded domain $D \subset V$, such that $\mathcal{Z} 
    \subset \varphi^{-1}(D)$, there are continuous functions $\tilde b_j: 
    D\times D \to 
    \cpx$, 
    $j = 1, 2$,  
    such that:
    \begin{enumerate}
        \item If $\varphi(\mathcal{Z}) = \left\{z_1, z_2\right\}$, where 
        $z_1$ $(z_2)$ is the vortex (antivortex),
        \begin{align}
        b_j(z_1, z_2) = \frac{-2\,\sign_j}{\bar z_1 - \bar z_2} + \tilde
        b_j(z_1, z_2), 
        \end{align}
        where $b_j$, $j = 1, 2$, are the non-trivial coefficients in the 
        $\Lsp^2$ 
        metric, defined in lemma~\ref{lem:coef-sym}.
        
        \item 
        \begin{align}
        \lim_{|z_1 - z_2| \to 0} \tilde b_j(z_1, z_2) = 0.
        \end{align}
    \end{enumerate}
\end{lemma}

\begin{proof}
    On $\varphi^{-1}(D)$, Green's function
    can be written as
    \begin{align}
    G(x_1, x_2) = \frac{1}{2\pi}\,\log\,\abs{\varphi(x_1) - \varphi(x_2)} + 
    \tilde G(x_1, 
    x_2),
    \end{align}
    with a smooth regular part $\tilde G: \varphi^{-1}(D)\times 
    \varphi^{-1}(D) \to \reals$.  Therefore, the solution $h$ to the Taubes  
    equation can be written as
    \begin{align}
    h(x; x_1, x_2) = \tilde h(x; x_1, x_2) + \log\,\abs{\varphi(x) -
        \varphi(x_1)}^2 - \log\,\abs{\varphi(x) - \varphi(x_2)}^2 + \tilde v(x; 
    x_1, x_2), 
    \end{align}
    where
    \begin{align}\label{eq:tilde-v-vav-sigma}
    \tilde v(x; x_1, x_2) = 4\pi\, \tilde G(x, x_1)
    - 4\pi\, \tilde G(x, x_2),
    \end{align}
    and $\tilde h(x; x_1, x_2)$ can be extended in $C^1$ to the coincidence set 
    $x_1 = x_2$ by corollary~\ref{cor:lim-tilde-h}. 
    Denoting $h(\varphi^{-1}(z); \varphi^{-1}(z_1), 
    \varphi^{-1}(z_2))$ and $\tilde h(\varphi^{-1}(z); \varphi^{-1}(z_1), 
    \varphi^{-1}(z_2))$ as $h$, $\tilde h$, etcetera,
    \begin{align}
    b_j(z_1, z_2) &= 2\,\conj\del|_{z=z_j}
    \pbrk{\sign_j\,h - \log\abs{z - z_j}}\nonumber\\
    &= 2\,\conj\del_{z=z_j}
    \pbrk{\sign_j\tilde h - \log\abs{z - z_k} + \sign_j\tilde v}\nonumber\\
    &= \frac{-2}{\conj z_j - \conj z_k} 
    + 2\,\sign_j\conj\del|_{z=z_j}(\tilde h + \tilde v)\nonumber\\
    &= \frac{-2\,\sign_j}{\conj z_1 - \conj z_2} + \tilde b_j,
    \end{align}
    where the regular part $\tilde b_j$ is continuous in $D\times D$. This 
    proves the first statement. The second statement is a consequence of 
    corollary~\ref{cor:lim-tilde-h} and the fact that     
    by~\eqref{eq:tilde-v-vav-sigma},
    \begin{align}
    \lim_{|z_1 - z_2| \to 0} \conj\del\lvert_{z=z_j}\pbrk{\tilde 
    v(\varphi^{-1}(z); 
    \varphi^{-1}(z_1), \varphi^{-1}(z_2))} = 0.
    \end{align}
\end{proof}

Suppose $\gamma:U_{1} \subset \cpx \to U_{2} 
\subset \cpx$ is  
a holomorphic change of coordinates in ambient space, such that
$z_{k}\in U_{1}$ 
for all cores. There are pairs of corresponding
coefficients $b_{s}(z_{1},\ldots,z_{n})$,
$b'_{s}(z'_{1},\ldots,z'_{n})$ 
in each of the charts. Let $z'=\gamma(z)$,
$z'_{k}=\gamma(z_{k})$, as in~\cite{romao2001quantum}, we have the 
transformation rule  
\begin{equation}
b'_j = \frac{1}{\conj{\gamma'_j}}b_j-\frac{\conj{\gamma''_j}} {\left(
    \conj{ \gamma'_j}\right)^{2}}.\label{eq:b-coeffs-change-of-coords}
\end{equation}

Manton and Nasir noted in~\cite{manton1999c} that  
equation (\ref{eq:b-coeffs-change-of-coords}) is similar to the transformation
rule for the Levi-Civita connection on $\sphere$ and resembles the
topological nature of the coefficients $b_j$. In the sphere, the group of 
isometries is large, in the sense that it is a Lie group, and each of this 
isometries induces a holomorphic change of coordinates on the moduli space. 
We exploit this remark to prove the following lemmas. 

\begin{lemma}\label{lem:coeff-sym}
In the projective chart, the coefficients  $b_j$ satisfy the identities, 
\begin{gather}
\sum_k(2\,\conj z_k + \conj z_k^2\,b_k + \conj
b_k) = C,\label{eq:bs-gen-rel}\\
\sum_k\conj z_kb_k\in\reals,\label{eq:z-axis-symmetry-b}
\end{gather}
for some constant $C$. For a vortex-antivortex pair, $C = 0$. 
\end{lemma}

Rom\~ao deduced similar identities for vortices of a modified Chern-Simons 
model 
on the sphere in~\cite{romao2001quantum}, employing the action of $SO(3)$ on 
 the moduli space.

\begin{proof}
\newcommand*{\ldv}{\mathcal{L}}

 Let us consider a rotation  $\gamma:\sphere\to\sphere$. In a 
 stereographic projection chart, $\gamma$ can be represented as a
 Möbius transformation,
\begin{equation}
\gamma(z)=\frac{az+b}{-\conj bz+\conj a},
\end{equation}
for some coefficients $a, b \in \cpx$, 
such that $\abs a^{2}+\abs b^{2}=1$. Since $\gamma:\cpx\setminus\left\{
  \conj a / \conj b\right\} \to\cpx\setminus\left\{ -a/\conj b \right\} $ is a 
  holomorphic
change of coordinates, a rotation of the core positions 
in the sphere reads,
\begin{equation}
  b'_j=(-b\conj z_j+a)^{2}b_j-2b(-b\conj z_j+a).
  \label{eq:b-rotation-formula}
\end{equation}




Invariance of the solutions to the Taubes equation under
the group of isometries means  
that the vector fields generated by $SO(3)$ in the moduli space by
diagonally acting on the cores' positions are Killing fields. These fields are generated by the 1-parameter families of matrices,
\begin{equation}
\begin{aligned}
U_{X}(\alpha) & =\begin{pmatrix}\cos\left(\frac{\alpha}{2}\right) & -i\sin\left(\frac{\alpha}{2}\right)\\
-i\sin\left(\frac{\alpha}{2}\right) & \cos\left(\frac{\alpha}{2}\right)
\end{pmatrix}, &
U_{Y}(\beta) & =\begin{pmatrix}\cos\left(\frac{\beta}{2}\right) & -\sin\left(\frac{\beta}{2}\right)\\
\sin\left(\frac{\beta}{2}\right) & \cos\left(\frac{\beta}{2}\right)
\end{pmatrix},\\
U_{Z}(\gamma) & =\begin{pmatrix}e^{-i\frac{\gamma}{2}} & 0\\
0 & e^{i\frac{\gamma}{2}} &
\end{pmatrix},
\end{aligned}
\end{equation}
$\alpha, \beta, \gamma \in \reals$. We can compute conservation equations 
corresponding to the generators
of the Lie algebra $\mathfrak{su}(2)$. These equations correspond to conservation
of angular momentum in the moduli space. The generating Killing fields
in the moduli space are,
\begin{equation}
\label{eq:killing-fields-moduli}
\begin{split}
\xi_{X} &= \frac{i}{2}\sum_j(z_j^{2}-1)\del_{z_j}-(\conj z_j^{2}-1)
\del_{\conj z_j},\\ 
\xi_{Y} &= -\frac{1}{2}\sum_j(z_j^{2}+ 1) \del_{z_j} + (\conj z_j^{2}
+ 1 )\del_{\conj z_j},\\
\xi_{Z} &=-i\sum_jz_j\del_{z_j}-\conj z_j\del_{\conj z_j}.  
\end{split}
\end{equation}

By \eqref{eq:b-rotation-formula}, the Lie derivatives of the coefficients
are,

\begin{equation}
\label{eq:bs-lie-derivative}
\begin{split}
\ldv_{\xi_{X}}b_j & =i(\conj z_jb_j+1),\\
\ldv_{\xi_{Y}}b_j & =\conj z_jb_j+1,\\
\ldv_{\xi_{Z}}b_j & =-ib_j.\\  
\end{split}
\end{equation}

Hence the coefficients $b_j$ satisfy the identities,
\begin{align}
\frac{1}{2}\sum_k(z_k^{2}-1)\del_{z_k}b_j- (\conj
z_k^{2}-1)\del_{\conj z_k} b_j & =\conj z_jb_j + 1\\
-\frac{1}{2}\sum_k(z_k^{2} + 1)\del_{z_k}b_j+(\conj z_k^{2}+1)\del_{\conj z_k}b_j & =\conj z_jb_j+1,\\
\sum_kz_k\del_{z_k}b_j - \conj z_k\del_{\conj z_k}b_j &= b_j.
\end{align}




%
Recall the coefficients $b_j$ have the symmetries,
\begin{equation}
  \del_{z_k}b_j=\del_{\conj z_j}\conj b_k,
  \qquad\del_{\conj z_k}b_j=\del_{\conj z_j}b_k.\label{eq:delz-bs-identities}
\end{equation}

Hence,
\begin{align}
\sum_k(z_k^{2}- 1) \del_{\conj z_j}\conj b_k-(\conj
z_k^2 - 1)\del_{\conj z_j}b_k &= 2(\conj z_jb_j+1),\label{eq:3}\\
\sum_k(z_k^2 + 1) \del_{\conj z_j}\conj b_k+ (\conj
z_k^2 + 1)\del_{\conj z_j}b_k &= -2(\conj z_jb_j + 1),\label{eq:4}\\
\sum_kz_k\del_{\conj z_j}\conj b_k-\conj z_k\del_{\conj
z_j}b_k &= b_j.\label{eq:5}
\end{align}

Adding equations \eqref{eq:3} and \eqref{eq:4} and also subtracting
and conjugating the same pair of equations,
\begin{align}
\sum_kz_k^2\del_{\conj z_j}\conj b_k+\del_{\conj z_j}b_k &= 0,\\
\sum_k\del_{z_j}b_k + z_k^2\del_{z_j}\conj b_k &= -2(z_j\conj b_j + 1).
\end{align}

From these two equations, we deduce,
\begin{align}
\partial_{z_j}\sum_k(2\conj z_k + \conj z_k^2b_k + \conj b_k) = 0, &&
\partial_{\bar z_j}\sum_k(2 \bar z_k + \conj z_k^2b_k + \conj b_k) = 0,
\end{align}
hence $\sum_k(2 \bar z_k + \conj z_k^2b_k + \conj b_k)$
is constant. 

Equation \eqref{eq:5} implies,
\begin{equation}
\del_{z_j}\sum_k\left(z_k\conj b_k - \conj
  z_kb_k\right) = 0.\label{eq:delzbars-imag-zrbr}
\end{equation}

From this equation and its conjugate, $\sum_k\left(\conj
  z_kb_k-z_k\conj b_k\right)$
is constant, but this quantity must be zero when all the vortices
and antivortices are located on the real line. Therefore,
\begin{equation}
\sum_k\conj z_kb_k\in\reals.
\end{equation}

Finally, for a vortex-antivortex pair at positions $z_\pm = \pm \epsilon$, 
we have $b_\pm(\epsilon, -\epsilon) \in \reals$ and by~\eqref{eq:bs-gen-rel},
\begin{align}
b_+(\epsilon, -\epsilon) + b_-(\epsilon, -\epsilon) = \frac{C}{1 + \epsilon^2}.
\end{align}

By lemma~\ref{prop:b-sing}, there are continuous functions 
$\tilde b_\pm: \reals \to \reals$ such that,
\begin{align}
b_\pm(\epsilon, -\epsilon) = \mp \frac{1}{\epsilon} + \tilde b_\pm(\epsilon),
\end{align}
and $\lim_{\epsilon \to 0}\tilde b_\pm(\epsilon) = 0$, 
hence,
\begin{align}
\lim_{\epsilon \to 0} (b_+(\epsilon, -\epsilon) + b_-(\epsilon, -\epsilon))
= \lim_{\epsilon \to 0} (\tilde b_+(\epsilon) + \tilde b_-(\epsilon)) = 0.
\end{align}

Therefore, $C = 0$ for a vortex-antivortex pair.

\let\ldv\undefined

\end{proof}

Let $\mathbb{S}^2_\Delta$ be the diagonal in the product $\sphere
\times \sphere$. The orthogonal group acts diagonally on the moduli space 
$\moduli^{1,1}(\sphere) \cong (\sphere \times \sphere) 
\setminus \mathbb{S}^2_\Delta$ by isometries. We can always assume
there is a projective chart such that the pair is located with the
vortex at $z_1 = \epsilon$ and the antivortex at $z_2 =
-\epsilon$. From \eqref{eq:bs-gen-rel} and the fact that

\begin{equation}
b_{j}(\epsilon, -\epsilon) =
\conj{b_j}(\epsilon, -\epsilon),\label{eq:refl-invariance} 
\end{equation}
we conclude,
\begin{equation}
b_{1}(\epsilon,-\epsilon) +
b_{2}(\epsilon,-\epsilon)=0.\label{eq:b1-b2-complementarity}
\end{equation}


The $\Lsp^2$ metric in $\moduli^{1,1}(\sphere)$ is Kähler and invariant
under the diagonal action of $O(3)$, given any
pair $(z_1, z_2) \in \moduli^{1,1}(\sphere)$, we can always find a
rotation of $\sphere$ such that in south pole stereographic
projection, $z_1 = \epsilon$, $z_2 = -\epsilon$. In this way, we have
a diffeomorphism,
\begin{equation}
  (\sphere \times \sphere) \setminus \mathbb{S}^2_\Delta \cong (0, 1]
  \times SO(3),
  \label{eq:moduli11-so3-param}
\end{equation}
hence, the moduli space can be parametrized as $(0, 1]\times SO(3)$. 

\begin{lemma}
\label{lem:kahler-metric-can-form}
Let $g$ be a Kähler metric in $\sphere\times\sphere$ such that if $o
\in O(3)$ and $(z_1, z_2) \in \sphere\times\sphere$, then the action  
\begin{equation}
    o * (z_1, z_2) = (o * z_1, o * z_2),
\end{equation}
is by isometries. Let $E_0 = \del_\epsilon$ and let $E_j \in
\mathfrak{so}(3)$ be the left invariant vector field corresponding to 
rotations with respect to the $j$-th coordinate axis in $\reals^3$.
Then there exists a function
\begin{equation}
    A : (0, 1] \to \reals,
\end{equation}
and a real constant $c$ such that in the parametrization
\eqref{eq:moduli11-so3-param},  
\begin{multline}
\label{eq:so3-diag-inv-metric}
 g = A\left(\frac{1-\epsilon^{2}}{1+\epsilon^{2}}\,(\sigma^{1})^{2} +
   \frac{1+\epsilon^{2}}{1-\epsilon^{2}}(\sigma^{2})^{2}\right) -
 \frac{1}{\epsilon}
 \frac{dA}{d\epsilon}
 \left((\sigma^{0})^{2}+\epsilon^{2}(\sigma^{3})^{2}\right)\\
 + \frac{c}{1+\epsilon^{2}}\left(\sigma^{0}\sigma^{2}
   + \frac{\epsilon(1-\epsilon^{2})}{1+\epsilon^{2}} \sigma^{1}
   \sigma^{3}\right),
\end{multline}
where $\sigma^k\in T^*((0,1]\times SO(3))$ is the  co-vector dual to
$E_k$, $k = 0, \ldots, 3$. For this metric, the volume is, 
\begin{equation}
\label{eq:vol-formula}
\vol\left(\sphere\times\sphere\right) =
4\pi^{2}\lim_{\epsilon\to0}A(\epsilon)^{2}-c^{2}\pi^{2}.  
\end{equation}
\end{lemma}

\begin{proof}
\newcommand{\Vol}{\mathrm{Vol}}

This lemma is similar to \cite[Prop. 5.1]{romao2018}, but for $\tau \neq 0$, 
the 
swapping map $(z_1, z_2) \mapsto (z_2, z_1)$ is no longer a symmetry of the 
metric, instead, we consider the action of orientation reversing isometries 
of the sphere on the moduli space.


  


A general symmetric bilinear form in
$T\left((0,1)\times SO(3)\right)$ invariant under the diagonal $SO(3)$ action, 
will be a linear combination
\begin{align}
A_{rs}\sigma^{r}\sigma^{s},
\end{align}
with $A_{rs}=A_{sr}$. Let $q(\epsilon)=(\epsilon,-\epsilon)$,
$\epsilon\in(0,1]$. 
Denoting by $(X,Y,Z)$ coordinates in $\reals^{3}$,  the basis $E_{j}$ can be 
represented in the canonical embedding of $\sphere$ as the unit sphere in 
$\reals^3$ as 
\begin{align}
E_{0} & =
\frac{2(1-\epsilon^2)}{\left(1+\epsilon^{2}\right)^{2}}\left(\frac{\del}{\del
    X_{1}}-\frac{\del}{\del
    X_{2}}\right)-4\frac{\epsilon}{\left(1+\epsilon^{2}\right)^{2}}\left(-\frac{\del}{\del
    Z_{1}}-\frac{\del}{\del Z_{2}}\right),\\ 
E_{1} & =-\frac{1-\epsilon^{2}}{1+\epsilon^{2}}\left(\frac{\del}{\del
    Y_{1}}+\frac{\del}{\del Y_{2}}\right),\\ 
E_{2} & =-\frac{1-\epsilon^{2}}{1+\epsilon^{2}}\left(\frac{\del}{\del
    X_{1}}+\frac{\del}{\del
    X_{2}}\right)+\frac{2\epsilon}{1+\epsilon^{2}}\,\left(-\frac{\del}{\del
    Z_{1}}+\frac{\del}{\del Z_{2}}\right),\\ 
E_{3} & =\frac{2\epsilon}{1+\epsilon^{2}}\left(\frac{\del}{\del
    Y_{1}}-\frac{\del}{\del Y_{2}}\right). 
\end{align}

A short calculation yields,
\begin{align}
\Jop E_{0}=\frac{1}{\epsilon}E_{3},\qquad\Jop 
E_{1}=\frac{1-\epsilon^{2}}{1+\epsilon^{2}}E_{2},
\end{align}
where $\Jop$ is the pseudo-complex structure on 
$T((\sphere\times\sphere) \setminus \mathbb{S}^2_\Delta )$. 
If the metric is Kähler, we deduce,
\begin{align}
A_{03}=A_{12}=0,\qquad A_{33}=\epsilon^{2}A_{00},\qquad
A_{11}=\left(\frac{1-\epsilon^{2}}{1+\epsilon^{2}}\right)^{2}A_{22}. 
\end{align}

Let $C:\sphere\to\sphere$ be the reflection map $Y \mapsto -Y$ on the $XZ$ 
plane. $C$
acts on $\sigma^k$ as follows, 
\begin{align}
C^{*}\sigma^{0}=\sigma^{0},\qquad C^{*}\sigma^{1}=-\sigma^{1},\qquad
C^{*}\sigma^{2}=\sigma^{2},\qquad C^{*}\sigma^{3}=-\sigma^{3}. 
\end{align}

From reflection invariance we further obtain,
\begin{align}
A_{01}=A_{23}=0.
\end{align}

Let $A=A_{00}$, $B=(1+\epsilon^{2})^{-2}\,A_{22}$, then the metric is,
\begin{multline}
g = A \left( (\sigma^{0})^{2}+\epsilon^{2}(\sigma^{3})^{2}\right) +
B\left( (1-\epsilon^{2})^{2}\, (\sigma^{1})^{2} +
  (1+\epsilon^{2})^{2}\, (\sigma^{2})^{2} \right) \\
+ A_{02}\,\sigma^{0}\sigma^{2}+A_{13}\,\sigma^{1}\sigma^{3}.   
\end{multline}

If $\omega=g(\Jop\cdot,\cdot)$ is the Kähler form of the metric,
then 
\begin{equation}
\omega = \epsilon\,A\,\sigma^{0}\wedge\sigma^{3} +
(1-\epsilon^{4})\,B\,\sigma^{1}\wedge\sigma^{2} +
\frac{1}{\epsilon}A_{13}\,\sigma^{0}\wedge\sigma^{1} -
\frac{1+\epsilon^{2}}{1 -
  \epsilon^{2}}A_{13}\,\sigma^{2}\wedge\sigma^{3},\label{eq:kahler-form}
\end{equation}
provided $\epsilon A_{02}=(1+\epsilon^{2})(1-\epsilon^{2})^{-1}A_{13}$,
to account for skew-symmetry of $\omega$. The $SO(3)$ valued forms
$\sigma^{1},\sigma^{2},\sigma^{3}$, are related by
$d{\sigma}^{1}=-\sigma^{2}\wedge\sigma^{3}$ and cyclic permutations of
this identity. Kähler forms are closed. For $\omega$ this 
is true provided the coefficients in (\ref{eq:kahler-form}) are solutions
to the equations,
\begin{align}
\epsilon A &= -\frac{d}{d\epsilon}\left( (1-\epsilon^{4})
  \,B\right),\qquad
\frac{1}{\epsilon} A_{13} = \frac{d}{d\epsilon}\left(\epsilon
  A_{02}\right).\label{eq:kahler-form-coeffs} 
\end{align}

Regularity of the metric as $\epsilon\to1$ implies
$\lim_{\epsilon\to1}(1-\epsilon^{4})B(\epsilon)=0$. From the second
equation in (\ref{eq:kahler-form-coeffs}) and the algebraic relation
of the coefficients $A_{13}$, $A_{02}$, we infer 
\begin{align}
A_{02}=\frac{c}{1+\epsilon^{2}},
\end{align}
for some real constant $c$. Redefining the function $(1-\epsilon^{4})B$
as $A(\epsilon)$, the metric has the form
\eqref{eq:so3-diag-inv-metric}. Since 
$\int_{SO(3)}\sigma^{1}\wedge\sigma^{2}\wedge\sigma^{3}=8\pi^{2}$
\cite{romao2018} and 
\begin{align}
\int_{0}^{1}\frac{\epsilon\,(1-\epsilon^{2})}
{(1+\epsilon^{2})^{3}}d\epsilon=\frac{1}{8},  
\end{align}
for this metric, the volume form is
\begin{align}
\vol = -\left( A\,A' + \frac{c^{2}\epsilon( 1 - \epsilon^{2} )}{ (1
    + \epsilon^{2})^{3}}\right) d\epsilon \wedge \sigma^{1} \wedge
\sigma^{2} \wedge\sigma^{3}.
\end{align}

After integration, the total volume of the metric is 
\begin{equation}
\Vol\left(\sphere\times\sphere\right) =
4\pi^{2}\lim_{\epsilon\to0}A(\epsilon)^{2}-c^{2}\pi^{2}.  
\end{equation}

\let\Vol\undefined


\end{proof}

Applying lemma~\ref{lem:kahler-metric-can-form} to the $\Lsp^2$ metric, 
we obtain, 
 \begin{lemma}
 \label{lem:samols-metric-in-eps-so3}
 The $\Lsp^2$ metric on $\moduli^{1,1}(\sphere)$ has the structure 
 provided by Lemma~\ref{lem:kahler-metric-can-form}, with 
 \begin{gather}
     A = 2\pi \left(
     \frac{4R^2}{1 + \epsilon^2} - \epsilon\,b_1 - 2R^2 - 1
     \right),
	\label{eq:A-formula}\\
     c = 8\pi R^2\tau.
     \label{eq:c-formula}
 \end{gather}
 \end{lemma}
 
\begin{proof}
To compute the constant $c$, we calculate
$g\left(E_{0},E_{2}\right)$. Tangent vectors $E_{0}$, $E_{2}$ in
projective coordinates $(z_1, z_2) \in \sphere \times \sphere$ with
respect to the south pole are,  
\begin{equation}
  E_{0} = \frac{\del}{\del x_1} - \frac{\del}{\del x_{2}}, \qquad
  E_{2} = \frac{1+\epsilon^2}{2}\,\left( \frac{\del}{\del x_{1}} +
    \frac{\del}{\del x_2} \right).\label{eq:E0-E2-proj}
\end{equation}
where $z_k = x_k + i y_k$. Thence, 
\begin{align}
g\left(E_{0},E_{2}\right)
 & =\frac{1+\epsilon^{2}}{2}g\left(\frac{\del}{\del
x_{1}}-\frac{\del}{\del x_{2}},\frac{\del}{\del
x_{1}}+\frac{\del}{\del x_{2}}\right)\nonumber \\
 & =\frac{1+\epsilon^{2}}{2}2\pi\left(\cf(1+\tau)-\cf(1-\tau)+\frac{\del
 b_{1}}{\del z_{1}}+\frac{\del b_{2}}{\del z_{1}}-\frac{\del
 b_{1}}{\del z_{2}}-\frac{\del b_{2}}{\del
 z_{2}}\right)\label{eq:gl-e0-e2-expand} 
\end{align}

To simplify (\ref{eq:gl-e0-e2-expand}), we use the symmetries of the
coefficients $b_j$, lemma~\ref{lem:coeff-sym},
\begin{align}
	\sum_j\left(\pdv{z_1}{b_j} - \pdv{z_2}{b_j}\right) &=
	\frac{1}{2} \sum_j{\dv\epsilon}b_j(\epsilon,-\epsilon) -
\frac{i}{2}\sum_j 
	\left(\pdv{y_1}{b_j} - \pdv{y_2}{b_j}\right)\nonumber\\
	&=\frac{1}{2}{\dv \epsilon} (b_1 + b_2) + \frac{1}{2\epsilon}
(b_1 + b_2)\nonumber\\ 
	&= 0.
\end{align}

Hence,
\begin{equation}
	\label{eq:gl2-e0-e2}
	\moduliMetric(E_0, E_2) = \frac{8\pi R^2\tau}{1 + \epsilon^2}
\end{equation}
and consequently $c = 8\pi R^2\tau$. Let us compute
$\moduliMetric(E_0, E_0)$, 
\begin{align}
	\nonumber
	\moduliMetric(E_0, E_0) &= \moduliMetric \left( \pdv{x_1} -
	\pdv{x_2}, \pdv{x_1} -\pdv{x_2} \right)\\
	&= 2 \pi \left( \cf (1 + \tau) + \cf (1 - \tau)
	+  \pdv{z_1}{b_1} - \pdv{z_1}{b_2}
	- \pdv{z_2}{b_1} + \pdv{z_2}{b_2} \right).
\end{align}

Again by symmetry,
\begin{equation}
	\label{eq:db1-db2-prop-at-qeps}
	\pdv{z_1}{b_j} - \pdv{z_2}{b_j} = \frac{1}{2}
        \frac{d b_j}{d\epsilon} 
	+\frac{1}{2\epsilon}b_j.
\end{equation}

Hence,
\begin{equation}
  \label{eq:gl2-e0-e0}
  \moduliMetric(E_0, E_0) = 2\pi \left( \frac{8R^2}{(1 +
      \epsilon^2)^2} + 
    \frac{d b_1}{d\epsilon} 
    + \frac{1}{\epsilon}b_1 \right).
\end{equation}

Comparing \eqref{eq:gl2-e0-e0} and \eqref{eq:so3-diag-inv-metric},
\begin{equation}
  \label{eq:A-diff-eq}
  - \frac{1}{\epsilon} \frac{d A}{d\epsilon} = 2\pi \left(
    \frac{8R^2}{(1 + \epsilon^2)^2} + \frac{db_1}{d\epsilon}
    + \frac{1}{\epsilon}b_1 \right),
\end{equation}

Solving this equation, we find,
\begin{equation}
	\label{eq:A-diff-eq-sol}
	A = \frac{8\pi R^2}{1 + \epsilon^2} - 2\pi\epsilon b_1 +
        \mathrm{const}. 
\end{equation}

From the regularity condition $\lim_{\epsilon\to 1}A(\epsilon) = 0$ used
to compute the formula for the volume of the moduli space  and the
explicit formula \eqref{eq:bpm-antipodal-position} for $b_1$ in the
antipodal case, the constant is 
\begin{equation}
	\label{eq:A-sol-const}
	\mathrm{const}. = -2\pi (2R^2 + 1).
\end{equation}

Therefore,
\begin{equation}
	A = 2\pi \left( \frac{4R^2}{1 + \epsilon^2} - \epsilon b_1  - 2R^2 -
	1\right). 
\end{equation}
\end{proof}

We claim that
\begin{equation}
  \label{eq:lim-eps-b1}
  \lim_{\epsilon\to 0} \epsilon b_1 = -1 
\end{equation}
as can be seen numerically in figure~\ref{fig:a-priori} for the
symmetric case in the unit sphere.  

For a vortex-antivortex pair,
\begin{equation}
  b_1(\epsilon, -\epsilon) = 2 \eval{\frac{\del}{\del
      x}}{z=\epsilon} h_{\epsilon} - \frac{1}{\epsilon}. 
  \label{eq:vav-b1}
\end{equation}

Since $h_{\epsilon} \to \mu$ in $C^1$ as $\epsilon\to 0$,
\begin{equation}
  \lim_{\epsilon\to 0} \epsilon\,b_1(\epsilon, -\epsilon) = -1.
  \label{eq:the-limit-conclusion}
\end{equation}

Applying lemmas~\ref{lem:kahler-metric-can-form}
and~\ref{lem:samols-metric-in-eps-so3}, the volume of the moduli space is
\begin{equation}\label{eq:volume-conjecture-sphere-formula}
	\vol \left( \moduli^{1,1}(\sphere) \right) =  \left( 8\pi^2
          R^2\right)^2 (1 - \tau^2).
\end{equation}
Notice that another way to express the volume is as $4\pi^2(1-\tau^2)\vol 
(\sphere)$, which corresponds to the volume of a product of spheres, 
each factor weighted by $2\pi (1\pm\tau)$, the effective mass of a core, 
hence, it is expected that as $\tau \to \pm 1$, the volume vanishes, 
because of the negligible weight of one of the factors.




\begin{figure}[p]
  \centering
    \includegraphics[width=.55\textwidth]{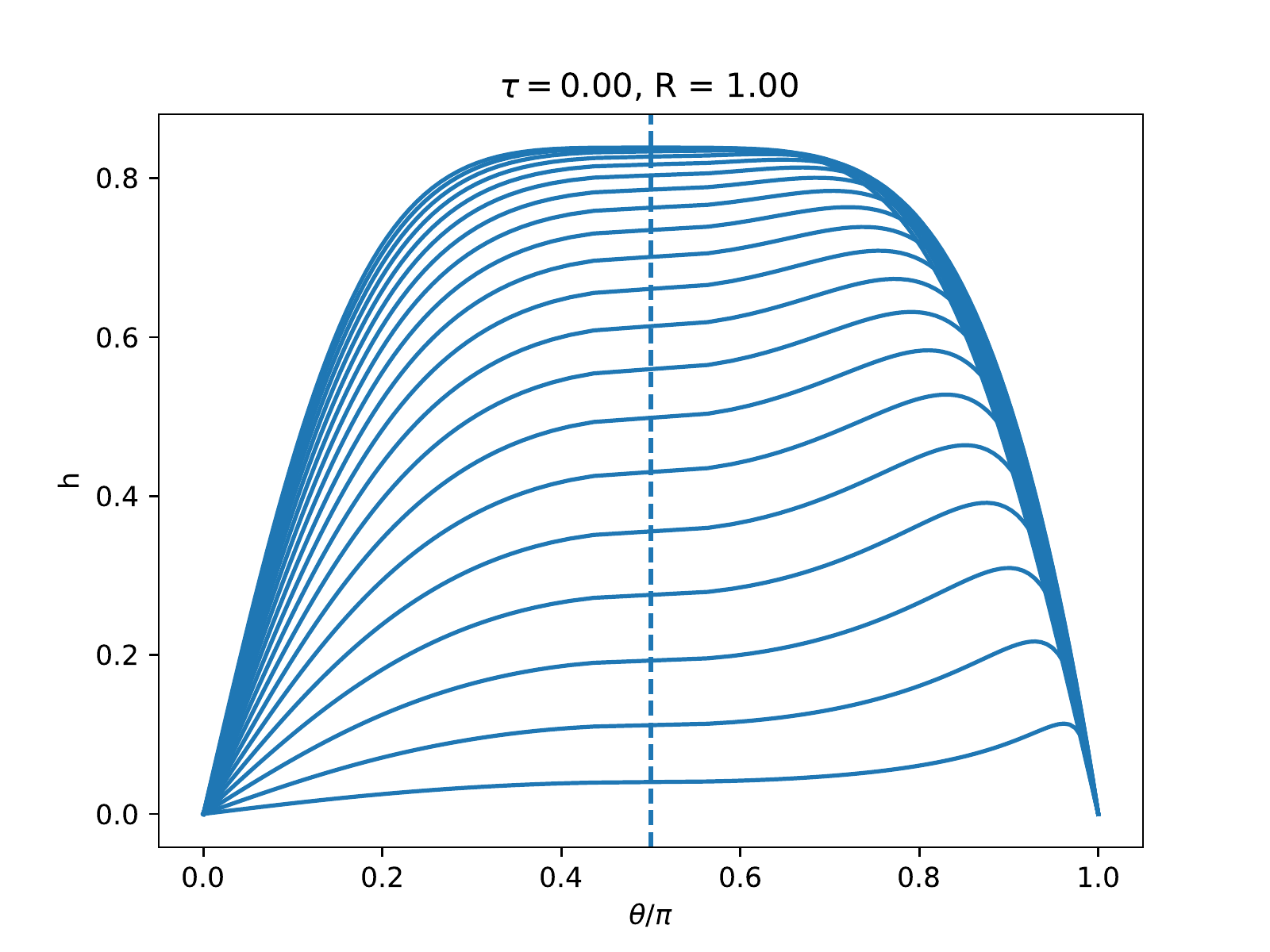}
    \includegraphics[width=.55\textwidth]{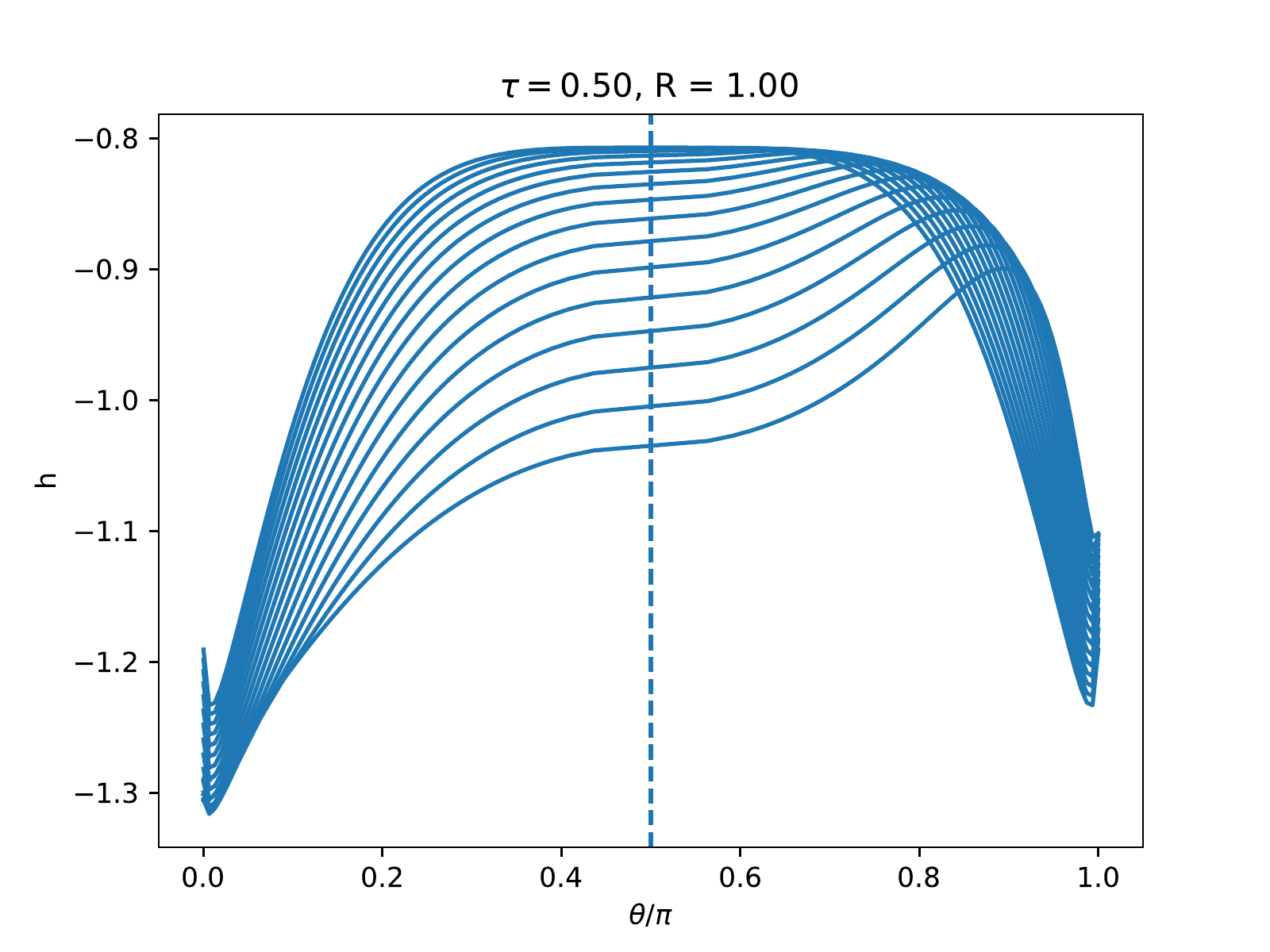}
    \includegraphics[width=.55\textwidth]{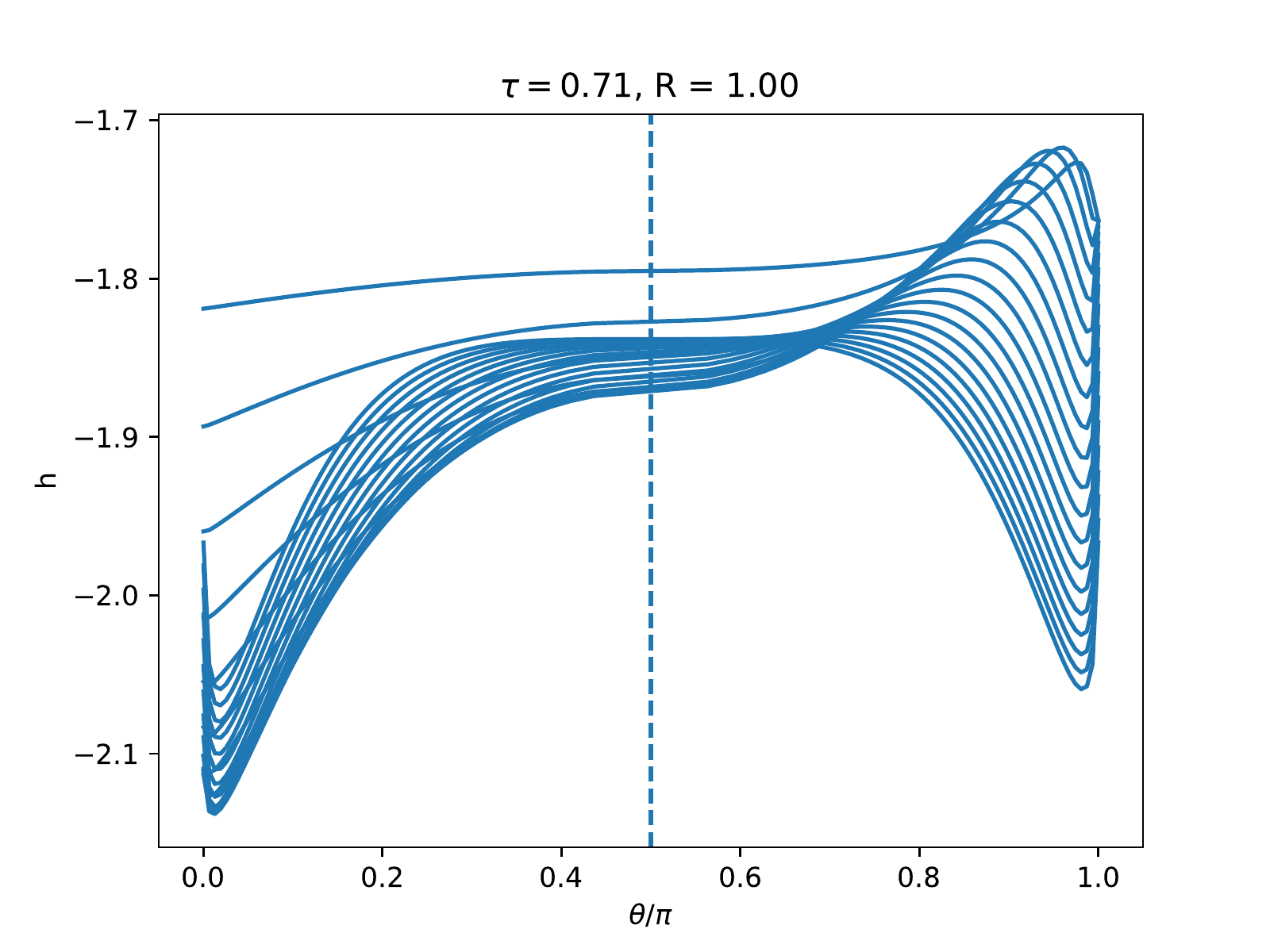}
    \caption{
      Three views of the declination data of $\tilde h_{\epsilon}$, 
      the regular part of the solution to the Taubes equation, for
      three different values of the asymmetry parameter $\tau$ on the 
      unit sphere. 
      \textbf{Top.} Vortex and
      antivortex are symmetric, with the same effective
      mass. \textbf{Middle and bottom}. The antivortex becomes more
      massive. We solved 
      from $\epsilon = 1$ down to $0.05$ in steps of 
    $0.05$, except that for $\tau = 0.5$, the computation
    stopped at $\epsilon = .20$ due to algorithm divergence. 
    As $\epsilon \to 0$ the data shows how $\tilde h_\epsilon$ flattens 
    as expected.
      }\label{fig:dec-data}
\end{figure}

\begin{figure}[p]
  \centering
  \includegraphics[width=.78\textwidth]{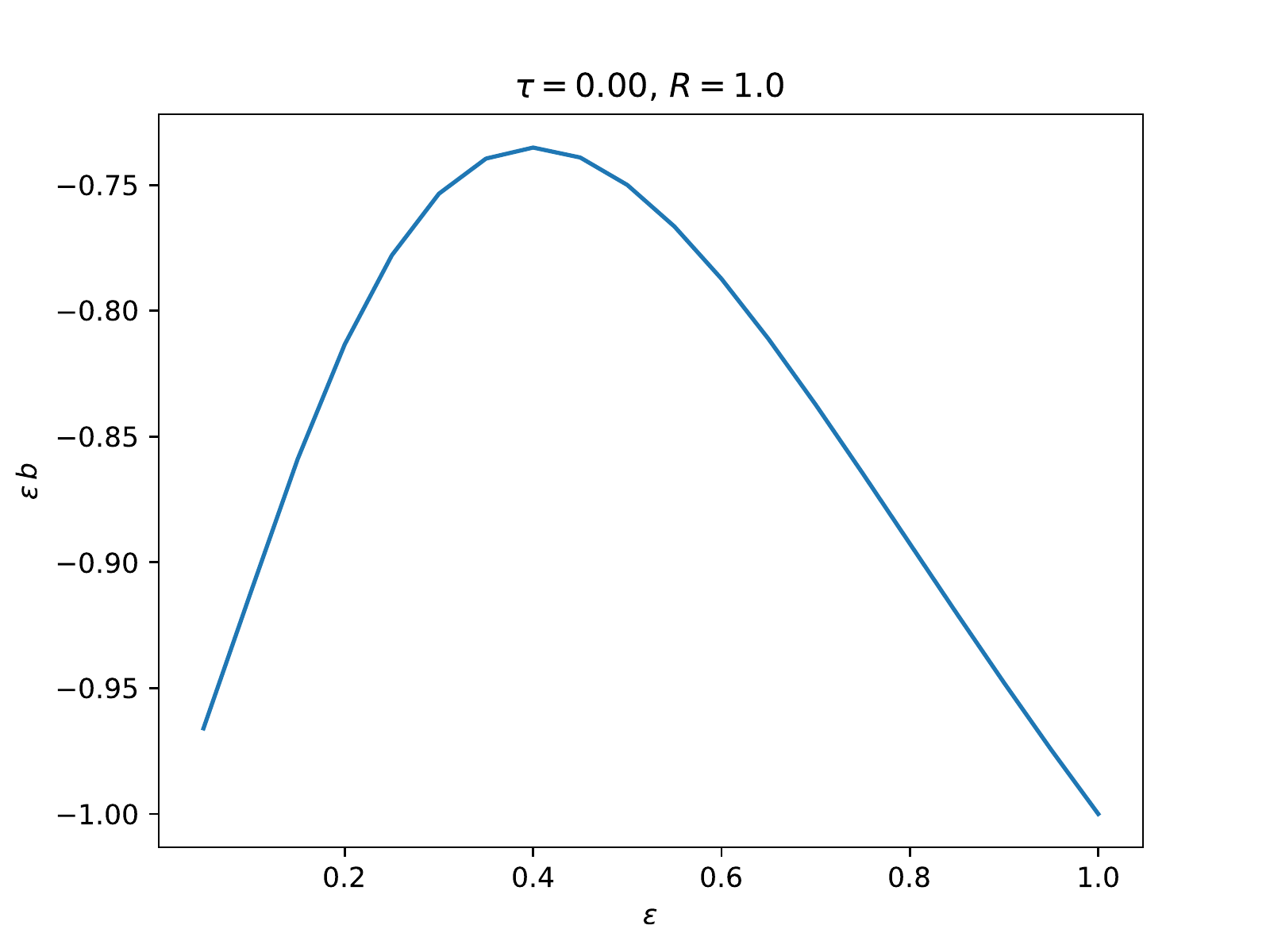}
  \includegraphics[width=.78\textwidth]{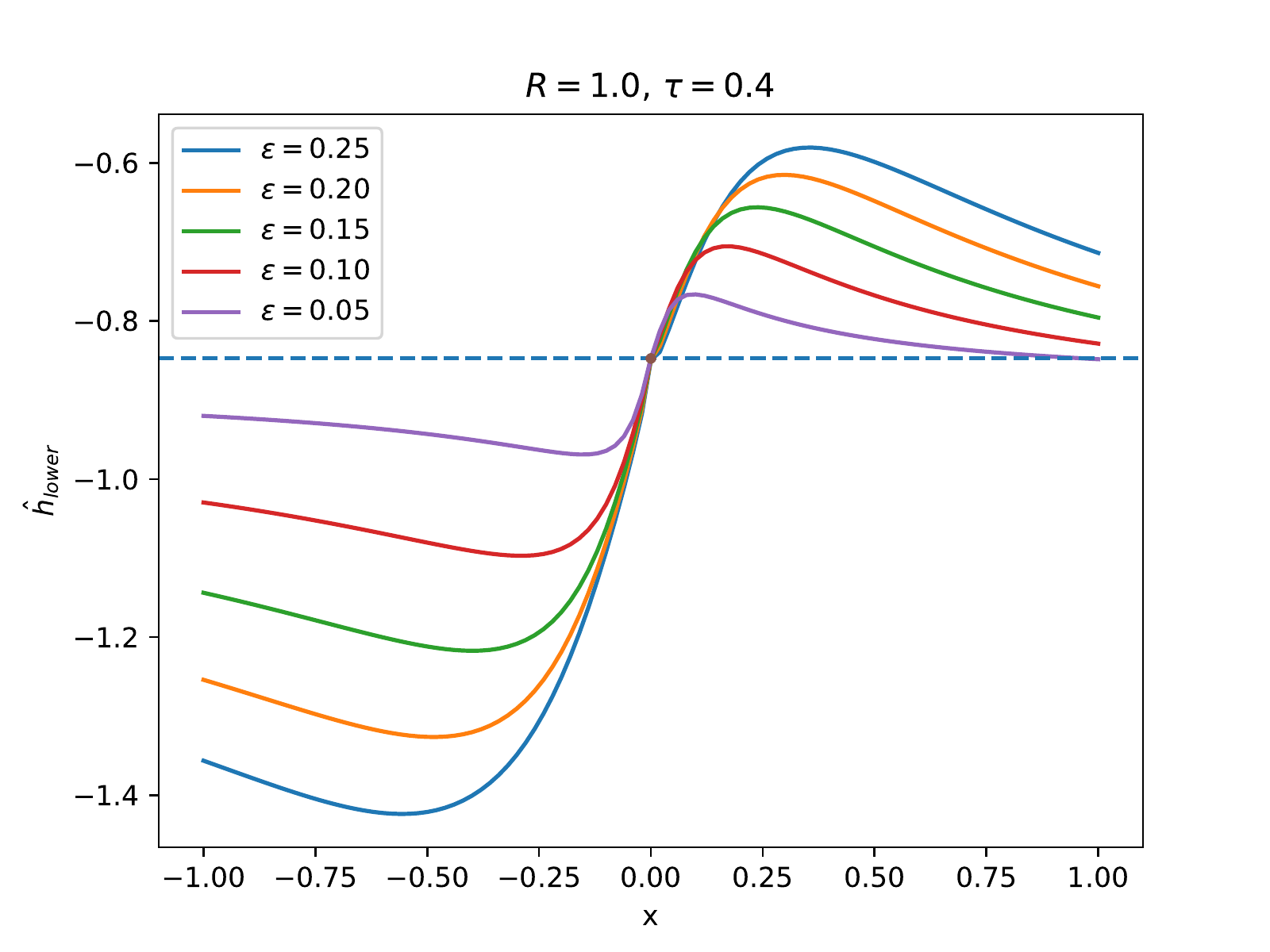}
  \caption{
    \textbf{Top}. Real profile of $\epsilon b$ in the symmetric
    case. The limit $\lim_{\epsilon\to 0} \epsilon b = -1$ is apparent
    in the numerical data. \textbf{Bottom}. Real profile of a
    vortex-antivortex pair located at $\pm \epsilon$ on the real axis
    of the extended complex plane for several values of
    $\epsilon$. In both cases, the domain is the unit sphere, the bottom 
    plot shows the behaviour of the real profile of
    $\tilde h$ as $\epsilon \to 0$ in the south pole of the domain. The dashed 
    horizontal line is $\log \left( (1 - \tau)(1 + 
      \tau)^{-1} \right)$. The data shows how the
    regular part of the solution to the Taubes equation converges to this
    constant value as the pair collides at the north pole. 
    }\label{fig:a-priori} 
\end{figure}

\subsection{Flat tori}

In this section we compute the volume of the moduli space for a flat 
tori, to 
this end, we extend the coefficients $b_q$ in 
the $\Lsp^2$ metric to a global object and relate it to the volume of 
$\moduli^{1,1}(\mathbb{T}^2)$ in lemma~\ref{lem:volume-integral}. 
Consider a holomorphic chart $\varphi: U \subset \mathbb{T}^2 \to \cpx$ 
on an open and dense set $U$, with coordinates $z = \varphi(x)$, $x \in U$. Let 
us define,
\begin{align}
b_U = b_j\,d\bar z^j \in \Omega^{(0,1)}((U \times U) \setminus \diag_U).
\end{align}

In general $b_U$ is only well defined on a chart, however, flat tori admit 
atlases such that the holomorphic changes of coordinates are 
translations. Since
translations have trivial second derivatives, by 
\eqref{eq:b-coeffs-change-of-coords}  $b_U$ extends to a global form $b \in
\Omega^{(0,1)}(\moduli^{1,1}(\mathbb{T}^2))$. By the symmetries of the 
coefficients $b_j$, this form is holomorphic, as the following short 
calculation shows in coordinates:
\begin{align}
\bar\del b_U &= \sum_{i,j} \bar\del_{z_i}b_j\,d\bar z^i\wedge d\bar 
z^j\nonumber\\
&= -\sum_{i,j} \bar\del_{z_j}b_i\,d\bar z^j\wedge d\bar z^i\nonumber\\
&= -\bar\del b_U,
\end{align}
hence, $\bar \del b_U = 0$.

To compute the volume of flat tori, we will use the $(1,1)$-form $\del b$ to
define another form in the moduli space which is more convenient
for calculations. Let $\proj_j:\mathbb{T}^2\times \mathbb{T}^2\to \mathbb{T}^2$ 
be the canonical projection map onto the $j$-th factor of the product. Let
us define the form
\begin{align}
  \kform_0 = 2\pi\,(1-\tau)\,\proj_1^{*}\,\kform_{\mathbb{T}^2} + 2\pi\,(1 +
  \tau)\, \proj_2^{*}\,\kform_{\mathbb{T}^2}.
\end{align}

The K\"ahler form on the moduli space can be written as,
\begin{align}
  \kform &= \kform_0 + \pi i\,\del b \in
  \Lambda^{1,1}(\moduli^{1,1}(\mathbb{T}^2)). 
\end{align}

Notice that,
\begin{align}
  \vform &= \half\,\kform\wedge\kform\nonumber\\
         &= \vform_0 + \pi i\,\kform_0\wedge \del b -
           \frac{\pi^2}{2}  \del b \wedge \del b,
\end{align}
where $\vform_0 = \half \kform_0\wedge\kform_0$ is the restriction of
the volume form in the product $\mathbb{T}^2\times\mathbb{T}^2$ to the moduli 
space.

    \begin{lemma}\label{lem:volume-integral}
      Let $\diag_\epsilon$ be the $\epsilon$-tubular neighbourhood of
      the diagonal set of $\mathbb{T}^2\times\mathbb{T}^2$ for small
      $\epsilon$. The volume of the moduli space can be computed as,
      \begin{multline}
        \Volume(\moduli^{1,1}(\mathbb{T}^2)) = 
        4\pi^2 (1 - \tau^2)\,\Volume(\mathbb{T}^2)^2 \\
        + \lim_{\epsilon\to 0}\int_{\mathbb{T}^2 \times \mathbb{T}^2
        \setminus \diag_\epsilon}  
        \pbrk{\pi i\, \kform_0\wedge \del b
        - \frac{\pi^2}{2}  
        \del b \wedge \del b}.\label{eq:vol-int-lim}
      \end{multline}
    \end{lemma}

    \begin{proof}
      \begin{align}
        \Volume(\moduli^{1,1}(\mathbb{T}^2))
        &= \lim_{\epsilon\to 0}  
          \int_{\mathbb{T}^2 \times \mathbb{T}^2 \setminus \diag_\epsilon}
          \vform\nonumber\\
        &= \int_{\mathbb{T}^2\times\mathbb{T}^2}\vform_0
          + \lim_{\epsilon\to 0}\int_{\mathbb{T}^2 \times
          \mathbb{T}^2 \setminus \diag_\epsilon}
          \pbrk{\pi i\,\kform_0\wedge \del b
          - \frac{\pi^2}{2} \del b \wedge \del b}.
      \end{align}

      On the other hand,
      \begin{align}
        \vform_0 = 4\pi^2(1 - \tau^2)\,\proj_1^{*}\,\kform_{\mathbb{T}^2} \wedge
        \proj_2^{*}\,\kform_{\mathbb{T}^2}.
      \end{align}

      Applying Fubini and the change of variables theorems,
      \begin{align}
        \int_{\mathbb{T}^2\times\mathbb{T}^2}\vform_0
        &= 4\pi^2(1-\tau^2)\pbrk{\int_{\mathbb{T}^2}\kform_{\mathbb{T}^2}}^2
          = 4\pi^2(1-\tau^2)\Volume(\mathbb{T}^2)^2.
      \end{align}

      This concludes the proof of the lemma.
    \end{proof}

 According to lemma~\ref{lem:volume-integral}, to compute the volume
 of $\moduli^{1,1}(\mathbb{T})$, we must compute the two non-trivial
 terms in~\eqref{eq:vol-int-lim}.

\begin{lemma}\label{lem:vav-comp-b1-b2-torus}
    Let $\pi: \cpx \to \mathbb{T}^2$ be the canonical covering map and let 
    $R \subset \cpx$ be an open parallelogram such that $\pi|_R: R \to 
    \mathbb{T}^2$ is a bi-holomorphism onto its image and $U = \pi|_R(R)$ is 
    open and dense. On the local coordinates $\pi|_R^{-1}:U \to R$, there is a 
    constant 
    $c \in \cpx$ such that for any pair of 
    different points $z_1, z_2 \in R$,
  \begin{align}
    b_1(z_1, z_2) + b_2(z_1, z_2) &= c,
  \end{align}
\end{lemma}

\begin{proof}
  If $\mathcal{I} : \mathbb{T}^2 \to \mathbb{T}^2$ is an isometry, the Taubes 
  equation is invariant under $\mathcal{I}$, 
  \begin{align}
    h(\mathcal{I}(x); \mathcal{I}(x_1), \mathcal{I}(x_2)) = h(x; x_1, x_2),
  \end{align}
  $x, x_1, x_2 \in \mathbb{T}^2$, $x_1 \neq x_2$. 
  By construction, there is a $v \in \cpx$ such that 
  $\mathcal{I}_\varphi = \varphi\circ \mathcal{I} \circ \varphi^{-1}(z) = z + 
  v$ for $z \in \varphi(\mathcal{I}^{-1}(U)\cap U)$. For small $v$, the 
  translation $\mathcal{I}_\varphi$ maps a  neighbourhood, 
  not necessarily connected,
  $N \subset R$ of $x_1$ and $x_2$ into $R$. 
  This implies $b_j$ has the symmetries, 
  \begin{align}
    b_j(z_1 + v, z_2 + v) = b_j(z_1, z_2),
  \end{align}
  $v$ small. Hence,
  \begin{align}
    \del_{z_1}b_j + \del_{z_2}b_j = \bar\del_{z_1}b_j + \bar\del_{z_2}b_j = 0.
  \end{align}

  Applying the symmetries of the coefficients $b_j$,
  \begin{align}
\del_{z_j}(b_1 + b_2) = \bar\del_{z_1}\bar b_j + \bar\del_{z_2}\bar b_j = 0.
\end{align}

Similarly,
  \begin{align}
    \bar\del_{z_j}(b_1 + b_2) = 0.
  \end{align}
  
  Hence $b_1 + b_2$ is constant on the connected neighbourhood $R$.
\end{proof}

\begin{proposition}
  In a flat torus $\mathbb{T}^2$, for the $(1,1)$ form $\del b$ we
  have,
  \begin{align}
    \del b \wedge \del b = 0.
  \end{align}
\end{proposition}

\begin{proof}
  We apply the previous lemma to prove the proposition. 
  By lemma~\ref{lem:vav-comp-b1-b2-torus}, there is an open and dense 
  set $U \subset \mathbb{T}^2$ and a chart $\varphi: U \to R\subset \cpx$, 
  $R$ an open parallelogram, such that in this local coordinates $b_1 + b_2$ is 
  a constant. Denoting points in $R$ as $z_j$, a direct calculation shows,
  \begin{align}
    b_U \wedge \del b_U
    &=
      (b_2\,\del_{z_1}b_1 - b_1\,\del_{z_1}b_2)\,
    dz_1\wedge d\bar z_1 \wedge d\bar z_2\nonumber\\
    &\quad + (-b_2\,\del_{z_2}b_1 + b_1\,\del_{z_2} b_2)\,
      d\bar z_1 \wedge dz_2 \wedge d\bar z_2\nonumber\\
    &=
      -c\,\del_{z_1}b_2\, dz_1\wedge d\bar z_1 \wedge d\bar z_2
    -c\,\del_{z_2}b_1\,d\bar z_1 \wedge dz_2 \wedge d\bar z_2.
  \end{align}

  Since $b_1$ and $b_2$ add to a constant,
  \begin{align}
    \del b_U \wedge \del b_U
    =
    -c\,(\del_{z_2}\del_{z_1} b_2 + \del_{z_1}\del_{z_2}b_1)\,
    dz_1 \wedge d\bar z_1 \wedge dz_2 \wedge d\bar z_2 = 0.
  \end{align}

  Since $U$ is dense, we conclude $\del b \wedge \del b \equiv 0$.
\end{proof}

By this proposition and lemma~\ref{lem:volume-integral}, to compute the volume 
of the moduli space, we have to integrate 
$\kform_0\wedge \del b$. 

\begin{theorem}\label{thm:vol-torus}
  For a flat torus $\mathbb{T}^2$, the volume of the moduli space is,
  \begin{align}
    \Volume(\moduli^{1,1}(\mathbb{T}^2))
    =
    4\pi^2(1 - \tau^2)\,\Volume(\mathbb{T}^2)^2 +
    16\pi^3\,\Volume(\mathbb{T}^2). 
  \end{align}
\end{theorem}

Notice that the first term of the formula is 
similar to the case of the 
sphere~\eqref{eq:volume-conjecture-sphere-formula}, however, the second 
term is new, bearing in mind 
the volume conjecture,~\ref{conj:volume-conjecture}, one can argue the 
extra term is related to the genus of the base surface, however, it  
is not clear how to relate our computation to this fact and 
the relation is open to 
future work.

\begin{proof}
    Let,
    \begin{align}
    \mathbb{T}^2({\epsilon})
    &= (\mathbb{T}^2 \times \mathbb{T}^2) \setminus \diag_{\epsilon},\\
    \kform_j &= \proj^{*}_j\kform_{\mathbb{T}^2}, \qquad j = 1, 2,
    \end{align}
  and let $k$ be the complementary index of $j$, such that
  $\set{j, k} = \set{1, 2}$. By Fubini's theorem,
  \begin{align}
  \int_{{\mathbb{T}^2}(\epsilon)}\kform_0\wedge \del b
  &=
    2\pi \sum_j (1 - \sign_j\tau)\int_{\mathbb{T}^2}\pbrk{
    \int_{\mathbb{T}^2 \setminus
    \disk_{\epsilon}(x_j)} \inc_k^{*}\del b}\, \kform_{\mathbb{T}^2},
\end{align}
where for any given $x_j \in \mathbb{T}^2$, $\inc_k: \mathbb{T}^2
 \hookrightarrow
\mathbb{T}^2\times\mathbb{T}^2$ is the 
inclusion of the torus as the k-th factor of the product anchored at $x_j$. 
Since
$b$ is well defined globally,
\begin{align}
  \int_{\mathbb{T}^2\setminus\disk_{\epsilon}(x_j)} \inc_k^{*}\del b
  =
  \int_{\del \mathbb{T}^2\setminus\disk_{\epsilon}(x_j)} \inc_k^{*}b
  =
  -\int_{\del \disk_{\epsilon}(x_j)} \inc_k^{*}b,
\end{align}
where we always orient a submanifold by the outward pointing
normal. Let $\varphi: U \to \cpx$ be a holomorphic chart defined on an open 
and dense set $U$. If $x_j \in U$, for small $\epsilon$, $\disk_\epsilon(x_j) 
\subset U$. Assume $j = 1$, $k = 2$, in the chart,
\begin{align}
(\varphi^{-1})^*\inc_2^{*}b = b_2 d\bar z.
\end{align}

 If $z_1 = \varphi(x_1)$ and $D(z_1) \subset \cpx$ is a 
bounded domain and neighbourhood of $z_1$, by lemma~\ref{prop:b-sing}, 
\begin{align}
b_2 = \frac{2}{\conj z_1 - \conj z_2} + \tilde b_2(z_1, z_2),\qquad 
z_2 \in D(z_1).
\end{align}

If $D_\epsilon(z_j) = \varphi(\disk_\epsilon(x_j))$, by Cauchy's residue 
theorem,
\begin{align}
\int_{\del \disk_{\epsilon}(x_1)} \inc_2^{*}b = 
-2\int_{\del D_\epsilon(z_1)} \frac{d\conj z}{\conj z - \conj z_1} 
+ \int_{\del D_\epsilon(z_1)}\tilde b_2(z_1,z)d\conj z
= 4\pi i + \int_{\del D_\epsilon(z_1)}\tilde b_2(z_1,z)d\conj z.
\end{align}

If $j = 2$, $k = 1$, we find a similar result,
\begin{align}
\int_{\del \disk_{\epsilon}(x_2)} \inc_1^{*}b = 
4\pi i + \int_{\del D_\epsilon(z_2)} \tilde b_1(z, z_2) d\conj z.
\end{align}

Since $\tilde b_k$ is a continuous function in a neighbourhood of each 
$z_j \in \cpx$, 
\begin{align}
\lim_{\epsilon \to 0} \int_{\del D_\epsilon(z_j)} \tilde b_k\,d\conj z = 0.
\end{align}

Hence, since $U$ is dense in $\mathbb{T}^2$, 
\begin{align}
lim_{\epsilon\to 0}
\int_{\mathbb{T}^2}\pbrk{
    \int_{\mathbb{T}^2 \setminus
        \disk_{\epsilon}(x_j)} \inc_k^{*}\del b} \kform_{\mathbb{T}^2}
    &= -4\pi i \Volume(\mathbb{T}^2)\nonumber\\
    &\quad - \frac{i}{2} \int_{\cpx}
    lim_{\epsilon\to 0} \pbrk{\int_{\del D_{\epsilon}(z_j)} \tilde b_k}
    e^{\Lambda(z_j)} dz_j\wedge d\conj z_j\nonumber\\
    &= -4\pi i \Volume(\mathbb{T}^2).
\end{align}

Finally,
\begin{align}
  \int_{\moduli^{1,1}(\mathbb{T}^2)}\kform_0\wedge\del b
  &= 2\pi\sum_j(1 - \sign_j\tau)
lim_{\epsilon\to 0}
\int_{\mathbb{T}^2}\pbrk{
  \int_{\mathbb{T}^2 \setminus
    \disk_{\epsilon}(x_j)} \inc_k^{*}\del b} \kform_{\mathbb{T}^2}.\nonumber\\
&= 2\pi\sum_j(1 - \sign_j\tau)
\left(-4\pi i \Volume(\mathbb{T}^2) \right)\nonumber\\
&= -16\pi^2\,i\,\Volume(\mathbb{T}^2).
\end{align}

By lemma~\ref{lem:volume-integral}, we conclude the volume formula.
\end{proof}

\let \sign    \undefined
\let \Wsp     \undefined
\let \wto     \undefined
\let \diag    \undefined
\let \vform   \undefined
\let \Volume  \undefined
\let \Wsp     \undefined
\let \Bop     \undefined
\let \fnF     \undefined
\let \diff    \undefined
\let \fnV     \undefined
\let \fnv     \undefined
\let \gp      \undefined
\let \Jop     \undefined
\let \fbundle \undefined
\let \pbundle \undefined
\let \Diff    \undefined
\let \hf      \undefined
\let \potentialE \undefined
\let \Fstable    \undefined
\let \vq         \undefined
\let \vp         \undefined
\let \cf         \undefined
\let \vol        \undefined
\let \Hsp        \undefined
\let \htilde     \undefined
\let \ctilde     \undefined
\let \pconnection \undefined
\let \Energy      \undefined
\let \Xsp         \undefined
\let \suchthat    \undefined
\let \dist        \undefined
\let \moduliMetric \undefined
\let \domain       \undefined
\let \isometry     \undefined
\let \pSpace       \undefined
\let \proj         \undefined
\let \kform        \undefined
\let \inc          \undefined







\label{ch:vav-compact}

\chapter{Chern-Simons deformations of vortices}
\label{ch:cs-moduli}



\newcommand*\nf{N}
\newcommand*\np{n}
\newcommand*\hf{\phi}%
\newcommand*\gp{a}%
\newcommand*\Diff{\mathcal{D}}%


\newcommand*\modulikk{{\moduli^{k_+,k_-}}}


\newcommand*{\Hsp}{\mathrm{H}}
\newcommand*\Wsp{\mathrm{W}}
\newcommand*\Top{\mathrm{T}}
\newcommand*\Lop{\mathrm{L}}



\newcommand*\vol{\mathrm{Vol}}



In this chapter we consider Chern-Simons deformations of vortices of the $O(3)$ 
Sigma model and of the Abelian Higgs model. We will consider deformations 
relying on 
a deformation constant $\kappa$. There are several results in the literature 
 about existence of solutions to the field equations, for both types of models. 

In section~\ref{sec:cs-intro} we address existence and uniqueness of solutions 
to the field equations for  
 deformations of the $O(3)$ Sigma model. 
  On the 
 plane,  
Han and Nam prove in \cite{han2005topological} that the field equations 
admit a solution up to some upper and lower bound for $\kappa$. If
there are only vortices or antivortices, 
Han and Song prove in~\cite{han2011existence} 
existence of solutions for any $\kappa$. On a flat 
torus, Chae-Nam~\cite{chae2001condensate}  and Chiacchio-Ricciardi 
prove~\cite{chiacchio2007multiple}
the existence of a bound on the constant for the existence of
solutions as well as the existence of multiple solutions if the number of 
vortices and antivortices on the surface is different. We extend the technique 
used by Flood and Speight 
in~\cite{flood2018chern} for 
Chern-Simons deformations of  the Abelian Higgs model to show 
the existence of a minimal
deformation constant, independent of the position of the vortices, if the 
surface is compact. We know from chapter~\ref{c:vav-compact} that for 
$\kappa=0$, the moduli space is incomplete, 
imposing some technical difficulties in the techniques 
used for deformations of the Abelian Higgs model. 
In subsection~\ref{sec:small-deform}, we show that on a compact surface, 
small deformations of the solution to the Taubes equation vary 
smoothly with $\kappa$. In subsection~\ref{sec:proof-cnt-bound}, we 
show the existence of a positive lower bound for $|\kappa|$, independent of 
the 
position of the cores. In subsection~\ref{sec:unbalanced-case} we 
focus on the unbalanced case where the number of vortices and antivortices 
differ and show that the possible constants $\kappa$ for which 
the field equations admit a solution are bounded. By means of several bounds in 
the norm of solutions to the 
governing elliptic problem, we show the existence of multiple 
solutions in subsection~\ref{sec:exist-mult-sols}. We finalize in 
subsection~\ref{sec:symm-deform-sphere} with numerical evidence on the sphere 
supporting a conjecture about the existence of solutions to the field equations 
for any $\kappa$ if the number of vortices and antivortices coincide.

In section~\ref{ch:cs-loc} we study low energy dynamics of both 
the Abelian Higgs 
and the $O(3)$ Sigma model vortices with a Chern-Simons deformation. 
The results 
discussed in section~\ref{sec:cs-intro} guarantee this is a well posed problem 
for small $\kappa$. Previous work in this direction includes the models by 
Kim-Lee~\cite{kim_vortex_1994} and by Collie-Tong~\cite{collie_dynamics_2008} 
for deformations of Abelian vortices. 
From the work of Alqahtani-Speight~\cite{alqahtani2015ricci} we know 
the model of Kim-Lee cannot extend to the coincidence set. We show 
our formula can be extended and compare it with 
the model of Collie-Tong for deformations of Abelian vortices, showing that 
our computation leads to different dynamics for pairs of vortices on the plane. 

In subsection~\ref{subsec:mhcs-model} we introduce the 
Maxwell-Higgs-Chern-Simons model, as we will see, the introduction of 
a Chern-Simons term in the field equations induces a connection 
term affecting the dynamics in moduli space. In 
subsections~\ref{sec:cs-localisation}
 and \ref{sec:loc-cs} we find the localization 
 formula~\eqref{eq:loc-conn-form} for this term and compare it 
 with the Collie-Tong connection. In subsection~\ref{subsec:cs-loc-o3-sigma} 
 we show how to extend our arguments to include the $O(3)$ Sigma model. 
 Finally, in subsection~\ref{sec:ellipt-reg}, we show that the 
 connection term can be extended to coalescence points, in the 
 case of the $O(3)$ Sigma model, provided the cores are of the 
 same type.

\section{Chern-Simons deformations of the O(3) Sigma model}
\label{sec:cs-intro}

Recall the construction of the $O(3)$ Lagrangian~\eqref{eq:o3-lag} in
section~\ref{sec:field-th}. To the pair $(\hf, A)$ 
of a field and a connection, we 
add an additional neutral field $N\in C^{\infty}(\surface)$ and
to avoid a name collision, in this chapter we denote the north pole section as 
$\np$, so that we modify the $O(3)$ Lagrangian as,
\begin{multline}
  L_{O(3),CS} = \half\left(\norm{D_t\hf}^2  +
  \norm{e}^2 + \norm{\dot N}^2 - \left(\norm{D \hf}^2 + \norm{B}^2  +
    \norm{dN}^2 
  \right.\right.\\
\left.\left. + \norm{\kappa \nf + \tau - \hf_3}^2
+ \norm{\nf X_{\hf}}^2 \right)\right),
\end{multline}
where $\hf_3$ is the gauge invariant product $\lproduct{\np, \hf}$ and  
$X_\hf$ is defined in a local trivialization $\hf_{\alpha}: U_\alpha \to 
\sphere$ as the section such that,
\begin{align}
X_{\hf_\alpha} = \vb e_3 \times \hf_\alpha.
\end{align}

 We add a Chern-Simons term to the Lagrangian, 
\begin{equation}
  L_{CS} = \half \brk(\lproduct{a, *e} + \lproduct{a_0, *B}).
\end{equation}

This Chern-Simons term is not gauge invariant, however, any two gauge
related terms differ by a divergence. The product is the
$\Lsp^2$ product induced in the exterior algebra by 
the metric in $\surface$. With this notation, the Kim-Lee-Lee
Lagrangian \cite{kimm1996anyonic} is,
\begin{equation}\label{eq:o3-sigma-cs-lagrangian}
  L = L_{O(3),CS} + \kappa L_{CS}.
\end{equation}
For the Abelian Higgs model, there are two ways to introduce a Chern-Simons 
term in the theory, one is due to Jackiew-Lee-Weinberg~\cite{jackiw1990} 
and the other to Lee-Lee-Min~\cite{lee_self-dual_1990}. In the first case,  
the connection term is replaced by a Chern-Simons term and the 
potential term is replaced by a sextic potential that admits 
a set of Bogomolny equations. It is known that several difficulties arise 
to study solutions to this model~\cite{flood2018chern}. For the second 
model, the extension to Chern-Simons 
deformations of the $O(3)$ Sigma model given 
in~\eqref{eq:o3-sigma-cs-lagrangian} is well established, yet the 
model has only being explored for compact tori due to the difficulties 
that non-compactness of the moduli space impose. We address those 
difficulties in this section, let us start considering variations with respect 
to $a_0$ yielding the Gauss law,
\begin{equation}
  d^{*} e = -\lproduct{D_t \hf, X_{\hf}}  + \kappa *B,
\end{equation}
where $d^{*} = -*d*$ is the codifferential. 
Instead of computing the field 
equations by the variational method, we note that these equations admit the use 
of the Bogomolny trick. Assume Gauss's law holds,  
the total energy of a tuple $(\hf, A, \nf)$ is,
\begin{align}
\mathrm{E} = \frac{1}{2}\, \left(\norm{D_t\phi}^2 + 
\norm{e}^2 + \norm{\dot{N}}^2 + 
\norm{D\phi}^2 + \norm{B}^2 + \norm{d N}^2
\right.\nonumber\\
\left.
+ \norm{\kappa N + \tau - \phi_3}^2
+ \norm{N X_\phi}^2
\right).
\end{align} 

Let us define the quadratic form,
\begin{align}
\mathrm{Q} &=\frac{1}{2}\,\left(\norm{D_t\phi - N\,X_\phi}^2 + 
\norm{e - d N}^2 + \norm{\dot{N}}^2 
+
\norm{*B + \kappa N + \tau - \phi_3}^2
\right)  + \norm{\bar \partial_A\phi}^2,
\end{align}
where $\bar\partial_A\hf$ is the $(0,1)$ component of $D\hf$ with respect to 
the 
almost complex structures of $\surface$ and the target $\sphere$. 
We can simplify $\mathrm{Q}$ as follows,
\begin{align}
\mathrm{Q} &=\mathrm{E} - \lproduct{D_t\phi, N X_\phi} - \lproduct{e, 
    d N} + \lproduct{*B, \kappa N + \tau - \phi_3} 
+ \norm{\bar \partial_A\phi}^2 - \frac{1}{2}\norm{D\phi}^2\nonumber\\
&= \mathrm{E} - \lproduct{\,\lproduct{D_t\phi,X_\phi} + d^*e - \kappa *B, 
     N} + \lproduct{*B, \tau - \phi_3}
+ \norm{\bar \partial_A\phi}^2 - \frac{1}{2}\norm{D\phi}^2\nonumber\\
&= \mathrm{E} + \lproduct{*B, \tau - \phi_3}
+ \norm{\bar \partial_A\phi}^2 - \frac{1}{2}\norm{D\phi}^2,
\end{align}
where we have used Gauss's law in the second equation. 

Since,
\begin{align}
D\phi &= \partial_A\phi + \bar\partial_A\phi,
\end{align}
we deduce,
\begin{align}
\mathrm{Q} &= \mathrm{E} + \lproduct{*B, \tau - \phi_3}
+ \frac{1}{2}\norm{\bar \partial_A\phi}^2 - \frac{1}{2}\norm{\partial_A\phi}^2.
\end{align}

Consider a trivialization 
$ \psi : \pi^{-1}(U) \to U \times \sphere $, such that 
$\psi\circ\phi(x) = (x, \tilde 
\phi(x))$, where $ \tilde \phi: U \to \sphere $ and $ U $ is an open, simply 
connected, dense set such that $Z = \hf_3^{-1}(\pm 1) \subset U$.  
Assume in this chart the connection is represented 
by a form $ a \in \Omega^1(U) $, 
let $ \theta: \sphere \setminus \left\{ \pm \np \right\}  \to \reals $ be the 
azimuthal angle on the sphere, Rom\~ao and Speight show in~\cite{romao2018} 
that there is a 
well defined, gauge invariant form $ \xi \in \Omega^1(\surface\setminus Z)$, 
such that on $U$, 
\begin{align}
\xi = \phi_3 \cdot (a - \tilde \phi^* d\theta).
\end{align}

A short computation in local coordinates shows
\begin{align}
\frac{1}{2}\,*(|\partial_A\phi|^2 - |\bar \partial_A\phi|^2) =
d\xi - \tau\,B + (\tau - \phi_3)\,B,
\end{align}
hence, 
\begin{align}
\mathrm{Q} &= \mathrm{E} - \int_\surface (d\xi - \tau B)\nonumber\\
&= \mathrm{E} - (2\pi\,(1-\tau)k_+ - 2\pi\,(1+\tau)k_-),
\end{align}
where the last integral was also computed in~\cite{romao2018}. Therefore, as 
for the $O(3)$ Sigma model without deformation, 
\begin{align}
\mathrm{E} \geq 2\pi\,(1-\tau)k_+ - 2\pi\,(1+\tau)k_-,
\end{align}
with equality if and only the following Bogomolny  
equations are satisfied,
\begin{align}
  \dot \nf &= 0, 
  \\
  e &= d \nf, 
  \\
  \Diff_t\hf &= \nf X_{\hf}, 
  \\
\conj{\del}_A\hf &= 0,\label{eq:hf}\\
*B &= -(\kappa \nf + \tau - \lproduct{n, \hf}).
\end{align}

Since equation
\eqref{eq:hf} holds, by the 
result~\cite[~p.~8]{sibner2010} of Sibner et al., $Z$ is finite, 
moreover, if we
consider a holomorphic chart $\varphi:U \subset
\surface \to  V \subset \cpx$ about $c \in Z$, and a trivialisation 
$\psi: \pi^{-1}(U) \to U \times \sphere$, such that 
$\psi\circ \hf|_U = (\mathrm 
{id}, \tilde \hf)$, 
then the degree of the map $\tilde\hf\circ \varphi^{-1} : V \to \widehat{\cpx}$ 
at $\varphi(c)$ is independent of the holomorphic coordinates chosen, as for 
the $O(3)$ Sigma 
model. We call this the degree of the
section $\hf$ at $c$. 
As in the $O(3)$ Sigma model, we define the sets $\vset = 
\hf_3^{-1}(1)$, $\avset =
\hf_3^{-1}(-1)$ of vortices and antivortices and  
denote by $k_+ =
\abs\vset$, $k_- = 
\abs\avset$ the size of the core sets, counted with multiplicity, where
a core $c$ is repeated as many times as the degree of $\hf(c)$. If
we choose the gauge $a_0 = -\nf$, the fields become stationary: 
$\dot\hf = \dot a = 0$. Defining the function $u:\surface\setminus
\vset \cup \avset \to \reals$,
\begin{equation}
  u = \log \brk(\frac{1 - \hf_3}{1 + \hf_3}),
\end{equation}
from the Bogomolny equations and the Gauss's law, we find that $(u, N)$ is a 
solution to the elliptic problem, 
\begin{equation}
  \begin{aligned}
-\laplacian u &= 2 \brk(
 \kappa N + \tau + \frac{e^u - 1}{e^u + 1}
) + 4\pi \sum_{p\in\vset} \delta_p - 4\pi\sum_{q\in\avset}\delta_q,\\
-\laplacian N &= \kappa \brk(
\kappa N + \tau + \frac{e^u - 1}{e^u + 1}
) + \frac{4 e^u}{(e^u + 1)^2}N,
\end{aligned}
\label{eq:ell-cs}
\end{equation}
%

As a consequence of the second equation, if $\kappa =
0$, $N \equiv 0$ and the equation for $u$ reduces to the elliptic
problem of the abelian $O(3)$ Sigma model, in which case we 
know that there exists exactly one solution, provided Bradlow's bound
holds. Given a set of disjoint divisors $D = (P, Q)$, we define
\begin{align}
  \kappa_{*}(D) = \sup \set{\kappa > 0 \,\mid 
  \text{there exists a solution of the field equations}}. 
\end{align}

This is a non-negative number we aim to prove that satisfies the
inequality,
\begin{align}
  \inf\,\set{\kappa_{*}(D)\,\mid {D\in\moduli^{k_+,k_-}}} > 0.
\end{align}

Moreover, if $\abs{k_+ - k_-} > 0$ and Bradlow's bound is fulfilled,
it will turn out that the supremum
\begin{align}
  \sup\,\set{\kappa_{*}(D)\,\mid {D\in\moduli^{k_+,k_-}}},
\end{align}
is bounded and for small but positive $\kappa$, there are two
solutions to the field equations, one close to the solution $(u_0, 0)$
of $BPS$ solitons and another with arbitrarily large norm, in a sense
to be defined on the following sections. Similar statements hold for
negative $\kappa$.


\subsection{Small deformations of $\kappa$}
\label{sec:small-deform}

In this section we prove that small deformations of the solution 
$h_0$ to the Taubes equation vary smoothly with $\kappa$. In order to do 
this, we will define a suitable operator and use it together 
with the implicit function 
theorem. 
Recall for any holomorphic chart $\varphi: U\subset \surface \to \cpx$ and 
bounded domain $D \subset \cpx$, there is a smooth function 
$\tilde G: \varphi^{-1}(D) \to \reals$ such that if $x, y \in \varphi^{-1}(D)$,
\begin{equation}
  G(x, y) = \frac{1}{2\pi}\log\,|\varphi(x) - \varphi(y)| + \tilde G(x, y),
\end{equation}

Hence, for the functions,
\begin{align}
  v_{\vb c} = 4\pi\,\sum_i\,G(x, c_i), \qquad \vb{c} = (c_1,\ldots, c_{k_\pm}),
\end{align}
$e^{v_{\vb c}}$ varies smoothly with $\vb c$.  We denote by
$\Delta_{k_+, k_-}$ the $(k_+, k_-)$ diagonal of 
$\surface^{k_+}\times \surface^{k_-}$. The space of solutions of
the elliptic problem at $\kappa = 0$ is the moduli space of asymmetric
vortices and antivortices described on Chapter~\ref{ch:pre}. We define
the function 
\begin{align}
  v :  \modulikk \to C^{\infty}(\surface, \conj\reals), && 
  v(\vb p, \vb q) = v_{\vb p} - v_{\vb q},
\end{align}
where $C^{\infty}(\surface,\conj\reals)$ means the set of smooth
functions, except at a finite set of points at which we have divergences to 
$\pm \infty$. If $F:\reals
\to \reals$ is the function 
\begin{equation}
  F(t) = \frac{e^t - 1}{e^t + 1} + \tau,
\end{equation}
then solving equation \eqref{eq:ell-cs} is equivalent to
finding a pair of functions $(h, N)$ such
that,
\begin{align}
  \laplacian h + 2 (\kappa N + F(v + h)) + \frac{4\pi}{\abs \surface}(k_+ -
k_-) &= 0,
\label{eq:h}\\
  \laplacian N + \kappa \brk(\kappa N + F(v + h)) + 2 F'(v + h)\,N &= 0.
\label{eq:N}
\end{align}


We introduce the potential function $V(t) = 2F'(t)$ such that, if
$(\vb p, \vb q) \in \modulikk$, $h \in \Hsp^r$ and $v = v(\vb p, \vb
q)$, then,
\begin{align}
V(v + h) = \frac{4 e^{v_{\vb p} + {v_{\vb q}} + h}}{(e^{v_{\vb p} + h} + 
e^{v_{\vb q}})^2}.
\end{align}
As shown in the proof of Theorem~\ref{thm:h-reg-param-dep}, by 
equation~\eqref{eq:green-logd-smoothpart} the functions $e^{v_{\vb p}}$, 
$e^{v_{\vb q}}$ are smooth and vary smoothly
with $(\vb p, \vb q)$. We observe that 
the spaces $\Hsp^{r}$, $r \geq 2$, are algebras, a prove can be found 
in \cite{flood2018chern} where Flood and Speight use this result to 
prove that $e^{h}$ is a smooth function
$\Hsp^r \to \Hsp^r$, the claim follows because $e^h$ is the limit of 
the absolutely converging power series $\sum_{n=0}^\infty h^n/n!$. 
As a consequence,  $V(v + h) 
\in \Hsp^r$. Likewise, $F(v + h) \in \Hsp^r$ if $h \in \Hsp^r$. For any
given pair of disjoint sets $\vset$, $\avset$, let us define the operator,
\begin{equation}
 \Top: \reals \times \Hsp^r \times \Hsp^r \to \Hsp^{r-2} \times
  \Hsp^{r-2},  
\end{equation}

\begin{align}
   \Top (\kappa, h, N) &= \left(
  \laplacian h + 2\brk(\kappa N + F(v + h)) + \frac{4\pi}{\abs
    \surface}(k_+ - k_-),\right.\\
&\quad\left.
  \laplacian N + \kappa\brk(\kappa N + F(v + h)) + V(v + h) N \right)
\end{align}

$\Top$ is a smooth mapping between Hilbert spaces. For any given $h \in \Hsp^r$,
$r\geq 2$, we define the operator
\begin{align}
 \Lop: \Hsp^r \to \Hsp^{r-2}, && \Lop = \laplacian + V(v + h).
\end{align}
The derivative of the restriction
\begin{align}
  \Top|:\Hsp^r\times \Hsp^r \to \Hsp^{r - 2} \times \Hsp^{r-2}, &&
  \Top|(h, N) = \Top(0, h, N),
\end{align}
at a point $(h, 0)$ is $d\Top|_{(h,0)} = \Lop \oplus \Lop$.

\begin{lemma}
  For any set of core points $(\vb p, \vb q)$ in the moduli space, the operator
  $\Lop$ is a Hilbert space isomorphism $\Hsp^r\to \Hsp^{r - 2}$.
\end{lemma}

\begin{proof}
By Sobolev's embedding, $h \in C^0(\surface)$, hence $V \geq 0$ is a
continuous function which is 
only zero at the finite set $\vset\cup \avset$. By
Lemma~\ref{lem:h2-l2-iso}, for any $\psi \in \Hsp^{r -2}$ there is
exactly one $\varphi \in \Hsp^2$ such that,
\begin{equation}
  \Lop \varphi = \psi,
\end{equation}
but by Schauder's estimates,
\begin{equation}
  \norm{\varphi}_{\Hsp^r} \leq C \brk(\norm{\laplacian
    \varphi}_{\Hsp^{r - 2}} + \norm{\varphi}_{\Lsp^2}),
\end{equation}
for some constant $C$, hence $\varphi \in \Hsp^r$ and $\Lop$ is a 
bijective bounded operator. By the open mapping theorem,
$\Lop^{-1}$ is also continuous, hence bounded and the claim
follows.
\end{proof}

\begin{proposition}
  \label{thm:small-def} Assume Bradlow's bound holds, then 
  there is a positive constant 
  $\kappa_0(\vset, \avset)$ such that if $\abs \kappa < \kappa_0$,
  the elliptic problem \eqref{eq:ell-cs} has a solution. Moreover, for
  any open neighbourhood $U$ of $\vset \cup 
  \avset$, the restriction of the solutions $(h, N)$ to $\surface
  \setminus \overline{U}$ varies smoothly with $\kappa$.   
\end{proposition}

\begin{proof}
Let $h_0 \in 
\Hsp^r$ be the solution of equation \eqref{eq:h} with $\kappa = N =
0$, i.e. $h_0$ is the solution to the regularised Taubes equation of
the abelian Sigma model, since $\Top(0, h_0, 0) = 
(0, 0)$, by the implicit function theorem, there is an interval
$(-\kappa_{0}, \kappa_{0})$ such that the map,
\begin{equation}
  (-\kappa_{0}, \kappa_{0}) \ni \kappa \mapsto  (\kappa, h_\kappa,
  N_\kappa) \in (-\kappa_{0}, \kappa_{0}) \times \Hsp^r\times \Hsp^r,
\end{equation}
is smooth and $\Top(\kappa, h_\kappa, N_\kappa) = (0, 0)$. Therefore,
each pair $(h_\kappa, N_\kappa)$ is a solution to the regular
elliptic problem \eqref{eq:h}, \eqref{eq:N} in $\Hsp^r \times
\Hsp^r$. By a bootstrap argument, each $(h_{\kappa}, N_{\kappa})$ is in
$\Hsp^k\times \Hsp^k$ for any $k \geq r$. Hence, by Sobolev's
embedding, $(h_{\kappa}, N_{\kappa})$ is smooth, moreover, the
function $u_{\kappa} = v + h_{\kappa}$ varies smoothly as a
function of $\kappa$ and $(\vb p, \vb q)$ if $p_j, q_k\in U$
are such that $p_j \neq q_k$ for each vortex and antivortex position.
\end{proof}

Thus if $\kappa$ is small, there is a family of solutions to the field
equations close to the BPS soliton determined by $h_0$, in the sense
that $(h_{\kappa} - h_0, N_{\kappa})$ is small in the $\Hsp^r \times
\Hsp^r$ norm for any $r > 0$.

\subsection{A positive gap for $\kappa_{*}(D)$}
\label{sec:proof-cnt-bound}

By proposition~\ref{thm:small-def}, $\kappa_{*}(D) > 0$ for any
distribution of the divisors. On this section we will prove the
existence of a positive lower bound for $\kappa_{*}$, independent of
the core positions. Thus, localization of vortex-antivortex systems
makes sense globally for small deformations $\kappa$ of the BPS
model as is for the case of Ginzburg-Landau
vortices~\cite{flood2018chern}, even though in this case
the moduli space should be incomplete as is the case for the BPS model at
$\kappa = 0$. We prove several technical lemmas first, in order to 
find bounds for the norm of $\Top'$, the derivative of the 
operator defined in the previous section.

\begin{lemma}
  \label{lem:h0-bound}
  The solutions $h$ of equation~\eqref{eq:h} with $\kappa = 0$ are
  uniformly bounded on $\Hsp^2$. 
\end{lemma}

\begin{proof}
  Let,
  \begin{equation}
    c(h) = \frac{1}{\abs \surface} \int_{\surface} h \,\vol,
  \end{equation}
  be the average of $h$ on $\surface$. $h - c$ is
  orthogonal to the kernel of $\laplacian$, Schauder's estimates in
  this case give,
  \begin{equation}
    \norm{h - c}_{\Hsp^2} \leq C \norm{\laplacian h}_{\Lsp^2},
  \end{equation}
  where we denote by $C$ a positive constant, independent of the function $h$. 
  The function $F(t)$ is bounded, hence
  \begin{equation}
    \laplacian h = -2F(v + h) - \frac{4\pi}{\abs\surface}(k_+ - k_-)
  \end{equation}
  is uniformly bounded in $\Lsp^2$. Therefore, the set of functions $\set{h -
  c}$ is bounded in $\Hsp^2$ and by Sobolev's
  embedding also in $C^0(\surface)$. We claim  that the averages
  are also bounded. Assume otherwise towards a
  contradiction. Let $\tilde h 
  = h - c$, then
  there are sequences $v_n$, $\tilde h_n$, $c_n$ such that $\abs
  {c_n} \to \infty$. Suppose  $c_n \to \infty$, and let $(\vb p_n, \vb
  q_n)  \in \modulikk$ be the points
  defining $v_n$. Since $\surface$ is compact, we can assume the
  convergence $(\vb p_n, \vb q_n) \to (\vb p_{*}, \vb q_{*}) \in
    \surface^{k_+} \times \surface^{k_-}$. We have pointwise convergence $v_n \to
  v_{*} = v_{\vb p_{*}} - v_{\vb q_{*}}$, except possibly at points on
  the surface which are in $\vb p_{*}$ and $\vb q_{*}$ if there is
  any. Since the functions $\tilde h_n$ are uniformly bounded, we also have,
  \begin{equation}
    2F(v_n + \tilde h_n + c_n) \to 2(1 + \tau), \qquad \text{pointwise a.e.}
  \end{equation}
  
  By the dominated convergence theorem,
  \begin{equation}
    \int_{\surface}2F(v_n + \tilde h_n + c_n) \vol \to 2(1 + \tau)\abs\surface,
  \end{equation}
  but by the divergence theorem,
  \begin{equation}
    \int_{\surface}2F(v_n + \tilde h_n + c_n) \vol =
    -\int_{\surface}\brk(\laplacian h + \frac{4\pi}{\abs\surface}(k_+
    - k_-))\vol 
    = - 4\pi (k_+ - k_-),
  \end{equation}
  and this contradicts Bradlow's bound. If $c_n\to -\infty$ the same
  argument gives another contradiction. Therefore the set of averages 
  $\set{c(h)}$ is bounded. The lemma follows.
\end{proof}

\begin{lemma}\label{lem:pot-bound}
  For any $\epsilon > 0$ there is a positive constant $C(\epsilon)$,
  such
  that for any set of divisors and any $h \in \Hsp^2$ with $\norm{h}_{\Hsp^2} < 
  \epsilon$, 
  \begin{equation}
   \lproduct{V(v + h), 1}_{\Lsp^2} \geq C.
 \end{equation}
 
\end{lemma}

\begin{proof}
We will omit the subindex in the product
$\lproduct{V, 1}$ since it 
  is clear that we refer to $\Lsp^2(\surface)$. The potential is a
  non negative 
  function, hence $\lproduct{V(v + h), 1} \geq 0$. Assume towards a 
  contradiction the 
  existence of sequences $\set {v_n}$, $\set{h_n}$,  where
  $\norm{h_n}_{\Hsp^2} < \epsilon$ 
  and with vortices and antivortices at positions $\vb p_n$, $\vb q_n$, such 
  that for the sequences of potentials,
  \begin{equation}
   V_n = V(v_n + h_n),
  \end{equation}
  we have $\lproduct{V_n, 1} \to 0$. As in the previous lemma,
  we can assume $\vb p_n \to \vb p_* \in \surface^{k_+}$ and $\vb q_n \to \vb 
  q_* \in \surface^{k_-}$ together with pointwise convergence $v_n \to v_{*}$, 
  except possibly
  at points $x \in \surface$ belonging to the set of coordinates of $\vb p_{*}$ 
  or $\vb q_{*}$. Let $C_0 > 0$ be Sobolev's
  constant, such that,
  \begin{equation}
    \norm{h}_{C^0(\surface)} \leq C_0\,\norm{h}_{\Hsp^2}.
  \end{equation}

  Hence,
  \begin{equation}
    0 \leq \lproduct{V(\abs{v_{\vb p_{n}}} + \abs{v_{\vb q_{n}}} +
      C_0\epsilon), 1} \leq \lproduct{V_n, 1} \to 0.
  \end{equation}

  On the other hand, $V(\abs{v_{\vb p_{n}}} + \abs{v_{\vb q_{n}}} +
      C_0\epsilon)$ is a sequence of bounded functions converging pointwise
      to the continuous function $V(\abs{v_{\vb
          p_{*}}} + \abs{v_{\vb q_{*}}} + C_0\epsilon)$. By the
      dominated convergence theorem,
      \begin{equation}
        \lproduct{V(\abs{v_{\vb p_{*}}} + \abs{v_{\vb q_{*}}} +
      C_0\,\epsilon), 1} = 0,
  \end{equation}
  a contradiction.
\end{proof}

Given any pair $(D, h) \in \modulikk\times \Hsp^r$, $r \geq 2$, the
potential $V(v + h)$ is a non negative continuous function such that
$0 < \lproduct{V(v + h), 1}$. If $\psi \in \Lsp^2$, by
Lemma~\ref{lem:h2-l2-iso}, there is exactly one $\varphi \in \Hsp^2$
such that, 
\begin{equation}
  (\laplacian + V(v + h))\,\varphi = \psi,
\end{equation}
and a positive constant $C'$, independent of
$V(v + h),\, \varphi$ and $\psi$, such that,
\begin{align}
\label{eq:lax-milgram-bound}
\norm{\varphi}_{\Hsp^1} \leq \frac{C'}{\lproduct{V(v + h),1}} \brk(
 \frac{\norm{V(v + h)}^2_{\Lsp^2}}{\lproduct{V(v + h), 1}} + \norm{V(v + 
 h)}_{\Lsp^2}
+ 1
) \norm{\psi}_{\Lsp^2}.
\end{align}

We let the point in the moduli space vary in order to define the operator,
\begin{align}
 \mathcal{L}  : \moduli^{k_+,k_-}\times \Hsp^2 \to B(\Hsp^2, \Lsp^2), &&
\mathcal{L}(\vb p, \vb q, h) = L,
\end{align}
where $L = \laplacian + V(v + h)$ was defined previously. 
$\mathcal{L}$ is a continuous map such that each $\mathcal{L}(\vb p,
\vb q, h)$ is invertible. Since inversion of bounded invertible
operators is continuous, the map  
\begin{align}
  \mathcal{L}':\moduli^{k_+,k_-}\times \Hsp^2 \to B(\Lsp^2, \Hsp^2), &&
  \mathcal{L}'(\vb p, \vb q, h) = \Lop^{-1},
\end{align}
is also continuous. 

\begin{lemma}\label{lem:bound-cstar}
  Given $\epsilon > 0$, let $\Omega = \moduli^{k_+,k_-}\times
  B_{\epsilon}(0)$ and let,
\begin{equation}
  C_{*}(\epsilon) = \sup_{\Omega}\,\norm{\mathcal{L}'}, 
\end{equation}
then $C_{*}$ is finite. 
\end{lemma}

\begin{proof}
  By lemma~\ref{lem:pot-bound}, there is a constant $C(\epsilon)$ such that
  $\lproduct{V(v + h), 1} \geq C$ for any $(D, h) \in \Omega$ . If
  $\varphi = \mathcal{L}'\,\psi$, with $\norm{\psi}_{\Lsp^2} = 1$ by
  \eqref{eq:lax-milgram-bound} we have the bound
  \begin{equation}
  \norm{\varphi}_{\Hsp^1} \leq \frac{C'}{C}\brk(
  \frac{\abs\surface}{C} + \abs{\surface}^{1/2} + 1).
\end{equation}

By Schauder's estimates, 
\begin{align}
  \norm{\varphi}_{\Hsp^2} &\leq C \brk(\norm{\laplacian
    \varphi}_{\Lsp^2} + \norm{\varphi}_{\Lsp^2})\nonumber\\
  &= C \brk(\norm{-V(v + h)\varphi + \psi}_{\Lsp^2} + 
  \norm{\varphi}_{\Lsp^2})\nonumber\\
  &\leq C,
\end{align}
where the last constant is not necessarily equal to the first one.  
Therefore, $\norm{\mathcal{L}'|_{\Omega}} \leq C$, hence $C_{*} \leq C$.
\end{proof}

For given $\kappa$, let us define $\Top_{\kappa}|:\Hsp^2\times
\Hsp^2 \to \Lsp^2\times \Lsp^2$ as the restriction
$\Top_{\kappa}|(h, N) = \Top(\kappa, h, N)$, then we have,
\begin{align}
d\Top_{\kappa}|_{(h, N)} &= \Lop \oplus \Lop + \Top',
\end{align}
where,
\begin{equation}
  \Top'(h', N') = \brk(2\kappa N', \kappa^2 N' + \frac{\kappa}{2} V(v
  + h) h' + V'(v + h) N h').
\end{equation}

Since $V'(t)$ and $V(t)$ have range $[0, 1]$,
\begin{align}
  \norm{\Top'(h', N')}_{\Lsp^2\times\Lsp^2}\leq (2\abs\kappa +
\kappa^2)\norm{N'}_{\Lsp^2}  + \brk(
\frac{\abs{\kappa}}{2} + \norm{N}_{\Lsp^2}
) \norm{h'}_{\Lsp^2}.
\end{align}

By Cauchy-Schwarz,
\begin{align}
 \norm{T'} \leq \brk(\kappa^2(2 + \abs{\kappa})^2 +
\brk(\abs\kappa / 2 + \norm{N}_{\Lsp^2})^2)^{1/2}.
\end{align}

We use
lemma~\ref{lem:h0-bound} and choose $\epsilon_0$ such that $\norm{h_0}_{\Hsp^2}
< \epsilon_0$ for any solution $(h_0, 0)$ of equation \eqref{eq:h}
with $\kappa = 0$. We also take the  
constant $C_{*}(\epsilon_0)$ of lemma~\ref{lem:bound-cstar} and define
\begin{equation}
  \epsilon = \min\set{\epsilon_0, \frac{1}{2\, C_{*}\cdot (\frac{7}{2} +
      \epsilon_0)}}. 
\end{equation}

\begin{proposition}\label{prop:kappa-*-lb}
  For any set of divisors $D \in \moduli^{k_+,k_-}$,
  \begin{equation}
  \kappa_{*}(D) \geq \epsilon \min\set{
    1, \frac{1}{2\,C_{*} (2\epsilon(1 + \epsilon) + \max\set{1 \pm
        \tau}\abs\surface^{1/2})} 
  }.
\end{equation}
\end{proposition}

\begin{proof}
  For $\abs{\kappa} < \epsilon$, $\norm{h}_{\Lsp^2} < \epsilon$ and
$\norm{N}_{\Lsp^2} < \epsilon$, we have,
\begin{equation}
  \norm{\Top'} < \frac{1}{2C_{*}} \leq \frac{1}{2 \norm{(\Lop\oplus
      \Lop)^{-1}}},
\end{equation}
hence, as in the proof of Lemma 5 in \cite{flood2018chern}, we can
conclude that $d\Top_{\kappa}|$ is invertible, independently of the
point in the moduli space and
\begin{equation}
  \norm{(d\Top_{\kappa}|)^{-1}} \leq 2C_{*}.
\end{equation}

If $\chi(\kappa) = (h_{\kappa}, N_{\kappa})$ is the curve of solutions to
equations \eqref{eq:h},\eqref{eq:N}, guaranteed to exist by
proposition~\ref{thm:small-def}, then by the implicit function theorem,
this curve can be extended whenever $d\Top_{\kappa}|$ is invertible at
$\chi(\kappa)$. This is the case if $\abs{\kappa} < \epsilon$ and
$\norm{h}, \norm{N} < \epsilon$. So, for any $D \in \moduli^{k_+,
  k_-}$, let $\kappa_0 > 0$ be the right end of the maximal interval
$[0, \kappa_0)$ on which this curve can be extended. Either $\kappa_0
\geq \epsilon$,  
or there exists a $\kappa_1$ with $\abs{\kappa_1} < \kappa_{*} <
\epsilon$ such that $\norm{\dot\chi(\kappa)}_{\Hsp^2\times\Hsp^2} \geq \epsilon
\kappa_{*}^{-1}$. In the later case,
\begin{align}
  \norm{\dot\chi_{\kappa_1}}_{\Hsp^2\times\Hsp^2} &=
  \norm{(d\Top_{\kappa}|_{\chi(\kappa)})^{-1}\del_{\kappa}
    \Top|_{(\kappa_1,\chi(\kappa_1))}}_{\Hsp^2\times\Hsp^2}\nonumber\\
 &\leq 
  2C_{*}\norm{\del_{\kappa}\Top|_{(\kappa_1,\chi(\kappa_1))}}_{\Lsp^2\times\Lsp^2}. 
\end{align}

But,
\begin{equation}
  \del_{\kappa}\Top = (2N, 2\kappa N + F(v + h)),
\end{equation}
hence,
\begin{align}
  \norm{\del_{\kappa}\Top}_{\Lsp^2\times\Lsp^2} &\leq 2\norm{N}_{\Lsp^2}
  + 2\abs \kappa \norm{N}_{\Lsp^2} +
  \sup \set{\abs F} \abs{\surface}^{1/2}\nonumber\\
  &\leq 2\epsilon(1 + \epsilon) + \max\set{1 \pm \tau}\abs\surface^{1/2}.
\end{align}

Therefore, either $\kappa_0 \geq \epsilon$ or 
\begin{equation}
  \kappa_0 \geq \frac{\epsilon}{2C_{*} (2\epsilon(1 + \epsilon) +
    \max\set{1 \pm 
        \tau}\abs\surface^{1/2})} 
  .
\end{equation}

Since $\kappa_{*}(D) \geq \kappa_0$ we conclude the claimed lower bound. 
\end{proof}


\newcommand*\ux{\underline x}
\newcommand*\ox{\overline x}

\newcommand*\Ik{I_{\kappa}}
\newcommand*\Aset{\mathcal{A}}
\newcommand*\avh{\overline{h}}
\newcommand*\Kset{\mathcal{K}}
\newcommand*\opPhi{\Phi}
\newcommand*\opL{\mathrm{L}}
\newcommand*\opI{\mathrm{I}}
\newcommand*\lsdeg{\mathrm{deg}_{\mathrm{LS}}}
\newcommand*\av{c}

\subsection{The unbalanced case}\label{sec:unbalanced-case}

In this section we assume $k_+ \neq k_-$. In this case the family of
deformation constants is bounded, contrasting with the 
euclidean case, where there are examples for which the elliptic
problem can be solved for any $\kappa$~\cite{chen2019analysis}. 
We will prove the existence of multiple solutions of the field equations, 
the first step will be to describe the possible limit points of 
sequences $(h_{\kappa_n}, \kappa_n N_{\kappa_n})$ as $\kappa_n \to 0$.

It will be convenient to redefine the neutral field as follows. Let $N' =
\kappa N$, equation~\eqref{eq:N} can be rewritten as, 
\begin{equation}
  \label{eq:N'}
  \laplacian N' + (\kappa^2 + 2F'(v + h))\,N' + \kappa^2F(v + h) = 0.
\end{equation}

\begin{proposition}\label{prop:kappa-*-ub}
  If $\abs{k_+ - k_-} > 0$, $\kappa_{*}(D)$ is uniformly bounded,
  \begin{align}
 \kappa_{*}(D) \leq \pbrk{\frac{\max\set{1 \mp \tau}
   \abs\surface}{2\pi\,\abs{k_+ - k_-}}}^{1/2}.    
  \end{align}
 
\end{proposition}

\begin{proof}
  Let $\ux$ be such that $N'(\ux) = \min_\surface N'$ and likewise, 
  let $\ox$ be such that $N'(\ox) = \max_\surface N'$.
  %
  
  Let us denote $F(v(x) + h(x))$ as $F(x)$. Likewise, we denote $F'(v(x)
  + h(x))$ as $F'(x)$.

  By the maximum principle,
\begin{align}
-\frac{\kappa^2F(\ux)}{\kappa^2 + 2F'(\ux)} \leq
N'(x) \leq
- \frac{\kappa^2 F(\ox)}{\kappa^2 + 2F'(\ox)}.
\end{align}

Since $F'(x) \geq 0$, we conclude the uniform bounds,
\begin{align}\label{eq:N'-bound}
-(1 + \tau) \leq N' \leq (1 - \tau).
\end{align}

From equations \eqref{eq:h}  and \eqref{eq:N'} we obtain,
\begin{align}
\laplacian N'- \frac{\kappa^2}{2} \brk(\laplacian h + \frac{4\pi}{\abs
\surface} (k_+ - k_-)) + 2F'(v + h)N' = 0. 
\end{align}

Integrating this equation,
\begin{align}
-2\pi\kappa^2(k_+ - k_-) + 2 \lproduct{F'(v + h), N'} = 0.
\end{align}

Using the fact that $F'$ is a positive function bounded by
$1/2$ and the uniform bound for $N'$,
\begin{equation}
\kappa^2 = \frac{\lproduct{F'(v + h), N'}}{\pi\,(k_+ - k_-)}
  \leq \frac{\max\set{1 \mp \tau} \abs\surface}{2\pi\,\abs{k_+ - k_-}}.
\end{equation}
\end{proof}

As a consequence of this proposition, we have the following lemma,

\begin{lemma}\label{lem:h-n-bound}
  If $\abs{k_+ - k_-} > 0$, and $(h_{\kappa}, \nf'_{\kappa})$ denotes a 
  solution to the pair of equations~\eqref{eq:h},~\eqref{eq:N'} with
  deformation parameter $\kappa$, for any $p \geq 2$ there is a
  uniform constant 
  $C(p)$, such that,
  \begin{align}
    \norm{h_{\kappa} - c(h_{\kappa})}_{\Wsp^{2, p}} +
    \norm{\nf'_{\kappa}}_{\Wsp^{2, p}} \leq C,
  \end{align}
  where $c(h_{\kappa})$ is the average of $h_{\kappa}$ on the
  surface. 
\end{lemma}

\begin{proof}
  We will assume the constant $C$ 
 can change from one line to the next. Since $\kappa_{*}$ is
 bounded, from \eqref{eq:N'} we
have a uniform bound for the norm of $N'$ in Sobolev's space, 
\begin{equation}
  \norm{\laplacian N'}_{\Lsp^p} \leq C.
\end{equation}

By Calderon-Sygmund theory, this implies,
\begin{equation}
  \norm{N'}_{\Wsp^{2,p}} \leq C \brk(
  \norm{\laplacian N'}_{\Lsp^p} + \norm{N'}_{\Lsp^p}
  ) \leq C.
\end{equation}

Similarly, from~\eqref{eq:h} we deduce the existence of an upper
bound for the set
$\set{h_{\kappa} - c(h_{\kappa})}$ in $\Wsp^{2,p}$. 
\end{proof}

If we fix any $p > 2$, Sobolev's theory says that the embedding
$\Wsp^{2,p} \to C^1(\surface)$ is continuous, thence there 
exists a constant independent of $\kappa$, such that,
\begin{equation}
\norm{h_{\kappa} - \av(h_{\kappa})}_{C^1} +  \norm{N'}_{C^1} \leq C. 
\end{equation}

Let $(h_n, N'_n)$ denote a sequence of solutions to the
elliptic equations \eqref{eq:h}-\eqref{eq:N'} with a corresponding
sequence of parameters $\set{\kappa_n}$. We are interested in
describing the behaviour of these solutions as $\kappa_n \to
0$. Although the sequence $(h - \av(h), \nf')$ is uniformly bounded in 
$C^1$ we cannot rule out the possibilities $\av(h_n) \to \pm
\infty$. In the following lemmas we deal with the three cases arising
on this analysis.

\begin{lemma}\label{lem:hn-N'n-conv}
  If $(h_n, N'_n)$ is a sequence of solutions to the elliptic
  equations \eqref{eq:h}-\eqref{eq:N'} with parameters $\kappa_n \to
  0$ such that the sequence is bounded in $\Hsp^1\times\Hsp^1$, then
  for any $p \geq 2$ the sequence converges  
  to $(h_0, 0)$ strongly in $\Wsp^{2,p} \times \Wsp^{2,p}$, where
  $h_0$ is the solution to the regularized Taubes equation.
\end{lemma}

In particular, this means the convergence is uniform in $C^1\times
C^1$. 

\begin{proof}
  By the Banach-Alaoglu theorem, for any subsequence $(h_{n_k},
  \nf'_{n_k})$, after passing to another subsequence if necessary, we can
  assume $(h_{n_k}, \nf'_{n_k}) \to (h_{*},\nf'_{*})$ weakly in
  $\Hsp^1\times \Hsp^1$ and by the Rellich-Kondrachov lemma,
  strongly in $\Lsp^2\times\Lsp^2$. Let $(u, w) \in
  \Hsp^1\times\Hsp^1$, equations~\eqref{eq:h}~-~\eqref{eq:N'} can
  be expressed in weak form as,
  \begin{equation}
    \begin{gathered}
      \lproduct{\grad u, \grad h_{n_k}} + \lproduct*{u, 2\,(\nf'_{n_k}  +
      F(v + h_{n_k})) +  \frac{4\pi}{\abs\surface}\,(k_+ - k_-)} =
  0,\\
      \lproduct{\grad w, \grad \nf'_{n_k}} + \lproduct*{w,
        (\kappa_{n_k}^2 + 2\,F'(v + h_{n_k}))\,\nf'_{n_k} +
        \kappa^2_{n_k}\,F(v + h_{n_k})} = 0.
    \end{gathered}
  \end{equation}

Weak convergence in $\Wsp^1$ plus strong convergence in
$\Lsp^2$ imply
\begin{align}
  \lproduct{\grad u, \grad h_{n_k}} \to \lproduct{\grad u, \grad
    h_{*}}, &&
  \lproduct{\grad w, \grad \nf'_{n_k}} \to \lproduct{\grad w, \grad
    \nf'_{*}}.
\end{align}

After passing to another subsequence if necessary, we can assume $h_{n_k}
\to h_{*}$ pointwise almost everywhere. By the 
dominated convergence theorem,
\begin{align}
  \lproduct{u, F(v + h_{n_k})} = \lproduct{u, F(v + h_{*})},
\end{align}
and simmilarly for $w$. Therefore, $(h_{*}, \nf'_{*})$ is a weak solution
to the equations,

  \begin{align}
    \laplacian h_{*} + 2\,(\nf'_{*} + F(v + h_{*})) + \frac{4\pi}{\abs
      \surface}\,(k_+ - k_-) = 0,\label{eq:h*}\\
    \laplacian \nf'_{*} + 2\,F'(v + h_{*})\,\nf'_{*} = 0.\label{eq:N'*}
  \end{align}

Ellipticity guarantees the solution is in fact strong, hence, by the usual
elliptic estimates, $(h_{*}, \nf'_{*}) \in \Hsp^2\times \Hsp^2$ and  
the solution is continuous. This together with~\eqref{eq:N'*} implies
$\nf'_{*} \equiv 0$. Therefore~\eqref{eq:h*} is the regularised
Taubes equation, whose unique solution is $h_{*} \equiv
h_0$.

Since any subsequence of $(h_n, \nf'_n)$ can be refined to a
convergent subsequence to 
$(h_0, 0)$ in $\Lsp^2\times \Lsp^2$, we
obtain the limit,
\begin{align}
  \norm{h_n - h_0}_{\Lsp^2} + \norm{\nf'_n}_{\Lsp^2} \to 0.
\end{align}

In particular, $\nf'_n \to 0$ in $\Lsp^2$, this limit and the
boundedness of the functions $F(t)$ and $F'(t)$ imply by means of
equation~\eqref{eq:N'} the limit,
\begin{align}
  \norm{\laplacian \nf'_n}_{\Lsp^2} \to 0.
\end{align}

Hence, by the usual combination of Schauder's estimates and Sobolev's
embedding, we find two constants such that,
\begin{align}
  \norm{\nf'_n}_{C^0} \leq C_1\,\norm{\nf'_n}_{\Hsp^2} \leq
  C_2\,(\norm{\laplacian \nf'_n}_{\Lsp^2} + \norm{\nf'_n}_{\Lsp^2}) \to 0.
\end{align}

Whereas by equation~\eqref{eq:h}, 
\begin{align}
  -\laplacian (h_n - h_0) = 2\,\nf'_n + 2\,(F(v +
  h_n) - F(v + h_0)).\label{eq:lapl-hn-h0-nfp-Fn-F0}
\end{align}

Note that by the mean value theorem, for any $x \not\in Z$, 
\begin{align}
F(v(x) + h_n(x)) - F(v(x) + h_0(x)) = F'(\xi)\,(h_n(x) - h_0(x)),
\end{align}
for some $\xi$ between $h_n(x)$ and $h_0(x)$, whereas for $x \in Z$,
\begin{align}
F(v(x) + h_n(x)) - F(v(x) + h_0(x)) = 0,
\end{align}
since $F(v(x) + h_n(x)) = F(v(x) + h_0(x)) = \pm 1 + \tau$ in this case, 
hence, there is 
a constant $C > 0$, such that,
\begin{align}
|F(v + h_n) - F(v + h_0)| \leq C\,|h_n - h_0|,\label{eq:abs-Fn-F0-leq-hn-h0}
\end{align}
by~\eqref{eq:lapl-hn-h0-nfp-Fn-F0} and~\eqref{eq:abs-Fn-F0-leq-hn-h0},
\begin{align}
  \norm{\laplacian(h_n - h_0)}_{\Lsp^2} &\leq 2\,\norm{\nf'_n}_{\Lsp^2} +
  2\,\norm{F(v +
    h_n) - F(v + h_0)}_{\Lsp^2}\nonumber\\
  &\leq C\,(\norm{\nf'_n}_{\Lsp^2} + \norm{h_n-h_0}_{\Lsp^2}).
\end{align}

We deduce $\norm{\laplacian(h_n-h_0)}_{\Lsp^2} \to
0$. Repeating the elliptic estimate argument we find,
\begin{align}
  \norm{h_n- h_0}_{\Hsp^2} \to 0,
\end{align}
and consequently also $h_n \to h_0$ in $C^0(\surface)$. Finally, we
follow a bootstrap argument. Knowing the convergence $(h_n, \nf'_n)
\to (h_0, 0)$ is uniform in $C^0$, we can repeat the previous
computations for the norm of the Laplacian, this time in the $\Lsp^p$
norm and deduce the claimed limit.
\end{proof}

In case $(h_n, \nf'_n)$ is not bounded, necessarily $\set{c(h_n)}$ has
a subsequence diverging to $\pm \infty$. We consider each possibility in
the following lemma.

\begin{lemma}\label{lem:hn-N'n-div}
    If $\set{(h_n, \nf'_n)}$ is a sequence of solutions to
    equations~\eqref{eq:h}~\eqref{eq:N'} such that $\kappa_n \to 0$ 
   and $c(h_n) \to \infty$, and if 
  $p \geq 2$ is a fixed but otherwise arbitrary constant, the
  following limit holds,
  \begin{align}
    \norm{h_n - c(h_n)}_{\Wsp^{2,p}} + \norm{\nf'_n - \alpha_+}_{\Wsp^{2,p}} 
    \to 0,
  \end{align}
  where
  \begin{align}
    \alpha_+ = -1  - \tau - \frac{2\pi\,(k_+ - k_-)}{\abs \surface},
  \end{align}
  and the condition $k_+ - k_- < 0$ is necessary. Similar statements
  hold if $c(h_n) \to -\infty$, where the constant $\alpha$ in this case
  is,
  \begin{align}
    \alpha_- = 1  - \tau - \frac{2\pi\,(k_+ - k_-)}{\abs \surface},
  \end{align}
  and the condition $k_+ - k_- > 0$ is necessary.
\end{lemma}

\begin{proof}
  We proceed as in lemma~\ref{lem:hn-N'n-conv}. Let $\tilde h_n = h_n
  - \av(h_n)$, we consider a subsequence $(\tilde h_{n_k},
  \nf'_{n_k})$. By lemma~\ref{lem:h-n-bound}, after passing to another
  subsequence if necessary, we can assume $(\tilde h_{n_k}, \nf'_{n_k})
  \to (\tilde 
  h_{*}, \nf_{*}')$ weakly in $\Hsp^1\times \Hsp^1$ and strongly in
  $\Lsp^2\times \Lsp^2$. Repeating the argument of
  lemma~\ref{lem:hn-N'n-conv}, this time using the fact that
  $\av(h_{n_k}) \to \infty$, we find $(\tilde h_{*}, \nf'_{*})$ is a
  strong solution to the equations,

    \begin{gather}
      \laplacian \tilde h_{*} + 2\,(\nf'_{*} + 1 + \tau) +
      \frac{4\pi}{\abs \surface}\, (k_+ - k_-) = 0,\label{eq:tilde-h-n}
      \\
      \laplacian \nf'_{*} = 0.\nonumber
    \end{gather}

  Since $\nf'_{*}$ is in the kernel of the Laplacian, it has to be a
  constant function $\alpha$. By a bootstrap argument, $\tilde
  h_{*}$ is smooth. To determine the constant, we integrate
  equation~\eqref{eq:tilde-h-n} by means of the divergence theorem and find,
  \begin{align}
    \alpha = - 1 - \tau - \frac{2\pi}{\abs \surface}\,(k_+ - k_-),
  \end{align}
  and,
  \begin{align}
    \laplacian \tilde h_{*} = 0.
  \end{align}

  Since each $\tilde h_n$ has zero average the same holds for $\tilde
  h_{*}$, therefore $\tilde h_{*} \equiv 0$.

  We follow once more the pattern of lemma~\ref{lem:hn-N'n-conv} to
  conclude strong convergence $(\tilde h_{n_k},\nf'_{n_k}) \to (0,
  \alpha)$ in $\Lsp^2\times \Lsp^2$ and deduce that for any $p \geq
  2$, 
  \begin{align}
    \norm{\tilde h_n}_{\Wsp^{2,p}} + \norm{\nf'_n - \alpha}_{\Wsp^{2,p}} \to 0.
  \end{align}

  Finnally, by~\eqref{eq:N'-bound},
  \begin{align}
    -(1 + \tau) \leq \alpha \leq 1 - \tau,
  \end{align}
  hence,
  \begin{align}
    -(1 + \tau) \leq -(1 + \tau) - \frac{2\pi}{\abs \surface}\,(k_{+}
    - k_-),
  \end{align}
  implying $k_+ < k_-$. It is clear these arguments can be rearranged for
  the case $\av(h_n) \to -\infty$. 
\end{proof}

Lemmas~\ref{lem:hn-N'n-conv},~\ref{lem:hn-N'n-div} thus prove the
following proposition.

\begin{proposition}\label{prop:cs-lim-points}
  If $\set{(h_n, \nf'_n)}$ is a sequence of solutions to
  equations~\eqref{eq:h}~\eqref{eq:N'} such that $\kappa_n \to 0$, the
  only possible limit points are $(h_0, 0)$ and $(0, \alpha_{\pm})$,
  for $\alpha_{\pm}$ defined on lemma~\ref{lem:hn-N'n-div}. If
  $\set{(h_n, \nf'_n)}$ has a bounded (unbounded) subsequence, then
  $(h_0, 0)$ ($(0, \alpha_{\pm})$) is a limit point. 
\end{proposition}

\subsection{Existence of multiple solutions}\label{sec:exist-mult-sols}

In this section we prove the existence of multiple solutions to the field 
equations using theorem~\ref{thm:ls-deg-continuum}. In order to do this, 
 we define an operator $\Phi$ satisfying the hypothesis of the 
 theorem trough a series of technical lemmas.

\begin{lemma}\label{lem:kappa-bounded}
  For any $0 < \kappa_0 < \kappa_{*}(D)$ and $p \geq 2$, the set
  $\set{(h_{\kappa}, N'_{\kappa}) \mid  \kappa > \kappa_0}$ is
  bounded in $\Wsp^{2,p}\times\Wsp^{2,p}$.
\end{lemma}

\begin{proof}
Assume otherwise towards a contradiction. 
If $(h_n, \nf'_n)$ is an unbounded sequence, by
lemma~\ref{lem:h-n-bound} we can suppose
$\av(h_n) \to \pm \infty$ depending on the sign of $k_+ -
k_-$. Without loss of generality, let us assume
$k_+ - k_- > 0$. In this case $\av(h_n) \to -\infty$ by
lemma~\ref{lem:hn-N'n-div}. Since $\set{\kappa_n}$
is bounded, we can assume $\kappa_n \to k_* \neq 0$. Let $\tilde h_n =
h_n - \av(h_n)$, going through the
steps of the proof of lemma~\ref{lem:hn-N'n-conv}, we deduce the
existence of a strong limit $(\tilde h_n, \nf'_n) \to (\tilde h_{*},
\nf'_{*})$ in $\Wsp^{2,p}\times\Wsp^{2,p}$, such that $(\tilde h_{*},
\nf'_{*})$ is a solution to the problem,
\begin{align}
(\laplacian + \kappa_*^2)\, N'_* + \kappa_*^2\,(-1 + \tau) &= 0,\\
\laplacian \tilde h_* + 2 N'_* + 2\brk(-1 + \tau +
\frac{2\pi}{\abs\surface} (k_+ 
- k_-)) &= 0.
\end{align} 

By elliptic regularity the pair $(\tilde h_{*}, N'_*)$ is
smooth. Since $\kappa_*^2 > 0$, 
integrating the first equation, we obtain,
\begin{equation}
  \lproduct{N_*', 1} + (-1 + \tau)\abs\surface = 0.
\end{equation}

Integrating the second equation, we have,
\begin{equation}
  \lproduct{N_*', 1} + (-1 + \tau)\,\abs\surface + 2\pi(k_+ - k_-) = 0.
\end{equation}

Hence $2\pi\,(k_+ - k_-) = 0$, a contradiction. 
\end{proof}

We will prove the
existence of multiple solutions to the field equations in the
unbalanced case, which can be seen in figure~\ref{fig:sols} on the
right column, adapting the argument from \cite{han2014multiplicity}
which relies on Leray-Schauder's degree. 

We define the operators,
\begin{align}
 \opL &= (-\laplacian
 -\lambda, -\laplacian - \lambda), &
\opPhi_{\kappa}(h, N) &=
(f_{\kappa}(h,N), g_\kappa(h,N)),
\end{align}
where
\begin{align}
  f_{\kappa}(h, N) &= 2(N + F(v + h)) + \frac{4\pi}{\abs\surface}(k_+
- k_-) - \lambda h,\\
  g_{\kappa}(h, N) &= (\kappa^2 + 2F'(v + h)) N + \kappa^2F(v + h) -
\lambda N,
\end{align}
and $\lambda$ is a positive constant.  
Recall a continuous non-linear map $T : X \to Y$ of Banach spaces is
said to be compact if it maps any bounded subset $A\subset X$ to a
precompact set $TA \subset Y$, 

\begin{lemma}\label{lem:op-phi-cont}
  The operator $\opPhi : [0, \kappa_{*}] \times \Hsp^2\times \Hsp^2
  \to \Lsp^2\times\Lsp^2$ such that $\opPhi(\kappa, \cdot,\cdot) =
  \opPhi_{\kappa}$ is continuous.
\end{lemma}

\begin{proof}
We
show the component functions $f_{\kappa}$ and $g_{\kappa}$ are
continuous as follows. Notice, 
\begin{align}
  \norm{f_{\kappa_2}(h_2, \nf_2) - f_{\kappa_{1}}(h_{1}, \nf_{1})}_{\Lsp^2}
  &\leq 2\,\norm{\nf_2 - \nf_1}_{\Lsp^2} + 2\,\norm{F'(\xi)\,(h_2 - 
  h_1)}_{\Lsp^2} +
  \lambda\,\norm{ h_2 - h_1}_{\Lsp^2}\nonumber\\
  &\leq C\,(\norm{N_2 - N_1}_{\Lsp^2} + \norm{h_2 - 
  h_1}_{\Lsp^2}),\label{eq:f-k1-k2}
\end{align}
where $\xi$ is well defined almost everywhere. For $g_{\kappa}$ we have,
\begin{align}
  \norm{g_{\kappa_2}(h_2, \nf_2) - g_{\kappa_1}(h_1, \nf_1)}_{\Lsp^2}
  &\leq \norm{\kappa_2^2\,\nf_2 - \kappa_1^2\,\nf_1}_{\Lsp^2}\nonumber\\
  &\quad + \norm{
    \kappa_2^2\,F(v + h_2) - \kappa_1^2\,F(v + h_1) 
  }_{\Lsp^2}\nonumber\\
  &\quad +  2\,\norm{\nf_2\,F'(v + h_2) - \nf_1\,F'(v + 
  h_1)}_{\Lsp^2}\nonumber\\
  &\quad + \lambda\,\norm{\nf_2 - \nf_1}_{\Lsp^2}\nonumber\\
  &\leq \abs{\kappa_2^2 - \kappa_1^2}\,\norm{\nf_2}_{\Lsp^2}
  + \kappa_1^2\,\norm{\nf_2 - \nf_1}_{\Lsp^2}\nonumber\\
  &\quad \abs{\kappa_2^2 - \kappa_1^2}\,\norm{F(v + h_2)}_{\Lsp^2} +
  \kappa_1^2\, \norm{F(v+h_2) - F(v + h_1)}_{\Lsp^2}\nonumber\\
  &\quad 2\,\norm{(\nf_2 - \nf_1)\,F'(v + h_2)}_{\Lsp^2}\nonumber\\
  &\quad + 2\,\norm{\nf_1\,(F'(v
    + h_2) - F'(v + h_1))}_{\Lsp^2}\nonumber\\
  &\quad + \lambda\,\norm{\nf_2 - \nf_1}_{\Lsp^2}.
\end{align}

$F(v + h_2)$ and $F'(v + h_2)$ are uniformly bounded by a constant
independent of $h_2$. Also, there exist functions $\xi$, $\eta$ well
defined except at core positions, such that,
\begin{align}
  \abs{F(v + h_2) - F(v + h_1)} = |F'(\xi)\,(h_2 - h_1)| \leq
  C\,\abs{h_2 - h_1}, \\
  \abs{F'(v + h_2) - F'(v + h_1)} = |F''(\eta)\,(h_2 - h_1)| \leq
  C\,\abs{h_2 - h_1}.
\end{align}

By Sobolev's embedding, $\nf_1$ is continuous and the norm
$\norm{\nf_1}_{C^0}$ is controlled by $\norm{\nf_1}_{\Hsp^2}$, hence,
\begin{align}
  \norm{\nf_1\,(F'(v + h_2) - F'(v + h_1))}_{\Lsp^2} &\leq
  \norm{\nf_1}_{C^0}\,\norm{F'(v + h_2) - F'(v + h_1)}_{\Lsp^2}\nonumber\\
  &\leq C\,\norm{\nf_1}_{\Hsp^2}\,\norm{h_2 -  h_1}_{\Lsp^2}.
\end{align}

Hence, there is a constant $C$, independent of $(\kappa_j, h_j,
\nf_j)$, such that, 
\begin{multline}
  \norm{g_{\kappa_2}(h_2, \nf_2) - g_{\kappa_1}(h_1, \nf_1)}_{\Lsp^2}
  \leq
  C\,(\abs{\kappa_2^2 - \kappa_1^2}\,(1 + \norm{\nf_2}_{\Lsp^2})\\
  + \norm{\nf_2 - \nf_1}_{\Lsp^2}\,(1 + \kappa_1^2)
  +
  \norm{h_2 - h_1}_{\Lsp^2}\,(\kappa_1^2 + \norm{\nf_1}_2)).\label{eq:g-k1-k2}
\end{multline}

\eqref{eq:f-k1-k2} and~\eqref{eq:g-k1-k2} prove the component
functions are continuous.
\end{proof}

\begin{proposition}\label{lem:oplphi-compact}
  The operator $T = \opL^{-1}\circ \opPhi : [0, \kappa_{*}] \times
  \Hsp^2\times \Hsp^2 \to \Hsp^2 \times \Hsp^2$ is compact.
\end{proposition}

\begin{proof}
  By the Cauchy-Schwarz inequality and
  the standard elliptic estimates, if $u, f \in \Lsp^2$, and
  \begin{align}
    (\laplacian + \lambda)\,u = f
  \end{align}
  in the weak sense, then $u \in \Hsp^2$ and there is a constant $C$,
  independent of $(u, f)$ such that, 
  \begin{align}
    \norm{u}_{\Hsp^2} \leq C\,\norm{f}_{\Lsp^2}. 
  \end{align}

  This shows  $(\laplacian + \lambda)^{-1}: \Lsp^2
  \to \Hsp^2$ is continuous, therefore $\opL$ is also continuous and
  by lemma~\ref{lem:op-phi-cont} also $T$. If $T(\kappa, h, \nf) = (u,
  w)$, then we have, 
  \begin{align}
    (\laplacian + \lambda)\,u &= -f_{\kappa}, &
    (\laplacian + \lambda)\,w &= -g_{\kappa}. \label{eq:u-w-op-sig}
  \end{align}
  
Let $A \subset [0, \kappa_{*}]\times \Hsp^2 \times \Hsp^2$ be bounded
and closed and let $R 
> 0$ be sufficiently large such that if $(\kappa, h, \nf) \in A$, then
$\norm{h}_{\Hsp^2} + \norm{\nf}_{\Hsp^2} \leq R$. If $\set{(u_n, w_n)}$ is a
sequence in $T(A)$, such that $(u_n, w_n) = T(\kappa_n, h_n, \nf_n)$
for $(\kappa_n, h_n, \nf_n) \in A$, 
we can find a subsequence $(u_{n_k}, w_{n_k})$ such that $\kappa_{n_k}
\to \kappa' \in [0, \kappa_{*}]$ and  $(h_{n_k},
\nf_{n_k}) \to (h_{*}, \nf_{*})$ weakly in $\Hsp^1\times\Hsp^1$ and
strongly in $\Lsp^2\times\Lsp^2$. By equations \eqref{eq:f-k1-k2}
and~\eqref{eq:g-k1-k2}  and the fact that $\set{\nf_n}$ is bounded in
$\Hsp^2$, the sequence $\set{\opPhi(\kappa_{n_k}, h_{n_k}, \nf_{n_k})}$ is
Cauchy in $\Lsp^2 \times \Lsp^2$. By~\eqref{eq:u-w-op-sig},
$\set{(u_{n_k}, w_{n_k})}$ is Cauchy in $\Hsp^2\times \Hsp^2$, therefore 
convergent in $\overline{T(A)}$.
\end{proof}

By this proposition, for any bounded open set $\Omega \subset
\Hsp^2\times \Hsp^2$ such that $(I -
\opL^{-1}\circ\opPhi_{\kappa})^{-1}(0)\not\in \del\Omega$, the  degree 
\begin{align}
  \mathrm{deg}(I - \opL^{-1}\circ\opPhi_{\kappa}, \Omega, 0),
\end{align}
is well defined and a homotopical invariant for $\kappa$ restricted to
any subinterval $[a, b] \subset [0, \kappa_{*}]$ such that,
\begin{align}
  (I -
\opL^{-1}\circ\opPhi_{\kappa})^{-1}(0)\not\in \del\Omega\, \qquad
\forall \kappa \in [a, b].
\end{align}

Notice $(h, \nf) \in (I
- \opL^{-1}\circ \opPhi_{\kappa})^{-1}(0)$ if and only if it is a solution
to the governing elliptic problem, \eqref{eq:h}, \eqref{eq:N'}.

\begin{lemma}\label{lem:oplphi-cont-k}
For any ball $B \subset \Hsp^2 \times \Hsp^2$ and for any $\epsilon > 0$, 
there is a $\delta > 0$ such that for any $(h, N) \in B$,  
$|\kappa_2 - \kappa_1| < \delta$ implies  $|T(\kappa_2, h, N) - T(\kappa_1, h, 
N)| < \epsilon$.
\end{lemma}

It is said that the operator $T$ is continuous in $\kappa$ uniformly with 
respect to $(h, N)$ in balls in $\Hsp^2\times\Hsp^2$

\begin{proof}
 In order to prove the lemma, it is sufficient to consider balls centred at 
 the origin. Let $R > 0$, and $(h, N) \in B_R(0)$. If $(u, w) = 
 \opL^{-1}\circ
 \Phi(\kappa, h, \nf)$, 
 then $u$ and $w$ are solutions to~\eqref{eq:u-w-op-sig}.
 Let $(u_j, w_j) = \opL^{-1}\circ\opPhi(\kappa_j, h, \nf)$, $j =
 1, 2$,  then $u_1 = u_2$
 because $f_{\kappa}(h, \nf)$ is  
 independent of $\kappa$. For $w_j$, by~\eqref{eq:g-k1-k2} there is a
 constant $C = C(R)$ independent of $(h, \nf)$, such that,
 \begin{align}
   \norm{g_{\kappa_2}(h,\nf) - g_{\kappa_1}(h, \nf)}_{\Lsp^2} \leq
   C\,\abs{\kappa_2 - \kappa_1}^2.
 \end{align}
 
 By~\eqref{eq:u-w-op-sig}, 
 \begin{align}
    \norm{\grad(w_2 - w_1)}^2_{\Lsp^2} + \lambda\,\norm{w_2 - w_1}^2_{\Lsp^2}
    &= -\lproduct{(w_2 - w_1), g_{\kappa_2}(h, N) - g_{\kappa_1}(h, 
    N)}\nonumber\\
    &\leq \norm{w_2 - w_1}_{\Lsp^2}\cdot
    \norm{g_{\kappa_2}(h, \nf) - g_{\kappa_1}(h, \nf)}_{\Lsp^2}
 \end{align}
 
Whence, there exists another constant, independent of $(h, N)$, such 
 that,
 \begin{align}
     \norm{w_2 - w_1}_{\Lsp^2} \leq C\,\norm{g_{\kappa_2}(h, \nf) -
       g_{\kappa_1}(h, \nf)}_{\Lsp^2}.
 \end{align}
 
 By Schauder's estimates,
 \begin{align}
     \norm{w_2 - w_1}_{\Hsp^2} 
     &\leq C\,\brk(
     \norm{\laplacian(w_2 - w_1)}_{\Lsp^2} + \norm{w_2 - 
     w_1}_{\Lsp^2})\nonumber\\
     &\leq C\brk(
     \norm{g_{\kappa_2}(h, N) - g_{\kappa_1}(h, N)}_{\Lsp^2}
     + (\lambda + 1)\,\norm{w_2 - w_1}_{\Lsp^2})\nonumber\\
     &\leq C\,\norm{g_{\kappa_2}(h,\nf) - g_{\kappa_1}(h, 
     \nf)}_{\Lsp^2}\nonumber\\
      &\leq C\abs{\kappa_2^2 - \kappa_1^2}.
 \end{align}
 Therefore, $\norm{(u_2, w_2) - (u_1, w_1)}_{\Hsp^2\times\Hsp^2} \to 0$ 
 uniformly as $\kappa_2 \to \kappa_1$.
\end{proof}

If $0 < \kappa_0 < \kappa_*(D)$, we know by 
lemma~\ref{lem:kappa-bounded} that there exists an $R > 0$ such that
for any $\kappa \in [\kappa_0, \kappa_*(D)]$ the 
solution to equations \eqref{eq:h}, \eqref{eq:N'} is in the
interior of the disk $\disk(0, R) \subset \Hsp^2\times \Hsp^2$. Since for
any $\epsilon > 0$ there is no solution to the equations for
$\kappa_*(D) + \epsilon$, by the homotopy invariance of the degree, we
conclude, 
\begin{equation}
  \lsdeg(\opI{} - \opL^{-1}\circ\opPhi_{\kappa}, \disk(0, R), 0) = 0, \qquad
  \kappa \in [\kappa_0, \kappa_*(D)].
\end{equation}

By proposition~\ref{thm:small-def} we know there is a
neighbourhood $U$ of $(h_0, 
0)$ such that for $\kappa$ small enough, there is exactly one solution
$(h_\kappa, N'_\kappa)$ to equations \eqref{eq:h} and \eqref{eq:N'} in
$U$ and this solution varies smoothly in $\Hsp^2\times\Hsp^2$ with
$\kappa$.

\begin{lemma}\label{lem:index-T0}
  $\abs{\mathrm{ind}(I - T(0,\cdot,\cdot), (h_0, 0), 0)} = 1$. 
\end{lemma}

\begin{proof}
  At $\kappa = 0$, the derivative of $\opPhi_0$ at $(h_0, 0)$ has components,
  \begin{align}
    f_0'(h_0, 0)\cdot(\delta h, \delta N) &= 2\,(\delta N + F'(v +
    h_0)\,\delta h) - \lambda\, \delta h,\\
    g_0'(h_0, 0)\cdot(\delta h, \delta N) &= 2\,F'(v + h_0)\,\delta N
    - \lambda \,\delta N.
  \end{align}

  The operator $\opL^{-1}$ is linear, hence, the derivative of $T_0 =
  T(0, \cdot, \cdot)$ at $(h_0, 0)$ is,
  \begin{align}
    T_0'(h_0, 0) = \opL^{-1}\circ\,\opPhi_0'(h_0, 0).
  \end{align}

  If $(\delta h, \delta N) \in \mathrm{Ker}(I - T_0'(h_0, 0))$,
  then $(\delta h, \delta N)$ is the solution to the elliptic problem,
  \begin{align}
    -\laplacian\,\delta h &= 2\,\delta N + 2\,F'(v +
    h_0)\,\delta h,\\
    -\laplacian\,\delta N &= 2\,F'(v + h_0)\,\delta N.
  \end{align}

  By lemma~\ref{lem:h2-l2-iso}, the operator $\laplacian + 2\,F'(v +
  h_0): \Hsp^2 \to \Lsp^2$ is an isomorphism. Therefore, $\delta h =
  \delta N = 0$. By theorem~\ref{thm:leray-schauder-index},
  \begin{align}
    \mathrm{ind}(I - T(0,\cdot,\cdot), (h_0, 0), 0) = \pm 1,
  \end{align}
  where the sign depends on the multiplicities of the eigenvalues
  $\lambda > 1$ of $I - T'_0(h_0, 0)$. 
\end{proof}

\begin{proposition}\label{prop:existence-cs-sol-infty}
  There is a $\kappa_0 > 0$ such that, if $0 < \kappa < \kappa_0$,
  equations~\eqref{eq:h},\eqref{eq:N'} have exactly two continuous
  families of solutions. As 
  $\kappa \to 0$ one of the families is convergent to $(h_0, 0)$, the
  solution to of the regularised Taube's equation and the second
  family is such that $(h_{\kappa} - \av(h_{\kappa}), \nf'_{\kappa})
  \to (0, \alpha_{\pm})$ and $\av(h_{\kappa}) \to \mp\infty$, where
  $\alpha_{\mp}$ and the sign of the divergence depend on the sign of
  $k_+ - k_-$ as in lemma~\ref{lem:hn-N'n-div}. Moreover, for any $R >
  0$, there is a $\kappa' > 0$ such that if $0 < \kappa < \kappa'$,
  then $\norm{(h_{\kappa}, \nf_{\kappa})}_{\Hsp^2\times \Hsp^2} > R$ for
  at least one pair of solutions to the equations.
\end{proposition}

\begin{proof}
    Proposition~\ref{lem:oplphi-compact}, 
    and Lemma~\ref{lem:oplphi-cont-k} show $T$ satisfies the hypotesis of 
    theorem~\ref{thm:ls-deg-continuum}. By proposition~\ref{thm:small-def}, 
    there exists $\kappa_0 \in (0,
  \kappa_{*}]$ and an open bounded set $U \subset
  \Hsp^2\times\Hsp^2$ of $(h_0, 0)$ such that the restriction of
  $(h_{\kappa}, \nf_{\kappa})$ to $U$ varies smoothly with $\kappa \in
  [0, \kappa_0)$. By lemma~\ref{lem:index-T0}, if $\mathrm{diam}(U)$
  is small, $\mathrm{deg}(I - T(0,
  \cdot, \cdot), U, 0) \neq 0$ and there is no other solution for
  $\kappa = 0$ in $\overline U$. By
  theorem~\ref{thm:ls-deg-continuum}, there is a connected closed set
  $\mathcal{C} \subset [0, \kappa_{*}] \times \Hsp^2 \times \Hsp^2$,
  such that $(0, h_0, 0) \in \mathcal{C}$ and either $\mathcal{C}$ is
  unbounded or $\mathcal{C} \cap (\set{0} \times (\Hsp^2\times \Hsp^2
  \setminus \overline U)) \neq \emptyset$. Since for $\kappa = 0$
  there is only one solution to equations~\eqref{eq:h}~\eqref{eq:N'},
  we rule out the second possibility. As $\kappa_{*} < \infty$, by
  lemma~\ref{lem:kappa-bounded} there is a second family $(h_n,
  \nf_n)$ of solutions to the equations, such that $\kappa_n \to 0$
  and $\norm{(h_{\kappa},
    \nf_{\kappa})}_{\Hsp^2\times \Hsp^2} \to \infty$. By
  proposition~\ref{prop:cs-lim-points}, $(0, \alpha_{\pm})$ is a limit
  point. In order to prove the last claim, assume towards a
  contradiction, the existence of $R > 0$ and a sequence $\kappa_n \to
  0$ such that $\norm{(h_{\kappa_n} - h_0, \nf_{\kappa_n})}_{\Hsp^2
    \times \Hsp^2} \leq R$ for all solutions with parameter
  $\kappa_n$. By
  lemma~\ref{lem:hn-N'n-conv}, the set of solutions
  $\set{(h_{\kappa}, \nf_{\kappa}) \mid \norm{(h_{\kappa} - h_0,
    \nf_{\kappa})} = R}$ can not accumulate at $\kappa = 0$. Let
$\kappa_R > 0$ be such that if $\norm{(h_{\kappa} - h_0,
  \nf_{\kappa})}_{\Hsp^2 \times \Hsp^2} = R$, then $\kappa >
\kappa_R$ and let us choose $n$ such that $\kappa_n <
\kappa_R$. Consider the relatively open set, 
\begin{align}
 V = \set{(\kappa, h, \nf) \in [0, \kappa_{*}] \times \Hsp^2 \times
   \Hsp^2 \mid \norm{(h - h_0, \nf)}_{\Hsp^2 \times \Hsp^2} > R, \; 0
   <\kappa  < \kappa_n} \cap \mathcal{C}. 
\end{align}

$V$ is not empty because there is a divergent sequence in
$\mathcal{C}$ with deformation parameter converging to $0$, we claim
$V$ is  also closed, because if 
$\set{(\mu_n, h_n, \nf_n)} \subset V$ has an accumulation point
$(\mu_{*}, h_{*}, \nf_{*})$, then $(\mu_{*}, h_{*}, \nf_{*}) \in \mathcal{C}$
because this set is closed. At the same time, $\norm{(h_{*} - h_0,
  \nf_{*})}_{\Hsp^2 \times \Hsp^2} \geq R$ and $0 \leq \mu_{*} \leq
\kappa_n$. Since $\norm{(h_{*} - h_0,
  \nf_{*})}_{\Hsp^2 \times \Hsp^2} = R$ implies $\mu_{*} > \kappa_R$,
we can discard this case. If $\mu_{*} = 0$ then $(h_{*}, \nf_{*}) =
(h_0, 0)$ which is 
impossible. If $\mu_{*} = \kappa_n$, then we also have $\norm{(h_{*} -
  h_0, \nf_{*})}_{\Hsp^2 \times \Hsp^2} \leq R$, which is
absurd. Therefore $(\mu_{*}, h_{*}, \nf_{*}) \in V$ and this set is
open and closed. Since $\mathcal{C}$ is connected, $V =
\mathcal{C}$. A contradiction.
\end{proof}

In view of this proposition, we can define $\kappa_{\mathcal{C}}(D)$
as the supremum, 
\begin{align}
  \kappa_{\mathcal{C}}(D) = \sup \set{\kappa > 0 \mid (\kappa,
    h_{\kappa}, \nf_{\kappa}) \in \mathcal{C}}. 
\end{align}

Proposition~\ref{prop:existence-cs-sol-infty} shows
$\kappa_{\mathcal{C}}$ is actually a maximum, moreover, since the set
of solutions for which $(h_{\kappa}, \nf_{\kappa}) \in U$ is contained
in $\mathcal{C}$, the same lower bound for $\kappa_{*}$ is also valid
for $\kappa_C$. We summarise the results of this section in the
following theorem.

\begin{theorem}\label{thm:limit-kappa-0-hf}
  Let $(\kappa, \hf_{\kappa}, A_{\kappa}, \nf_{\kappa})$ be a
  solution of the Bogomolny equations of the BPS soliton equations with
  Chern-Simons deformation constant $\kappa$. Assume $k_+ - k_- \neq
  0$ and Bradlow's bound is satisfied, then the following properties hold,
  \begin{enumerate}
  \item $\kappa$ is bounded by a constant independent of the position
    of the divisors as given in proposition~\ref{prop:kappa-*-ub},

  \item If $\kappa$ is small, for each divisor there are at least two 
  gauge inequivalent families of solutions to the Bogomolny equations. 
  
  \item There are an $\epsilon > 0$ and $\kappa_0 > 0$ such that, if
    $\abs \kappa < \kappa_0$ there is exactly one gauge equivalence class 
    $(\kappa, \hf_{\kappa}, A_{\kappa}, \nf_{\kappa})$ of solutions
    to the Bogomolny equations, such that
    \begin{align}
      \norm{h_{\kappa} - h_0}_{C^1} + \norm{\nf_{\kappa}}_{C^1} < \epsilon.
    \end{align}

    Outside any closed neighbourhood $\overline U$ of the core set
    $\vset \cup \avset$, this family of solutions varies smoothly with
    $\kappa$. 
  \item For any  neighbourhood $U$ of the core set,
    we have the following property: For any $\epsilon > 0$, there is a
    $\kappa' > 0$, such that if $\abs{\kappa} < \kappa'$, there is a
    solution $(\kappa, \hf_{\kappa}, A_{\kappa}, \nf_{\kappa})$ to the 
    Bogomolny equations, such that
    \begin{align}
      \norm{\hf_3^{\kappa} \mp 1}_{C^1(\surface\setminus U)} +
      \norm{\nf_{\kappa} - \alpha_{\pm}}_{C^1} < \epsilon,
    \end{align}
    where the signs chosen and the constant $\alpha_{\pm}$ depend on
    the sign of the difference $k_+ -  k_-$  as defined on
    lemma~\ref{lem:hn-N'n-div}.     
  \end{enumerate}
\end{theorem}

\let \ux \undefined
\let \ox \undefined

\let\Ik\undefined
\let\Aset\undefined
\let\avh\undefined
\let\Kset\undefined
\let\opPhi\undefined
\let\opL\undefined
\let\opI\undefined
\let\lsdeg\undefined
\let\av\undefined

%
%

\subsection{Symmetric deformations on the sphere}
\label{sec:symm-deform-sphere}

In this section we study the deformation constant in the sphere. It is
known that in the Euclidean plane, there is a solution to the
elliptic problem for any $\kappa \in \reals$. Hence, the existence of
an upper bound for $\kappa_{*}$ in a compact surface is a nontrivial
task. We will suppose all the vortices are located at the north pole
of the domain sphere, and all the antivortices at the south pole. We
choose trivialisations $\hf_{\pm}: U_{\pm} \to \sphere$ at
$U_{\pm} = \sphere \setminus \set{(0, 0, \mp 1)}$, which
stereographically project from the south or north pole respectively,
as $\varphi_{\pm}:U_{\pm} \to \cpx$. This projections are related by a
gauge transformation, which by spherical symmetry is,
\begin{equation}
  \varphi_+ = \frac{e^{in\theta}}{\varphi_-}, \qquad x \in U_+\cap U_-. 
\end{equation}

Whereas for the connection, if it is represented locally by $\gp_{\pm} \in
\Omega^1(U_{\pm})$,
\begin{equation}
  \gp_+ = \gp_- + n d\theta, \qquad x \in U_+ \cap U_-. 
\end{equation}

Stereographic coordinates in the domain sphere will be denoted
accordingly $x_{\pm} = r_{\pm} e^{i\theta_{\pm}}$. Hence, $x_+x_- = 1$
in $U_+\cap U_-$ and $\theta_+ =-\theta_-$. We choose the ansatz,
\begin{align}
  \varphi_{\pm} &= f_{\pm}(r_{\pm})e^{i k_{\pm}\theta_{\pm}} &
  \gp_{\pm} &= \gp_{\pm}(r_{\pm})\,d\theta_{\pm},
\end{align}
which is justified by the equivariant rotational symmetry of the
problem. Compatibility of the fields then requires $n = k_+ -
k_-$. The Bogomolny equations reduce to a system of ODES which we aim to
integrate,
\begin{align}
f_{\pm}' &= \frac{1}{r} (k_{\pm} \mp \gp_{\pm})\, f_{\pm},\\
\gp_{\pm}' &= r \Omega(r) B_{\pm},\\
N_{\pm}'' &= -\Omega(r) \brk(\kappa B_{\pm} - \frac{4 f_{\pm}^2N_{\pm}}{(1 +
f_{\pm}^2)^2}) - \frac{1}{r}N'_{\pm},
\end{align}
where,
\begin{align}
\Omega(r) &= \frac{4R^2}{(1 + r^2)^2}, \\
B_{\pm} &= - \brk(\kappa N_{\pm} + \tau \pm 1 \mp \frac{2}{1 + f_{\pm}^2}).
\end{align}

We solved the Bogomolny equations in the punctured disk
$\disk_1(0)\setminus\set{0}$ adding the compatibility conditions,
\begin{equation}
\begin{aligned}
f_+(1)f_-(1) = 1, &\qquad&
\gp_+(1) + \gp_-(1) = n, \\
N_+(1) = N_-(1), &\qquad&
N_+'(1) = - N_-'(1),
\end{aligned}
\end{equation}
together with the lowest order approximation to the fields at $r = 0$,
\begin{align}
f_{\pm} &= q_{\pm}r^{k_{\pm}} + \order\brk(r^{k_{\pm} + 1}),
\\
B_{\pm} &=
\begin{cases}
  -\brk(\kappa p_{\pm} + \tau \pm 1 \mp \frac{2}{1 + q_{\pm}^2}) + \order(r), &
  k_{\pm} = 0,\\
  -(\kappa p_{\pm} + \tau \pm 1 \mp 2) + \order(r), &
  k_{\pm} \neq 0,\\
\end{cases}
\\
\gp_{\pm}  &= 2 B_{\pm}R^2r^2 + \order(r^3),
\\
N_{\pm} &=
\begin{cases}
  p_{\pm} + \brk(
  -\kappa B_{\pm} + \frac{4q_{\pm}^2p_{\pm}}{(1 + q_{\pm}^2)^2}
  )\,R^2r^2 + \order(r^3),
  & k_{\pm} = 0, \\
  p_{\pm} - 
  \kappa B_{\pm}
  \, R^2r^2 + \order(r^3),
  & k_{\pm} \neq 0.
\end{cases}
\end{align}

To find the initial stable solution, we used the shooting method in
the interval $[\delta, 1]$ for a small value $\delta > 0$. Given
initial conditions $Z = 
(q_+, q_-, p_+, p_-)$ for the parameters, we solved the Bogomolny 
equations and defined a map $M: \reals^4 \to \reals^4$,
\begin{align}
Z \mapsto (f_+(1)f_-(1) - 1, \gp_+(1) + \gp_-(1) - n,
N_+(1) - N_-(1), N'_+(1) + N'_-(1)),
\end{align}
whose zero determines suitable initial conditions for a solution to the 
Bogomolny equations compatible at the boundary of the disk. Next, we
applied the pseudo-arclength continuation method, as described in
\cite{flood2018chern}. Given initial data $(\kappa_0, Z_0) \in \reals^5$, we
sought a nearby point $(\kappa, Z)$ such that,
\begin{equation}
  \dot Z_0 \cdot (Z - Z_0) + \dot \kappa_0\,(\kappa - \kappa_0) =
  \delta s,
\end{equation}
for a small positive constant $\delta s$. We restricted ourselves to
positive $\kappa$ and solved the Bogomolny 
equation in the vortex-antivortex case and the case $k_+ = 2$, $k_- =
0$. We solved both cases on a sphere of radius $2$. The results can be
seen on Figure~\ref{fig:kappa}. We found that for the
vortex-antivortex case, the data suggests  $\kappa$ is unbounded. This
would be the case if for all the solutions, the function $h = \log
f^2$ have bounded average. If rotationally symmetric solutions are
unique, invariance of Taube's equation under isometries of the sphere
implies the average is actually zero. Therefore, we conjecture
$\kappa$ unbounded for this configuration of cores' on the sphere. On the
 $(2, 0)$ case, arclength continuation started growing fairly quickly
 until it reached a maximum value and started 
decreasing towards zero as expected.  In Figure~\ref{fig:sols} on the
right column can be seen the two limiting solutions of
theorem~\ref{thm:limit-kappa-0-hf}. At $\kappa = 0$  we obtained the
solution to the Taubes equation as expected and a limiting solution, as
the averages of $h$ diverged towards infinity, the gauge invariant
component $\hf_3$ of the Higgs field $\hf$ started converging to
constant $1$, in other words, $\hf$ converged to the north pole
section, while $\kappa\nf$ converged to the expected limit
\begin{align}
  \alpha_- = \frac{3}{4}.
\end{align}



\begin{figure}
  \centering
  \includegraphics[width=.85\textwidth]{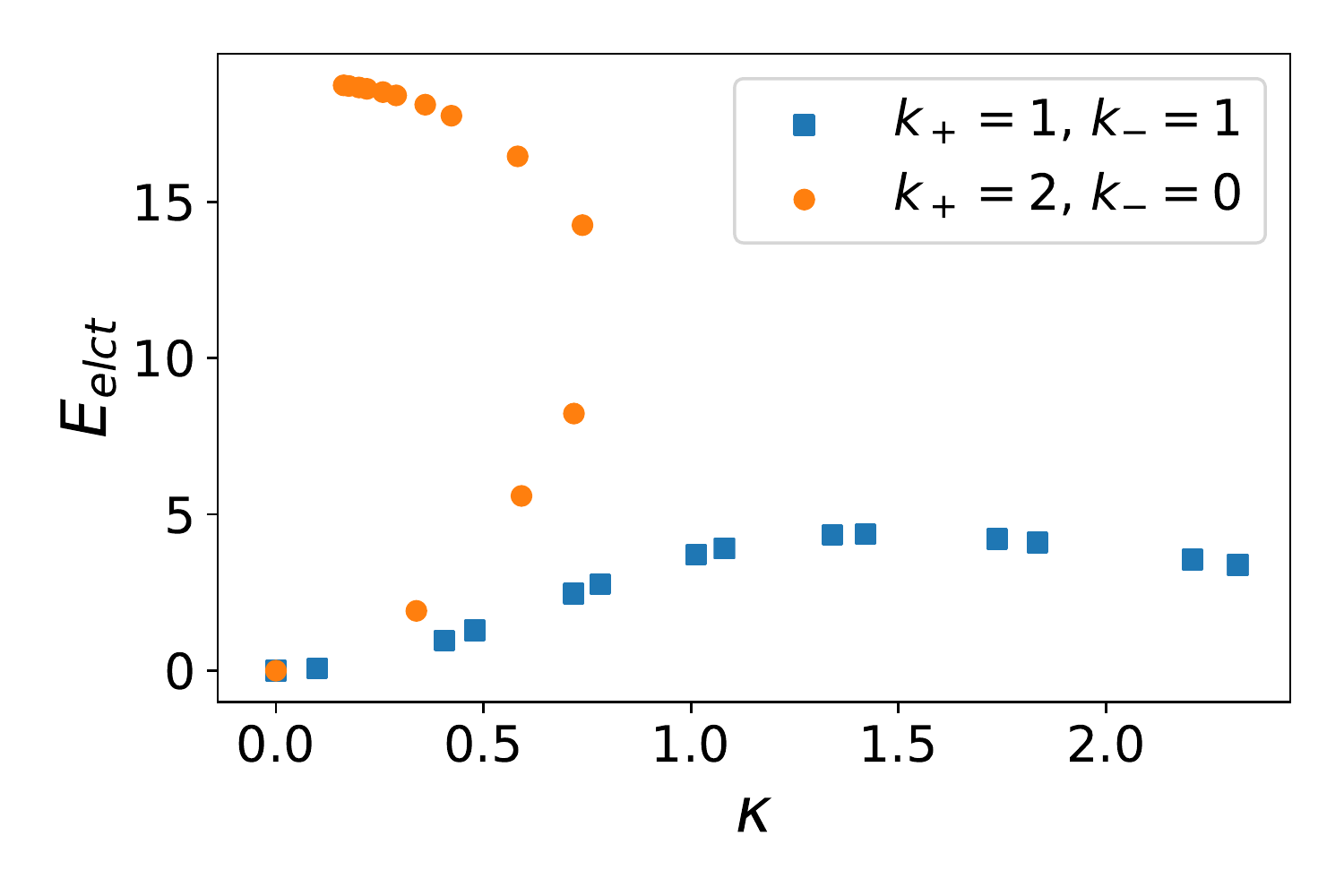}
  \caption{Comparison of the electrostatic energy in the balanced and
    unbalanced cases. The existence of two types of solutions if
    $k_+ \neq k_-$ is evident from the graph, while the energy also suggests
    uniqueness of the solution $(h_{\kappa}, N_{\kappa})$ for each
    $\kappa$ in the balanced case.}
  \label{fig:kappa}
\end{figure}

\begin{figure}
  \centering
  \begin{tabular}{cc}
    \includegraphics[width=2.4in]{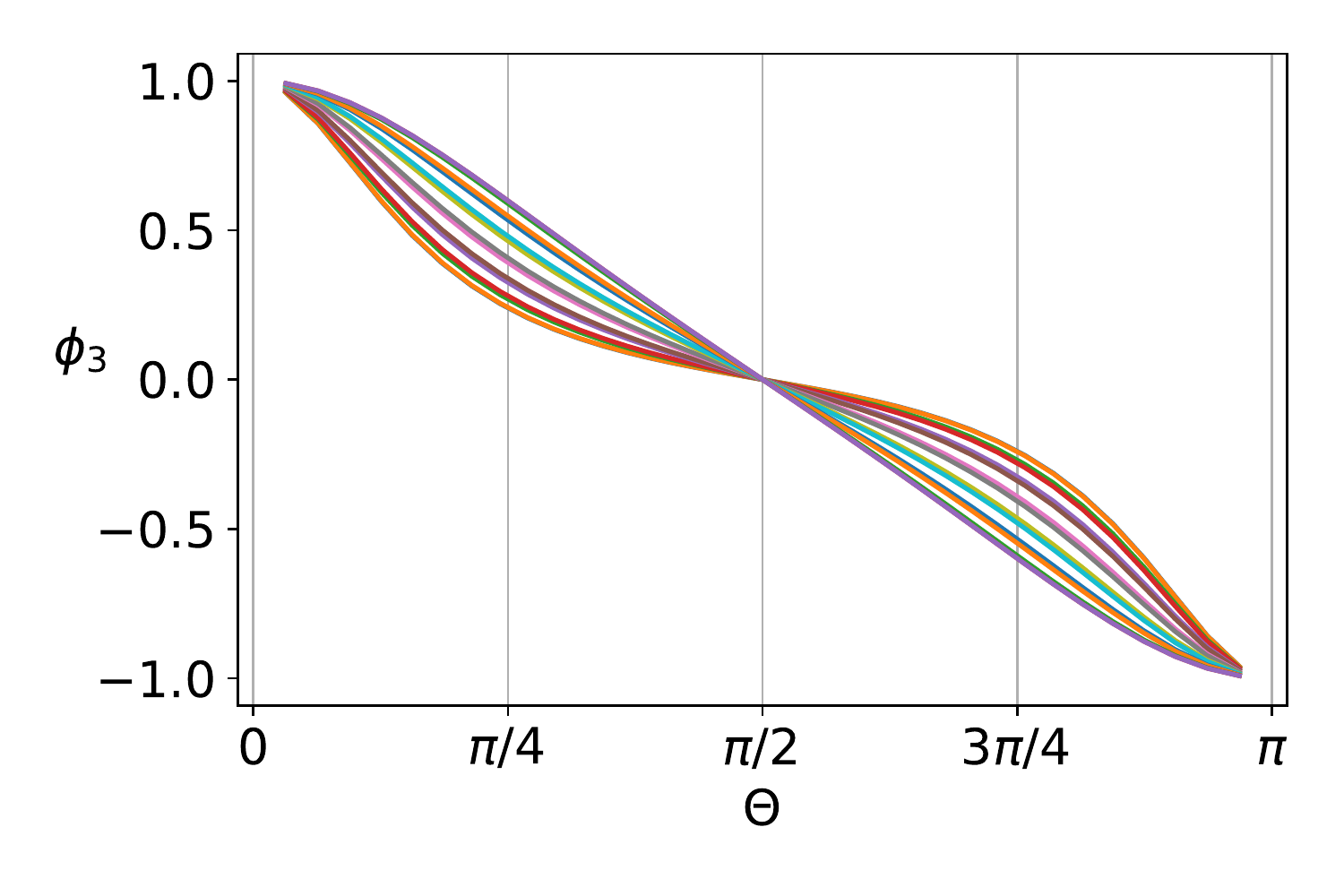} &
    \includegraphics[width=2.4in]{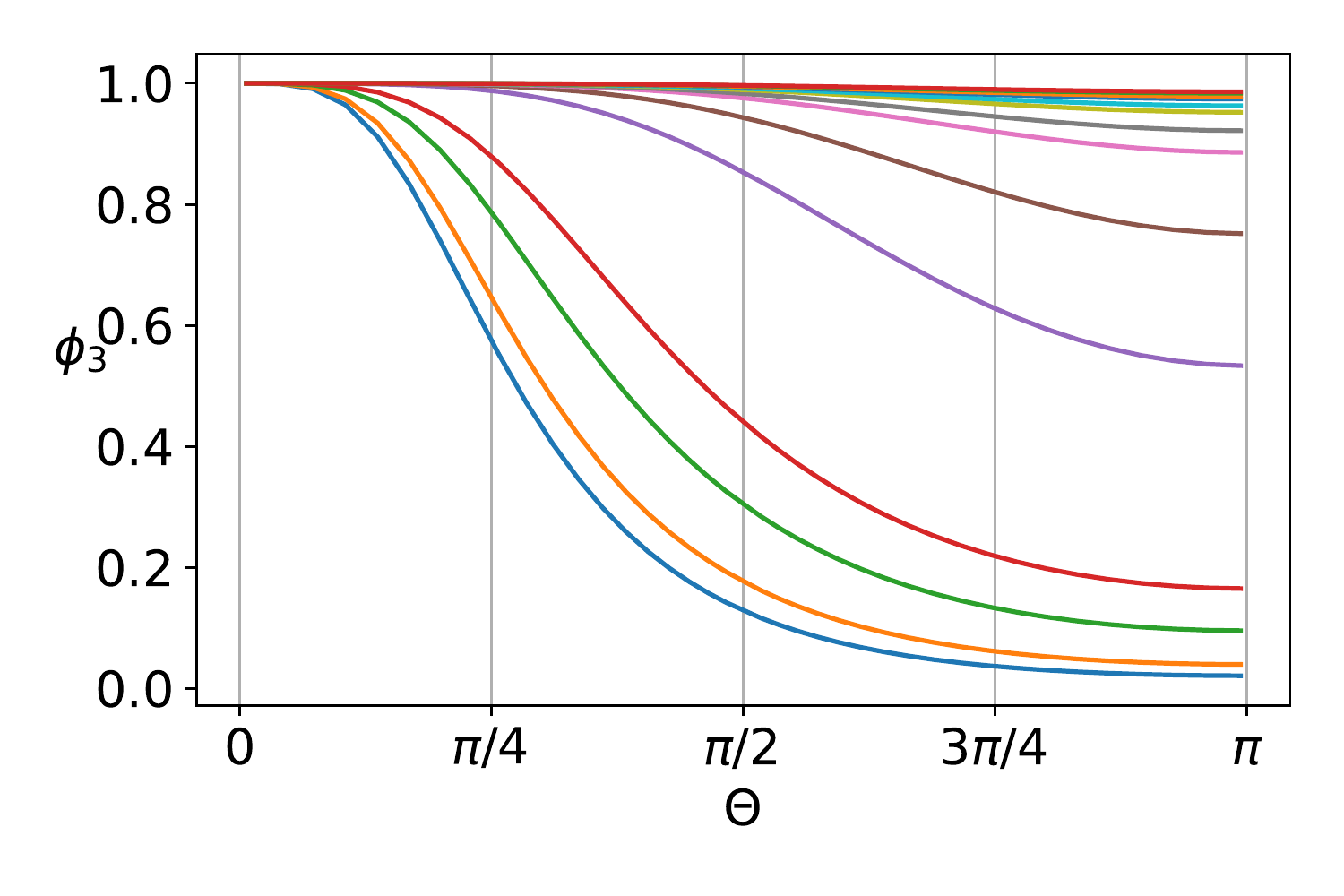} \\
    \includegraphics[width=2.4in]{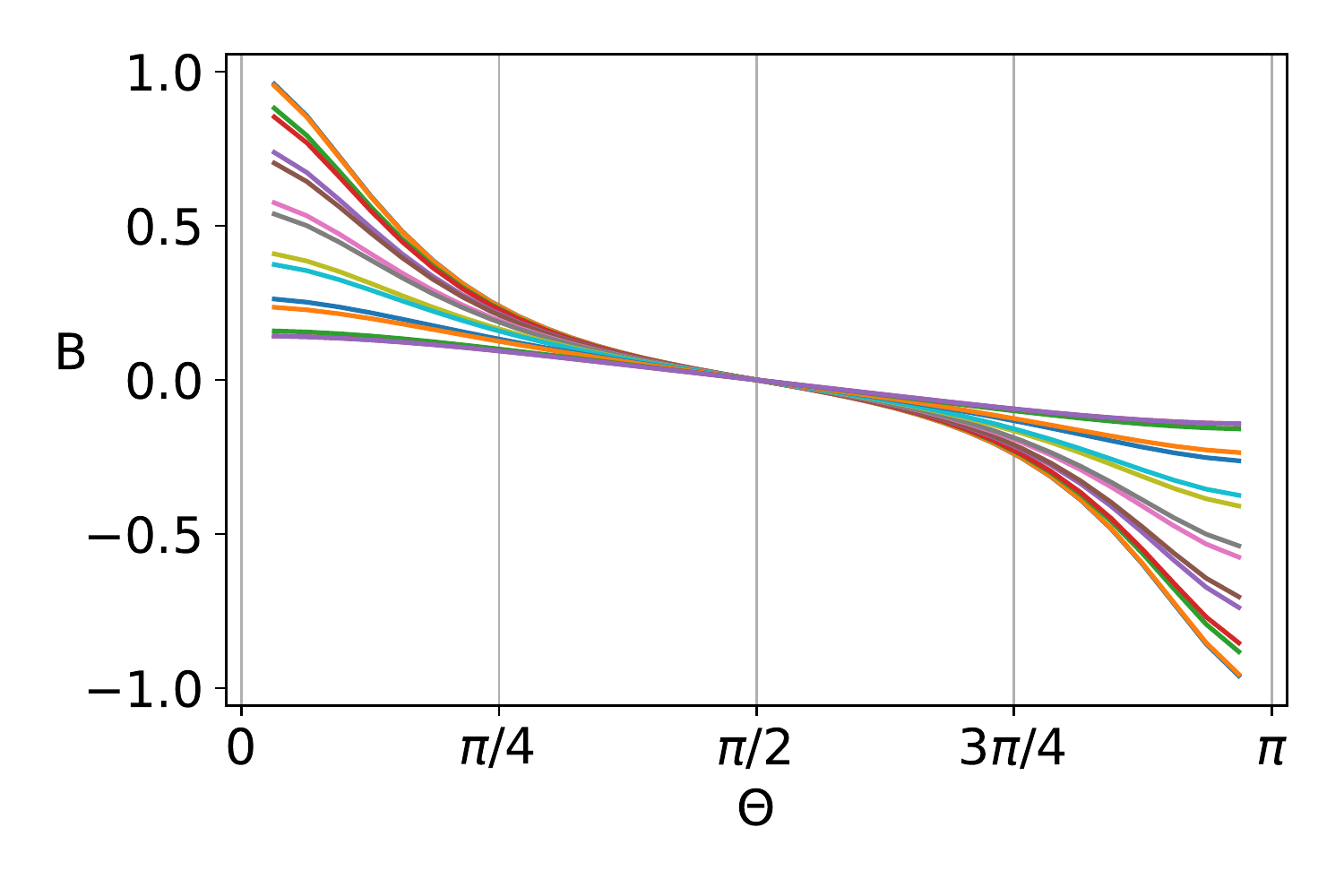}    &
    \includegraphics[width=2.4in]{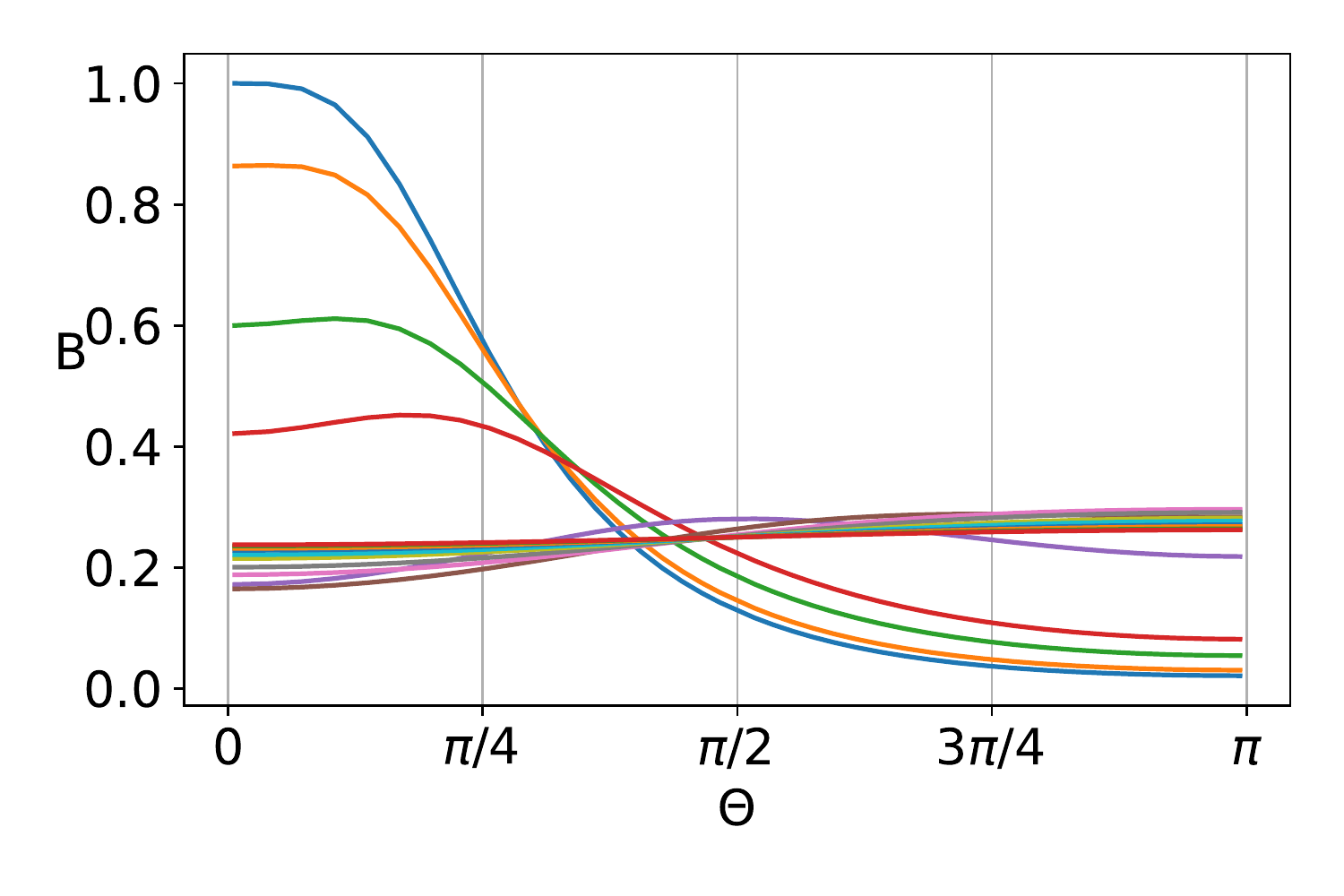}    \\
    \includegraphics[width=2.4in]{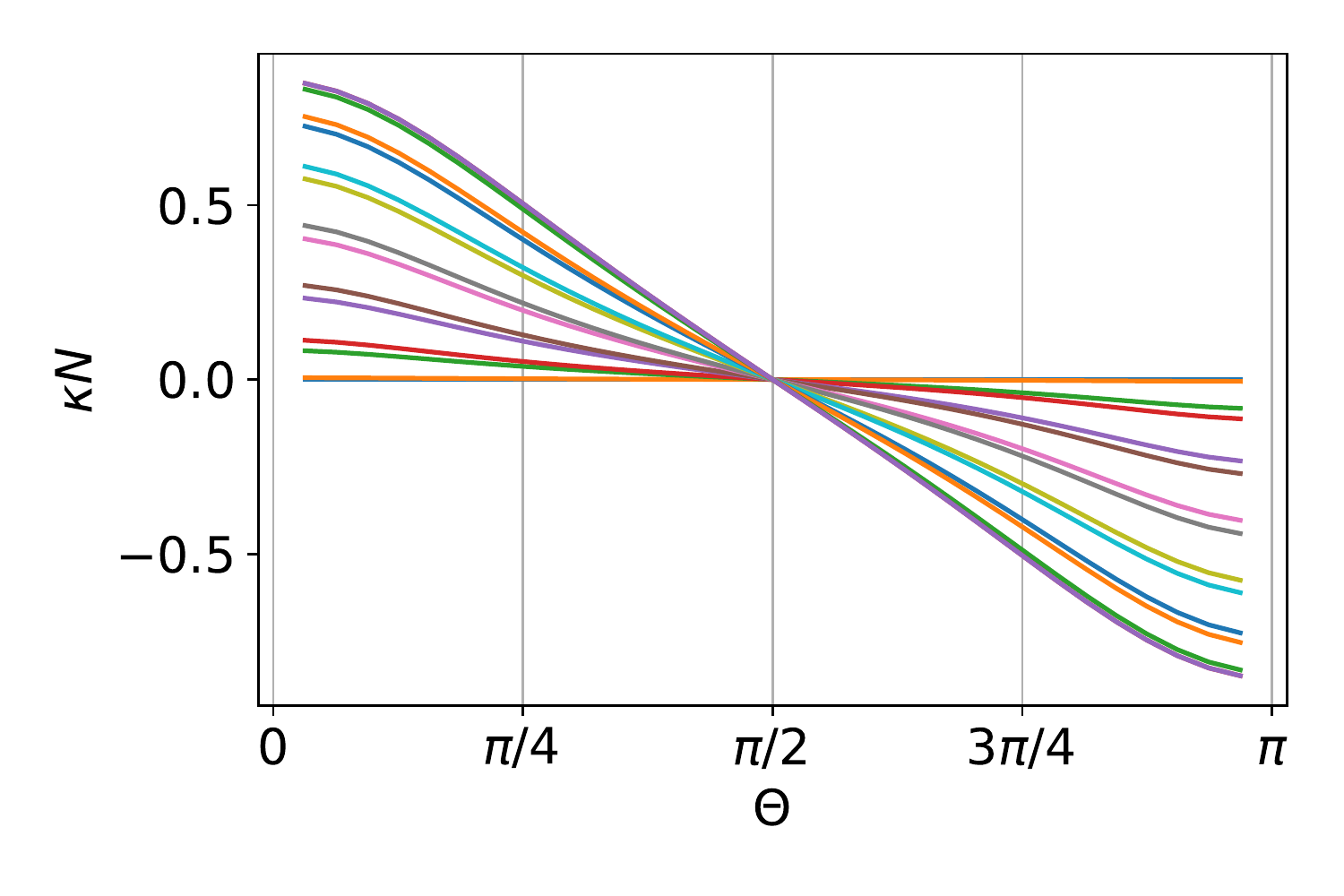}  &
    \includegraphics[width=2.4in]{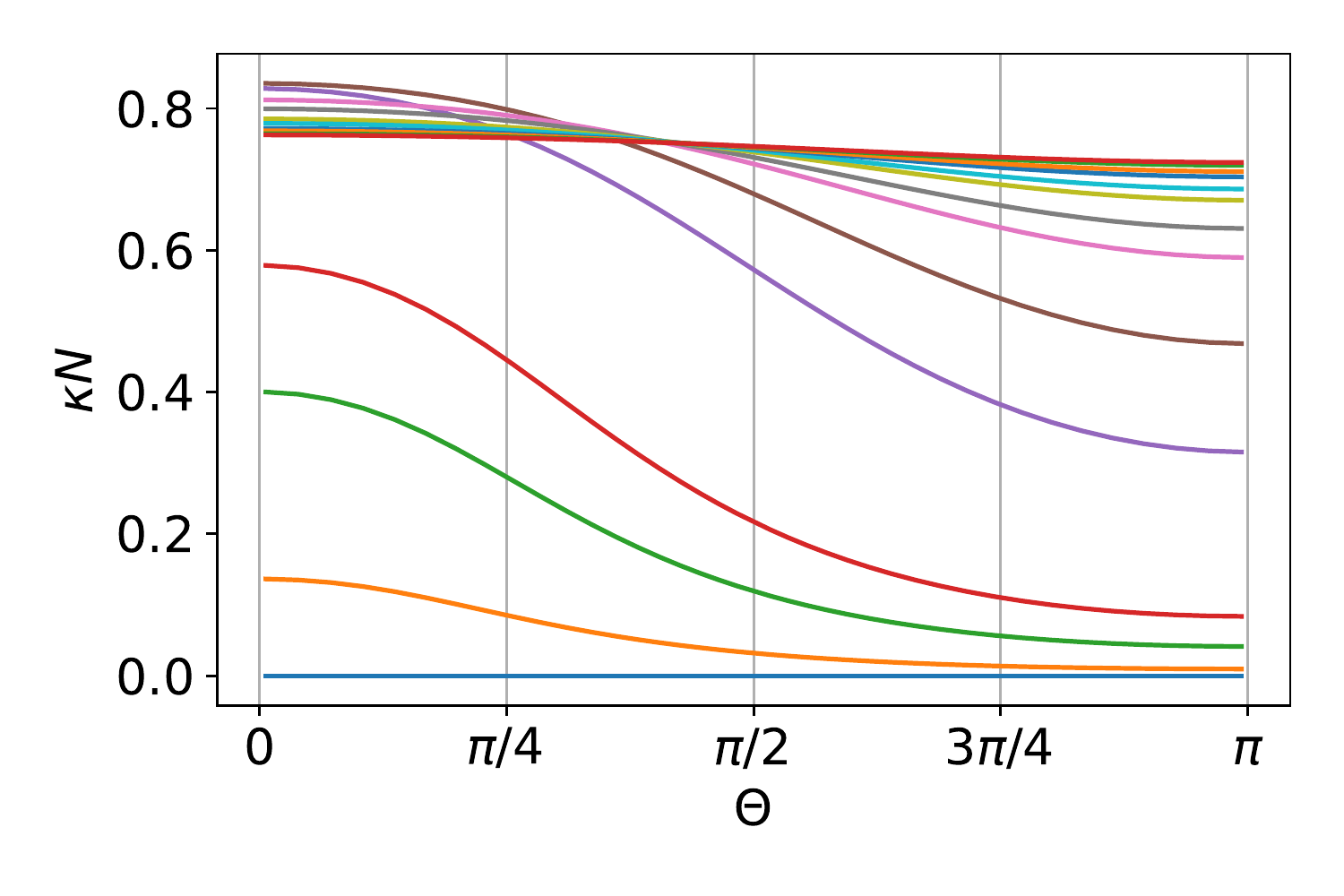}  \\
    \includegraphics[width=2.4in]{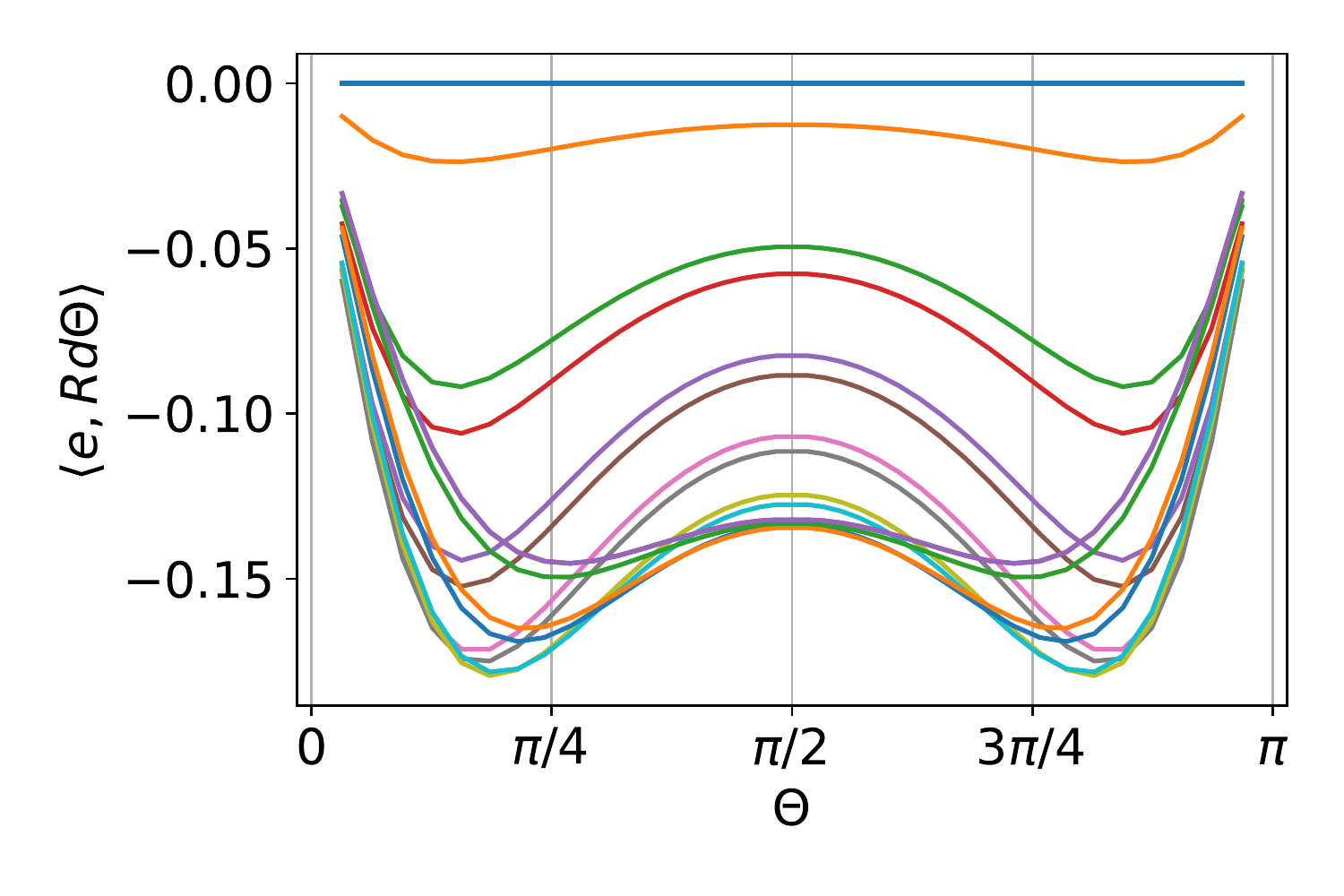} &
    \includegraphics[width=2.4in]{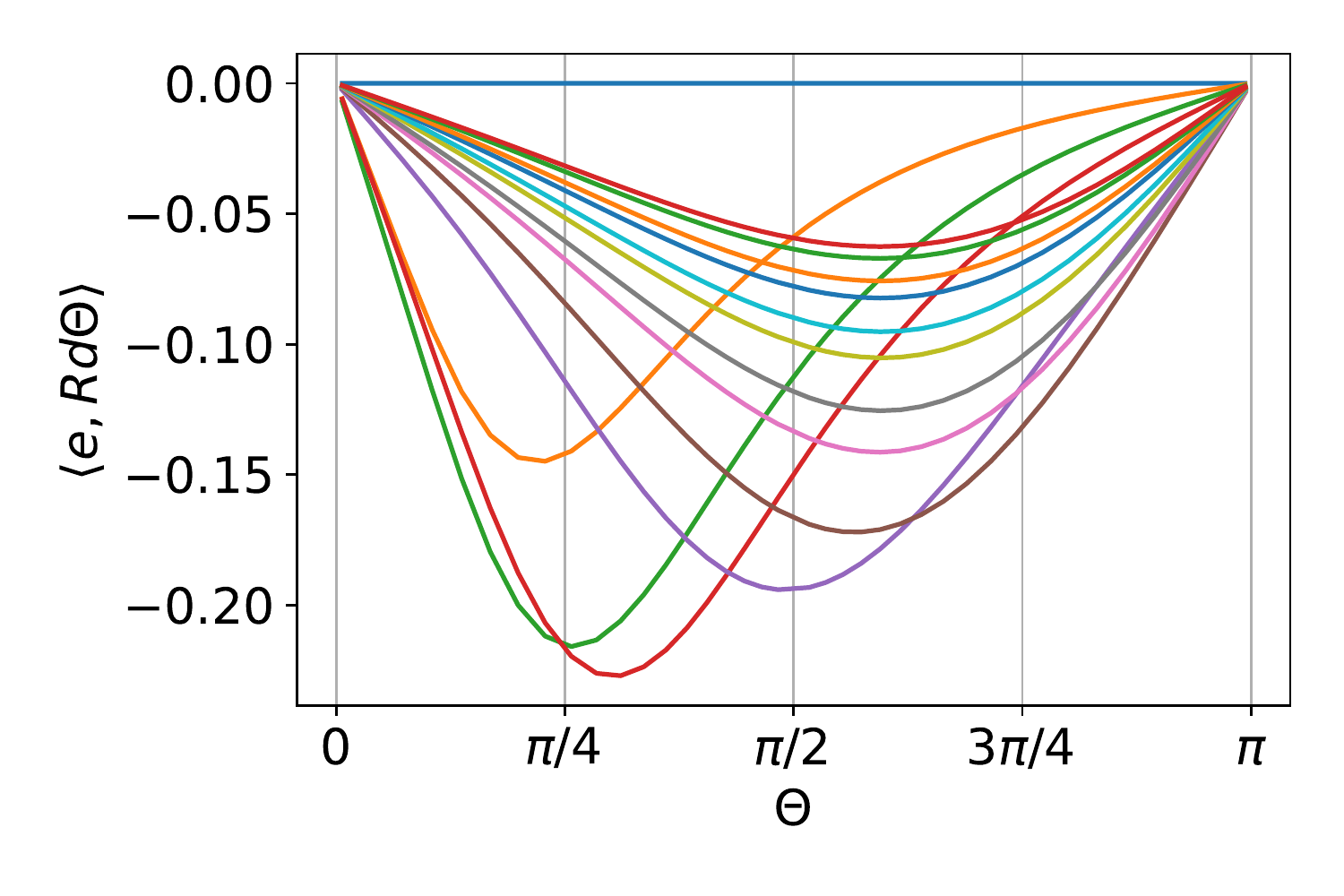}
  \end{tabular}
    
  \caption{Snapshots of solutions to the Bogomolny equations on the
    sphere along the declination angle $\Theta$. The radius was set to
  $R = 2$. \textbf{Left.} Vortex-antivortex case. \textbf{Right.} Two
  vortices at north pole and no antivortices. The asymmetry parameter
  was set to $\tau = 0$.}
  \label{fig:sols}
\end{figure}


\let         \nf          \undefined   
\let         \np          \undefined   
\let         \hf          \undefined   
\let         \gp          \undefined   
\let         \Diff        \undefined   
\let         \modulikk    \undefined   
\let         \Hsp         \undefined   
\let         \Wsp         \undefined   
\let         \Top         \undefined   
\let         \Lop         \undefined   
\let         \vol         \undefined   


\section{Dynamics of the moduli space of Ginzburg-Landau vortices  with a
  Chern-Simons term}
\label{ch:cs-loc}




\def \higgsField              {\phi}                  
\def \gaugePotential          {a}                     
\def \electricField           {e}                     
\def \magneticField           {B}                     
\def \spacetimeGaugePotential {A}                     
\def \curvatureField          {F}                     
\def \neutralField            {N}                     
\def \fieldConfSpace          {\mathcal{A}}           
\def \gaugeTransfSpace        {\mathcal{G}}           
\def \configurationSpace      {\mathscr{C}}           
\def \bigTantent              {\mathcal{T}}           
\def \vertSpace               {\mathcal{V}}           
\def \connectionSpace         {\mathscr{A}}           
\def \gaugeTransfSpace        {\mathscr{G}}           
\def \LagrangianDensity       {\mathcal{L}}
\def \sign                    {s}

\newcommand* \hf           {\higgsField}             
\newcommand* \gp           {\gaugePotential}         
\newcommand* \ef           {\electricField}          
\newcommand* \mf           {\magneticField}          
\newcommand* \cf           {\curvatureField}         
\newcommand* \nf           {\neutralField}           
\newcommand* \confSp       {\configurationSpace}     
\newcommand* \gaugeSp      {\gaugeTransfSpace}       
\newcommand* \vertSp       {\mathcal{V}}             
\newcommand* \gaugeGp      {\gaugeTransfSpace}       
\newcommand* \fieldSp      {\fieldConfSpace}         
\newcommand* \connSp       {\connectionSpace}        
\newcommand* \connForm     {\omega}                  
\newcommand* \bigTangent   {\mathcal{T}}             
\newcommand* \pair         {\Phi}                    
\newcommand* \stgp         {\spacetimeGaugePotential}
\newcommand* \hreg         {\tilde h}                
\newcommand* \htilde       {\hreg}                   
\newcommand* \chireg       {\tilde \chi}             
\newcommand* \ConTr        {\Omega_\moduliSp} 
\newcommand* {\Energy    } {E}
\newcommand* {\KinEn     } {T}
\newcommand* {\PotEn     } {V}
\newcommand* {\KinTr     } {K}
\newcommand* {\Lagrangian} {L}


\def\covariantDerivative{\nabla}

\newcommand* \cdv      {\covariantDerivative}
\newcommand* \vp       {\mathbf{p}}          
\newcommand* \conFc    {\Omega}              
\newcommand* \dist     {\mathrm{d}}          
\newcommand* \manifold {M}


\newcommand*{\moduliSp   }{\mathcal{M}}
\newcommand*{\slaplacian }{\Delta_\surface}
\newcommand*{\dvr        }[1]{\nabla\cdot\mathbf{#1}} 
\newcommand*{\Diff       }{\mathcal{D}}
\newcommand*{\vdel       }{\del}
\newcommand*{\dTaubes    }{\del_v\!\TaubesOp}
\newcommand*{\afn        }{a}
\newcommand*{\bfn        }{b}
\newcommand*{\ueps       }{u}
\newcommand*{\Xsp        }{\mathcal{X}}
\newcommand*{\hilbSp     }{\mathcal{H}}
\newcommand*{\fFunc      }{\mathcal{F}}
\newcommand*{\intManifold}{\int_\manifold}
\newcommand*{\Hn         }{\mathrm{H}}
\newcommand*{\wto        }{\rightharpoonup}
\newcommand*{\ElpOp      }{E}
\newcommand*{\Hsp        }{H}
\newcommand*{\TaubesOp   }{L}
\newcommand*{\Proj       }{\Pi}







\NewDocumentCommand\hprod{sm}{
  \IfBooleanTF#1
  {
    \left\langle {#2} \right\rangle
  }
  {
    \langle {#2} \rangle
  }
}

\newcommand*\aForm{B}

In section~\ref{sec:cs-intro} we proved the existence of a minimum 
constant $\kappa_{\min}$ such that regardless of core positions on the 
moduli space $\moduli^{k_+, k_-}(\surface)$ of vortices and
antivortices of the $U(1)$-gauged $O(3)$ Sigma model with a Chern-Simons deformation, there
exists a solution to the Bogomolny equations close to the
solution at $\kappa = 0$. This result extends similar claims for the
Ginzburg-Landau model obtained in~\cite{flood2018chern} and justifies
the possible existence of a localization formula similar to the one
obtained in section~\ref{subsec:loc-bps-solitons} for BPS 
solitons of the gauged $O(3)$ Sigma model, with the addition of a term 
dependent on the Chern-Simons
constant $\kappa$. In this section we apply the general framework of
section~\ref{sec:loc-formulas} to compute the extra $\kappa$ term in
the localization formula. As the calculations are similar for both 
models, we compute a localization
formula for each one and finalise discussing about the extension of
our formula to the coincidence set.

Previous work on the subject for the Ginzburg-Landau functional
includes models of Kim-Min and
Kim-Lee~\cite{kim199281,kim_vortex_1994}, where the authors considered
a related model with a different type of the Chern-Simons
interaction,~\cite{kim_first_2002}, where Kim and Lee analysed the dynamics of
the Ginzburg-Landau model with a neutral field  
on the plane and~\cite{collie_dynamics_2008} where Collie and Tong 
addressed motion on the moduli space of abelian vortices 
in the presence of a magnetic field and concluded that the extra
Chern-Simons term in the localization formula is the Ricci
form of the metric on the reduced moduli space. Later, 
in~\cite{alqahtani2015ricci} Alqahtani and Speight showed that the deformation 
term of Kim-Lee cannot extend to the coincidence set of modulli space, whereas 
the term from Collie-Tong can, and thus Kim-Lee and Collie-Tong deformations 
of the Abelian Higgs model are different.



\subsection{The Maxwell Higgs Chern Simons model}\label{subsec:mhcs-model}

We work on a Riemann surface $\surface$ that can be either compact or
the Euclidean plane. The setup is as in
section~\ref{sec:loc-formulas}. We assume the existence of a principal
bundle, 
\begin{equation}
\label{eq:p-bundle}
U(1) \to P \to \reals\times\surface
\end{equation}
and denote by $\rho$ the representation of $U(1)$ as isometries of the
complex plane, 
\begin{equation}
\label{eq:rep}
\rho: U(1) \to \Aut(\cpx).
\end{equation}

Let $F =
(\reals\times\surface) \times_{\rho} \cpx$ be the hermitian bundle
associated to $\rho$. We fix the metric in $\reals\times\surface$ as the
product
\begin{equation}
\label{eq:st-metric}
dt^2 - g,
\end{equation}
where $g$ denotes the Riemannian metric in $\surface$. Let $\tilde\Diff$ be
the connection induced by $\rho$,
\begin{equation}
\label{eq:diff-def}
\tilde\Diff: \Gamma F \to  \Gamma\brk((\reals \oplus T^{*}\surface)
\otimes TF). 
\end{equation}

Given $\hf \in \Gamma F$, in a local trivialisation this is the map
\begin{equation}
\label{eq:diff-def-loc}
\tilde\Diff \hf = d\hf - i \tilde \stgp\otimes \hf.
\end{equation}

As in section~\ref{sec:cs-intro}, we add a neutral scalar field,
\begin{equation}
\label{eq:nf}
\nf \in C^{\infty}(\reals \times\surface).
\end{equation}

It will be convenient to make the division of space and time explicit,
we denote by $\Diff_t\hf\in C^{\infty}(\reals\times 
\surface)$ the time component of $\tilde\Diff\hf$ and from now onwards
$d\hf: \reals \to \Gamma(T^{*}\surface\otimes TF)$ will be the spatial
component of $d\hf$ as a function of time. The spatial component of
$\tilde\Diff\hf$, denoted by $\Diff\hf$, is
\begin{align}
\label{eq:diff-hf-def}
\Diff\hf: \reals \to \Gamma\brk(T^{*}\surface\otimes TF), &&
\Diff\hf = d\hf - i A\otimes\hf,
\end{align}

The Maxwell-Higgs Lagrangian is,
\begin{multline*}
\label{eq:mh-lag}
\Lagrangian_{MH} = \half\brk(\norm{\Diff_t\hf}^2 + \norm{\ef}^2 +
\norm{\dot\nf}^2 - \norm{\Diff \hf}^2 - \norm{\mf}^2 - \norm{d\nf}^2)
\\ - \lproduct{1, U},
\end{multline*}
the norms and the product in the Lagrangian are in the
$\Lsp^2$ sense. The potential function $U$ is given by,
\begin{equation}
\label{eq:u-pot}
U = \frac{1}{8} \brk( -2\kappa\nf + 1 - \abs\hf^2)^2 + \half \abs{\nf\hf}^2.
\end{equation}

We add a Chern-Simons term to the Lagrangian,
\begin{equation}
\label{eq:lag-cs}
\Lagrangian_{CS} = \half\,\brk( \lproduct{A,*\ef} + \lproduct{\tilde A_0,
  *\mf} ),
\end{equation}
the Maxwell-Higgs-Chern-Simons Lagrangian is,

\begin{equation}
\label{eq:mhcs-lag}
\Lagrangian = \Lagrangian_{MH} + \kappa \Lagrangian_{CS}.
\end{equation}

As for the $O(3)$ Sigma model, the Chern-Simons term is not gauge invariant,
however if two gauge 
potentials differ by a gauge transformation, the corresponding
Lagrangians will differ by a total divergence, hence they will yield
the same field equations. If $\surface$ is compact, $A$ and $\tilde A_0$
are only defined locally, however, we can always choose an open and
dense subset of $\surface$, diffeomorphic to the unit disk by the
Riemann mapping theorem, in which the connection is trivializable. In
any case, the Lagrangian is well defined up to gauge
equivalence. Variating $\Lagrangian$ with respect to $\tilde A_0$,  
Gauss's law is,
\begin{equation}
\label{eq:gauss-law}
d^{*} \ef = -\lproduct{\Diff_t\hf, i\hf} + \kappa\, *\mf.
\end{equation}

In terms of the gauge potential this is the same as,
\begin{equation}
\label{eq:gauss-law-gp}
-\brk(\laplacian + \abs{\hf}^2) \tilde A_0 = \frac{1}{2i}
\brk(
\hf \dot\hf^{\dagger} - \hf^{\dagger}\dot\hf
) - d^{*} \dot A + \kappa *\mf.
\end{equation}

Let $\KinEn$ and $\PotEn$ be the kinetic and potential energy of the
fields,
\begin{equation}
  \label{eq:kin-pot-en}
  \begin{aligned}
\KinEn &= \half(
\norm{\Diff_t\hf}^2 + \norm{\ef}^2 + \norm{\dot\nf}^2
), \\
\PotEn &= \half(\norm{\Diff\hf}^2 + \norm{\mf}^2 + \norm{d\nf}^2)
+ \lproduct{1, U},    
  \end{aligned}
\end{equation}
the total conserved energy of the fields is
\begin{equation}
\label{eq:total-energy}
\Energy = \KinEn + \PotEn,
\end{equation}

If $\surface$ is the Euclidean plane, we assume the following
convergence at infinity, 
\begin{align}
\label{eq:conv-infinity-fields}
  \dot A, A \in \Hsp^1\brk(\Omega^1(\plane)), &&
  \tilde A_0, \nf \in \Hsp^1\brk(\plane), \\
  \dot\nf, \dot\hf \in\Lsp^2\brk(\plane), &&
  \lim_{\abs{x}\to\infty} \abs{\hf(x)}^2 = 1.
\end{align}

Let $\del_A\hf$ be the projection of $\Diff\hf$ in the sub-space $(1,
0)$ of the complexification of the bundle $T^{*}\surface\otimes TF$,
in local coordinates,
\begin{equation}
\label{eq:del-gp}
\conj\del_{A}\hf = \half\brk(\Diff_1\hf + i\,\Diff_2\hf ).
\end{equation}

We apply the Bogomolny trick to obtain a set of equations obeyed by the
fields,

\begin{equation}
\begin{aligned}[b]
  0 &\leq \half\left(\norm{\Diff_t\hf - i\nf\hf}^2
  + \norm{\dot\nf}^2 + \norm{-d\nf + \ef}^2\right.\\
    &\quad\left. + \norm*{*\mf -   \half \brk(-2\kappa\nf + 1
    - \abs\hf^2)}^2\right) + \norm{\conj\del_A\hf}^2\\
    &= \Energy - \lproduct{\Diff_t\hf, i\nf\hf} - \lproduct{d\nf, \ef} -
    \half \lproduct{*\mf, -2\kappa\nf + 1 - \abs\hf^2}\\
    &\quad + \norm{\conj\del_A\hf}^2 - \half\norm{\Diff\hf}^2 \\
&= \Energy - \lproduct*{\nf, -*d*\ef +\lproduct{\Diff_t\hf, i\hf} -
  \kappa*\mf } - \half\lproduct{*\mf, 1}\\
&\quad + \half\lproduct{*\mf, \abs{\hf}^2}
+ \norm{\conj\del_A\hf}^2 - \half\norm{\Diff\hf}^2\\
&= \Energy - n\pi.
\end{aligned}
\end{equation}

To obtain the last equation, we discarded several
divergences and used Gauss's law and the identity,
\begin{equation}
*\mf\,\abs\hf^2 = \abs{\del_A\hf}^2 - \abs{\conj\del_A\hf}^2.
\end{equation}

We also used that,
\begin{equation}
\int_{\surface} \mf = 2n\pi, \qquad n \in \mathbb{Z}.
\end{equation}

If $\surface$ is compact, this is due to the fact that $\mf$ is the
curvature of the line bundle $F$, in the case that $\surface$ is the
Euclidean plane, this comes from the assumptions of decaying of the
gauge potential at infinity and the nontrivial winding of $A$ at the
circle at infinity. Hence, a set of fields  $\pair = (\hf, \nf,
A) \in \fieldSp$, is a minimal with energy 
\begin{equation}
  \Energy = n\pi,
\end{equation}
provided it satisfies the Bogomolny equations,
\begin{align}
\dot\nf &= 0,\\
\label{eq:bog-dnf}
 \ef &= d\nf, \\
\label{eq:bog-d0-hf}
\Diff_t\hf &= i\nf\hf,\\
\label{eq:bog-conj-del-gp-hf}
\conj\del_A\hf &= 0, \\
\label{eq:bog-mf}
*\mf &= \half\brk(-2\kappa\nf + 1 - \abs\hf^2).
\end{align}

From equations \eqref{eq:bog-dnf}, \eqref{eq:bog-d0-hf}  and Gauss's
law, $\nf$ is a solution to the elliptic problem,
\begin{equation}
\label{eq:nf-elliptic-problems}
\brk(\laplacian + \abs\hf^2) \nf = \kappa*\mf. 
\end{equation}

If $\pair$ is a solution to the Bogomolny equations and Gauss's
law, by conservation of Energy and the Bogomolny equations,
\begin{equation}
\label{eq:lag-bog-spc}
\Lagrangian = 2\,\KinEn + \kappa\,\Lagrangian_{CS} - \Energy
= \frac{\kappa}{2} \brk(\tilde A_0 + \nf, *\mf) - n\pi.
\end{equation}

If we take the radiation gauge,  $\tilde A_0 = -\nf$, $\pair$ is an
extremal of the Lagrangian. The solution is stationary by the 
Bogomolny equations. Hence $\pair$ is a solution to the field
equations.



\subsection{Low energy dynamics with a Chern-Simons
  term}\label{sec:cs-localisation} 

To apply the low energy approximation, we will work in the
space $\fieldSp'$ of fields $\pair = (\nf,\hf, \tilde A)$ which are
solutions to the Bogomolny equations. We have two models of the moduli
space, the $O(3)$ Sigma model, for which localization was discussed
on section~\ref{sec:loc-formulas} and the Ginzburg-Landau model,
studied by Samols on the plane~\cite{samols1992}. The Bogomolny equations 
for both soliton types have a similar structure, allowing to compute
the contribution to the $\Lsp^2$ metric by the Chern-Simons term in
both models. In this section, we do so and compute the energy
contribution for the space of solutions to the Bogomolny equations
$\fieldSp'$, our computation will be valid in both cases. In the next
section we specialise into the MHCS moduli space and later extend our
result to the $O(3)$ Sigma model. Given $\pair \in  
\fieldSp'$, the formal tangent space $T_{\pair}\fieldSp'$ is the space
of solutions to the linearization of the Bogomolny equations at
$\pair$. We introduce the $\Lsp^2$ metric on $T_{\pair}\fieldSp'$
induced by the metrics in $\surface$ and target space. We have the
inclusion
\begin{align}
\label{eq:gsp-into-tsp}
\gaugeSp \hookrightarrow T_{\pair}\fieldSp', &&
\alpha\in \gaugeSp \mapsto (\nf, e^{i\alpha}\hf, A + d\alpha) \in
T_{\pair}\fieldSp',
\end{align}
defining the vertical bundle $\gaugeSp \to \vertSp \to
\fieldSp'$. Suppose $\pair_s:\reals \to \fieldSp'$ is a differentiable
curve in $\fieldSp'$, 
meaning that as a function in an open dense set $U\subset
\surface$,  $\reals \times U \to \reals\times F
\times \Omega^1(U)$ is
differentiable. Let us define $\beta = \tilde A_0 + \nf$, by Gauss's law
and~\eqref{eq:nf-elliptic-problems}, $\beta$ 
is a solution to
\begin{equation}
\label{eq:bog-get-eq}
\brk(\laplacian + \abs\hf^2) \beta = -\frac{1}{2i}\brk(
\hf\dot\hf^{\dagger} - \hf^{\dagger}\dot\hf
) + d^{*} \dot A.
\end{equation}

Equation \eqref{eq:bog-get-eq} means $\beta$ is the orthogonal
projection of $\dot\pair$ onto $\vertSp$. Recall in $\fieldSp'$ the energy is
conserved. By the Bogomolny equations energy is given by the expression,
\begin{equation}
\label{eq:energy-fieldsp}
\Energy = \half\brk(\norm{\nf\hf}^2 + \norm{d\nf}^2) + \PotEn = n\pi.
\end{equation}

As in section~\ref{sec:loc-formulas}, we assume variations of fields in
$\fieldSp'$ are good approximations to slowly moving
vortices.  We work perturbatively in the deformation
parameter. Assume $\kappa$ 
is small, by equation \eqref{eq:nf-elliptic-problems},
\begin{equation}
\label{eq:nf-kappa-pert}
\nf = \kappa\,\nf_{\kappa} + \order(\kappa^2).
\end{equation}

Discarding terms of order $\kappa^2$ and a divergence, the kinetic
energy of a field $\pair \in\fieldSp'$ is,
\begin{equation}
  \label{eq:kinen-field-sp}
  \begin{aligned}[b]
    \KinEn &= \half\brk( \norm{\dot\hf^{\perp} - i\nf\hf}^2 +
    \norm{\dot A^{\perp} - d\nf}^2)\\
    &= \half\brk( \norm{\dot\hf^{\perp}}^2 + \norm{\dot A^{\perp}}^2
    + \norm{\nf\hf}^2 + \norm{d\nf}^2)
  \end{aligned}
\end{equation}

For the Chern-Simons term we have,
\begin{equation}
\label{eq:eq-cs-approx-fieldsp}
\begin{aligned}[b]
  \Lagrangian_{CS} &= \half\brk(
  \lproduct{ A, *\ef} + \lproduct{\tilde A_0, *\mf}
  )\\
  &= \half\brk(
  \lproduct{A, *\dot A - *d\tilde A_0} + \lproduct{\tilde A_0, *\mf}
  )\\
  &= \half
  \lproduct{A,*\dot A} + \lproduct{\tilde A_0,*\mf}
\end{aligned}
\end{equation}

To first order in $\kappa$, the Lagrangian can be approximated as, 
\begin{equation}
  \label{eq:approx-lagrangian}
  \begin{aligned}[b]
    \Lagrangian' &= \KinEn - \PotEn + \kappa\Lagrangian_{CS}\\
    &= \half\brk(\norm{\dot\hf^{\perp}}^2 + \norm{\dot A^{\perp}}^2) +
    \norm{\nf\hf}^2 + \norm{d\nf}^2 - n\pi + \frac{\kappa}{2}
    \lproduct{ A, *\dot A} + \kappa\,\lproduct{\tilde A_0, *\mf} \\
    &= \half(\norm{\dot\hf^{\perp}}^2 + \norm{\dot A^{\perp}}^2)
    - n\pi
     +\kappa\,\lproduct{\beta, *\mf} + \frac{\kappa}{2}
    \lproduct{A, *\dot A},
  \end{aligned}
\end{equation}

where we discarded another divergence and used
\eqref{eq:nf-elliptic-problems}. Let us introduce the kinetic 
and connection terms, $\KinTr, \ConTr: T\fieldSp' \to \reals$, defined as,
\begin{align}
\label{eq:kin-term-field-sp}
\KinTr &= \half(\norm{\dot\hf^{\perp}}^2 + \norm{\dot A^{\perp}}^2),
&
\ConTr &= \kappa\brk(\lproduct{\beta, *\mf} + \half
    \lproduct{A, *\dot A}),
\end{align}
in geometric terms, $\KinTr$ plays the role of a metric on tangent space, 
on the other hand, $\ConTr$ is a  
 connection, deviating the motion of the fields from geodesic motion, 
 as will be evident in the next subsection when 
we obtain a formula for this term. 
Therefore, the effective Lagrangian at low energy is,
\begin{equation}
\label{eq:eff-lag}
\Lagrangian_{eff} = \KinTr + \ConTr.
\end{equation}

$\KinTr$ is gauge invariant and if $\kappa = 0$,  it is the kinetic energy
term in the Samols approximation to Ginzburg-Landau theory or the
kinetic energy term computed in section~\ref{sec:loc-formulas}. 
However if $\kappa \neq 0$, this
term does not render the same energy as the extra $\kappa$ term in the 
Bogomolny equations deforms the fields. Although $\ConTr$ is gauge
dependent because of the $\beta$ 
factor, $\Lagrangian_{eff}$ determines the dynamics in a gauge
invariant way, because any gauge transformation contributes a total
divergence.


\subsection{A formula for the connection term}
\label{sec:loc-cs}

In this section we focus on the Maxwell-Higgs-Chern-Simons model. Let
$\vset \subset \surface$ be the set of zeros of $\hf$. We assume the 
zeros are simple. If the energy of a solution
to the Bogomolny equations is $n\pi$, there are $n$
vortices on $\surface$. We work in a chart $\varphi:U \to \cpx$ defined on an 
open and dense subset $U \subset
\surface$ and assume that $\vset \subset U$. We denote by $z = \varphi(x)$ 
points on $\cpx$ and assume the metric takes the form
\begin{equation}
\label{eq:metric-isoth-coords}
g = e^{\Lambda(z)} \brk(dz_1^2 + dz_2^2),\qquad z \in \cpx.
\end{equation}

Since $U$ is contractible, the restriction  $F\vert_U$ is
trivial. Let $U' = U\setminus \vset$, we define 
the fields $h, \chi \in C^{\infty}(U')$, $\eta \in
C^{\infty}(U',\cpx)$, such that,  
\begin{align}
\hf &= e^{\frac{h}{2} + i\chi}, &&
\eta = \frac{\dot h}{2} + i\dot\chi.
\end{align}

As for the $O(3)$ Sigma model $\chi$ is only well defined modulo $2\pi$, 
however, $h$, $\eta$ and
$d\chi$ are well defined functions on $U'$. Since the zeros of 
$\hf$ are simple, for any $p\in\vset$ there is a coordinate
neighbourhood $U_p$ and a smooth function $\tilde\hf_p\in
C^{\infty}(U_p, \cpx\setminus\set{0})$, such that, 
\begin{align}
\hf(x) = (\varphi(x) - \varphi(p))\,\tilde\hf(x), \qquad x \in U_p.
\end{align}

Let $r_p  = \log\,\abs{\varphi(x) - \varphi(p)}$, 
$\theta_p = \Arg(\varphi(x) - \varphi(p))$,  $x \in U'_p
= U_p\setminus\set{p}$, we have local expansions, 
\begin{align}
\label{eq:loc-h-chi}
h(x) = \log r_p^2 + \tilde h_p(x), &&
\chi = \theta_p(x) + \tilde\chi_p(x),
\end{align}
where the regular parts are functions $\tilde h_p, \tilde\chi_p \in
C^{\infty}(U_p)$. Locally, by
\eqref{eq:bog-conj-del-gp-hf} the gauge potential can be expressed in
terms of $d\chi$ and $dh$,
\begin{equation}
\label{eq:dgp-dchi-dh}
A = d\chi - \half *dh.
\end{equation}

Hence in $U'$, $h$ and $\chi$ satisfy the equations,
\begin{align}
\label{eq:elliptic-chi-h}
\laplacian h = 2*\mf, &&
\laplacian \chi = d^{*} A.
\end{align}

By Gauss's law, on $U'$ we have the following relation between
$\beta$ and $\dot\chi$,
\begin{equation}
  \label{eq:dot-chi-to-beta}
  \brk(\laplacian + \abs\hf^2)\dot\chi =
  \brk(\laplacian + \abs\hf^2) \beta.
\end{equation}

Note that $\beta$ is a smooth function defined on $U$
whereas $\dot\chi$ has divergences at vortex positions. Let
$\disk_{\epsilon}$ denote a collection of small $\epsilon$ 
geodesic disks, each one centred at one vortex position. The
orientation in each geodesic disk given by
the outward unit normal. Let $U_{\epsilon} = U\setminus
\disk_{\epsilon}$ be the surface with the holes left by removing the
disks. The orientation of $\del U_{\epsilon}$ is given by the
outward unit normal and if $\surface$ is the euclidean plane, we
assume the fields $\beta$ and $\nf$ converge fast enough at infinity.
Using Green's second identity and discarding divergences in the following
integral, we find,
\begin{equation}
\label{eq:int-beta-b}
\begin{aligned}[b]
  \kappa \int_{\surface} \beta \mf &= \int_{\surface}
  \beta\,(\laplacian + \abs\hf^2)\nf \\
  &= \int_{\surface} (\laplacian + \abs\hf^2)\beta \cdot \nf\\
  &= \lim_{\epsilon\to 0} \int_{U_{\epsilon}} (\laplacian
    + \abs\hf^2)\,\dot\chi \cdot \nf\\ 
  &= \lim_{\epsilon\to 0} \kappa\int_{U_{\epsilon}}
  \dot\chi \mf + \lim_{\epsilon\to 0} \int_{\del\disk_{\epsilon}} 
  (-\dot\chi *d\nf + \nf*d\dot\chi).
\end{aligned}
\end{equation}

On the other hand,
\let \gpalt \gp
\let \gp A
\begin{equation}
  \begin{aligned}[b]
    \gp\wedge\dot\gp &= \gp\wedge d\dot\chi - \half\gp\wedge*d\dot h\\
  &= \gp\wedge d\dot\chi + \half \dot\gp \wedge *d h - \half \dv{t}
\brk( \gp \wedge *dh)\\
  &= \gp\wedge d\dot\chi + \half \brk(
   d\dot\chi \wedge *dh - \half *d\dot h \wedge * dh
)  - \half \dv{t} \brk( \gp \wedge *dh)\\
  &= 2 \gp \wedge d\dot\chi - d\chi \wedge d\dot\chi + \frac{1}{4} dh
\wedge d\dot h  - \half \dv{t} \brk( \gp \wedge *dh)\\
  &= 2 \gp \wedge d\dot\chi + d\brk(\dot\chi d\chi) - \frac{1}{4} d
\brk(\dot h d h)  - \half \dv{t} \brk( \gp \wedge *dh)\\
  &= 2 \dot\chi \mf - 2 d\brk(\dot\chi \gp) + d\brk(\dot\chi d\chi) -
\frac{1}{4} d \brk(\dot h d h)  - \half \dv{t} \brk( \gp \wedge *dh).
  \end{aligned}
\end{equation}

Discarding the time derivative, we find,
\begin{equation}
\label{eq:int-gp-dot-gp}
\int_{\surface} \gp\wedge \dot\gp = 2\lim_{\epsilon\to
  0}\int_{U_{\epsilon}}\dot\chi\mf + \lim_{\epsilon\to
  0}\int_{\del\disk_{\epsilon}} \brk( 2\dot\chi\gp -\dot\chi d\chi +
\frac{1}{4} \dot h dh  ).
\end{equation}

Thence,
\begin{equation}
  \label{eq:omega-ver-me}
  \begin{aligned}[b]
    \ConTr &= \kappa\int_{\surface}\beta\mf - \frac{\kappa}{2}
    \int_U \gp \wedge \dot\gp\\
    &=  \lim_{\epsilon\to 0}\int_{\del\disk_{\epsilon}} 
\brk(
  -\dot\chi *d\nf + \nf*d\dot\chi  - \kappa\dot\chi\gp +
  \frac{\kappa}{2} \dot\chi d\chi - \frac{\kappa}{8}\dot h dh
).
  \end{aligned} 
\end{equation}

If $\alpha \in\Omega^1(U_p)$, in local, polar coordinates at $U_p$,
\begin{align}
\label{eq:eta-polar-coords}
\alpha &= \alpha_rdr_p + \alpha_{\theta} r_pd\theta_p, &
\dot\theta_p &= - \frac{-\sin\theta_p\,\dot p_1 + \cos\theta_p\,\dot
p_2 }{r_p},
\end{align}
where for a time varying point, $p(s)$, 
in local coordinates $\varphi_*\dot p = \dot p_1\,\del_1 + \dot p_2\,\del_2$.  
We deduce,
\begin{equation}
\label{eq:lim-dtheta-eta}
\lim_{\epsilon\to 0} \int_{\del\disk_{\epsilon}(p)}\dot\theta_p\,\alpha =
-\pi \brk(\alpha_1(p)\,\dot p_1 + \alpha_2(p)\,\dot p_2).
\end{equation}

Let $\moduliSp' \subset \moduliSp$ be the open subset of non coalescent
vortices, which we can identify with the set,

\begin{equation}
\label{eq:non-coalesc}
\set{\left.
    (p_1,\ldots, p_n) \in \surface^n
  \;\right\rvert\; p_j \neq p_k \; if\; j \neq k}.
\end{equation}

We define the projector,
\begin{align}
\label{eq:def-projector}
\Proj_j : \moduliSp' \to \cpx, && \Proj_j(\vp) = p_j.
\end{align}

Thence,
\begin{equation}
\label{eq:projector-integral}
\lim_{\epsilon\to 0} \sum_{j} \int_{\del\disk_{\epsilon}(p_j)}
 \dot\theta_{p_j}\, \alpha = - \pi \sum_j\lproduct*{\Proj_j^{*}\alpha, \dot\vp},
\end{equation}
where $\lproduct{\cdot,\cdot}$ is the
pairing of the pullback of $\alpha$ with the tangent vector $\dot\vp
\in T_{\vp}\moduliSp'$.

In equation \eqref{eq:omega-ver-me}, all the regular parts of the forms will
converge to zero as $\epsilon \to 0$, thus, the only terms to consider
are the singular parts. We compute those singular parts in the
following equations, 
\begin{equation}
\label{eq:int-dotchi-forms}
\begin{aligned}[b]
  \lim_{\epsilon\to 0}\int_{\del\disk_{\epsilon}}\dot\chi (-*d\nf -
 \kappa\gp) &=
\lim_{\epsilon\to 0} \sum_{j}\int_{\del\disk_{\epsilon}(p_j)} \dot\theta_{p_j}
(-*d\nf - \kappa\,\gp)\\ 
 &= -\pi \sum_j \lproduct{\Pi^{*}_j(-*d\nf - \kappa\gp), \dot\vp}.
\end{aligned}
\end{equation}

We also have,
\begin{equation}
\label{eq:sddottheta}
\begin{aligned}[b]
  \lim_{\epsilon\to 0}\int_{\disk_{\epsilon}} \nf *d\dot\chi &=
 \lim_{\epsilon\to 0} \sum_{j} \int_{\del\disk_{\epsilon}(p_{(j)})} \nf
 *d\dot\theta_{p_{(j)}}\\
 &= \lim_{\epsilon\to 0}\sum_{j}
 \int_{\del\disk_{\epsilon}(p_{(j)})} \nf \brk(
   \frac{-\sin\theta_{p_{(j)}}\, \dot p_{(j)1} + \cos\theta_{p_{(j)}}\, \dot 
   p_{(j)2}}{\epsilon}
   ) d\theta_{p_{(j)}},
\end{aligned}
\end{equation}
where we denote $p_j$ as $p_{(j)}$ to avoid confusion with the role of both 
subindexes. Taylor's expansion of  
$\nf(x)$ in a neighbourhood $V$ of a point $p$ is,
\begin{equation}
\label{eq:nf-taylor}
\nf(x) = \nf(p) + \del_1(\nf\circ\varphi^{-1})(\varphi(p))\,r_{p}\cos\theta_{p} 
+ \del_2(\nf\circ\varphi^{-1})(\varphi(p))\,r_{p}\sin\theta_{p} 
+ \order(r_{p}^2),
\end{equation}
for $x \in V$. Let $\nf_\varphi = \nf\circ\varphi^{-1}: U \to \reals$, thence, 
\begin{multline}
\label{eq:int-nf-sddotchi-disk}
\lim_{\epsilon\to 0}\int_{\del\disk_{\epsilon}(p)} \nf
\brk(
   \frac{-\sin\theta_p\, \dot p_1 + \cos\theta_p\, \dot p_2}{\epsilon}
   ) d\theta_p =
   \pi \brk(-\del_2\nf_\varphi(\varphi(p))\, \dot p_1 
   + \del_1\nf_\varphi(\varphi(p))\, \dot p_2).
\end{multline}

We conclude,
\begin{equation}
\label{eq:int-nf-sddotchi-algebraic-form}
\lim_{\epsilon\to 0} \int_{\disk_{\epsilon}} \nf *d\dot\chi = \pi
\sum_j \lproduct*{\Proj_j^{*}(*d\nf), \dot\vp}.
\end{equation}

Similarly, for the term $\dot h dh$ at $p \in \vset$ we have,
\begin{align}
\label{eq:doth-dh-local}
\dot h &= - 2\brk(
 \frac{\cos(\theta_p) \dot p_1 + \sin(\theta_p) \dot p_2}{r_p}
) + \del_t\hreg_p, &
dh &= \frac{2 dr_p}{r_p} + d\hreg_p.
\end{align}

Discarding the integrals of regular parts,
\begin{equation}
\begin{aligned}[b]
  \lim_{\epsilon\to 0}\int_{\del\disk_{\epsilon}} \dot h dh &=
  -2 \lim_{\epsilon\to 0} \sum_{j}
\int_{\del\disk_{\epsilon}(p_{(j)})} 
\brk(
  \frac{\cos(\theta_{p_{(j)}}) \dot p_{(j)1} + 
      \sin(\theta_{p_{(j)}}) \dot
    p_{(j)2}}{\epsilon}
  ) d\hreg_p\\
  &= -2 \lim_{\epsilon\to 0} \sum_{j}\int_{\del\disk_{\epsilon}(p_{(j)})}
  \left(
  -\del_1(\hreg_{p_{(j)}}\circ\varphi^{-1})(\varphi(p_{(j)}))\cdot
  \sin^2\theta_{p_{(j)}}\cdot\dot p_{(j)2} 
  \right.\\
  &\qquad \left. + 
  \del_2(\hreg_{p_{(j)}}\circ\varphi^{-1})(\varphi(p_{(j)}))
  \cdot\cos^2\theta_{p_{(j)}}\cdot\dot p_{(j)1}
  \right)
  d\theta_{p_{(j)}}\\
  &= 2\pi\sum_j \lproduct*{\Proj_j^{*}(*d\hreg_{p_{(j)}}), \dot\vp}.
\end{aligned}
\end{equation}

For the $\dot\chi\,d\chi$ term we have,
\begin{align}
\label{eq:chi-local-sing-reg-decomp}
\dot\chi &= \dot\theta_p + \del_t\chireg_p,&
d\chi &= d\theta_p + d\chireg_p.
\end{align}

We deduce the remaining integral is,
\begin{equation}
\begin{aligned}[b]
  \lim_{\epsilon\to 0}\int_{\del\disk_{\epsilon}} \dot\chi d\chi &=
  \lim_{\epsilon\to 0}\sum_{j}\int_{\del\disk_{\epsilon}(p_j)} (
  \dot\chi d\theta_{p_j} + \dot\chi d\chireg_{p_j}
  )\\
 &= \lim_{\epsilon\to 0} \sum_{j}\int_{\del\disk_{\epsilon}(p_j)}
 \brk(
 \del_t\chireg_{p_j}\cdot d\theta_{p_j} + \dot\chi d\chireg_{p_j})\\
 &= \pi\sum_j \brk(
 2\,\del_t\chireg_{p_j}(p_j) 
 - \lproduct{\Proj_j^{*}\brk(d\chireg_{p_j}),\dot\vp}
 ).
\end{aligned}
\end{equation}

Collecting all the pieces, we find the following expression for the
connection on the moduli space,
\begin{equation}
\label{eq:loc-conn-form}
\ConTr = \pi\sum_j \brk(
\Proj_j^{*}\brk(
2*d\nf +\kappa\gp - \frac{\kappa}{2}\,d\chireg_{p_j}
- \frac{\kappa}{4} *d\hreg_{p_j}
)
+ \kappa \del_t\chireg_{p_j}(p_j)).
\end{equation}

Because of this term, motion of a set of vortices on the moduli space deviates 
from geodesic motion according to a force
given by the two form,

\begin{align}
\nonumber
  d\ConTr &= \pi\sum_j \brk(
\Proj_j^{*}\brk(
2 d*d\nf +\kappa \mf - \frac{\kappa}{4} d*d\hreg_{p_j}
)
+ \kappa\, d\del_t\chireg_{p_j}(p_j))\\
\nonumber
&= \pi\sum_j \brk(
\Proj_j^{*}\brk(
2\,\abs\hf^2 *\nf - 2\kappa\mf +\kappa \mf - \frac{\kappa}{4} (-2\mf)
)
+ \kappa\, d\del_t\chireg_{p_j}(p_j))\\
\label{eq:cs-force}
&= \kappa\pi \sum_j\brk( -\half \Proj_j^{*} \mf +
d\del_t\chireg_{p_j}(p_j)). 
\end{align}

To simplify this equation, we used \eqref{eq:nf-elliptic-problems} and
\eqref{eq:dgp-dchi-dh}. if $z_j = \varphi(p_j)$, $\chireg_{z_j} = 
\chireg_{p_j}\circ\varphi^{-1}$ are local expression on the chart,
discarding higher order terms in $\kappa$, from \eqref{eq:bog-mf} we find,
\begin{equation}
\label{eq:simp-proj-s-kappa-mf}
\kappa\,\Proj_j^*\mf = \Proj_j^{*} \brk(\frac{\kappa}{2}(1 -
\abs\hf^2)\,e^{\Lambda(z)} dz^1\wedge dz^2)
= \frac{\kappa}{2}e^{\Lambda(z_j)}\,dz_j^1\wedge dz_j^2.
\end{equation}

Hence,
\begin{equation}
\label{eq:domega-mf-simp}
d\ConTr = \kappa\pi  \sum_{j}\brk(
- \frac{1}{4}e^{\Lambda(z)}dz_j^1\wedge dz_j^2 + d\del_t\chireg_{z_j}(z_j)
).
\end{equation}

To obtain an explicit formula for the remaining terms
$d\del_t\chireg_{z_j}(z_j)$, we will work to lowest order in $\kappa$. From
\eqref{eq:bog-mf} and \eqref{eq:elliptic-chi-h},  $h$ and $\nf$ are
solutions to the following system of equations in the sense of
distribution,  
\begin{align}
-\laplacian h &= e^h - 1 + 2\kappa\nf + 4\pi\sum_{j}
\delta_{p_j},\label{eq:elliptic-prob-h}\\
-\laplacian \nf &= e^h\nf - \frac{\kappa}{2} (1 - 
e^h),\label{eq:elliptic-prob-nf}
\end{align}

Let $(h_0,\nf_0)$ be the solution to this system at $\kappa =
0$, we know  $\nf_0 = 0$ by the Julia-Zee theorem
\cite{spruck-proof-2009} and  $h_0$ is the solution to the Taubes 
equation for the Ginzburg-Landau functional \cite{taubes1980},
\begin{equation}
\label{eq:cs-taubes}
-\laplacian h_0 = e^{h_0} - 1 + 4\pi\sum_{j}\delta_{p_j},
\end{equation}
if $(\del_{\kappa}h, \del_{\kappa}\nf)$ is the next order solution at
$\kappa = 0$, then in $U'$,
\begin{align}
  -\laplacian \del_{\kappa}h &= e^{h_0}\del_{\kappa}h,\\
\label{eq:elliptic-problem-order-kappa-nf}
  -\laplacian \del_{\kappa}\nf &= e^{h_0}\del_{\kappa}\nf -
  \half + \frac{\kappa}{2} e^{h_0}\del_{\kappa}h.
\end{align}

We make the assumption $(\del_{\kappa}h, \del_{\kappa}\nf)
\in \Lsp^2\times \Lsp^2$, in this case elliptic regularity implies
these are smooth functions, hence $\del_{\kappa}h \equiv 0$ and to
lowest order,
\begin{equation}
\label{eq:elliptic-problem2-order-kappa-nf}
\brk(\laplacian + e^{h_0})\del_{\kappa}\nf = \half.
\end{equation}

Therefore $\del_{\kappa}\nf \neq 0$, moreover,
\begin{align}
h &= h_0 + \order(\kappa^2), &
\nf &= \kappa\,\del_{\kappa}\nf + \order(\kappa^2).
\end{align}

For the remaining of the argument, we will assume $h$ is the solution to
 the Taubes equation without further notice. We can get and explicit
formula for the nontrivial term in \eqref{eq:domega-mf-simp}
introducing complex coordinates on $\moduli'$. If
we consider the singularity of $h$ at $p_j$, we have the following
equations in the sense of
distributions,  
\begin{align}
\label{eq:bog-h-lin}
-\brk(\laplacian + e^h) \del_{z_j} h &= 4\pi\,\del\delta_{p_j}, &
-\brk(\laplacian + e^h) \conj{\del_{z_j}} h = 4\pi\,\conj\del\delta_{p_j}.
\end{align}

Let
\begin{equation}
\label{eq:bog-eta-def}
\eta = \half \dot h + i \dot\chi,
\end{equation}
from \eqref{eq:dot-chi-to-beta} and \eqref{eq:loc-h-chi}, we deduce, 
\begin{equation}
\label{eq:bog-eta-eq}
\eta = \sum_{j} \dot p_j\,\del_{z_j}h + i\beta.
\end{equation}

We can expand $h$ in a neighbourhood of each $p_j \in \vset$ as in the
$O(3)$ Sigma model,
\begin{equation}
\label{eq:bog-loc-h}
h = \log r_j^2 + a_j + \half\brk( \conj{b}_j\, (z - z_j) + b_j (\conj z -
\conj z_j) ) + \order(r_j^2).
\end{equation}

Hence,
\begin{equation}
  \label{eq:bog-dotp-hp-local}
  \begin{aligned}[b]
    \del_t\chireg_{p_i}(p_i) &= 
    \Im \brk(\dot q\cdot \del_q\hreg_{p_i}(p_i)) + \beta
    \\
    &= \Im \brk(
    \brk(
      \sum_{j} \dot q_j \del_{z_j}a_i
      ) - \half\, \conj{b_i}\,\dot z_i
      ) + \beta.
  \end{aligned}
\end{equation}

So far, we have not used our gauge freedom in $\fieldSp'$, we can do
so now and discard the $\beta$ term to ease the final expression of
the computation. Let us define the complex form $\omega_c$, 
\begin{equation}
\label{eq:bog-omegac}
\omega_c = \sum_{i,j} \brk(
\del_{z_j}a_i - \half \conj{b}_i\,\delta_{ij}) dz^j
= \vdel a - \half \conj b,
\end{equation}
where,
\begin{align}
\label{eq:vdel}
\vdel &= \sum_{i} dz^i\otimes\del_{z_i}, &
a &= \sum_{i} a_i, &
\conj b &= \sum_{i} \conj{b}_i\,dz^i,
\end{align}

The imaginary part of $\omega_c$ is the nontrivial term in $\ConTr$, 

\begin{equation}
\label{eq:d-omegac}
d\Im\brk( \omega_c ) = \Im\brk( d\omega_c ) = \Im\brk(
\conj\vdel \vdel a - \half \brk(\vdel \conj b +
\conj\vdel \conj b)).
\end{equation}

The coefficients $b_i$ have the symmetries,
\begin{align}
\del_{z_j}\conj{b}_i &= \del_{z_i}\conj{b}_j, &
\conj{\del}_{z_j}\conj{b}_i &= \del_{z_i} b_j,
\end{align}
proved in~\cite{manton_topological_2004} by Manton-Sutcliffe for the
Euclidean plane. The proof can be adapted to 
compact manifolds and is essentially the same as the proof of 
lemma~\ref{lem:coef-sym}. Whence,
\begin{align}
\label{eq:del-b-sym}
\vdel\conj b &= 0, &
\conj\vdel \conj b &= -\vdel b.
\end{align}

Hence $\vdel b \in \Lambda^2(U,i\reals)$. Since $a$ is real, it is
also valid that $\conj\vdel\vdel a \in \Lambda^2(U, i\reals)$. Hence, the
curvature induced by the Chern-Simons term is written locally as,
\begin{equation}
\label{eq:bog-curv-form}
d\ConTr = - \kappa\pi i\, \brk(
\frac{1}{8}
\sum_{i}e^{\Lambda(z_i)} dz^i\wedge {d\conj z^i} + \conj\vdel \vdel a
+\half \vdel b.
)
\end{equation}

To lowest order, the metric is the $\Lsp^2$ metric, 
\begin{equation}
\label{eq:samols-approx}
ds^2 = \pi \sum_{i,j} \brk(
 e^{\Lambda(z_i)}\delta_{ij}  + 2 \del_{z_i} b_j) dz^i\,{d\conj z^j}.
\end{equation}
whose symplectic form is \cite[p.~212]{manton_topological_2004},
\begin{equation}
\label{eq:samos-symp-form}
\omega_0 = \frac{i\pi}{2} \brk(
\sum_{i} e^{\Lambda(z_i)}\delta_{ij}\,dz^i \wedge {d\conj z^j} + 2 \vdel b
).
\end{equation}

Therefore, the Chern-Simons curvature in $\moduliSp'$ is related to
the symplectic form of the $\Lsp^2$ metric by,
\begin{equation}\label{eq:bog-curv-form-d}
  d\ConTr = 
  -\frac{\kappa}{2} \omega_0 - \kappa\pi i \brk(
  -\frac{1}{8}  \sum_{i}
  e^{\Lambda(z_i)}\delta_{ij}\,dz^i\wedge {d\conj z^j} + \conj\vdel \vdel a
  ).
\end{equation}

\let \gp \gpalt
\let \gpalt \undefined

\subsubsection{Comparing with the Collie-Tong connection}

With our choice of notation, 
Collie and Tong proposed for vortices of the 
Abelian Higgs model~\cite{collie_dynamics_2008} that $d\ConTr = \kappa \rho$, 
where $\rho$ is the Ricci form of the metric in the moduli space, then 
properties of the dynamics of vortices with a Chern-Simons interaction term 
were studied by Krusch-Speight~\cite{krusch2010exact} and 
Alqahtani-Speight~\cite{alqahtani2015ricci} in all cases assuming the 
dynamics is modified by the Ricci form, however, little is 
known in the literature about how good the Ricci form approximation is. 
 We do not compare the dynamics of the Collie-Tong proposal with the 
connection term found due to lack of time, the problem remains open for 
future work, instead, 
if we consider~\eqref{eq:bog-curv-form-d}, we can see the connection 
term is not the Ricci form in the case of the moduli space 
$\moduli^{2}(\plane)$ of the MHCS 
model. For a pair $(z_1, z_2)$ of non-coalescent abelian vortices in $\plane$, 
define centre of mass coordinates $(Z, W)$, such that $z_1 = Z + W$, 
$z_2 = Z - W$, $W = \epsilon\,e^{i\theta}$, then the metric in the open and 
dense subset of non-coalescent vortices is,
\begin{equation}
g_{\Lsp^2} = 2\pi\,dZ\,d\overline Z + f(\epsilon)\,(d\epsilon^2 + 
\epsilon^2\,d\theta^2), 
\end{equation}
where the conformal factor is,
\begin{equation}
f(\epsilon) = 2\pi\,\pbrk{1 + \frac{1}{\epsilon}\,\frac{d}{d\epsilon}(\epsilon 
\tilde b(\epsilon))},
\end{equation}
and the coefficient $\tilde b(\epsilon)$ is defined as $\tilde b(\epsilon) = 
b_1(\epsilon, 
-\epsilon)$. In centre of mass coordinates, the Ricci form is,
\begin{equation}
\begin{aligned}[b]
\rho &= i\,\partial\bar\partial \log \sqrt{|g_{\Lsp^2}|}\\
&= i\,\partial\overline{\partial}\,\log f(\epsilon)\\
&= -K(\epsilon)\,f(\epsilon)\,\epsilon d\epsilon\wedge d\theta,
\end{aligned}
\end{equation}
where, $K(\epsilon)$ is the Gaussian curvature of the subspace of vortex pairs 
with $Z = 0$,
\begin{equation}
K = -\frac{1}{2\epsilon f(\epsilon)}\,\frac{d}{d\epsilon}\pbrk{\epsilon 
\frac{d}{d\epsilon}\log f(\epsilon)}.
\end{equation}

On the other hand, by~\eqref{eq:bog-curv-form-d},
\begin{multline}\label{eq:dConTr-pair}
d\ConTr = -\frac{\kappa}{2}\,i\pi\,dZ \wedge d\overline{Z}  
- \frac{\kappa}{2} f(\epsilon)\,\epsilon d\epsilon \wedge d\theta\\ 
-\kappa\,\pi\,i\pbrk{-\frac{1}{4}\,dZ \wedge d\overline{Z} + 
\frac{i}{2}\,r\,dr\wedge d\theta + \bar\partial \partial a}.
\end{multline}

On the plane, $(h, N)$, 
the solution to  
equations~\eqref{eq:elliptic-prob-h}-\eqref{eq:elliptic-prob-nf}, is 
invariant under isometries. For small $\kappa$ and small 
perturbations of 
$(h_0, 0)$, where $h_0$ is the solution to the Taubes equation for the Abelian 
Higgs model, the result of Flood-Speight~\cite{flood2018chern} shows 
the existence of exactly one solution $(h, N)$ to the field equations 
on a compact surface. It is sensible to assume the same statement holds on the 
plane, 
this implies $a$ is invariant under isometries of $\plane$, hence,
\begin{align}
\partial_{z_1}a + \partial_{z_2}a = 0, &&
\bar\partial_{z_1}a + \bar\partial_{z_2}a = 0.
\end{align}

From these equations, we deduce,
\begin{align}
\partial a = \partial_{z_1}a\,dz^1 + \partial_{z_2}a\,dz^2 
= -\partial_{z_1}a\,(dz^2 - dz^1).
\end{align}

Likewise,
\begin{equation}
\begin{aligned}[b]
\bar\partial\partial a &= 
-(\bar\partial_{z_1}\partial_{z_1}a\,d\bar z^1 
+ \bar\partial_{z_2}\partial_{z_1}a\,d\bar z^2)
\wedge (dz^2 - dz^1)\\
&= \bar\partial_{z_1}\partial_{z_1}a
\,(d\bar z^2 -  d\bar z^1)\wedge (dz^2 - dz^1)\\
&= 2\,i\,\bar\partial_{z_1}\partial_{z_1}a\cdot \epsilon\,d\epsilon\wedge 
d\theta.
\end{aligned}
\end{equation}

Let $\tilde a(\epsilon) = a(\epsilon, -\epsilon)$, isometric invariance of 
$a$ implies,
\begin{equation}
\bar\partial_{z_1}\partial_{z_1}a = \frac{1}{4\epsilon}\,\frac{d}{d\epsilon} 
\pbrk{\epsilon\,\frac{d \tilde a}{d\epsilon}}.
\end{equation}

Going back to equation~\eqref{eq:dConTr-pair}, we find,
\begin{multline}\label{eq:dConfTr-2v-CM}
d\ConTr = -\frac{\kappa\,\pi\,i}{4}\,dZ \wedge d\bar Z \\
-\frac{\kappa \pi}{2}\,\pbrk{1 + 
\frac{2}{\epsilon}\frac{d}{d\epsilon}(\epsilon\,\tilde b(\epsilon))
- \frac{1}{\epsilon}\,\frac{d}{d\epsilon}\pbrk{\epsilon\,\frac{d\tilde 
a}{d\epsilon}}}\,\epsilon\,d\epsilon\wedge d\theta.
\end{multline}

Equation~\eqref{eq:dConfTr-2v-CM} shows $d\ConTr \neq \kappa\rho$, since 
$\rho$ has no $dZ\wedge  d\bar Z$ component.

\subsection{Chern-Simons localization on the O(3) Sigma 
model}\label{subsec:cs-loc-o3-sigma}

\let \gpalt \gp
\let \gp A

We can adapt our previous arguments to the $O(3)$ Sigma model with some
minor adjustments. Most of our previous deduction follows without
change, since the Bogomolny equations have the same structure, except
for the algebraic formula of $*\mf$. Gauss's law in this case has to be
replaced as in section~\ref{sec:cs-intro}. In this case, the
projection $\beta$ of a solution $\pair \in \fieldSp'$ to the static
Bogomolny equations onto vertical space is a solution to
equation~\eqref{eq:perp-condition}, which for variations of the core
positions is,
\begin{align}
  (\laplacian + \abs{X_{\hf}}^2)\,\beta = \lproduct{\dot\hf, X_{\hf}}
  + d^{*}\dot\gp,
\end{align}

In this case at each core $p_j \in \vset \cup \avset$, we have to take
into consideration the sign function $\sign_j$, otherwise our
computation on sections~\ref{sec:cs-localisation},~\ref{sec:loc-cs}
follow the same pattern and we find that for non-coalescent vortices, 
to first order in $\kappa$,
\begin{align}
    d\ConTr &= \kappa\pi \sum_j\,\sign_j\brk( -\half \Proj_j^{*} \mf +
    d\del_t\chireg_{p_j}(p_j))\nonumber\\
    &= \kappa\pi\,\sum_j\,\sign_j\,\pbrk{\half
      \Proj_j^{*}(*(\kappa\,\nf + \tau - \lproduct{n, \hf})) +
      d\del_t\chireg_{p_j}(p_j)
    }\nonumber\\
    &= \kappa\pi\,\sum_j\,\pbrk{-\frac{i}{4}
      (1 - \sign_j\tau)\,e^{\Lambda(z_j)}\,dz^j\wedge d\conj z^j + 
      \sign_j\,d\del_t\chireg_{z_j}(z_j)
    }.
\end{align}

To deduce a formula for the second term in the sum, we know by
section~\ref{sec:cs-intro} that for $\kappa = 0$, the only solution to
the governing elliptic problem is $(h_0, 0)$ where $h_0$ is the
solution to the Taubes equation for the $O(3)$ Sigma model. As is shown in
equation~\eqref{eq:eta-formula}, $\eta$ can be computed from the
derivatives of $h$, in accordance to~\eqref{eq:bog-eta-eq}. If we
recall equation~\eqref{eq:loc-h-coeffs-p}, we find for any small
holomorphic neighbourhood $U_j$ of $p_j \in \vset\cup\avset$, 
$z_j = \varphi(p_j)$,
\begin{align}
  \tilde h_{p_j}(\varphi(x)) = \sign_j\,a_j + \half\,\sign_j\,\pbrk{ \conj b_j
 \, (z - z_j) + b_j\,(\conj{z} - \conj z_j)} 
 + \order(r_j^2). 
\end{align}

Comparing with~\eqref{eq:bog-loc-h}, we deduce,
\begin{align}
  \sign_{j}\del_t\tilde\chi_{p_j}(p_j) = \Im\pbrk{\sum_{i}
    \dot z_i \del_{z_i}a_{j}  - \half\, \conj b_{j}\,\dot z_j}.
\end{align}

By lemma~\ref{lem:coef-sym}, $b$ has the same symmetries than for
Ginzburg-Landau solitons, therefore, 
\begin{align}
    d\ConTr &= -\kappa\pi\,i\,\sum_j\,\pbrk{\frac{1}{4}
      (1 - \sign_j\tau)\,e^{\Lambda(z_j)}\,dz^j\wedge {d\conj z^j} + 
      \conj\del\del\,a + \half\,\del b
    }.
\end{align}

\let \gp \gpalt
\let \gpalt \undefined






\subsection{Extending to the coalescence points}
\label{sec:ellipt-reg}

In the previous sections we developed two formulae for the extra term
in a localization formula of vortices in a model with a Chern-Simons
term. We assumed $\kappa$ small and found two related formulae on
$\moduli'$, the open and dense subspace of the moduli space of
non-coalescent vortices.
To extend $d\ConTr$ to the coalescence
points means to study the limit of the formula as any pair of vortices,
(of the same type for deformations of the $O(3)$ Sigma model) coalesce, meaning 
as $d(p_j, p_k)
\to 0$ for $p_j\neq p_k$ and $p_j$, $p_k$ vortices of the same
type. Let us consider first the $O(3)$ Sigma model. Let $\htilde \in
C^{\infty}(\surface)$ be the regular part of $h$,
theorem~\ref{thm:h-reg-param-dep} shows $\tilde h$ depends smoothly on
vortex positions as long as vortices and antivortices do not
coalesce. Recall in section~\ref{c:intro-gov-elliptic} we defined
smooth functions $v_j: \surface\setminus\set{p_j} \to \reals$, such that,
\begin{align}
  h = \sum_{j} \sign_j\,v_j + \tilde h.
\end{align}

In the compact case, we assume   $\vset\cup\avset \subset U$, $U$ an open and 
dense subset of the surface in which a holomorphic chart $\varphi: U \to \cpx$ 
is defined. Let $D\subset \cpx$ be a bounded domain containing $\varphi(\vset 
\cup \avset)$. Since each 
$v_j$ is a constant multiple of Green's function, there is a smooth function
$\tilde v : D \times D \to \reals$ such that for
any $z \in D$,
\begin{align}
  v_j(\varphi^{-1}(z)) = \log\,\abs{z - z_j}^2 + \tilde v(z, z_j).
\end{align}

Thus,
\begin{align}
  \sign_ia_i = \sum_{j \neq i}\sign_j\log\,\abs{z_j - z_i}^2 +
  \sum_j \sign_j\tilde v(z_j, z_i) + \tilde h(\varphi(z_i)),
\end{align}
and,
\begin{align}
  \sign_ib_i &= 2\,\conj\del_z\pbrk{
    \sum_{j \neq i}\sign_j\log\,\abs{z - z_j}^2 + \sum_j\sign_j\,\tilde
    v(z, z_j) + \tilde h(\varphi^{-1}(z))
  }(z_i)\nonumber\\
  &=\sum_{j \neq i}\frac{2\sign_j}{\conj z_i - \conj z_j} +
  2\sum_j\sign_j\,\conj\del_z\tilde v(z_j, z_i) + 2\,\conj\del_z\tilde h(z_i).
\end{align}
 
Hence, for non-coalescent cores at $D$,
\begin{equation}\label{eq:drv-reg}
  \begin{gathered}
\conj\vdel \vdel\, a =
\sum_{i,j}\sign_i\sign_j\,\conj\vdel\vdel\,\tilde v(z_j, z_i) +
\sum_i\sign_i\,\conj\vdel\vdel\,\tilde h(z_i),\\
\vdel\,b = 2\vdel\pbrk{\sum_{i,j}\pbrk{\sign_i\sign_j\conj\del_z\tilde
v(z_j,z_i)
 + \sign_i\delta_{ij}\conj\del_z\tilde h(z_i)}{d\conj z^i}}.
  \end{gathered}
\end{equation}

If $z \in D$ is such that for a fixed pair  of indices $j, k$,
the vortices at $z_j$, $z_k$ are of the same type and both 
converge to $z$, equation~\eqref{eq:drv-reg} implies the limit $
\lim_{|z_j-z_k|\to 0}d\ConTr$ exists and is unique, in fact, it
corresponds to solving the regularised Taubes equation with
configuration $\vb p$ such that $p_j = p_k = \varphi^{-1}(z)$.

For Ginzburg-Landau vortices the same argument is valid, it is
simpler, since in this case all the vortices are of the same type and
we can take $\sign_p = 1$ for all the cores. In this case we arrive to
an algebraic expression similar to \eqref{eq:drv-reg} and apply
proposition~\ref{prop:gl-reg-taubes-param-dep} to conclude $d\ConTr$
can also be
extended to the coalescence points.










\let \higgsField              \undefined              
\let \gaugePotential          \undefined              
\let \electricField           \undefined              
\let \magneticField           \undefined              
\let \spacetimeGaugePotential \undefined              
\let \curvatureField          \undefined              
\let \neutralField            \undefined              
\let \fieldConfSpace          \undefined              
\let \gaugeTransfSpace        \undefined              
\let \configurationSpace      \undefined              
\let \bigTantent              \undefined              
\let \vertSpace               \undefined              
\let \connectionSpace         \undefined              
\let \gaugeTransfSpace        \undefined              
\let \LagrangianDensity       \undefined              
\let \hf                      \undefined              
\let \gp                      \undefined              
\let \ef                      \undefined              
\let \mf                      \undefined              
\let \cf                      \undefined              
\let \nf                      \undefined              
\let \confSp                  \undefined              
\let \gaugeSp                 \undefined              
\let \vertSp                  \undefined              
\let \gaugeGp                 \undefined              
\let \fieldSp                 \undefined              
\let \connSp                  \undefined              
\let \connForm                \undefined              
\let \bigTangent              \undefined              
\let \pair                    \undefined              
\let \stgp                    \undefined              
\let \hreg                    \undefined              
\let \htilde                  \undefined              
\let \chireg                  \undefined              
\let \ConTr                   \undefined              
\let \Energy                  \undefined              
\let \KinEn                   \undefined              
\let \PotEn                   \undefined              
\let \KinTr                   \undefined              
\let \Lagrangian              \undefined              
\let \covariantDerivative     \undefined              
\let \cdv                     \undefined              
\let \vp                      \undefined              
\let \conFc                   \undefined              
\let \dist                    \undefined              
\let \manifold                \undefined              
\let \moduliSp                \undefined              
\let \slaplacian              \undefined              
\let \dvr                     \undefined              
\let \Diff                    \undefined              
\let \vdel                    \undefined              
\let \dTaubes                 \undefined              
\let \afn                     \undefined              
\let \bfn                     \undefined              
\let \ueps                    \undefined              
\let \Xsp                     \undefined              
\let \hilbSp                  \undefined              
\let \fFunc                   \undefined              
\let \intManifold             \undefined              
\let \Hn                      \undefined              
\let \wto                     \undefined              
\let \ElpOp                   \undefined              
\let \Hsp                     \undefined              
\let \TaubesOp                \undefined              
\let \Proj                    \undefined              
\let \hprod                   \undefined 
\let \sign                    \undefined


\chapter{Conclusion}

In this work we focused on geometric models of vortices and
antivortices of the $O(3)$ Sigma model. We emphasised the geometric nature of
the interaction of a vortex-antivortex pair on the moduli space.

We were able to prove that the $\Lsp^2$ metric in the moduli space is 
incomplete 
both on the euclidean plane and on a compact surface. We also analysed the 
dynamical
properties of the interaction on the plane, focusing on
scattering of vortex-antivortex pairs.

We also computed the volume of the moduli space on spheres and flat tori,
corroborating the work of Speight and R\~omao who conjectured a
formula for the volume of the moduli space for a general surface.

The fact that the moduli space is incomplete imposed some technical
difficulties on the proofs, that we overcame by analysing the
behaviour of solutions to the Taubes equations in the collision of a
vortex and an antivortex.

Finally, we added a Chern-Simons interaction term to our model and
applied the geodesic approximation ideas to determine the extra term in the
metric of the moduli space for small perturbations due to the
interaction. Our analysis indicates that the extra term 
can be extended to the coincidence set.

Some questions remain opened, representing an opportunity for future
work. The short range 
approximation formula for the metric on the space of vortex-antivortex
pairs of the euclidean plane relies on uniform convergence of the
family $\tilde h_{\epsilon}/\epsilon$, as $\epsilon \to 0$, where
$\tilde h_{\epsilon}$ is the regular part of the Taubes equation. 
Numerical evidence
suggests this conjecture is true. Should it be the case, we would be
able to prove formally that the Gaussian curvature of
$\moduli^{1,1}_0(\plane)$ diverges as $\epsilon \to 0$ 
as expected from the numerical evidence and we could
also justify analytically the effective potential of Ricci magnetic
geodesics. The equivalent conjecture for a compact surface would allow
to compute the volume formula for a general surface, where we no
longer have the extra symmetries that we used for the task.

In conclusion, geometric ideas to study field theory originated in the realm of
superconductivity with the Ginzburg-Landau functional at critical
coupling, but they have proved to be fruitful in a broader context. In
particular, for asymmetric vortex-antivortex systems of the $O(3)$ Sigma model, 
where with these ideas one can understand dynamics from a geometric point of
view.


\bibliographystyle{plainnat} 
\renewcommand{\bibname}{References} 
\bibliography{refs} 
\addcontentsline{toc}{chapter}{References} 

\end{document}